\documentclass{SciPost}

\hypersetup{
    colorlinks,
    linkcolor={red!50!black},
    citecolor={blue!50!black},
    urlcolor={blue!80!black}
}

\usepackage[bitstream-charter]{mathdesign}
\usepackage{import}
\usepackage[utf8]{inputenc}
\usepackage{cite}
\usepackage{physics}

\usepackage[dvipsnames]{xcolor}
\usepackage[most]{tcolorbox}

\usepackage[framemethod]{mdframed}

\usepackage[toc]{appendix}

\newcounter{defi}
\newcounter{prop}
\newcounter{lem}
\newcounter{thm}
\newcounter{cor}
\newcounter{ex}
\newenvironment{defi}[1][]{\begin{mdframed}[backgroundcolor=pink!20,topline = false,rightline=false,bottomline=false,linewidth=2.0pt,linecolor=red]
		\refstepcounter{defi}
		\ifstrempty{#1}
		{
			\textbf{Definition \thedefi.}
		}
		{
			\textbf{Definition \thedefi. (#1)}
		}
	}
	{\end{mdframed}}
\newenvironment{prop}[1][]{\begin{mdframed}[backgroundcolor=CornflowerBlue!20,topline = false,rightline=false,bottomline=false,linewidth=2.0pt,linecolor=NavyBlue]
		\refstepcounter{prop}
		\ifstrempty{#1}
		{
			\textbf{Proposition \theprop.}
		}
		{
			\textbf{Proposition \theprop. (#1)}
		}
	}
	{\end{mdframed}}
\newenvironment{lem}[1][]{\begin{mdframed}[backgroundcolor=CornflowerBlue!20,topline = false,rightline=false,bottomline=false,linewidth=2.0pt,linecolor=NavyBlue]
		\refstepcounter{lem}
		\ifstrempty{#1}
		{
			\textbf{Lemma \thelem.}
		}
		{
			\textbf{Lemma \thelem. (#1)}
		}
	}
	{\end{mdframed}}

\usepackage{chngcntr}
\counterwithin{equation}{section}
\counterwithin{defi}{section}
\counterwithin{thm}{section}
\counterwithin{lem}{section}
\counterwithin{prop}{section}
\counterwithin{ex}{section}
\counterwithin{cor}{section}
\counterwithin{table}{section}

\usepackage{amsthm}

\usepackage{amssymb} 
\usepackage{slashed} 
\usepackage{tensor} 
\usepackage{dsfont} 
\usepackage{tikz-cd} 
\usepackage[version=4]{mhchem} 
\usepackage{youngtab} 

\usepackage{empheq} 
\usepackage{enumerate} 
\usepackage{multicol} 
\usepackage{subfiles} 

\usepackage[super]{nth}
\usepackage{comment} 
\usepackage{stackrel}
\usepackage{siunitx} 

\usepackage{caption}
\usepackage{subcaption}
\usepackage{empheq} 
\usepackage{multirow}
\usepackage{hhline}

\usepackage{simplewick}
\usepackage{simpler-wick}
\usepackage[compat=1.1.0]{tikz-feynman}
\usepackage{feynmp}

\usepackage{tikz}
\usetikzlibrary{shapes,snakes}
\usetikzlibrary{decorations.markings}
\usepackage{pgfplots} 
\usetikzlibrary{patterns}

\usetikzlibrary{decorations.markings}
\tikzset{->-/.style={decoration={
			markings,
			mark=at position #1 with {\arrow{latex}}},postaction={decorate}}}
	
\tikzset{cross/.style={cross out, draw=black, fill=none, minimum size=2*(#1-\pgflinewidth), inner sep=0pt, outer sep=0pt}, cross/.default={2pt}}

\newcommand{\DrawLine}{%
  \begin{tikzpicture}
  \path[use as bounding box] (0,0) -- (\linewidth,0);
  \draw[color=blue!75!black,dashed,dash phase=2pt]
        (0-\kvtcb@leftlower-\kvtcb@boxsep,0)--
        (\linewidth+\kvtcb@rightlower+\kvtcb@boxsep,0);
  \end{tikzpicture}%
  }

\usetikzlibrary{quantikz}

\DeclareSymbolFont{usualmathcal}{OMS}{cmsy}{m}{n}
\DeclareSymbolFontAlphabet{\mathcal}{usualmathcal}

\fancypagestyle{SPstyle}{
\fancyhf{}
\lhead{\colorbox{scipostblue}{\bf \color{white} ~SciPost Physics Lecture Notes }}
\rhead{{\bf \color{scipostdeepblue} ~Submission }}

\fancyfoot[C]{\textbf{\thepage}}
}

\begin{document}

\pagestyle{SPstyle}

\begin{center}{\Large \textbf{\color{scipostdeepblue}{
Universal Quantum Computation with the $S_3$ Quantum Double:\\ A Pedagogical Exposition\\
}}}\end{center}

\begin{center}\textbf{
Chiu Fan Bowen Lo\textsuperscript{1$\star$},
Anasuya Lyons\textsuperscript{1$\dagger$},
Ruben Verresen\textsuperscript{2}, Ashvin Vishwanath\textsuperscript{1} and Nathanan Tantivasadakarn\textsuperscript{3}
}\end{center}

\begin{center}
{\bf 1} Department of Physics, Harvard University, Cambridge, Massachusetts 02138, USA
\\
{\bf 2} Pritzker School of Molecular Engineering, University of Chicago, Chicago, IL 60637, USA
\\
{\bf 3} Walter Burke Institute for Theoretical Physics and Department of Physics,
California Institute of Technology, Pasadena, CA, 91125, USA
\\[\baselineskip]
$\star$ \href{mailto:email1}{\small chiufanbowenlo@g.harvard.edu}\,,\quad
$\dagger$ \href{mailto:email2}{\small anasuya\_lyons@g.harvard.edu}
\end{center}

\section*{\color{scipostdeepblue}{Abstract}}

Non-Abelian topological order (TO) enables topologically protected quantum computation with its anyonic quasiparticles. Recently, TO with $S_3$ gauge symmetry was identified as a sweet spot---simple enough to emerge from finite-depth adaptive circuits yet powerful enough to support a universal topological gate-set. In these notes, we review how anyon braiding and measurement in $S_3$ TO are primitives for topological quantum computation and we explicitly demonstrate universality. These topological operations are made concrete in the $S_3$ quantum double lattice model, aided by the introduction of a generalized ribbon operator. This provides a roadmap for near-term quantum platforms.

\vspace{\baselineskip}







\vspace{10pt}
\noindent\rule{\textwidth}{1pt}
\setcounter{tocdepth}{2}
\tableofcontents
\noindent\rule{\textwidth}{1pt}
\vspace{10pt}


\section{Introduction} \label{sec:intro}

Realizing a topological quantum computing platform \cite{kitaev_fault-tolerant_2003,freedman_topological_2002,nayak_non-abelian_2008} that is robust against local errors is a holy grail in the field of quantum information. Any quantum computer requires at least two components: (1) a way to encode and store information (memory) and (2) a way to manipulate that information through gates (computation). The toric code \cite{kitaev_fault-tolerant_2003} and the related surface code \cite{bravyi_quantum_1998,freedman_projective_2001,fowler_surface_2012} are paradigmatic examples of topological quantum memories \cite{dennis_topological_2002}; information can be encoded in the degenerate ground state subspace in a topologically protected way. The power of the toric/surface code is due to its pointlike quasiparticle excitations having exchange or braiding statistics which is beyond those of bosons or fermions---such particles are called anyons \cite{leinaas_theory_1977,wilczek_quantum_1982,Halperin84,wu_general_1984,Arovas84}. In particular, in the toric code, braiding one particle around another can give rise to a $-1$ phase in the wavefunction, something that cannot occur for fundamental bosons or fermions. This generalized statistics implies ground state degeneracy \cite{wen_ground-state_1990,einarsson_fractional_1990} which can be used as a memory, upon which gates are implemented by braiding anyons \cite{kitaev_fault-tolerant_2003}. However, this simple $-1$ phase of the toric code is not powerful enough to achieve universal quantum computation using these excitations.
There are proposals which introduce external ``modules'' that power up toric code for universal quantum computation, including various combinations of magic state injection \cite{bravyi_universal_2005}, lattice surgery \cite{horsman_surface_2012,chamberland_universal_2022,litinski_lattice_2018}, or insertion of extrinsic defects \cite{bombin_topological_2010,barkeshli_twist_2013,barkeshli2015physical,benhemou_non-abelian_2022,yoder_surface_2017, kesselring_boundaries_2018, burton_genons_2024}.

However, are there other quantum states whose quasiparticle excitations are powerful enough to function as a universal quantum computer without the additional modules? This would constitute a genuine topological quantum computer, in the sense that the fault tolerance is fully provided by the topological protection, without having to devise alternate mechanisms to ensure the fault tolerance of external modules.

To look for more examples beyond toric code, we first note that the topological protection it enjoys is due to the fact that it is in a phase with topological order (TO); indeed, TOs can be identified by the presence of anyonic low-energy quasiparticle excitations \cite{wen1989,wen_topological_1990,wen_book_2004,sachdev_book_2023,simon_topological_2023}. One way of constructing more general examples of TOs is by considering the deconfined phase of an emergent lattice gauge theory \cite{Wegner71,Kogut79,Fradkin79}. For example, the toric code can be viewed as a $\mathbb Z_2$ lattice gauge theory at the deconfined fixed point \cite{kitaev_fault-tolerant_2003}. From this perspective, the $-1$ braiding phase is simply an Aharonov-Bohm phase between an emergent electric charge and a magnetic gauge flux \cite{AB}. Kitaev pointed out that one can consider lattice gauge theories for any finite group $G$, and that all such theories have anyons---realized by emergent gauge charges, gauge fluxes, and dyons with both charge and flux \cite{kitaev_fault-tolerant_2003}. Explicit lattice realizations of such gauge theories are realized by the so-called ``quantum double models'' \cite{kitaev_fault-tolerant_2003,propitius_discrete_1996,cui_universal_2015}. The simplest cases are where the group $G$ is Abelian. In such cases, the resulting anyons are Abelian, characterized by the nontrivial phase they incur when they exchange or braid with each other: they can pick up a $U(1)$ phase of $e^{i\theta}$, where $\theta \in [0, 2\pi)$ \footnote{Under an exchange, bosons correspond to $\theta = 0$ and fermions $\theta=\pi$.}\cite{leinaas_theory_1977,wilczek_quantum_1982,wu_general_1984}. However, Abelian anyons are still not powerful enough to implement the full universal gate set by just braiding, fusing and measuring the anyons \cite{mochon_anyons_2003,mochon_anyon_2004}.

This is where \emph{non}-Abelian topological order comes in, which arises, e.g., for non-Abelian gauge groups $G$ \cite{Goldin81,nayak_non-abelian_2008,simon_topological_2023,frohlich_statistics_1988,Moore1988,Fredenhagen1988, bais_flux_1980,wen_topological_1990,Moore1991,alexander_bais_quantum_1992,Gabbiani1992,alexander_bais_anyons_1993}. Non-Abelian anyons have an even more general form of exchange statistics--- they have internal Hilbert spaces (akin to a spin degree of freedom) which undergo unitary transformations when they braid with each other. Information can be encoded in these internal Hilbert spaces, and braiding provides a way to perform logical operations on the encoded information. In these notes, we will focus on the quantum double models for finite $G$ (in particular, the quantum double of $S_3$\footnote{Since $S_3$ is non-Abelian, its corresponding double model realizes a non-Abelian topological order.}), using the internal state of non-abelian anyons as the logical subspace.




\begin{figure}[h!]
\centering
\includegraphics[width=0.8\textwidth]{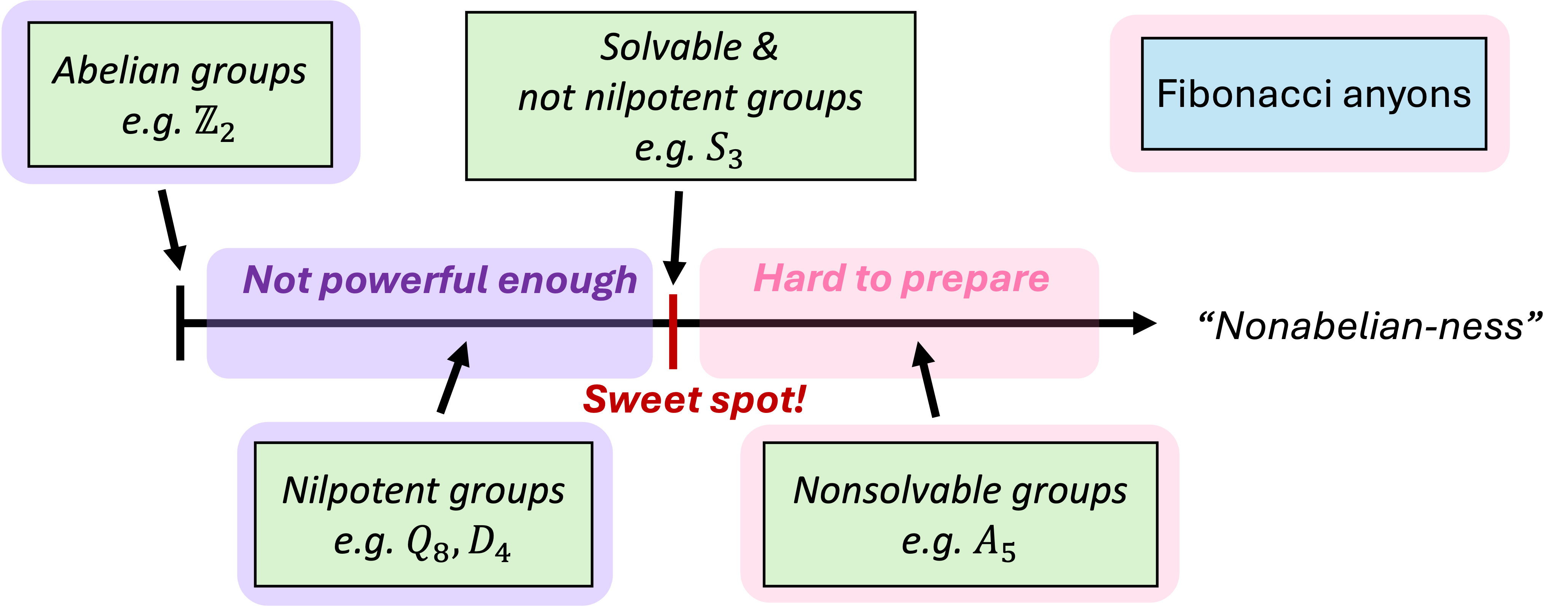}
\caption{An illustration for the ``non-Abelian-ness'' of groups, based on whether the group is Abelian, nilpotent, and solvable. Topological order with Abelian and nilpotent gauge groups are not powerful enough to host universal quantum computation; Topological order with nonsolvable gauge groups (and beyond quantum double, like Fibonacci anyons) are difficult to prepare experimentally. Topological order with solvable but non-nilpotent groups, like $S_3$, lie at the ``sweet spot''--- $S_3$ topological order hosts universal computation \textit{and} is relatively easy to prepare on a near-term quantum computer.}
\label{fig:nonabelian-spectrum}
\end{figure}

\subsection{Why S3?}\label{subsec:why-S3}

Not all non-Abelian topological orders are born equal---there are varying levels of computational power and experimental difficulty. Some orders do not support universal computation, and some are harder to realize. We want to find the sweet spot: a non-Abelian topological order that is universal for computation yet simple enough to be prepared in near-term quantum platforms. The ability of a non-Abelian quantum double model to host universal computation is related to the degree of ``non-Abelian-ness'' of the gauge group; more non-Abelian groups are more powerful \cite{mochon_anyon_2004}. What do we mean by ``non-Abelian-ness''? A rough diagram illustrating the spectrum of ``non-Abelian-ness'' can be found in Fig. \ref{fig:nonabelian-spectrum}. The least non-Abelian groups are the Abelian groups. The next level in the hierarchy of ``non-Abelian-ness'' are groups that are \textit{nilpotent}\footnote{A group $G$ is nilpotent if there exists an upper central series: $\{1\}\cong G_0 \trianglelefteq G_1\trianglelefteq \dots \trianglelefteq G_n\cong G$ such that $G_{i-1}$ is the center of $G_i$.  Iteratively quotienting out the center of the group takes us down this series from the full group $G$ to the trivial group $\{1\}$.}: in such theories, iteratively fusing an anyon with its antiparticle eventually results in an Abelian anyon. These Abelian anyons can then be fused together using a finite-depth circuit. We can implement some gates by braiding anyons arising from nilpotent gauge groups, but they will not be very powerful---the best we can do is a Pauli $X$ on the internal state of the anyon (or a generalization thereof for a higher-dimensional qudit)\cite{mochon_anyon_2004}\footnote{Here, we are considering a quantum double model where the underlying finite group is nilpotent. In the literature, general topological orders can be referred to as nilpotent if their anyons also always eventually fuse to Abelian anyons, even if these topological orders lie outside the quantum double paradigm. An example is the Ising topological order, which is based on the Ising fusion category. Such topological orders can be more powerful (Ising TO can realize all Clifford gates), although we are not aware of an example powerful enough for universal computation.}. 


If we go up even further in the ``non-Abelian-ness'' hierarchy,  considering groups that are not nilpotent (e.g. $S_3$ for solvable groups, $A_5$ for nonsolvable groups), then it is possible to realize a universal gate set by braiding and fusion of anyons \cite{mochon_anyons_2003, mochon_anyon_2004,ogburn_preskill_topological_1999}. We could even reach outside of the class of quantum double models and consider more exotic non-Abelian orders like Fibonacci topological order, for which braiding alone is sufficient for universal quantum computation \cite{freedman_modular_2002,freedman_two-eigenvalue_2002,trebst_short_2008}. 

Why not focus on non-solvable topological order, such as $A_5$ quantum double or Fibonacci then? We care about computational power, but we also care about the ease of state preparation. It turns out that ``non-Abelian-ness'' is also linked to how easy or hard it is to prepare a given topological order from a product state on a quantum simulator. Even with adaptive circuits (unitary circuits with mid-circuit measurements and feedforward), currently there are no finite-depth protocols to prepare nonsolvable topological order\footnote{The current best protocol for preparing non-solvable topological orders requires a circuit depth that scales with the logarithm of the system size \cite{lu_measurement_2022}.} (e.g. $A_5$ and Fibonacci), and it seems unlikely they exist \cite{tantivasadakarn_hierarchy_2023}. Among finite non-Abelian groups, only groups that are solvable, i.e. can be constructed from group extensions of Abelian groups, can be prepared via finite-depth adaptive circuits \cite{tantivasadakarn_hierarchy_2023}.

$S_3$, being non-nilpotent but solvable, lies right at the sweet spot: it is powerful enough to enable universal quantum computation \cite{mochon_anyon_2004}, while simple enough to be prepared using a finite-depth adaptive circuit \cite{verresen_efficiently_2022,tanti_long-range_2024,bravyi_adaptive_2022}, making the realization of $S_3$ topological order feasible on a near-term quantum computing platform. Additionally, since $S_3$ is the smallest non-Abelian group (it has six elements), its lattice model has the benefit of having a relatively small local Hilbert space, so it is more readily accessible by current small-scale quantum platforms. This combination of advantages motivates us to examine computation with the $S_3$ quantum double model in detail.

\subsection{Where do these notes stand in the literature?}\label{subsec:where}

Much of the content to follow is a review of known results. However,  we have not seen these results collated in a single set of notes. Given the wide interest in topological order across many fields (quantum information and computation, condensed matter theory, high-energy theory, mathematical physics, etc.), the relevant literature is spread out among many different communities with different priorities and interests. Our goal is to give a comprehensive explication of universal computing with $S_3$ anyons, at a conceptual and practical level. We aim to bridge the effective-field-theory-level understanding of anyon theories with explicit, lattice-level implementations to concretely demonstrate the universal gate set. In addition, this work also contains some new results, in particular an alternative demonstration of the universality of the topological gate set, as well as introducing a generalized ribbon operator that will be useful for lattice implementation of the universal gate set operations. 

The possibility of using $S_3$ quantum double to do universal quantum computation was established in a pioneering paper by Mochon \cite{mochon_anyon_2004}, where he demonstrated that $S_3$ topological order can host a universal qutrit gate set. The construction of this qutrit gate set is rather involved---we instead focus on a universal \emph{qubit} gate set constructed using the same logical encoding as that in Ref.~\cite{mochon_anyon_2004}. Intuitively, the use of the third qutrit basis state as a ``hidden state'' allows for more efficient implementations of qubit gates than is possible when trying to implement a full \emph{qutrit} gate set. This qubit gate set was first proposed by Kitaev in a problem set \cite{kitaev_anyons_nodate}, but as far as we are aware, it has not been explicated in great detail anywhere in the literature. We show that Kitaev's proposed set of anyonic braiding and fusion operations are universal by constructing explicit circuit implementations for the standard universal gate set (Cliffords generated by Hadamard, S, CZ, with CCZ acting as the non-Clifford gate). New developments in state preparation using measurement as resources \cite{raussendorf2001, Raussendorf_2005,ashkenazi_duality_2022,tanti_long-range_2024,verresen_efficiently_2022,lu_measurement_2022} lead to renewed interests in constructing quantum circuits for preparing states with nontrivial entanglement \cite{tantivasadakarn_hierarchy_2023,tanti_long-range_2024, verresen_efficiently_2022,bravyi_adaptive_2022,piroli_approximating_2024, piroli2021, zhu2023, lee2022decodingmeasurementpreparedquantumphases, li2023measuringtopologicalfieldtheories, gunn2023phasesmatrixproductstates, li2023} and manipulating the corresponding quasiparticle excitations \cite{bravyi_adaptive_2022,liu_StringNet_2022,Ren24}.
We present a new result along this thread, introducing a generalization of the ribbon operator formalism that allows for greater flexibility in state initialization. 

\subsection{Overview}\label{subsec:overview}

\begin{figure}[h!]
\centering
\includegraphics[width=0.7\textwidth]{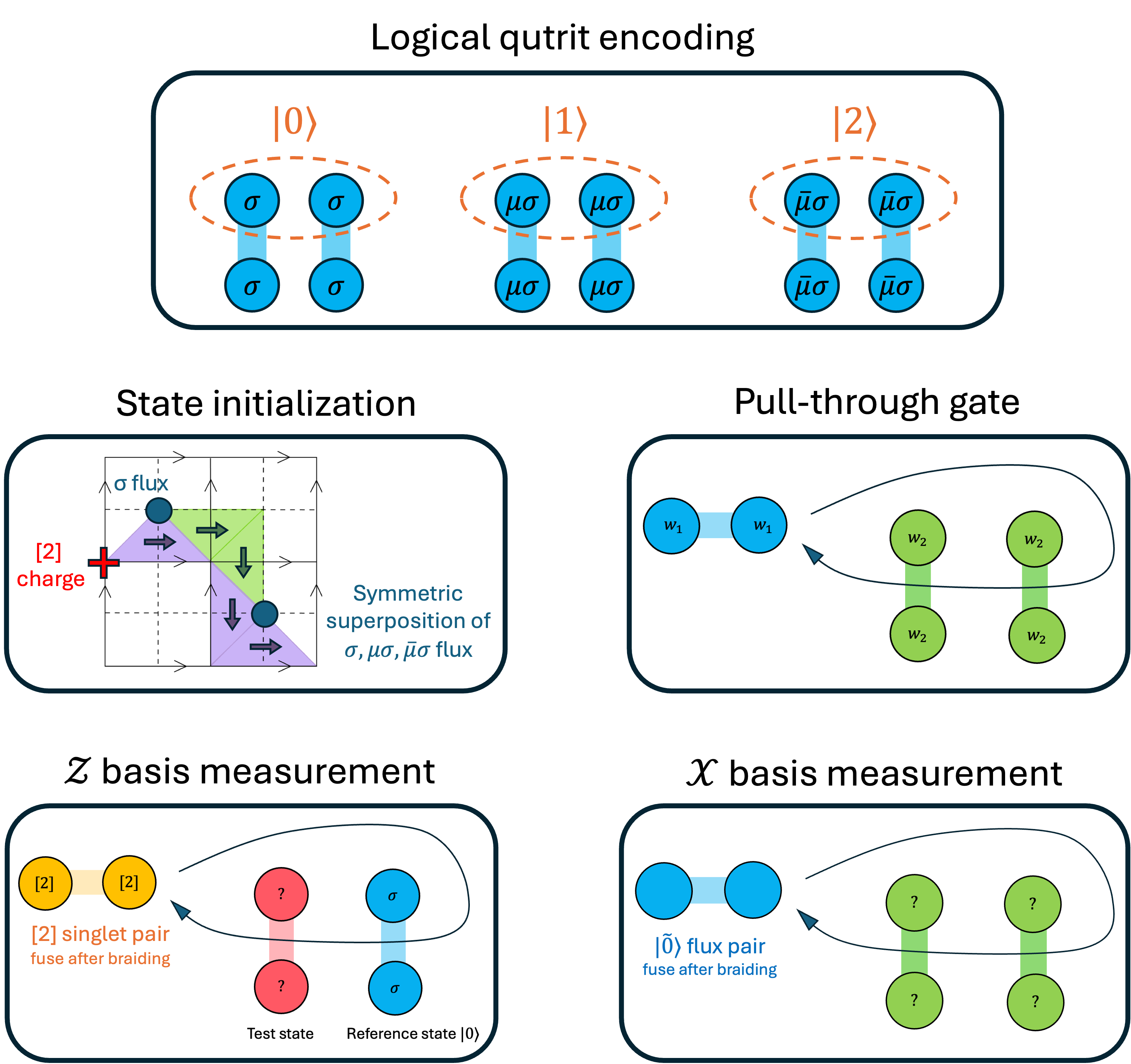}
\caption{A road map for the key ingredients of building the universal gate set of $S_3$ quantum double. The internal states of non-Abelian fluxes are used to encode logical states in the qutrit computational basis (Sec. \ref{state_initialization}). Winding and fusion of non-Abelian charges give rise to logical qutrit $\mathcal{Z}$ measurement and $\mathcal{X}$ basis measurement respectively (Sec. \ref{measurements}). Braiding of non-Abelian fluxes give rise to the pull-through gate, which is an fundamental operation used to implement other gates (Sec. \ref{flux-braiding}). The ribbon operator formalism (Sec. \ref{ribbons}) provides the lattice-level prescriptions for creating and manipulating the non-Abelian anyons.}
\label{fig:roadmap}
\end{figure}

To give a brief overview of the contents of these notes: we review the quantum double model and $S_3$ topological order in Section \ref{quantum-double} and discuss its ground state subspace in Section \ref{ground-state-subspace}. In Section \ref{excited-states}, we define the anyonic excitations present in quantum doubles and give some physical derivations for these definitions. Section \ref{universal_computing} will describe in detail the construction of the universal qubit gate set using $S_3$ anyons, where the key ingredients are shown in Fig. \ref{fig:roadmap}. As will be emphasized in Section \ref{state_initialization}, the internal Hilbert space of non-Abelian anyons will act as our logical space for computation. This is in contrast to the toric code, where the ground state subspace is used as the logical space for quantum memory. After establishing the anyon fusion (Section \ref{measurements}) and braiding operations (Section \ref{flux-braiding}), we then detail the full qubit universal gate set in Section \ref{universal-qubit-set}. Next, in Section \ref{ribbons} we discuss how to actually implement this gate set at the lattice level, using the so-called ribbon operators. Ribbons are the non-Abelian equivalent of string operators and are used to create anyons on top of the ground state---we introduce a novel generalization of these operators necessary to implement the universal gate set. Finally, we conclude with the current state of affairs for non-Abelian topological quantum computing, including experimental, error-correction, and fault-tolerance prospects. 

\section{The S3 Quantum Double Model}
\label{quantum-double}
\begin{tcolorbox}[enhanced, breakable,pad at break*=1mm]
   We start with an introduction to generic quantum double models for finite group $G$, defining the microscopic Hilbert space, the ground states, and low-lying excited states. These low-lying excited states contain anyons, whose properties are related to the gauge group $G$; we also give physical intuitions for these properties. As an example, we enumerate the anyon types in $S_3$ quantum double and their fusion rules. We highlight the possibility of single-particle excitations in non-Abelian quantum doubles on a torus.  Finally, we discuss the so-called ``neutrality conditions'' which constrain the physically-realizable states present in the model.
\end{tcolorbox}

A quantum double model\footnote{The material in this section is drawn from Refs.~\cite{kitaev_fault-tolerant_2003, preskill_lecture_2004, kitaev_anyons_nodate, cui_topological_2018,simon_topological_2023}.} $\mathcal{D}(G)$ is defined on an oriented graph\footnote{For $G \neq \mathbb{Z}_2$, $g \in G$ is not necessarily its own inverse and group multiplication does not have to be Abelian; when specifying the action of our operators, there is a difference then between left multiplication by $g$ and right multiplication by $\overline{g}$; the orientation of an edge determines the action.}: each edge hosts a Hilbert space $\mathbb{C}[G]$ of dimension $|G|$ with basis states $\ket{g}$ for each $g \in G$. In other words, there is a regular representation of $G$ living on each edge, whereby a group element $h \in G$ acts on a state $\ket{g} \mapsto \ket{hg}$. As mentioned in Sec. \ref{subsec:why-S3} in the introduction, we will view quantum doubles from the perspective of lattice gauge theory; these group-valued link variables can be interpreted as a $G$-valued gauge field. 

Associated with sets of edges are vertices $s$ and plaquettes $p$ (see Fig. \ref{fig:double-model1}a). We also define a ``site'' which consists of a pair $(s, p)$ of a neighboring vertex and plaquette. In these notes we will work with quantum double on the square lattice for simplicity, although the discussion applies to any oriented planar graph. 

\begin{figure}[h!]
\centering
\includegraphics[width=0.8\textwidth]{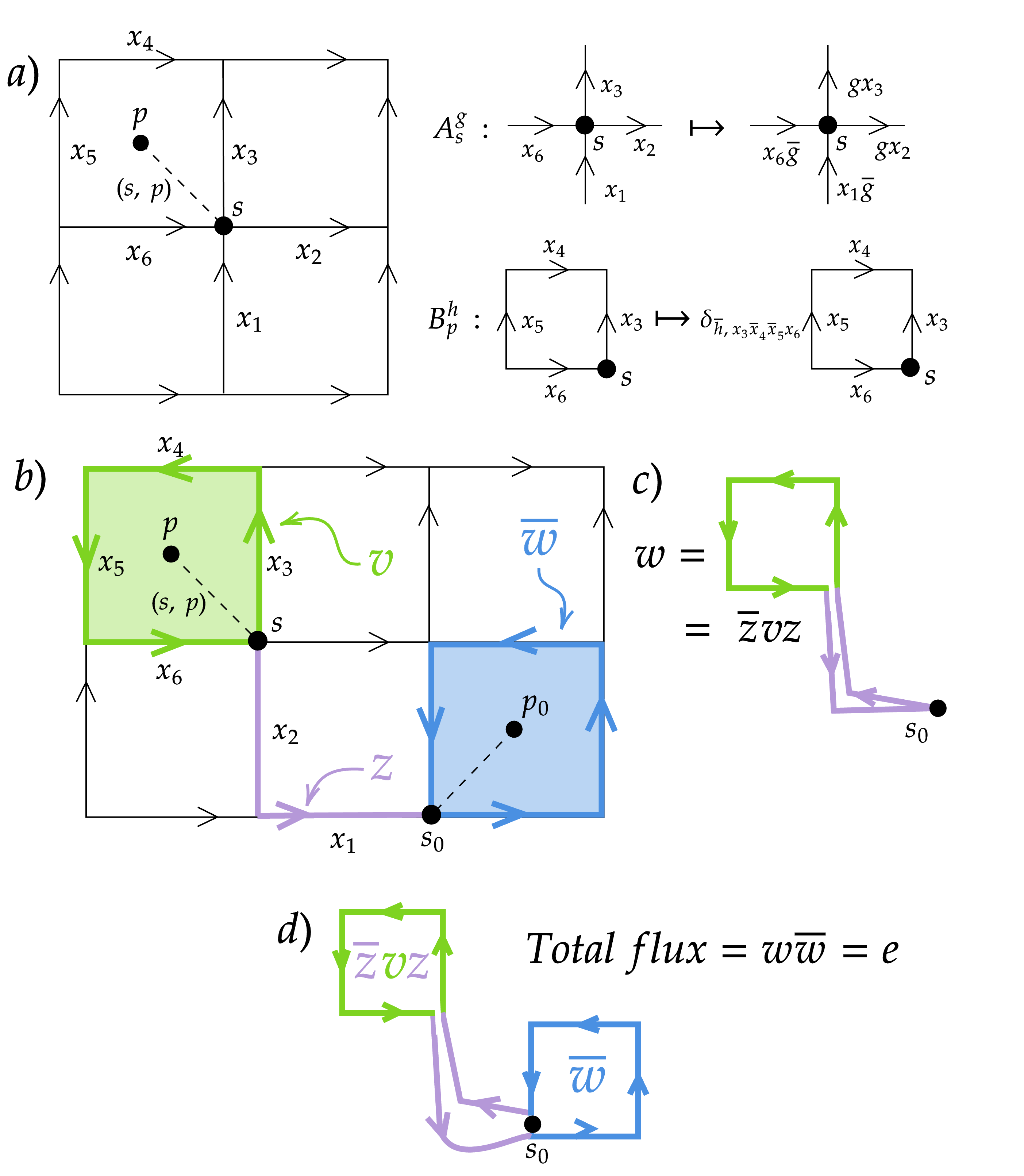}
\caption{a) Definitions of the vertex and plaquette operators. Each neighboring plaquette $p$ and vertex $s$ form a site, labeled by $(s, p)$. b) An elementary excitation pair is defined by two plaquette and/or vertex violations. We pick one excitation to live at the origin, and define the pair by properties of the excitation living away from the origin. The local flux $v$ is the product of edges counterclockwise around the plaquette $p$ (for the case in the figure, $v=x_3\bar{x}_4\bar{x}_5x_6$), starting from the vertex $s$, and $z$ is the product of group elements (inverse of a group element if the orientation is opposite to the direction of the path) along a path from the vertex $s$ to the origin. c)  The topological flux $w$ is the product of group elements along the closed loop from the origin around the excitation plaquette and back to the origin: $w = \bar z v z$. We will see in Sec. \ref{ribbons} that ribbon operators create such excitations with charge and flux contents given by $z,v$. d) An elementary anyon pair must fuse to the vacuum, which tells us that the local flux of the excitation living at the origin must be the inverse of the topological flux of the excitation away from the origin.}
\label{fig:double-model1}
\end{figure}

Next, we define the vertex operator $A^g_s$ which acts on all edges incident to the vertex $s$, and the plaquette operator $B^h_p$ which acts on the edges forming the boundary of plaquette $p$. The action of $A^g_s$ is to left-multiply each edge going \emph{out of} the vertex $s$ by $g$, and right-multiply each edge going \emph{into} the vertex $s$ by $\overline{g}$ (the inverse of the group element $g$). The operator $B^h_p$ projects the product of group elements around the plaquette $p$ onto the value $\bar{h}$; this product is taken starting at a given vertex $s$ and proceeding counterclockwise, with group elements multiplied from left to right. If the orientation of an edge is aligned with the  counterclockwise direction, it contributes its value to the product; if the orientation is is anti-aligned, that edge contributes its inverse. 

For the examples in Fig. \ref{fig:double-model1}a, $A^g_s$ and $B^h_p$ act in the following way:
\begin{equation}\label{eq:plaq_vert_ops}
\begin{aligned}
&A^g_s \ket{x_1, x_2, x_3, x_6,\dots} = \ket{x_1 \overline{g}, gx_2, gx_3, x_6 \overline{g},\dots}\\
&B^h_p \ket{x_3, x_4, x_5, x_6,\dots} = \delta_{\overline{h}, x_3 \overline{x}_4 \overline{x}_5 x_6} \ket{x_3, x_4, x_5, x_6,\dots}
\end{aligned}
\end{equation}
where ``$\dots$'' denotes the edges $x_i$ that are disjoint from $s$ and $p$. If they are associated with the same site $(s, p)$, the operators $A^g_s$ and $B^h_p$ generate the following algebra:
\begin{equation}\label{eq:drinfeld}
    A^{g_1}_s A^{g_2}_s = A^{g_1g_2}_s, \qquad B^{h_1}_pB^{h_2}_p = \delta_{h_1, h_2}B^{h_1}_p, \qquad A^g_s B^h_p = B^{gh\overline{g}}_p A^g_s
\end{equation}
The algebra generated by these commutation relations is called the \emph{Drinfeld double} of $G$, $\mathcal{D}(G)$; we can also think of it as the algebra of local operators at the site $(s, p)$ on the lattice.

Physically, we can view applying $A^g_s$ for all $g$ as imposing the Gauss law, checking for any "electric" charge living at the vertex $s$; whereas applying $B^h$ checks for any "magnetic" flux $h$ at the plaquette $p$. In the next section, we will see the specific form of the $A^g_s$ vertex operator and $B^h_p$ plaquette operators present in the Hamiltonian.

\subsection{Ground State Manifold}\label{ground-state-subspace}

The ground state subspace of this lattice model is defined by the span of states $\{|\mathrm{vac}\rangle\}$ satisfying the following stabilizer conditions:
\begin{equation}
    A_s^g|\mathrm{vac}\rangle = |\mathrm{vac}\rangle, \qquad B^h_{p}|\mathrm{vac}\rangle = \delta_{h, e}|\mathrm{vac}\rangle \qquad \forall s,p
\label{eq:stab-conditions}
\end{equation}

A Hamiltonian with such a ground state manifold (note that this Hamiltonian is not generically unique) can be written: 
\begin{equation}
    H_G = -\sum_s A_s - \sum_p B^e_p
\end{equation}
where $A_s = \frac{1}{|G|}\sum_{g \in G} A^g_s$ and $B^e_p$ are commuting projectors\footnote{Note that for non-Abelian $G$, the Hamiltonian of $\mathcal{D}(G)$ is not a stabilizer Hamiltonian, as there is no restriction that $A_s$, $B_p$ are products of Pauli operators. Stabilizer eigenvalues uniquely label \emph{all} states in the Hilbert space, not just the ground states. Commuting projector models, on the other hand, are still exactly solvable for the ground state manifold, since all the terms in the Hamiltonian commute. However, the full spectrum of excited states cannot be labeled uniquely by eigenvalues of these projectors.}.

The projectors ensure the ground states are free of excitations; a state that is annihilated by a vertex projector will have an ``electric'' charge living at the vertices, while a state annihilated by a plaquette term contains a``magnetic'' flux at that plaquette. To detect charge, the vertex operators $A_s$ check whether the state transforms nontrivially under the action of $G$. To detect flux, the plaquette operators $B^e_p$ check whether the product of group elements around a plaquette is $e$, i.e. trivial flux. In order to understand the properties of the ground state manifold, in particular its dimension, we need to enumerate the possible anyon excitations, which we will study next. 

\subsection{Excited states}
\label{excited-states}

We define an excitation as a violation of the vacuum stabilizer conditions localized to a given site (made up of a neighboring plaquette and vertex), as defined in Eq. \ref{eq:stab-conditions}. An elementary excitation involves a violation of at most two stabilizer conditions: one associated with a vertex operator and one associated with a plaquette operator. We will see later on that any local process always create a pair of excitations rather than a single excitation.

We will first consider a basis for the space of elementary excitations which we call the ``microscopic basis''. In this basis, we define an excitation by a pair of group elements, $z$ and $v$. The element $z$ gives the product of edge values along a path from the vertex associated with the excitation to an arbitrarily chosen but fixed origin\footnote{If you are familiar with the toric code, this is reminiscent of the string operators that create $e$ and $m$ particles; while it will turn out to be related, $z$ is a bookkeeping quantity here to label excited states. The true analog of the string operators will be the ribbon operators, which will be discussed in the next section.}, and $v$ gives the product of link values around the plaquette associated with the excitation (see Fig. \ref{fig:double-model1}b). We define the ``topological'' flux $w = \overline{z}vz$, which is the product of edge values around a loop from the origin around the plaquette of the excitation and back (see Fig. \ref{fig:double-model1}c). Why should we consider such a quantity? As we are dealing with non-Abelian group elements, when we take their product around a loop, where we start in that loop will impact the outcome. To properly compare the flux of different excitations, we need to measure their properties from the same starting point, which we define as the origin. 

Note that when we create an excitation at site $(s,p)$ with topological flux $w$ based at the origin $s_0$, we have also simultaneously created an excitation with local flux $\bar{w}$ at the origin $(s_0,p_0)$, such that the total flux is $w \cdot \bar{w}=e$, i.e. the total flux is trivial (Fig. \ref{fig:double-model1}d). With local operations, we always create pairs of anyon excitations in the vacuum fusion channel (trivial total flux and total charge) in order to satisfy the neutrality conditions, which is covered in more detail in Sec. \ref{neutrality}. 

We call $w$ topological (with respect to the excitation away from the origin at site $(s_0,p_0)$) because it is not affected by the actions of $A^g_s$ and $B^e_p$ operators local to the site of the excitation at $(s,p)$ \footnote{It is a useful exercise to check that this is the case.}. On the other hand, we call $v$ ``local'' because it will be affected by operators acting at the excitation site. The value of $z$, as the $z$ string stretches from the origin to the excitation site, will be affected both by local operators and operators with support far away. We call the degrees of freedom affected by local operators ``flavor'', and the nonlocal degrees of freedom ``color''. Ideally, we would like to deal only with color degrees of freedom because these will be unaffected by local processes like noise. Eventually, we will encode our logical information in these color degrees of freedom. 

We define the flavor operators (which affect local degrees of freedom):
\begin{equation}
    \begin{aligned}
        &A^g_{fl}|z, w\rangle = |gz, w\rangle\\
        &B^h_{fl}|z, w\rangle = \delta_{\overline{h}, zw\overline{z}}|z, w\rangle
    \end{aligned}
\end{equation}
The flavor operators are exactly the plaquette and vertex operators defined in Eq. \ref{eq:plaq_vert_ops} acting at the site $(s,p)$ of the excitation away from origin, written in terms of their action on $z$ and $w$.

The color operators (which affect global degrees of freedom) are defined:
\begin{equation}
    \begin{aligned}
        &A^g_{cl}|z, w\rangle = |z\overline{g}, gw\overline{g}\rangle\\
        &B^h_{cl}|z, w\rangle = \delta_{h, w}|z, w\rangle
    \end{aligned}
    \label{eq:color-operators}
\end{equation}
The color operators are exactly the plaquette and vertex operators defined in Eq. \ref{eq:plaq_vert_ops} acting at the origin $(s_0,p_0)$, where the other excitation of the pair is located. The color operators are chosen to commute with the flavor operators and to follow the same algebra given by Eq. \ref{eq:drinfeld}. We see that $z$ transforms under both the flavor and color operators; additionally, the states $\ket{z, w}$ are not eigenstates of either the flavor or color operators. Well-defined excitations should not change their type under the action of local operators--- we want a basis for excited states in terms of flavor and color degrees of freedom instead, and a way of labeling excitations which is invariant under the action of $A^g$ and $B^h$. We call this basis the ``anyon basis''. It will be exactly the basis corresponding to defined anyon types, which are irreducible representations of the quantum double.

\begin{table}
    \centering
    \begin{tabular}{|c|c|c|c|c|}
        \hline
        Conjugacy class & Centralizer & Charge & $d$ & Alternate label\\
        \hline
        \hline
         &  & $[+]$ & 1 & A\\
        $C_1$ & $Z(e) = S_3$ & $[-]$ & 1 & B\\
         &  & $[2]$ & 2 & C\\
         \hline
        $C_2$ & $Z((12)) \cong \mathbb{Z}_2$ & $[+]$ & 3 & D\\
        & & $[-]$ & 3 & E\\
        \hline
        & & $[+]$ & 2 & F\\
        $C_3$ & $Z((123)) \cong \mathbb{Z}_2$  & $[\omega]$ & 2 & G\\
        & & $[\omega^*]$ & 2 & H\\
        \hline
    \end{tabular}
    \caption{All possible anyon types in the $S_3$ quantum double model. The conjugacy class determines the flux content of the anyon type, while the irreducible representation of the centralizer determines the charge content of the anyon type. $d$ denotes the quantum dimension of the anyon; an anyon is classified as an Abelian anyon if it has a quantum dimension of $1$, and a non-Abelian anyon if its quantum dimension is greater than $1$. The final column lists alternate labels for each anyon type often used in the literature.}
    \label{tab:anyon_type}
\end{table}

Indexing irreducible representations (anyon types) of the Drinfeld double $\mathcal{D}(G)$ by $\alpha$, we can decompose Hilbert space into decoupled subspaces involving flavor and color degrees of freedom:
\begin{equation}
    \mathcal{H} = \bigoplus_{\alpha} \mathcal{L}_{\alpha} \otimes \mathcal{L}_{\overline{\alpha}}
\end{equation}
There are 8 anyon types in $D(S_3)$, as enumerated in Table \ref{tab:anyon_type}, with fusion rules listed in Table \ref{tab:fusion_rule}. Each anyon type $\alpha$ of $\mathcal{D}(G)$ is defined by the following:
\begin{enumerate}
    \item A \emph{conjugacy class} of $G$: In the case of $S_3$, there are three conjugacy classes: $C_e = \{e\}$, the 2-cycles $C_2 = \{(12), (23), (13)\}$ and the 3-cycles $C_3 = \{(123), (132)\}$. The conjugacy class defines the types of flux excitations.
    \item An \emph{irreducible representation of the centralizer} of a representative member of the conjugacy class; this defines the charge of the particle.\footnote{See appendix \ref{group-theory} for more on the representation theory of $S_3$.} 
\end{enumerate}
Here, we note that this accounting of anyon types is compatible with the physical Hilbert space available at the microscopic level. Each unit cell of the lattice has two qudits, yielding a total Hilbert space dimension of $6^2 = 36$. Furthermore, each unit cell can be associated with a unique plaquette-vertex pair $(s, p)$, which can host one anyon of any type. The total Hilbert space dimension needed to accomodate these anyon states is $1 + 1 + 2^2 + 3^2 + 3^2 + 2^2 + 2^2 + 2^2 = 36$. 

In the following sections, we will go through why these properties define anyon types, focusing on the physical intuition.\footnote{For more information on the mathematical aspects of Drinfeld doubles, see \cite{cui_topological_2018}.}

\begin{table}
    \centering
    \begin{tabular}{|c||c|c|c|c|c|c|c|c|}
        \hline
        $\otimes$ & $A$ & $B$ & $C$ & $D$ & $E$ & $F$ & $G$ & $H$\\
        \hhline{|=#=|=|=|=|=|=|=|=|} 
        $A$ & $A$ & $B$ & $C$ & $D$ & $E$ & $F$ & $G$ & $H$ \\ 
        \hline
        $B$ &  & $A$ & $C$ & $E$ & $D$ & $F$ & $G$ & $H$ \\
        \hline
        $C$ &  &  & $A\oplus B \oplus C$ & $D\oplus E$ & $D\oplus E$ & $G\oplus H$ & $F\oplus H$ & $F\oplus G$ \\
        \hline
        $D$ &  &  &  & \begin{tabular}{@{}c@{}}$A\oplus C \oplus F$ \\ $\oplus G\oplus H$\end{tabular}  & \begin{tabular}{@{}c@{}}$B\oplus C\oplus F$ \\ $\oplus G\oplus H$\end{tabular}  & $D\oplus E$ & $D\oplus E$ & $D\oplus E$ \\
        \hline
        $E$ &  &  &  &  & \begin{tabular}{@{}c@{}}$A\oplus C\oplus F$ \\ $\oplus G\oplus H$\end{tabular} & $D\oplus E$ & $D\oplus E$ & $D\oplus E$ \\
        \hline
        $F$ &  &  &  &  &  & $A\oplus B\oplus F$ & $C\oplus H$ & $C\oplus G$ \\
        \hline
        $G$ &  &  &  &  &  &  & $A\oplus B\oplus G$ & $C\oplus F$ \\
        \hline
        $H$ &  &  &  &  &  &  &  & $A\oplus B\oplus H$ \\
        \hline
    \end{tabular}
    \caption{Fusion rules of anyons in the $S_3$ quantum double model. For non-Abelian anyons, since their quantum dimension is greater than $1$, there are multiple possible fusion outcomes in general. Since the fusion is symmetric, i.e. $a \otimes b = b \otimes a$, we only enumerate the upper triangular part of the table.}
    \label{tab:fusion_rule}
\end{table}

\subsubsection{Conjugacy class and flux}

We first discuss why the flux of an anyon is described by the conjugacy class rather than an individual group element\footnote{The following argument is mostly based on John Preskill's lecture notes on topological quantum computing \cite{preskill_lecture_2004}.}. The flux of a particle can be measured by moving a charge along a closed loop around the flux--- exactly as in the Aharonov-Bohm effect, where we can infer the magnetic field of a solenoid by encircling it by an electron. Therefore, we can represent the flux of a particle by a closed loop we would use to measure that flux. Note that all loops need to start and end at the same origin due to the non-Abelian nature of our gauge group. This means the flux we are considering in this section is really the topological flux $w$ introduced in the previous section---the local flux $v$ is measured from a different place for each excitation and so does not allow for consistent comparison of different excitations.

\begin{figure}[h!]
\centering
\includegraphics[width=1\textwidth]{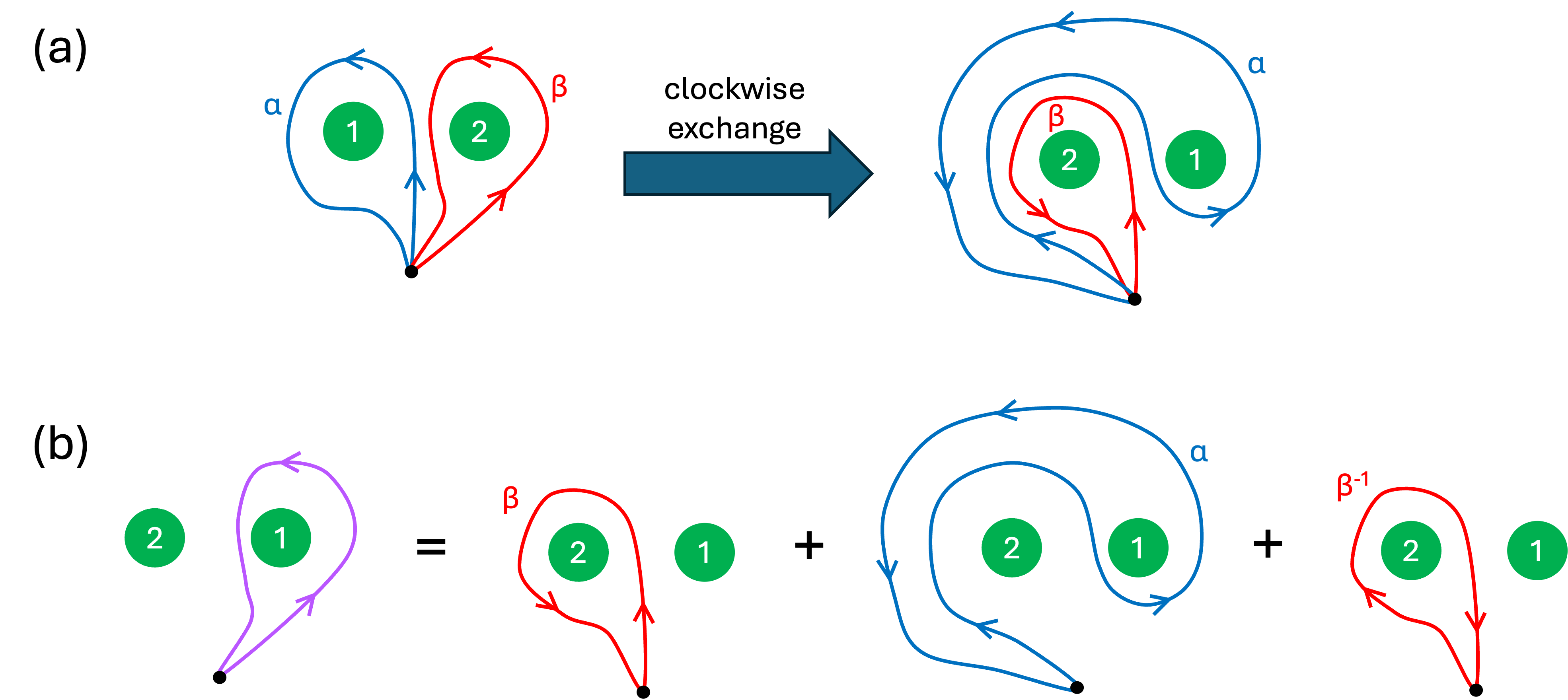}
\caption{Illustration for how braiding changes flux labeling. (a) Clockwise exchanging two flux particles. The loops corresponding to the original flux labels are continuously deformed during the braiding. (b) Decomposition of the loop around particle 1 (in its new location) in terms of previous loops $\alpha,\beta$.}
\label{fig:braided}
\end{figure}

Let ud consider two particles 1 and 2, enclosed by loops $\alpha$ and $\beta$ respectively ($\alpha,\beta\in G$), as shown in the left panel of Fig. \ref{fig:braided}(a). Now we exchange the two particles clockwise. The loops enclosing the particles are smoothly deformed, following the trajectory of the particles, as shown in the right panel of Fig. \ref{fig:braided}(a). In order to compare the fluxes before and after the exchange, we need to construct loops around them that look more like the left panel of Fig. \ref{fig:braided}(a). The loop around particle 2 does not need to be changed, and we see that particle 2 still has flux $\beta$. For particle 1, to measure its new flux, we need instead the purple flux loop in the leftmost panel of Fig. \ref{fig:braided}(b). To express it in terms of the original $\alpha$ and $\beta$ loops, we can decompose the purple loop into a concatenation of three loops, as shown on the right hand side of Fig. \ref{fig:braided}(b). This tells us that the new flux of particle 1 is $\bar\beta \alpha \beta$ (we use the convention of ordering from right to left).

Symbolically, we can express the result of the exchange as the action of some clockwise exchange operator $R$ on the state $\ket{\alpha,\beta}$:
\begin{equation}
    R\ket{\alpha,\beta} = \ket{\beta,\overline{\beta}\alpha\beta}
\end{equation}
Similarly, a counterclockwise exchange $R^{-1}$ would leave the first particle's flux unchanged, but conjugate the flux of particle 2:
\begin{equation}\label{eq:ccw-exchange-op}
     R^{-1}\ket{\alpha,\beta} = \ket{\alpha\beta\overline{\alpha},\alpha}
\end{equation}

Our choice of location of base point led to the asymmetric action of exchange $R$ on the two fluxes; however, this asymmetry is not physical---in order to measure the effect of exchanging two particles we need to ensure that the final configuration (in terms of the position of the particles) is the same as the initial configuration. Between two different flux $\alpha,\beta$, this is achieved by doing a full winding, i.e. clockwise exchanging the positions of the two particles twice, such that the particles return to their original locations after the operation. The overall effect of full winding will be to conjugate the particles by their total flux $\alpha \beta$:
\begin{equation}
    R^2\ket{\alpha,\beta} = \ket{\overline{(\alpha\beta)}\alpha(\alpha\beta),\overline{(\alpha\beta)}\beta(\alpha\beta)}
\end{equation}

Now, consider two observers Alice and Bob, who measure the flux information of the particles before and after the clockwise winding, respectively. Neither observer is aware of the prior history of the particle, and they need not measure the flux from the same base point. However, since the two particles are are in the same configuration according to Alice and Bob, they should agree on the particle types inferred from the flux measurements. Therefore, their two outcomes, $\alpha$ and $\overline{(\alpha\beta)}\alpha(\alpha\beta)$, should correspond to the same type of particle. This means we should define an equivalence relation between group elements related to each other by conjugation--- this is precisely a conjugacy class. So flux is only defined up to conjugacy class---the individual elements of a conjugacy class label the internal states of a flux, rather than the flux type itself. 

This thought experiment highlights a key property of the internal states of a non-Abelian anyon: they are not fundamental labels like the z-component of a spin, but rather a choice of convention. Different observers will generically assign different internal states to the same particle. However, all measurement results should be related by a consistent, \emph{global} change of convention; two observers should agree on whether two different particles have the same internal label or not.

\subsubsection{Irreducible representations and charge}

The charge information of an anyon is given by the irreducible representation of the centralizer (of a representative of a conjugacy class). To parse this involved definition, we need to first understand, why the centralizer? The centralizer of a group element $g \in G$ is the subgroup of all elements that commute with $g$: $Z(g) = \{h \in G| gh = hg\}$.\footnote{The centralizer is well-defined up to different choices of representatives within the same conjugacy class: given two elements $g_1,g_2$ in the same conjugacy class, their respective centralizer $Z(g_1)$ and $Z(g_2)$ are isomorphic. Therefore, by picking one representative from a given conjugacy class we can define the centralizer for the whole conjugacy class.} 

As with flux, we can gain some intuition by thinking about how we would physically measure the charge of a particle in experiment. Given an unknown charge, we can use a beam of calibrated fluxes to measure the charge: the fluxes can either pass to the left of the charge or to the right of the charge. There is a phase associated to each possible path, and the two topologically distinct paths (going to the left or right of the charge) will give rise to an interference pattern---the particulars of the interference pattern for different test fluxes give us information on the unknown charge. However, the interference pattern will be destroyed if the two topologically distinct paths are \emph{distinguishable}.

The two paths are distinguishable when the charge also carries a flux $h$ which does not commute with the test flux $g$; then, the interference pattern will be lost. This is due to the fact going left or right around the particle corresponds to a clockwise or counter-clockwise exchange. In the first case, the flux of the test particle changes from $g$ to $hg\overline{h} \neq g$. In the second, nothing happens to the flux $g$. We can tell which path the test particle has taken based on how its flux changes, so there will be no interference pattern. The interference pattern will only be preserved if $g$ commutes with $h$, so that $\bar h g h = g$. Thus, we can only define the charge of an anyon with flux $h$ in terms of the centralizer $Z(h)$. Note that, if we were to label the flux $h$ with a different element $h'$ in the same conjugacy class, the centralizer $Z(h')$ will not strictly be the same, but is related to $Z(h)$ by an isomorphism--- crucially, this means they have the equivalent representations.

The last piece of information to specify a charge is the irreducible representation. Moving a flux $g$ in a closed path around a charge $R$ will transform the state of the charge according to the unitary matrix $\Gamma^R(g)$, where $R$ is an irreducible representation of $S_3$. This unitary transformation will generically involve picking up some phase. The two paths the test fluxes can take in the interference experiment are related by such a closed loop--- going right around the charge is the same as going left and then fully around the charge counter-clockwise. This means there will be a phase shift due to the charge-flux topological interaction. By measuring the shift of the interference pattern due to this phase, we can extract the matrix elements of $\Gamma^R(g)$. Repeating with many different $g$ will allow us to determine the representation $R$. 

To summarize, if we have some particle with both charge and flux, the charge must be an irreducible representation of the centralizer $Z(h)$ (the subgroup containing all elements $g \in S_3$ that commute with the flux of the unknown particle), not of the whole group\footnote{This thought-experiment was drawn from John Preskill's lecture notes on anyons and topological quantum computing\cite{preskill_lecture_2004}. See \cite{overbosch_inequivalent_2001,bonderson_interferometry_2008} for further reading on non-Abelian anyon interferometry.}.

\subsection{Relating the anyon basis to the microscopic basis} \label{basis-change}

While the anyon basis allows understanding abstractly how the excitations in our model behave and interact, the microscopic basis is clearer for doing explicit calculation. Also, in any near-term realization of $S_3$ anyons on a quantum simulator, we will have access to the degrees of freedom on the lattice, which allows us to create and manipulate excitations in the microscopic basis directly, rather than having to rely on control over excitations in the anyon basis. Therefore, it is useful to write down the explicit map between these two bases.

As an example, let us try expressing flux excitations in the $C_3$ conjugacy class in terms of the microscopic basis. First, we enumerate all possible internal states in the anyon basis with $C_3$ flux. The internal Hilbert space $\mathcal{H}_{C_3}$ is 12 dimensional---each $C_3$ flux carries an internal Hilbert space of dimension $4$\footnote{Each anyon has a color and flavor degree of freedom, which are each Hilbert spaces of dimension $2$ in the case of $C_3$ flux.}, and there are three types of anyon with $C_3$ flux, differentiated by their charge ($[+]$, $[\omega]$, and $[\omega^*]$). This means, in terms of the anyon basis, we can decompose $\mathcal{H}_{C_3}$  into a direct sum over these different anyon types:
\begin{equation}
    \mathcal{H}_{C_3} = \left( \mathcal{L}_{[C_3,1]} \otimes  \mathcal{L}_{[C_3,1]} \right ) \oplus \left( \mathcal{L}_{[C_3, \omega^*]} \otimes  \mathcal{L}_{[C_3, \omega]} \right ) \oplus \left( \mathcal{L}_{[C_3, \omega]} \otimes  \mathcal{L}_{[C_3, \omega^*]} \right )
\end{equation}

The flavor-color basis states in each part of this decomposition will look like $|v, \overline{R}\rangle \otimes |w, R\rangle$, where $v$ is the local flux, $w$ is the topological flux, and $R$ is the representation of the charge. The first term in this decomposition corresponds to pure $C_3$ fluxes with trivial charge; i.e. the $[C_3,1]$ anyons. We will denote the four basis states for this subspace by (omitting the labeling for the $1$ representation):
\begin{equation}
\begin{aligned}
    \ket{\mu ; \mu} &\equiv \ket{\mu} \otimes \ket{\mu}\\
    \ket{\overline\mu ; \overline\mu} &\equiv \ket{\overline\mu} \otimes \ket{\overline\mu}\\
    \ket{\mu ; \overline\mu} &\equiv \ket{\mu} \otimes \ket{\overline\mu}\\
    \ket{\overline\mu ; \mu} &\equiv \ket{\overline\mu} \otimes \ket{\mu}
\label{eq:anyon-basis-states-c3}
\end{aligned}
\end{equation}
where $\mu = (123)$ and $\sigma = (23)$\footnote{We will use the cycle notation and the $\mu, \sigma$ notation interchangeably throughout the rest of the text--- see appendix \ref{group-theory} for notation and convention. The key identity to remember is that $\mu \sigma = \sigma \overline{\mu}$, which we will use a lot going forward.}. 

We want to express each of the states in equation \ref{eq:anyon-basis-states-c3} in the $|z, w\rangle$ basis. First, let us consider $|\mu ;\mu\rangle$. Because of the definition of global flux in terms of local flux, $w = \bar{z}vz$, this constrains the possible value of $z$; if $w = \mu$, $v = \mu$, then $z \in \{e, \mu, \overline{\mu}\}$. Therefore, the general form of $\ket{\mu;\mu}$ in terms of the microscopic states $\{\ket{z, w}\}$ is
\begin{equation}\label{eq:C3-anyon-basis-in-terms-of-microscopic-basis}
\begin{aligned}
|\mu ;\mu\rangle &= \sum_{z} c_z |z, w\rangle \\
&= c_{e} |e, \mu\rangle + c_{\mu} |\mu, \mu\rangle + c_{\overline{\mu}}|\overline{\mu}, \mu\rangle  
\end{aligned}
\end{equation}

We can determine the coefficients $c_z$ (up to normalization) by acting $A^{\mu}_{\mathrm{cl}}$ and $A^{\mu}_{\mathrm{fl}}$ on the state and demanding it transforms the way we expect. Considering the left-hand side of Eq. \ref{eq:C3-anyon-basis-in-terms-of-microscopic-basis} first, using the definition of color vertex operator (Eq. \ref{eq:color-operators}), and noting that the flavor-color basis is expressed by $\ket{v,w}$, we find that acting $A^{\mu}_{\mathrm{cl}}$ gives:
\begin{equation}
    A^{\mu}_{\mathrm{cl}}|\mu ;\mu\rangle = |\mu; \mu (\mu) \overline{\mu}\rangle = |\mu; \mu\rangle
\end{equation}

Why should applying $A^{\mu}_{\mathrm{cl}}$ act trivially on the local flux and conjugate the topological flux by $\mu$? As we will see in Sec. \ref{closed-ribbons}, the action of local $A^g$ operators is to wind a flux of $g$ around a given vertex. The color operator $A_{cl}^g$ winds a flux $g$ around the origin, affecting the topological degrees of freedom rather than the local ones.

Using Eq. \ref{eq:color-operators} again, acting $A^{\mu}_{\mathrm{cl}}$ on the right-hand side of Eq. \ref{eq:C3-anyon-basis-in-terms-of-microscopic-basis} gives:
\begin{equation}
    A^{\mu}_{\mathrm{cl}}\left(c_{e} |e, \mu\rangle + c_{\mu} |\mu, \mu\rangle + c_{\overline{\mu}}|\overline{\mu}, \mu\rangle  \right) = c_{e} |\overline{\mu}, \mu\rangle + c_{\mu} |e, \mu\rangle + c_{\overline{\mu}}|\mu, \mu\rangle
\end{equation}
Setting the right-hand side equal to the left-hand side, we see that $c_e = c_{\mu} = c_{\overline{\mu}}$. In order for the state to be normalized, we need $c_e = \frac{1}{\sqrt{3}}$. So we find:
\begin{equation}
    |\mu ;\mu\rangle = \frac{1}{\sqrt{3}} \left(|e, \mu\rangle + |\mu, \mu\rangle + |\overline{\mu}, \mu\rangle \right) 
\label{eq:mumu}
\end{equation}
We can repeat the same procedure for the other three states in this subspace:
\begin{equation}
\begin{aligned}
    &|\overline{\mu} ; \overline{\mu} \rangle = \frac{1}{\sqrt{3}} \left(|e, \overline{\mu}\rangle + |\mu, \overline{\mu}\rangle + |\overline{\mu}, \overline{\mu}\rangle \right) \\
    &|\overline{\mu} ; \mu \rangle = \frac{1}{\sqrt{3}} \left(|\sigma, \mu\rangle + |\overline{\mu} \sigma, \mu\rangle + |\mu \sigma, \mu\rangle \right)\\
    &|\mu ; \overline{\mu} \rangle = \frac{1}{\sqrt{3}} \left(|\sigma, \overline{\mu}\rangle + |\overline{\mu} \sigma, \overline{\mu}\rangle + |\mu \sigma, \overline{\mu}\rangle \right) 
\label{eq:otherstates}
\end{aligned}
\end{equation}

Next, let us look at states in the $\mathcal{L}_{[C_3, \omega^*]} \otimes \mathcal{L}_{[C_3, \omega]}$ subspace. These are states with color that transforms like the one-dimensional $[\omega]$ representation of the three-cycle centralizer $Z(C_3) \cong \mathbb{Z}_3$ (see Appendix \ref{rep-theory-z3} for the definition of the $[\omega],[\omega^*]$ representation). In the flavor-color basis, the four states in this subspace are:
\begin{equation}
\begin{aligned}
    &|\mu, \omega^* ; \mu, \omega \rangle\\
    &|\overline{\mu}, \omega^* ; \overline{\mu}, \omega \rangle\\
    &|\overline{\mu}, \omega^* ; \mu, \omega \rangle\\
    &|\mu, \omega^* ; \overline{\mu}, \omega \rangle
\end{aligned}
\end{equation}
Let us find the decomposition of the first state, $|\mu, \omega^* ; \mu, \omega \rangle$, in the $|z, w\rangle$ basis:
\begin{equation}
\begin{aligned}
     |\mu, \omega^* ; \mu, \omega \rangle &= \sum_z c_z |z, \mu\rangle\\
     &= c_e |e, \mu\rangle + c_{\mu}|\mu, \mu\rangle + c_{\overline{\mu}}|\overline{\mu}, \mu\rangle
\end{aligned}
\end{equation}
As before, we only sum over $z$ consistent with the condition $w = \overline{z}vz$. We can apply the operator $A^{\mu}_{\mathrm{cl}}$ to both sides and set them equal to find the coefficients $\{c_z\}$. The left-hand side becomes:
\begin{equation}
    A^{\mu}_{\mathrm{cl}} |\mu, \omega^* ; \mu, \omega \rangle = \omega  |\mu, \omega^* ; \mu, \omega \rangle
\end{equation}
Where does the overall phase $\omega$ come from? As we noted before, the $A^g$ operator is really winding a flux of $g$ around a given vertex: winding a flux around a charge transforms the charge according to its representation, in this case simply the one-dimensional representation $[\omega]$, where $\mu$ acts as a phase factor $\omega$.

Applying $A^{\mu}_{\mathrm{cl}}$ to the right-hand side and setting equal to the left-hand side yields:
\begin{equation}
\begin{aligned}
    A^{\mu}_{\mathrm{cl}} \left(c_e |e, \mu\rangle + c_{\mu}|\mu, \mu\rangle + c_{\overline{\mu}}|\overline{\mu}, \mu\rangle\right) &= c_e |\overline{\mu}, \mu\rangle + c_{\mu}|e, \mu\rangle + c_{\overline{\mu}}|\mu, \mu\rangle \\
    &= \omega \left (c_e |e, \mu\rangle + c_{\mu}|\mu, \mu\rangle + c_{\overline{\mu}}|\overline{\mu}, \mu\rangle \right)
\end{aligned}
\end{equation}
We see that $c_e = \omega c_{\overline{\mu}}$, $c_{\mu} = \omega c_e$, and $c_{\overline{\mu}} = \omega c_{\mu}$. We can take the convention that $c_e = 1$, and we find then that $c_{\mu} = \omega$ and $c_{\overline{\mu}} = \omega^*$. To normalize, we just divide everything by $\sqrt{3}$ to arrive at the final expression
\begin{equation}
    |\mu, \omega^* ; \mu, \omega \rangle = \frac{1}{\sqrt{3}} \left(|e, \mu\rangle + \omega|\mu, \mu\rangle + \omega^*|\overline{\mu}, \mu\rangle \right)
\end{equation}

We will not go through the rest of the $C_3$ states, nor the other anyon types, but the same method can be used to find their decompositions in terms of the microscopic basis. The general closed-form expression for the basis transformation can be found in appendix \ref{ap:basis-change}. One key aspect to stress here is the method we used to find these relationships. We can use the $A^g$, $B^h$ operators to check how any state transforms---given any composite state of many anyons, we can apply the same method to determine the charge and flux of the overall state.

\subsection{Single-particle excitations}

Perhaps counter-intuitively, single-particle excitations are possible in non-Abelian topological order when put on a non-trivial topology \cite{bombin_family_2008,iqbal_non-abelian_2024,Ritz24}. For example, for the $S_3$ quantum double, there are 8 possible ground states and $28$ possible single-particle excited states on a torus. They do not violate any neutrality conditions (as will be discussed in the next section), because generating a single-particle excitation requires an operator with extensive support over the entire system, whereas the neutrality conditions are local statements.

Here we detail how to understand the single-particle excitations in $D(S_3)$ specifically. Consider the $S_3$ quantum double on a torus with the smallest possible geometry---a $1$ by $1$ square lattice with periodic boundary conditions. On this small torus, there is only one vertex, one plaquette, and two independent edges. The vertex operator $A^g$ simply conjugates both edges by $g$, and the plaquette operator $B^h$ now enforces the group commutator of the two edges $g_1,g_2$ be $g_1g_2g_1^{-1}g_2^{-1} = h$.

Since there are $2$ edges, the full Hilbert space is of dimension $6^2=36$. Pair excitations are not possible since there is only one site; so the only possible states are ground states and single-particle excited states. We have already established that there are eight possible ground states in $D(S_3)$. We can now understand them as combinations of $g_1,g_2$ such that the group commutator $g_1g_2g_1^{-1}g_2^{-1}=e$. We can construct them by starting from a representative configuration that satisfies the plaquette operator (such that it is in the zero-flux sector). The eight representative configurations for each of the eight ground states are enumerated in Fig. \ref{fig:S3-ground-state-config}. To obtain the ground state, we apply the vertex projector $\sum_{g \in G} A_g$ to create a symmetric superposition of configurations, projecting to the zero-charge sector.

\begin{figure}[h!]
\centering
\includegraphics[width=0.6\textwidth]{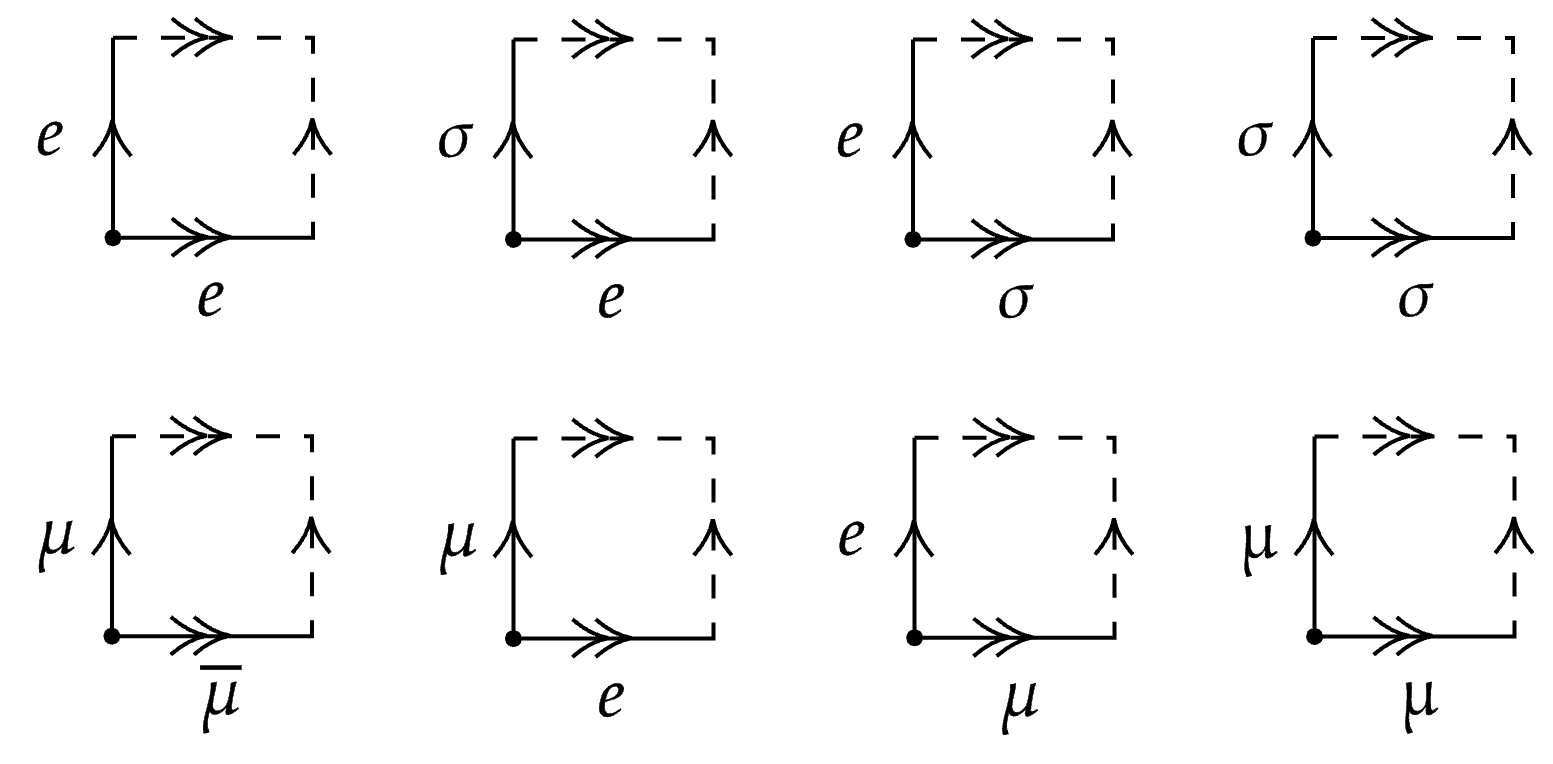}
\caption{Representative configurations for each possible ground state in the $S_3$ quantum double. The configurations satisfy the plaquette operator $B^e$, so they are in the zero-flux sector. Applying the vertex operator to project to the zero-charge sector will generate resonating terms to give an overall symmetric superposition, yielding a ground state. We remark that the number of distinct ground states on the torus is the exact same counting of the anyon labels, since the vertical group element is the flux and the horizontal group element is a conjugacy class of the centralizer, which we can (non-canonically) associate to the irreps of that centralizer.}
\label{fig:S3-ground-state-config}
\end{figure}

Since the vertex operator involves conjugating both the horizontal and vertical edges, it is impossible to start with one of the above representatives and obtain another representative as part of the resonating terms. Hence the above defines eight distinct ground states.

The other $36-8=28$ states are single-particle excited states. The single-flux excitation corresponds to a nontrivial group commutator $g_1g_2\bar g_1\bar g_2\neq e$; the single-charge excitation corresponds to non-symmetric superpositions of the resonating terms. From this, we can infer a selection rule for single-flux excited states: it is not possible to have single-flux excitation where the flux is not in the commutator subgroup of the gauge group. For example, in $S_3$ quantum double, the commutator subgroup of $S_3$ is $\mathbb{Z}_3=\{e,\mu,\bar\mu\}$. So it is impossible to have $C_2$ single-particle excitations, as the $C_2$ internal states label $\sigma,\mu\sigma,\bar\mu\sigma$ are not in the commutator subgroup. We are not aware of an explicit mention in the literature of such selection rule for single-particle excitations in quantum doubles, so we think it warrants a mention here.

\subsection{Neutrality conditions}\label{neutrality}

For the quantum double theory to make sense physically, we should expect some restrictions on how fractionalized excitations like anyons can be created. Such restrictions are called ``neutrality conditions''. Given any set of localized operators (their support is restricted to a contractible region) that create anyons, we expect the fusion channel of these anyons should be the vacuum. In other words, any other anyon braiding fully around the region in question should not detect anything but the vacuum. 

Therefore, starting with a ground state, the neutrality conditions ensure that all excitations that can be created locally can eventually be fused back to vacuum. A pair of excitations can be created out of the vacuum only if they satisfy the neutrality conditions.

At the lattice level, the condition of global neutrality amounts to ensuring two conditions: flux neutrality and charge neutrality \textit{in a local region}.\footnote{The reason why we restrict the neutrality conditions to local contractible region is to rule out the case of single-particle states as discussed in the previous sub-section.}
\begin{itemize}
    \item \textbf{Flux neutrality}: Given a base point, all global fluxes (determined by multiplying group elements along a loop starting and ending at the same base point) must multiply to the identity group element. This ensures that the action of any charge braiding around the entire system is trivial. 
    \begin{equation}
        A^g_{\mathrm{global}}|\psi\rangle = |\psi\rangle ~\forall~ g \in G
    \end{equation}
    where $A^g_{\mathrm{global}}$ acts as $A^g_{cl}$ on every excitation in the state. Its action on the microscopic basis is
    \begin{equation}
        A^g_{\mathrm{global}}|z_1, w_1; \cdots; z_n, w_n\rangle = |z_1\overline{g}, gw_1\overline{g}; \cdots; z_n\overline{g}, gw_n\overline{g}\rangle
    \end{equation}
    \item \textbf{Charge neutrality}: The charge neutrality condition in a local region asserts that winding any pure flux around the local region will act trivially on the state. 
    \begin{equation}
    B^h_{\mathrm{global}}|\psi\rangle = \delta_{h, e}|\psi\rangle
    \end{equation}
    where $B^h_{\mathrm{global}}$ projects the product of all the global fluxes in the state to the group element $h$. Its action on the microscopic basis is
    \begin{equation}
        B^h_{\mathrm{global}}|z_1, w_1; \cdots; z_n, w_n\rangle = \delta_{h, w_1\cdots w_n}|z_1, w_1; \cdots; z_n, w_n\rangle
    \end{equation}
\end{itemize}

We will often call a neutral pairs of excitations ``singlets'', or ``vacuum pairs'' going forward. As an example, consider the singlet state of charge $[2]$ anyons:
\begin{equation}
    |\psi\rangle = \frac{1}{\sqrt{2}}(|2_+\rangle|2_-\rangle + |2_-\rangle|2_+\rangle)
\end{equation}
Acting $A^{\sigma}_{\mathrm{global}}$ on this state yields:
\begin{equation}
    \begin{aligned}
        A^{\sigma}_{\mathrm{global}}|\psi\rangle &= 
        A^{\sigma}_{\mathrm{global}}\frac{1}{\sqrt{2}}(|2_+\rangle|2_-\rangle + |2_-\rangle|2_+\rangle)\\
        &= \frac{1}{\sqrt{2}} (|2_-\rangle|2_+\rangle + |2_+\rangle|2_-\rangle)\\
        &= |\psi\rangle
    \end{aligned}
\end{equation}
where we used the fact $\sigma$ acts like the Pauli operator $X$ in the two-dimensional irreducible representation of $S_3$. Acting $A^{\mu}_{\mathrm{global}}$ yields:
\begin{equation}
    \begin{aligned}
        A^{\mu}_{\mathrm{global}}|\psi\rangle &= 
        A^{\mu}_{\mathrm{global}}\frac{1}{\sqrt{2}}(|2_+\rangle|2_-\rangle + |2_-\rangle|2_+\rangle)\\
        &= \frac{1}{\sqrt{2}} (\omega \omega^* |2_+\rangle|2_-\rangle + \omega^* \omega |2_+\rangle|2_-\rangle)\\
        &= |\psi\rangle
    \end{aligned}
\end{equation}

Applying $A^g_{\mathrm{global}}$ for the other elements in the group will give similar results, as they are simply different combinations of $\sigma$ and $\mu$. As we are only dealing with charges in this example, we can immediately see $\ket{\psi}$ will be invariant under $B^e_{\mathrm{global}}$, as it enforces the flux neutrality. We will see in the next section that the globally neutral subspace will be used as the logical state space for universal quantum computation.

\section{Universal Computation with S3} \label{universal_computing}
\begin{tcolorbox}[enhanced, breakable,pad at break*=1mm]
    In this section, we discuss the resources and protocols necessary to construct a universal gate set with $S_3$ anyons. We summarize the construction of a universal \emph{qubit} gate set built from qutrit operations. This gate set was developed by Kitaev in unpublished work; there is also a problem set from one of his courses at Caltech where he goes through the construction of the gate set \cite{kitaev_anyons_nodate}. We have collated the results in this section from these resources--- we have also worked out in detail the construction of the usual universal gate set of Clifford gates (Hadamard and S gate) together with a non-Clifford gate (CCZ) to explicitly demonstrate the universality of the gate set, gate fidelity considerations, and state initialization procedures.
\end{tcolorbox}

\subsection{From lattice description to effective field theory perspective}

We defer the technical details of the lattice-level implementation in terms of the ribbon operator in the next section \ref{ribbons}. For this section, to simplify the discussion, we take an effective-field theory perspective and directly work with anyons to implement the gates. 

\subsection{Initializing states} \label{state_initialization}

\subsubsection{Computational basis}

We use $C_2$ flux singlets to encode logical qutrits\footnote{Later on, when constructing the qubit gates, we will exclude one of the internal states from the qubit computational basis, although it will still be very useful for implementing qubit gates.}. We want to construct our fundamental qutrits from globally neutral fluxes; otherwise, we would introduce unwanted entanglement simply by moving our qutrits through the system. Hence, we define the computational basis as follows:  
\begin{equation}\label{eq:computational-basis-def}
    \begin{aligned}
        \ket{0}&\equiv \ket{\sigma, \sigma}\\
        \ket{1}&\equiv \ket{\mu \sigma, \mu \sigma}\\
        \ket{2}&\equiv \ket{\bar \mu \sigma, \bar \mu \sigma}
    \end{aligned}
\end{equation}

Each qutrit is made up of a pair of fluxes with total neutral flux. These computational basis states are eigenstates of the qutrit $\mathcal{Z}$ gate, also known as a ``clock'' operator; the qudit $\mathcal{X}$ gate is called a ``shift'' operator\footnote{For notational clarity, we will denote qutrit Pauli operators with calligraphic font, whereas qubit operators will be denoted in regular font} (see \cite{gottesman_fault-tolerant_1999} for an overview of universal quantum computing with qudits):
\begin{equation}
\begin{aligned}
    &\mathcal{Z}\ket{i} = \omega^i \ket{i}\\
    &\mathcal{X}\ket{i} = \ket{i+1}
\end{aligned}
\end{equation}
where $\omega = e^{i\frac{2\pi}{3}}$ is a third root of unity, and addition is assumed to be modulo 3. Note that the generalized Paulis $\mathcal{Z}$ and $\mathcal{X}$ obey the following commutation relation:
\begin{equation}
    \mathcal{ZX} = \omega \mathcal{XZ}
\end{equation}

\subsubsection{Dual basis}

In addition to the computational basis, we define the \textit{dual} basis states:
\begin{equation}\label{eq:dual-basis-def}
\begin{aligned}
    \ket{\tilde 0}&\equiv \frac{1}{\sqrt{3}}\left(\ket{0}+\ket{1}+\ket{2}\right)\\
    \ket{\tilde 1}&\equiv \frac{1}{\sqrt{3}}\left(\ket{0}+\omega\ket{1}+\omega^*\ket{2}\right)\\
    \ket{\tilde 2}&\equiv \frac{1}{\sqrt{3}}\left(\ket{0}+\omega^*\ket{1}+\omega\ket{2}\right)
\end{aligned}
\end{equation}
These are the eigenstates of the $\mathcal{X}$ gate:  
\begin{equation}
\begin{aligned}
    &\mathcal{X}\ket{\Tilde{i}} = \omega^i\ket{\Tilde{i}}\\
    &\mathcal{Z}\ket{\Tilde{i}} = \ket{\widetilde{i+1}}
\end{aligned}
\label{qutrit-gates-dual-basis}
\end{equation}

Note that the dual basis states are still flux-neutral, as they are superpositions of flux-neutral states. Using the vertex and plaquette operators, one can show that they hold definite charge. The $\ket{\tilde 0}$ state is the trivial charge superposition for the $C_2$ conjugacy class, while $\ket{\tilde 1}$ and $\ket{\tilde 2}$ both transform like $[2]$ charges (they correspond to the $\ket{2_+}$ and $\ket{2_-}$ states, respectively).

\subsection{Measurements} \label{measurements}
A key part of our universal gate sets (both qubit and qutrit) will be the ability to measure in either the computational or dual bases. This will require being able to measure the flux or the charge of a given set of anyons, respectively.

\subsubsection{Measuring Flux} \label{measure-flux}

We can measure the flux by creating a $[2]$ charge singlet, braiding one of the resulting charges with the flux to be measured, and then trying to fuse the two charges again. If the flux acted trivially on the charge, the two charges will always fuse back to the vacuum. However, if the flux had some non-trivial action, there will be some probability that the two charges leave some remnant particle behind. Repeating this procedure many times, we can build up the probabilities of different outcomes, which tells us how the charge transforms under the action of the flux. We can then match the flux to its representation.

As a concrete example, consider a set of fluxes with total flux in the $C_3$ conjugacy class. This is used to perform measurement in the computational basis, as we will see in Sec. \ref{subsubsec:direct-measure}. For simplicity, suppose the topological flux is $w = \mu$. Now, we create a charge singlet 
\begin{equation}
|\psi_0\rangle = \frac{1}{\sqrt{2}}\left(|2_+, 2_-\rangle + |2_-, 2_+\rangle\right)
\end{equation}
and take the second charge in the pair around the flux. The effect of the flux is the group action of $w = \mu$ on the $[2]$ charge basis states (see section \ref{rep-theory-s3}):
\begin{equation}
    |\psi\rangle = \frac{1}{\sqrt{2}}\left(\omega^*|2_+, 2_-\rangle + \omega|2_-, 2_+\rangle\right)
\end{equation}
When we bring the two charges back together, they can only annihilate if their total charge is trivial; the probability this happens will be given by the overlap with the original $[2]$ charge singlet $|\psi_0\rangle$:
\begin{equation}
    |\langle \psi_0 | \psi \rangle|^2 = \frac{1}{4} |\omega+\omega^*|^2 = \frac{1}{4}
\end{equation}
The probability that the two charges combine to give a $[-]$ charge is given by the overlap with the total $[-]$ state $|\psi_-\rangle = \frac{1}{\sqrt{2}}\left(|2_+, 2_-\rangle - |2_-, 2_+\rangle\right)$:
\begin{equation}
    |\langle \psi_- | \psi \rangle|^2 = \frac{1}{4} |\omega-\omega^*|^2 = \frac{3}{4}
\end{equation}
As an aside, the fact we have some probability of getting a nontrivial remnant charge when we combine our two charges at the end of the protocol indicates that charge has been transferred to the flux; we know total charge must be conserved globally, and so the flux must pick up charge to account for the loss of the charge singlet. In general, we can find the probability of charge transfer by looking directly at the character of the representation of the flux, which is discussed in more detail in appendix \ref{ap:charge_transfer}. 

Since charge transfer is not deterministic, there is a $1/4$ probability of false negatives.  We can decrease our error rate by repeating the measurement procedure a few times; the chance of never getting a remnant particle even if we have a $C_3$ flux will be $(1/4)^N$ after $N$ repetitions. We only need to go to about $N=4$ repetitions to get our error rate below $0.5\%$.

This procedure distinguishes anyon types, but not internal states. If we were to use a $C_3$ flux with a different internal state, say $w = \overline{\mu}$, these probabilities would be the same:
\begin{equation}
\begin{aligned}
    &|\psi_1\rangle = \frac{1}{\sqrt{2}}\left(\omega|2_+, 2_-\rangle + \omega^*|2_-, 2_+\rangle\right)\\
    &|\langle \psi_0 | \psi_1 \rangle|^2 = \frac{1}{4} |\omega+\omega^*|^2 = \frac{1}{4}\\
    &|\langle \psi_- | \psi_1 \rangle|^2 = \frac{1}{4} |\omega-\omega^*|^2 = \frac{3}{4}
\end{aligned}
\end{equation}

Let us consider the same procedure, but with a $C_2$ flux instead; this will be used for dual basis measurement detailed in Sec. \ref{subsubsec:dual-measure}. The two charges will never annihilate when we bring them back together; unlike with the $C_3$ fluxes, there is always charge transfer. The action of $w = \sigma$ on the $[2]$ states is like the Pauli $X$:
\begin{equation}
    \begin{aligned}
        &\sigma |2_+\rangle = |2_-\rangle\\
        &\sigma |2_-\rangle = |2_+\rangle
    \end{aligned}
\end{equation}
So we see that the action of the flux on our initial state is:
\begin{equation}
    |\psi_1\rangle = \frac{1}{\sqrt{2}}\left(|2_+, 2_+\rangle + |2_-, 2_-\rangle\right)
\end{equation}
Taking the overlap of this state with the charge singlet or $[-]$ state will give $0$; however, we have a $50\%$ chance of ending up in the $|2_+\rangle$ or $|2_-\rangle$ state. By repeating the charge winding process many times and seeing the statistics of various remnant particles, we can determine what type of flux we have (trivial, $C_2$, or $C_3$).

\subsubsection{Computational Basis Measurements}\label{subsubsec:direct-measure}

We saw from the previous section that there is a good probability of charge transfer if we braid one half of a neutral $[2]$ charge pair around a $C_3$ flux. When we try to fuse the two $[2]$ charges after this procedure, they will have a $\frac{1}{4}$ chance of returning to the vacuum, but a $\frac{3}{4}$ chance of leaving a remnant $[-]$ charge. On the other hand, if we have a trivial flux particle, the same braiding procedure will have no charge transfer and the charge pair will always be able to fuse. This allows us to distinguish between $C_3$ flux and trivial flux, which in turn allows us to measure in the computational basis.

The procedure is as follows: consider two computational basis qutrits, $\ket{a}$, $\ket{b}$. If we group together one flux from each of the qutrits, the resulting anyon will have trivial flux only if $a = b$. Otherwise, it will have a $C_3$-valued flux. We then perform the flux measurement procedure; from the previous subsection, we see that the two qutrits are guaranteed to be different if we measure a remnant $[-]$ charge. The input state is projected onto the orthogonal subspace to that basis state. This occurs at a probability of $(3/4)^n$, where $n$ is the number of repetitions. To do a projective measurement onto the computational basis state, we initialize 3 reference states: $\ket{0},\ket{1},\ket{2}$. We repeat the flux measurement procedure for all three computational basis states $\{\ket{0},\ket{1},\ket{2}\}$ (which we will refer to as reference states from here on) until we obtain remnant $[-]$ charge from two of the reference states. Then the state is guaranteed to be in the third reference state, i.e. we have projectively measured the logical state to be in the third computational basis state.

\subsubsection{Measuring Charge (Dual Basis Measurement)}\label{subsubsec:dual-measure}

In the last section, we outlined a procedure for determining the flux of a given particle by measuring its effect on calibrated charges, which allows us to measure in the computational basis. In this section, we will describe dual basis measurements, which are equivalent to measuring the \emph{charge} of a given particle.

Consider a $C_2$ flux singlet state:
\begin{equation}
    |\tilde0\rangle = \frac{1}{\sqrt{3}}(|\sigma, \sigma\rangle + |\bar \mu \sigma, \bar \mu \sigma\rangle + |\mu \sigma, \mu \sigma\rangle)
\end{equation}
Similarly to the flux measurement procedure, we take one half of this singlet around the charge we want to measure, then we fuse the fluxes. The probability of getting various remnant particles helps characterize the charge.

Let's take the charge to be in the $|2_+\rangle$ state. Then the global state of our system is:
\begin{equation}
    |\psi\rangle = |\tilde0\rangle |2_+\rangle = \frac{1}{\sqrt{3}}(|\sigma, \sigma\rangle + |\bar \mu \sigma, \bar \mu \sigma\rangle + |\mu \sigma, \mu \sigma\rangle) |2_+\rangle
\end{equation}
Winding one of the $C_2$ fluxes around the charge gives:
\begin{equation}
    \begin{aligned}
        &\sigma |2_+\rangle = |2_-\rangle\\
        &\bar \mu \sigma |2_+\rangle = \omega|2_-\rangle\\
        &\mu \sigma |2_+\rangle = \omega^*|2_-\rangle
    \end{aligned}
\end{equation}
The resulting total state is:
\begin{equation}
    |\psi_f\rangle = \frac{1}{\sqrt{3}}(|\sigma, \sigma\rangle + \omega|\bar \mu \sigma, \bar \mu \sigma\rangle + \omega^*|\mu \sigma, \mu \sigma\rangle) |2_-\rangle = \ket{\tilde 1}\ket{2_-}
\end{equation}
The end state for the pair of fluxes now has nontrivial charge; it transforms as $|2_-\rangle$. This is straightforward to demonstrate. Winding a $\mu$ flux around $\ket{\tilde 1}$ maps the state to:
\begin{equation}
    \mu|\tilde 1\rangle = \frac{1}{\sqrt{3}}(|\bar \mu \sigma, \bar \mu \sigma\rangle + \omega|\mu \sigma, \mu \sigma\rangle + \omega^*|\sigma, \sigma\rangle) = \omega^* \ket{\tilde 1}
\end{equation}
which is exactly as we expect for a $|2_-\rangle$ charge.\footnote{We have used the relation $\mu\sigma = \sigma \bar \mu$, and $\mu^2 = \overline{\mu}$ in the above manipulations.} So if the anyon we are measuring carries $[2]$ charge, some of that charge will be transferred to the test flux pair, and the pair will be unable to fully fuse to the vacuum. If the particle to be measured has trivial charge, no charge transfer will occur, and the charges will fuse back to the vacuum.

This charge measurement procedure is almost the same as the flux measurement process, but with the roles of flux and charge swapped. We can use this measurement technique to perform dual basis measurements: we know that the $\ket{\Tilde{0}}$ state carries trivial total charge, while the $\ket{\Tilde{1}}$, $\ket{\Tilde{2}}$ states carry $[2]$ charge. Using a $C_2$ flux singlet, we can distinguish between these two cases. Note that here, there is no chance of a false negative--- no need to repeat the charge measurement process!

In the language of the following sections, we will call this measurement procedure ``comparing'' a $C_2$ flux pair $\ket{a}$ with $\ket{\Tilde{0}}$. If we get a yes (probe fluxes fuse to the vacuum), then $\ket{a}$ has been projected onto $\ket{\Tilde{0}}$. If we get a no (probe fluxes left a remnant charge), then $\ket{a}$ has been projected coherently onto the $\ket{\Tilde{1}}$ and $\ket{\Tilde{2}}$ subspace. We can compare with the other dual basis states using the nearly the same procedure. Suppose we would like to compare with the $\ket{\Tilde{1}}$ state; we can first apply the qutrit $\mathcal{Z}$ gate (or clock operator) to our input state twice. This will turn $\ket{\Tilde{1}}$ into $\ket{\Tilde{0}}$--- we then proceed with our $\ket{\Tilde{0}}$ comparison and apply a single $\mathcal{Z}$ at the end to reverse the first $\mathcal{Z}^2$. This works for $\ket{\Tilde{2}}$ as well if we reverse the roles of $\mathcal{Z}$ and $\mathcal{Z}^2$. We will explain the implementation of the $\mathcal{Z}$ gate in the next section, which only uses the braiding of fluxes.

\subsection{Flux Braiding}\label{flux-braiding}

We have described the procedures for measuring the flux or charge of an anyon. Using our particular qutrit encoding, these two protocols correspond to measurement in the computational and dual bases, respectively. The other fundamental tool in our toolkit is simply braiding two qutrits. 

The most basic gate realized with braiding we will call the ``pull-through'' gate. It will form the basis for many other entangling and control gates in our full gate set (both qutrit and qubit). 

\begin{defi}[Pull-through gate $U$]
\label{def:pull-through}
\begin{equation}
    U\ket{a,b} = \ket{a,-a-b}
\end{equation}

where $a,b$ are in the qutrit basis, so there is an implicit mod $3$ (the same goes for any following gates). 

\end{defi}

\begin{figure}[h!]
\centering
\includegraphics[width=1\textwidth]{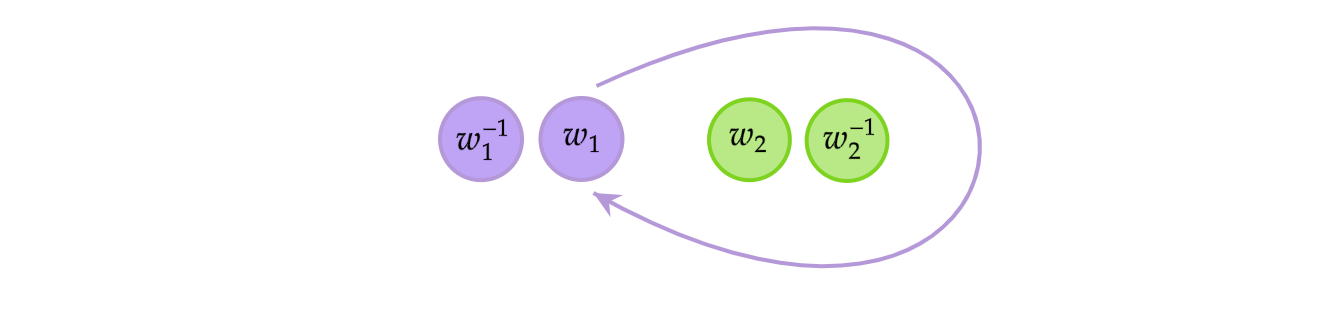}
\caption{Graphical depiction of the pull-through gate described in definition \ref{def:pull-through}. The gate is implemented by braiding one member of the pair on the left completely around both halves of the pair on the right.}
\label{fig:pull-through}
\end{figure}

To implement the pull-through gate, we use two vortex pairs, $\ket{w_1,\overline{w}_1,w_2,\overline{w}_2}$, and wind $w_1$ around the $w_2,\overline{w_2}$ vortex pair. This conjugates all three particles by their total flux $w=w_1w_2\overline{w}_2=w_1$ (the $\overline{w}_1$ particle is not changed, since it doesn't participate in the winding):
\begin{align*}
U\ket{w_1,\overline{w}_1,w_2,\overline{w}_2}
&=\ket{w_1w_1\overline{w}_1,\overline{w}_1, w_1w_2\overline{w}_1,w_1\overline{w}_2\overline{w}_1}\\
&=\ket{w_1,\overline{w}_1,w_1w_2\overline{w}_1,w_1\overline{w}_2\overline{w}_1}
\end{align*}

Plugging in our definitions for the computational basis, we find that:
\begin{equation}
    U\ket{a, b} = \ket{a, -a-b}
\end{equation}

We can construct variations of the pull-through gate by using ancillas; the first, the $U_+$ gate, is the qutrit generalization of the CNOT gate:
\begin{equation}
    U_+\ket{a,b}=\ket{a,b+a}
\end{equation}

\begin{tcolorbox}[enhanced, breakable,pad at break*=1mm, attach boxed title to top center={yshift=-3mm,yshifttext=-1mm}, colback=blue!5!white, colframe=blue!75!black, colbacktitle=blue!75!black, title=$U_+$ gate, fonttitle=\bfseries, boxed title style={size=small,colframe=blue!65!black}]

\begin{center}
    \begin{tikzpicture}
    \node[scale=1]{ 
        \begin{quantikz}[thin lines] 
            \lstick{$\ket{a}$} &  \ctrl{1} & \qw & \rstick{$\ket{a}$}\qw \\
            \lstick{$\ket{b}$} & \gate{U} & \gate{U} & \rstick{$\ket{a+b}$}\qw \\
            \lstick{$\ket{0}$} & \qw & \ctrl{-1} & \rstick{$\ket{0}$}\qw
        \end{quantikz}
    };
    \end{tikzpicture}
\end{center}

$U_+$ can be implemented by using the $U$ gate with an ancilla qubit:

\begin{enumerate}
    \item Start with the state $\ket{a,b,0}$.
    \item Apply $U$ gate to qubit 1\&2: $U_{12}\ket{a,b,0}=\ket{a,-a-b,0}$.
    \item Apply $U$ gate to qubit 3\&2: $U_{32}\ket{a,-a-b,0}=\ket{a,a+b,0}$.
    \item Discard the 3rd qubit to get the state $\ket{a,a+b}$.
\end{enumerate}

where $U_{ij}$ means that $i$ is the control qubit, and $j$ is the target qubit. Note that if we fix the second input to $b=1$, $U_+$ becomes the $\mathcal{X}$ gate (shift operator).

\end{tcolorbox}

\vspace{5mm}

We can also construct the inverse of $U_+$, the $U_-$ gate:
\begin{equation}
    U_-\ket{a,b}=\ket{a,b-a}
\end{equation}

\begin{tcolorbox}[enhanced, breakable,pad at break*=1mm, attach boxed title to top center={yshift=-3mm, yshifttext=-1mm}, colback=blue!5!white, colframe=blue!75!black, colbacktitle=blue!75!black, title=$U_-$ gate, fonttitle=\bfseries, boxed title style= {size=small,colframe=blue!65!black}]

\begin{center}
    \begin{tikzpicture}
    \node[scale=1]{ 
        \begin{quantikz}[thin lines] 
            \lstick{$\ket{a}$} & \ctrl{2} & \qw & \ctrl{2} &\rstick{$\ket{a}$}\qw \\
            \lstick{$\ket{b}$} & \qw & \gate{U_+} & \qw &\rstick{$\ket{b - a}$}\qw \\
            \lstick{$\ket{0}$} & \gate{U} & \ctrl{-1} & \gate{U} &\rstick{$\ket{0}$}\qw
        \end{quantikz}
    };
    \end{tikzpicture}
\end{center}

$U_-$ can be implemented by using the $U$ and $U_+$ gates with an ancilla qubit:
\begin{enumerate}
    \item Start with the state $\ket{a,b,0}$.
    \item Apply $U$ gate to qubit 1\&3: 
    $U_{13}\ket{a,b,0}=\ket{a,b,-a}$.
    \item Apply $U_+$ gate to qubit 2\&3: 
    $U_{+,32}\ket{a,b,-a}=\ket{a,b-a,-a}$.
    \item Disentangle the output state and the ancilla with another $U$ gate between qubits 1\&3:
    $U_{13}\ket{a,b-a,-a}=\ket{a,b-a,0}$.
\end{enumerate}

where $U_{ij}$ means that $i$ is the control qubit, and $j$ is the target qubit.

\end{tcolorbox}

We can construct a $\mathcal{Z}$ gate (the qutrit version of a qubit Pauli $Z$) using $U_-$:

\begin{equation}
    \mathcal{Z}\ket{a}=\omega^a\ket{a}
\end{equation}

\begin{tcolorbox}[enhanced, breakable, pad at break*=1mm, attach boxed title to top center={yshift=-3mm, yshifttext=-1mm}, colback=blue!5!white, colframe=blue!75!black, colbacktitle=blue!75!black, title= $\mathcal{Z}$ (qutrit) gate, fonttitle=\bfseries, boxed title style= {size=small,colframe=blue!65!black}]

\begin{center}
    \begin{tikzpicture}
    \node[scale=1]{ 
        \begin{quantikz}
            \lstick{$\ket{\psi}$} & \ctrl{1} & \rstick{$\mathcal{Z}\ket{\psi}$}\qw\\
            \lstick{$\ket{\Tilde{1}}$} & \gate{U_-} & \rstick{$\ket{\Tilde{1}}$}\qw
        \end{quantikz}
    };
    \end{tikzpicture}
\end{center}

$\mathcal{Z}$ can be implemented using the $U_-$ gate with a special ancilla qubit $\ket{\tilde 1}$:

\begin{enumerate}
    \item Start with the state $\ket{a,\tilde 1}=\frac{1}{\sqrt{3}}\left(\ket{a,0}+\omega\ket{a,1}+\bar\omega\ket{a,2}\right)$.
    \item Apply $U_-$ gate: $U_-\ket{a,\tilde 1} =\frac{1}{\sqrt{3}}\left(\ket{a,-a}+\omega\ket{a,1-a}+\bar\omega\ket{a,2-a}\right)$. \\
    
    Note that
    \begin{equation*}
        U_-\ket{a,\tilde 1} = \omega^a |a\rangle \frac{1}{\sqrt{3}}\big(\omega^{-a}|-a\rangle + \omega^{1-a}|1-a\rangle + \omega^{2-a}|2-a\rangle\big)
    \end{equation*}
    Consider the result given different values of $a$. If $a=0$, we see that the $U_-$ gate leaves the state invariant:
    \begin{equation*}
        U_- |0, \Tilde{1}\rangle = |0, \Tilde{1}\rangle= \omega^0|0, \Tilde{1}\rangle
    \end{equation*}
    If $a=1$, we find:
    \begin{equation*}
        U_-|1, \Tilde{1}\rangle = \omega |1\rangle \frac{1}{\sqrt{3}}\big(\omega^*|2\rangle + |0\rangle + \omega|1\rangle\big) = \omega|1, \Tilde{1}\rangle
    \end{equation*}
    and if $a=2$ we have:
    \begin{equation*}
        U_-|2, \Tilde{1}\rangle = \omega^* |2\rangle \frac{1}{\sqrt{3}}\big(\omega|1\rangle + \omega^*|2\rangle + |0\rangle\big) = \omega^*|2, \Tilde{1}\rangle = \omega^2|2, \Tilde{1}\rangle
    \end{equation*}

    \item From the above results, we see we obtain the right results simply by discarding the ancilla $|\Tilde{1}\rangle$ states.
\end{enumerate}

\end{tcolorbox}

\subsection{Universal Qubit Operations} \label{universal-qubit-set}

We have developed a basic set of qutrit state initialization, gates, and measurement protocols. These are enough to build up a universal gate set of qutrit gates\cite{mochon_anyon_2004}. However, it turns out to be simpler to demonstrate a set of universal \emph{qubit} gates built from our basic qutrit operations instead. We will build a universal gate set out of  Cliffords + one non-Clifford gate. We use the fact that $\{H,S,CZ\}$ generates the Clifford gate set (where $H$ is the Hadamard gate, $S=\text{diag}(1,i)$ is the phase gate, and $CZ$ is the control-Z gate\footnote{The more common gate set for demonstrating universality involves $CX$ as a generator of the Clifford gates; but replacing $CX$ with $CZ$ is another Clifford gate set since $H(CZ)H=CX$}.). The CCZ gate, which is a non-Clifford gate, completes the universal gate set. The dependency of the gates is shown in figure \ref{fig:enter-label}.

\begin{figure}[h!]
    \centering
    \includegraphics[width=0.6\textwidth]{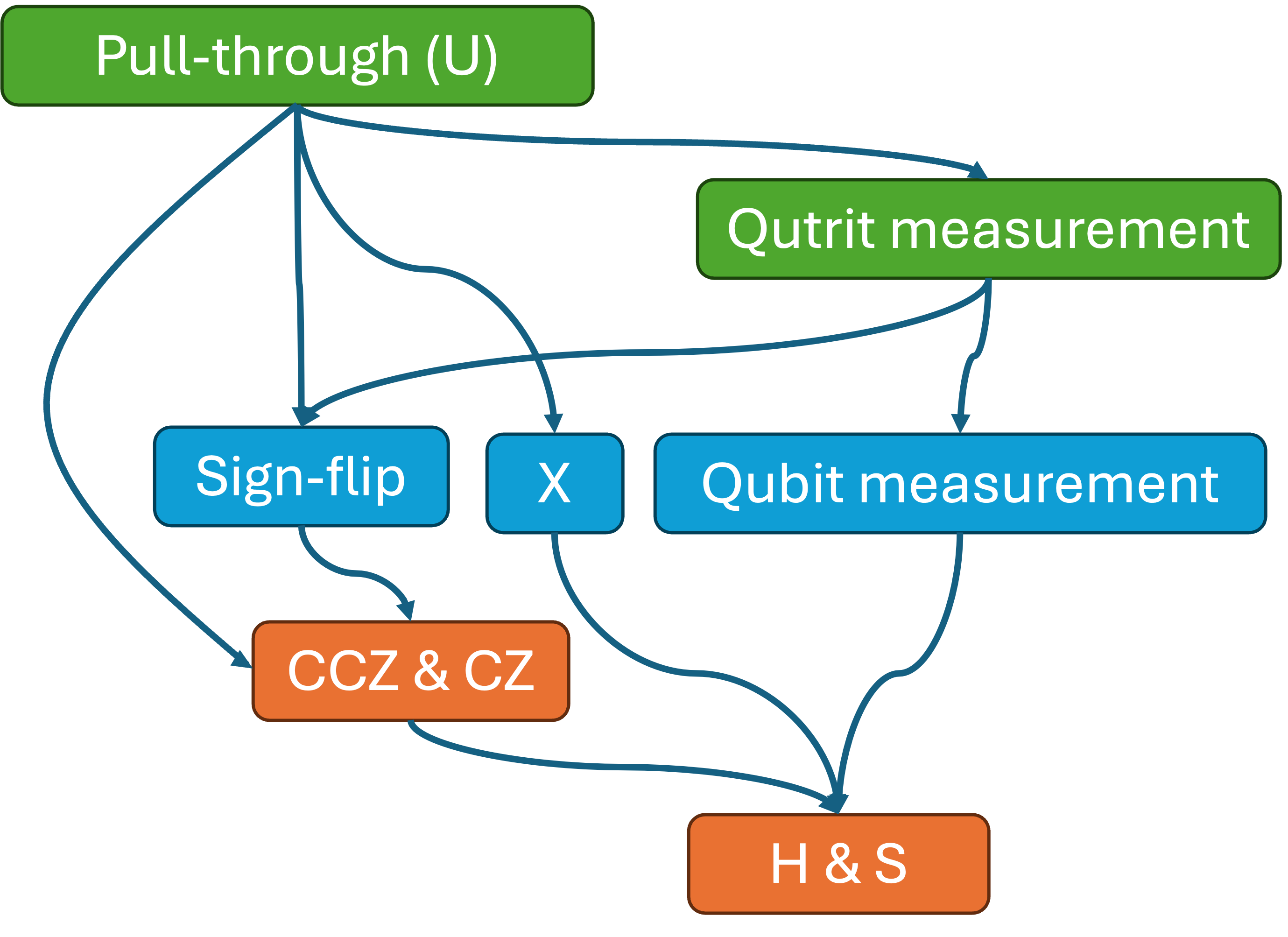}
    \caption{Dependency of the universal qubit gate set implementation. The foundational operations are the pull-through gate (U) and the qutrit measurement in the $\mathcal Z$ and $\mathcal X$ basis. Arrows indicate gates implemented based on the foundational operations.}
    \label{fig:enter-label}
\end{figure}

\subsubsection{Qubit Measurements}

 One necessary part of our qubit gate set is the ability to measure in either the qubit computational basis $\qty{\ket 0,\ket 1}$, or the qubit dual basis $\qty{\ket +, \ket -}$. The ability to measure in the qubit computational basis follows directly from being able to measure in the \emph{qutrit} computational basis, because they are really the same basis (for qubits we just ignore $\ket{2}$). 
 
 Measurement in the qubit dual basis is a little more complicated. The procedure is as follows:

 \begin{enumerate}
     \item Given a general state $\ket\psi=\alpha \ket+ + \beta\ket-$, we first compare with $\ket{\Tilde{0}}$. 
        \begin{itemize}
            \item If ``yes'', then we have measured the $\ket +$, since $\ket-$ state has zero overlap with $\ket{\Tilde{0}}$. 
            \item If ``no'', then we have effectively projected $\ket \psi$ to the orthogonal subspace of $\ket{\Tilde{0}}$: 
            $$
            \ket{\psi'} = (1- \ketbra{\Tilde{0}})\ket\psi = \alpha\ket{+'}+ \beta\ket{-'}
            $$
            where $\ket{+'}=\frac{1}{\sqrt 6}\qty(\ket0+\ket1-2\ket2)$, $\ket{-'}=\ket{-}$.
        \end{itemize}
    \item Then we compare with $\ket2$. 
        \begin{itemize}
            \item If ``yes'', then the state has to be $\ket{+'}$, since $\ket{-}$ has zero overlap with $\ket2$. 
            \item If ``no'', then the state is projected to $$\alpha\ket{+''}+\beta\ket{-''}$$
            where
            $\ket{+''} =(1- \ketbra{2}) \ket{+'} = \frac{1}{\sqrt 2}(\ket0+\ket1)=\ket+$ and $\ket{-''}=\ket-$.
        \end{itemize}
 \end{enumerate}

In general, the comparison result will be of the form ``no'',``no'',... and possibly end with a ``yes''. Whenever we obtain ``yes'' (at either comparison with $\ket{\tilde 0}$ or with $\ket2$), we can conclude we have measured the $\ket+$ state; whereas if we obtain ``no'', the state could be either the $\ket+$ or $\ket-$ state. So we need to repeat the measurement until we reach the desirable level of accuracy. At each step (which includes both comparisons), the probability of measuring at least one ``yes'' is $\frac{2}{3}+\frac{1}{3}\times\frac{2}{3}=\frac{8}{9}$. So up to the $n$th repetition, we will obtain the $\ket+$ state with probability $(1-\frac{1}{9}^n)|\alpha|^2$, and the $\ket-$ state with probability  $|\beta|^2 +\frac{1}{9}^n|\alpha|^2$. To achieve an better than 99\% accuracy in the measurement result, we just need to repeat for $n=3$.

\subsubsection{Qubit state initialization}

We will need the $\ket+$ state as an ancilla state in the implementation of some the qubit gates below, so we discuss here how to prepare it. To prepare a $\ket+$ state, we can \textit{repeat until success}, starting with a $\ket{\Tilde{0}}$ state and projecting the $\ket{2}$ component out using a comparison. The procedure is as follows: 
        
         \begin{enumerate}
             \item Start with a $\ket{\Tilde{0}}$ state.
             \item Compare with a $\ket{2}$ ancilla; if the result is ``no'' (with probability $\frac{2}{3}$) we have successfully prepared a $\ket{+}$ state, since $\ket{\Tilde{0}} = \frac{1}{\sqrt{3}} (\ket{0} + \ket{1} + \ket{2})$ without the $\ket{2}$ component (and fixing normalization) is $\ket{+}$. If the result is ``yes'' then we have to start over with a new $\ket{\Tilde{0}}$ state.
         \end{enumerate}

The probability of successfully preparing a $\ket{+}$ state is $1-\frac{1}{3}^n$, where $n$ is the number of rounds. To achieve a success rate of 99\%, we just need to repeat for $n=5$.

\subsubsection{Qubit Gates}

We want to construct a universal set of \emph{qubit} gates, using our underlying qutrit degrees of freedom. We will start by constructing a generalized version of the qubit Pauli $Z$ gate, which we call the ``sign-flip gate''. Its action on the computational basis states is given by
\begin{align}
    \sigma^z_j \ket{i} = (-1)^{\delta_{i,j}}\ket{i}
\end{align}

Given an input qutrit state $\ket{\psi}=c_0\ket{0}+c_1\ket{1}+c_2\ket{2}$, the gate $\sigma^z_j$ flips the sign of the coefficient $c_i$. Note that $\sigma^z_{1}$ is the same as the usual Pauli $Z$ operator in the qubit subspace spanned by $\{\ket0,\ket1\}$. The sign-flip gate will be implemented with the help of a special ancilla:
\begin{equation}
    |\xi\rangle = \frac{1}{\sqrt{3}}\big(|0\rangle - |1\rangle + |2\rangle\big)
\end{equation}
This is a qutrit ``magic-state''---it cannot be prepared solely via qutrit Clifford operations. In particular, we need both the qubit $\ket{+}$ state and measurement in the dual basis to prepare $\ket{\xi}$. 
\vspace{2.5mm}

\begin{tcolorbox}[enhanced, breakable, pad at break*=1mm, attach boxed title to top center={yshift=-3mm, yshifttext=-1mm}, colback=blue!5!white, colframe=blue!75!black, colbacktitle=blue!75!black, title= Sign-flip gate, fonttitle=\bfseries, boxed title style= {size=small,colframe=blue!65!black}]

\begin{center}
    \begin{tikzpicture}
    \node[scale=1]{ 
        \begin{quantikz}
            \lstick{$\ket{\psi}$} & \ctrl{1} & \qw & \rstick{$\sigma_{(i+2)}^z \ket{\psi}$} \qw \\
            \lstick{$\ket{\xi}$} & \gate{U_+} & \meter{$i=0,1,2$} \arrow[r] & \rstick{$\ket{i}$}
        \end{quantikz}
    };
    \end{tikzpicture}
\end{center}
\tcbsubtitle[before skip=\baselineskip]{Magic state preparation}
\begin{center}
    \begin{tikzpicture}
    \node[scale=1]{ 
        \begin{quantikz}
            \lstick{$\ket{+}$} & \gate{\mathcal Z} & \ctrl{1} & \meter{$\ket{\Tilde{0}}$} \arrow[r] & \rstick{$\mathrm{yes}$}\\
            \lstick{$\ket{+}$} & \gate{\mathcal Z^2} & \gate{U_+} & \qw & \rstick{$\ket{\xi}$}\qw
        \end{quantikz}
    };
    \end{tikzpicture}
\end{center}
We can construct $\ket{\xi}$ using the above circuit, starting with two copies of $|+\rangle$:
\begin{enumerate}
    \item Apply $\mathcal Z$ and $\mathcal Z^2$ to the two $|+\rangle$ copies:
    \begin{equation}
        \mathcal Z|+\rangle \otimes \mathcal Z^2|+\rangle = \frac{1}{\sqrt{2}}(|0\rangle + \omega|1\rangle) \otimes \frac{1}{\sqrt{2}}(|0\rangle + \omega^*|1\rangle) \equiv |\eta\rangle
    \end{equation}
where
    \begin{equation}
        \ket{\eta}=\frac{1}{2}(|00\rangle + \omega|10\rangle + \omega^*|01\rangle + |11\rangle)
    \end{equation}

    \item Apply $U_+$ to $|\eta\rangle$ with the first ancilla as the control:
    \begin{equation}
        \begin{aligned}
        U_+|\eta\rangle &= U_+\frac{1}{2}(|00\rangle + \omega|10\rangle + \omega^*|01\rangle + |11\rangle)\\
        &= \frac{1}{2}(|00\rangle + \omega|11\rangle + \omega^*|01\rangle + |12\rangle)
        \end{aligned}
    \end{equation}
    
    \item Compare the first qutrit with $|\Tilde{0}\rangle$ and post-select on a ``yes'' so the state of the second qutrit is now:
    \begin{equation}
        \begin{aligned}
        (\langle\Tilde{0}|\otimes\mathbb{I}) U_+|\eta\rangle &= (\langle\Tilde{0}|\otimes\mathbb{I})\frac{1}{2}(|00\rangle + \omega|11\rangle + \omega^*|01\rangle + |12\rangle)\\
        &= \frac{1}{2} \braket{\tilde 0}{0} (\ket 0 + \bar\omega\ket 1) + \frac{1}{2} \braket{\tilde 0}{1} (\omega \ket 1 + \ket 2)\\
        &=\frac{1}{2}\frac{1}{\sqrt{3}}(\ket 0 -\ket 1 + \ket 2) \\
        &= \frac{1}{2}\ket{\xi}
        \end{aligned}
    \end{equation}

The proportionality constant is the probability amplitude, so the probability of successfully obtaining the state $\ket \xi$ is $(1/2)^2 = 1/4$.

\end{enumerate}

\tcbsubtitle[before skip=\baselineskip]{Gate implementation}

Now that we have the needed ancilla, we can construct the sign flip gates. 

\begin{enumerate}
\item We start with a superposition of computational basis states: 
\begin{equation}
    |\psi\rangle = c_0|0\rangle + c_1|1\rangle + c_2|2\rangle
\end{equation}

\item Using the ancilla $|\xi\rangle$, we apply $U_+$:
\begin{equation}
    \begin{aligned}
    U_+ \ket\psi \otimes \ket\xi
    &= \frac{1}{\sqrt{3}}(c_0\ket{00} - c_0\ket{01} + c_0\ket{02} + c_1\ket{11} - c_1\ket{12} + c_1\ket{10} \\
    &\qquad \qquad \qquad \qquad + c_2\ket{22}-c_2\ket{20}+c_2\ket{21})\\
    &= \frac{1}{\sqrt{3}} ( \qty(c_0\ket0 + c_1\ket1-c_2\ket2)\ket0 +\qty(-c_0\ket0 + c_1\ket1+c_2\ket2)\ket{1}\\
    & \qquad \qquad \qquad \qquad  + (c_0\ket0 - c_1\ket{1}+c_2\ket2)\ket{2} )
\end{aligned}
\end{equation}

\item Now measure the second qutrit in the computational basis. 
\begin{itemize}
    \item If we compare with $|0\rangle$ and get a positive result then we have flipped $c_2$.
    \item If we compare with $|1\rangle$ and get a positive result then we have flipped $c_0$.
    \item If we compare with $|2\rangle$ and get a positive result then we have flipped $c_1$.
\end{itemize}
\end{enumerate}

Note that this sign flip gate acts as the normal Pauli $Z$ operator when we restrict to the space of states spanned by $\ket{0}, \ket{1}$ and get a yes result when comparing with $\ket{2}$. 

\end{tcolorbox}

\vspace{2.5mm}

We are not guaranteed to successfully implement the sign-flip gate that we want--- the correct output is contingent on getting a ``yes'' when we compare with a certain computational basis state. This is slightly concerning, as we do not want to have to post-select on a long string of certain measurement outcomes if we have a circuit that involves many uses of the sign-flip gate. Luckily, we do not have to resort to post-selection; instead, we can ``repeat until success'' as with the state preparation protocols. This involves repeating the gate as many times as necessary until we achieve a sequence of measurement results that cancel out in some sense to give the correct overall output. It turns out, to have a $99\%$ success rate, we need to repeat the gate around 35 times--- see appendix \ref{ap:repeat} for the details. 

The first purely qubit gate we will construct is Pauli $X$. We can construct $X$ without resorting to measurements, although we will take advantage of the qutrit $\ket{2}$ state to do some useful counting for us. As usual, the qubit $X$ gate is defined:
\begin{equation}
    X(a\ket{0} + b\ket{1}) = a\ket{1} + b\ket{0}
\end{equation}

\vspace{2.5mm}

\begin{tcolorbox}[enhanced, breakable, pad at break*=4mm, attach boxed title to top center={yshift=-3mm, yshifttext=-1mm}, colback=blue!5!white, colframe=blue!75!black, colbacktitle=blue!75!black, title= Qubit $X$ gate, fonttitle=\bfseries, boxed title style= {size=small,colframe=blue!65!black}]

\begin{center}
    \begin{tikzpicture}
    \node[scale=1]{ 
        \begin{quantikz}
            \lstick{$\ket{\psi}$} & \ctrl{1} & \qw & \gate{U_+} & \ctrl{1} & \qw \\
            \lstick{$\ket{0}$} & \gate{U_+} & \gate{U_+} & \ctrl{-1} & \gate{U_+} & \qw \\
            \lstick{$\ket{1}$} & \qw & \ctrl{-1} & \qw & \qw & \qw
        \end{quantikz}
    };
    \end{tikzpicture}
\end{center}

We can construct a qubit X gate using the $U_+$ gates and two ancillas. The ancillas will help us deal with the leakage into the $\ket{2}$ state that results from using our qutrits to do qubit operations. To show that the above circuit gives the right outputs, we calculate the results explicitly for a general input:

\begin{enumerate}
    \item $U_+$ between the state and $\ket{0}$ ancilla, the state as control:
        \begin{equation}
            \begin{aligned}
             U_+\ket{\psi}\ket{0} &= U_+(a\ket{0}+b\ket{1})\ket{0}\\
             &=a\ket{00} + b\ket{11}
            \end{aligned}
        \end{equation}
    \item $U_+$ between the two ancillas, the second ancilla as the control:
        \begin{equation}
             U_+(a\ket{00} + b\ket{11})\ket{1} = a\ket{011} + b\ket{121}
        \end{equation}
    \item $U_+$ between the state and first ancilla again, but with the ancilla as the control:
        \begin{equation}
            U_+(a\ket{011} + b\ket{121}) = a\ket{111} + b\ket{021}
        \end{equation}
    \item $U_+$ between the state and first ancilla a third time, but with the state as the control again. This will remove the entanglement between the state and the ancillas.
        \begin{equation}
            \begin{aligned}
                U_+(a\ket{111} + b\ket{021}) &= a\ket{121}+b\ket{021}\\
                &= (a\ket{1}+b\ket{0})\ket{21}\\
                &= (X\ket{\psi})\ket{21}
            \end{aligned}
        \end{equation}
\end{enumerate}

\end{tcolorbox}

\vspace{2.5mm}

Two more gates are necessary to complete the gate set. The first is the $CZ$ gate. It acts on qubit states in the following way:
\begin{equation}
CZ \ket{x,y} = (-1)^{xy}\ket{x,y}
\end{equation}
Notice that only when $x=y=1$, there is an overall sign flip.

\vspace{2.5mm}

\begin{tcolorbox}[enhanced, breakable, pad at break*=1mm, attach boxed title to top center={yshift=-3mm, yshifttext=-1mm}, colback=blue!5!white, colframe=blue!75!black, colbacktitle=blue!75!black, title= $CZ$ gate, fonttitle=\bfseries, boxed title style= {size=small,colframe=blue!65!black}]

\begin{center}
    \begin{tikzpicture}
    \node[scale=1]{ 
        \begin{quantikz}
            \lstick[2]{$\ket{\psi}$} & \ctrl{2} & \qw & \qw & \qw & \ctrl{2} & \qw \\
            & \qw & \ctrl{1} & \qw & \ctrl{1} & \qw & \qw\\
            \lstick{$\ket{0}$} & \gate{U_+} & \gate{U_+} & \gate{\sigma_{(2)}^z} & \gate{U_-} & \gate{U_-} & \qw
        \end{quantikz}
    };
    \end{tikzpicture}
\end{center}

The first two $U_+$ gates compute the sum $x+y$ in the ancilla qubit; notice that only the $\ket{11}$ input state will lead to a $\ket{2}$ ancilla. A general input state $\ket{\psi} = (a_{00}\ket{00} + a_{01}\ket{01} + a_{10}\ket{10} + a_{11}\ket{11})\ket{0}$ will be mapped to:
\begin{equation}
    \ket{\psi} = a_{00}\ket{000} + a_{01}\ket{011} + a_{10}\ket{101} + a_{11}\ket{112}
\end{equation}
where the third qubit is the ancilla. The sign-flip gate then flips the sign of the last coefficient only:
\begin{equation}
    \ket{\psi} = a_{00}\ket{000} + a_{01}\ket{011} + a_{10}\ket{101} - a_{11}\ket{112}
\end{equation}
The second set of control gates will undo the entanglement between the ancilla and the input state, as they subtract $x+y$ from the ancilla state:
\begin{equation}
    \ket{\psi} = (a_{00}\ket{00} + a_{01}\ket{01} + a_{10}\ket{10} - a_{11}\ket{11})\ket{0}
\end{equation}
So we have achieved the right coefficients; we only have a sign flip when both input qubits are in the $\ket{1}$ state.

\end{tcolorbox}

\vspace{2.5mm}

The second is the $CCZ$ gate, which acts in the following way:
\begin{equation}
CCZ \ket{x,y,z} = (-1)^{xyz}\ket{x,y,z}
\end{equation}
Notice that only when $x=y=z=1$, there is an overall sign flip.

\vspace{2.5mm}

\begin{tcolorbox}[enhanced, breakable, pad at break*=1mm, attach boxed title to top center={yshift=-3mm, yshifttext=-1mm}, colback=blue!5!white, colframe=blue!75!black, colbacktitle=blue!75!black, title= $CCZ$ gate, fonttitle=\bfseries, boxed title style= {size=small,colframe=blue!65!black}]

\begin{center}

    \begin{tikzpicture}
    \node[scale=1]{ 
        \begin{quantikz}
            \lstick[3]{$\ket{\psi}$} & \gate{\sigma_{(1)}^z} & \ctrl{3} & \qw & \qw & \qw & \qw & \qw & \ctrl{3} & \qw \\
            & \gate{\sigma_{(1)}^z} & \qw & \ctrl{2} & \qw & \qw & \qw & \ctrl{2} & \qw & \qw\\
            & \gate{\sigma_{(1)}^z} & \qw & \qw & \ctrl{1} & \qw & \ctrl{1} & \qw & \qw & \qw \\
            \lstick{$\ket{0}$} & \qw & \gate{U_+} & \gate{U_+} & \gate{U_+} & \gate{\sigma_{(1)}^z} & \gate{U_-} & \gate{U_-} & \gate{U_-} & \qw
        \end{quantikz}
    };
    \end{tikzpicture}
    
\end{center}

This circuit is conceptually similar to the $CZ$ circuit. We first apply the Pauli $Z$ operator to each input qubit individually--- this gives an overall minus sign only if we have an odd number of nonzero inputs. We have almost achieved the gate we want; we do get an overall minus sign when all three inputs $x=y=z=1$. However, we also get a minus sign when only one input is $1$. \\

To cancel this unwanted sign flip, we need to compute the sum of the inputs in the ancilla, using the three $U_+$ gates. The outcome will be $0$ if $x=y=z=1$, but will be $1$ if only one of $x,y,z=1$. So we flip the sign of the ancilla if its value is $1$, which cancels the wrong sign flip from before. Finally, we undo the entanglement between the ancilla and the inputs with the $U_-$ gates. We are left with a gate that only gives a global minus sign when $x=y=z=1$, as desired.\footnote{Once we have $CCZ$, we also automatically get $CZ$ if we fix one of the three inputs to $1$. So the previous $CZ$ implementation was not strictly necessary to demonstrate universality; however, it is a less resource intensive method than using $CCZ$, and so would be more practical in a physical realization.}

\end{tcolorbox}

\vspace{2.5mm}

The $CCZ$ gate, because it is non-Clifford\footnote{see appendix \ref{ap:lambda3} for a demonstration of why $CCZ$ is non-Clifford.}, is key to completing our universal gate set.

We can now build the qubit Hadamard using $CZ$. The Hadamard gate acts on a generic state by:
\begin{equation}
H(\alpha\ket0 + \beta\ket1) = \frac{1}{\sqrt{2}}(\alpha+\beta)\ket0 + \frac{1}{\sqrt{2}}(\alpha-\beta)\ket1    
\end{equation}

\vspace{2.5mm}

\begin{tcolorbox}[enhanced, breakable, pad at break*=1mm, attach boxed title to top center={yshift=-3mm, yshifttext=-1mm}, colback=blue!5!white, colframe=blue!75!black, colbacktitle=blue!75!black, title= $H$ gate, fonttitle=\bfseries, boxed title style= {size=small,colframe=blue!65!black}]

\begin{center}
    \begin{tikzpicture}
    \node[scale=1]{ 
        \begin{quantikz}
            \lstick{$\ket{\psi}$} & \gate[2]{CZ} & \qw & \meter{$\ket{\pm}$} \arrow[r] & \rstick{$\ket{+}$} \\
            \lstick{$\ket{+}$} & & \qw & \qw & \rstick{$H\ket{\psi} \text{ or }XH\ket{\psi}$}\qw\\
        \end{quantikz}
    };
    \end{tikzpicture}
\end{center}

We can implement $H$ by using $\ket+$ as an ancilla state. Given a state $\ket\psi=\alpha\ket0+\beta\ket1$:

\begin{enumerate}
    \item Start with state $\ket\psi \ket+$
    \item Apply $CZ$
    $$CZ(\alpha\ket0 + \beta\ket1)\ket+ = \ket+\qty[\frac{\alpha+\beta}{2}\ket0+\frac{\alpha-\beta}{2}\ket1] + \ket-\qty[\frac{\alpha-\beta}{2}\ket0+\frac{\alpha+\beta}{2}\ket1]$$
    \item Measure the 1st qubit in the $X$ basis: if the output is $\ket+$, we have implemented $H$; if the output is $\ket-$, we have implemented $XH$ and need to apply another qubit $X$ to correct the output.
\end{enumerate}

\end{tcolorbox}

As the implementation of $H$ contains a $CZ$ gate, which in turn contains a sign-flip gate, it will also be necessary in general to ``repeat until success''--- to go through many repetitions of the $CZ$ gate until we have obtained a sequence of measurement outcomes that result in the bare gate. 

\vspace{2.5mm}

We have now demonstrated how to build the Pauli $X, Z$ gates, $CZ$, and the Hadamard gate $H$. We also have a non-Clifford gate: $CCZ$. To fully generate the Cliffords, we need the phase gate $S$, which acts by:
\begin{equation}
    S(\alpha\ket0+\beta\ket1) = \alpha\ket0+i\beta\ket1
\end{equation}
The implementation of the phase gate is very similar to the Hadamard gate, as we again need a special ancilla---in this case, the eigenstate state of $Y$. We discuss the a protocol for the preparation of an $Y$ eigenstate in Appendix \ref{ap:y-preparation}.

\begin{tcolorbox}[enhanced, breakable, pad at break*=1mm, attach boxed title to top center={yshift=-3mm, yshifttext=-1mm}, colback=blue!5!white, colframe=blue!75!black, colbacktitle=blue!75!black, title= $S$ gate, fonttitle=\bfseries, boxed title style= {size=small,colframe=blue!65!black}]

\begin{center}
    \begin{tikzpicture}
    \node[scale=1]{ 
        \begin{quantikz}
            \lstick[1]{$\ket{\psi}$} & \gate[2]{CZ} & \meter{$\ket{\pm}$} \arrow[r] & \rstick{$\ket{+}$} \\
            \lstick[1]{$\ket{-_Y}$} & \qw  & \qw &\rstick{$S\ket{\psi}\text{ or } XS\ket{\psi}$} \qw
        \end{quantikz}
    };
    \end{tikzpicture}
\end{center}

We use the $-1$ eigenstate of $Y$, $\ket{-_Y}$, as an ancilla qubit. For convenience, we fix the global phase of the $-1$ eigenstate of $Y$ such that the state $\ket{-_Y}=\frac{1+i}{2}\ket0+\frac{1-i}{2}\ket1$. Given any state $\ket{\psi}= \alpha\ket0 + \beta\ket1$, 

\begin{itemize}
    \item Start with the state $\ket\psi\ket{-_Y}$
    \item apply $CZ$ to the physical and ancilla qubit
        \begin{equation*}
            CZ(\alpha\ket0 + \beta\ket1) (\frac{1+i}{2}\ket0+\frac{1-i}{2}\ket1)
            = \frac{1}{\sqrt{2}}(\alpha\ket0 + i\beta\ket1)\ket+ + \frac{1}{\sqrt{2}}(i\alpha\ket0 + \beta\ket1)\ket-
        \end{equation*}
    \item Measure the 1st qubit in the $X$ basis: if the output is $\ket+$, we have implemented $S$; if the output is $\ket-$, we have implemented $XS$ and need to apply another qubit $X$ to correct the output (up to a global phase).
\end{itemize}

\end{tcolorbox}

\vspace{2.5mm}
Thus, we have implemented all gates necessary for a universal gate set: $H, S, CZ,$ and $CCZ$. 

\section{Ribbon operator formalism}\label{ribbons}
\begin{tcolorbox}[enhanced, breakable,pad at break*=1mm]
    In this section we introduce the ribbon operators, which are lattice-level operators that create anyon excitations on top of the quantum double ground states. We discuss the iterative procedure for constructing ribbon operators, and explicitly write down the correspondence between ribbon operators and excited states in the microscopic and anyon bases. We show that the projectors in the quantum double Hamiltonian are nothing but small closed loop ribbon operators. We discuss how topological protection and logical gates arise from the properties of the ribbon operators. Finally, we introduce a generalization of the usual ribbon operator formalism that allows for consistent initialization of desired logical states. 
\end{tcolorbox}

In the previous section, we discussed, at the level of the effective field theory, how anyons in the $S_3$ quantum double can be used to implement a universal gate set. Now we ask the question: how can we actually create and manipulate these anyons at the lattice level? After all, to do any computation on a quantum device, we need to be able to initialize our anyons in particular internal states and braid them to implement unitary gates. 

In a non-Abelian quantum double model, the operators that create excitations are called ``ribbon operators''. They are called ribbons because they act on both the direct and dual lattices and so have a ``thickness'' (see Fig. \ref{fig:ribbon-operator-example} for an example of a ribbon). We will see that ribbons also have a recursive structure; each step of the ribbon depends on the previous ones. This is necessary to ensure our ribbons create isolated excitations at their ends, in analogy to the string operators in an Abelian model; the recursive structure compensates for the non-Abelian nature of the gauge group. It also ensures that excitations created by ribbon operators within a local region automatically satisfy the neutrality conditions as discussed in Sec \ref{neutrality}. On the other hand, the issue with moving non-Abelian anyons coherently, i.e. preserving the internal state of the anyon being moved, is trickier, as one needs to ensure that there is no remnant particle left at the original site of the anyon being moved. We leave a detailed discussion of how to coherently move non-Abelian anyons to an upcoming work. 

In this section, we discuss how to construct ribbon operators for creating anyons, and introduce various key ribbon properties. In particular, we highlight the importance of ``local orientation'' in defining ribbon operators consistently, which was only recognized recently \cite{yan_ribbon_2022}. For more detailed perspective on the original ribbon operator formalism, see \cite{kitaev_fault-tolerant_2003,bombin_family_2008,cui_topological_2018}. For a new gauging perspective in understanding the structure of ribbon operator, see \cite{lyons_protocols_2024}. 

\begin{figure}[h!]
\centering
\includegraphics[width=0.4\textwidth]{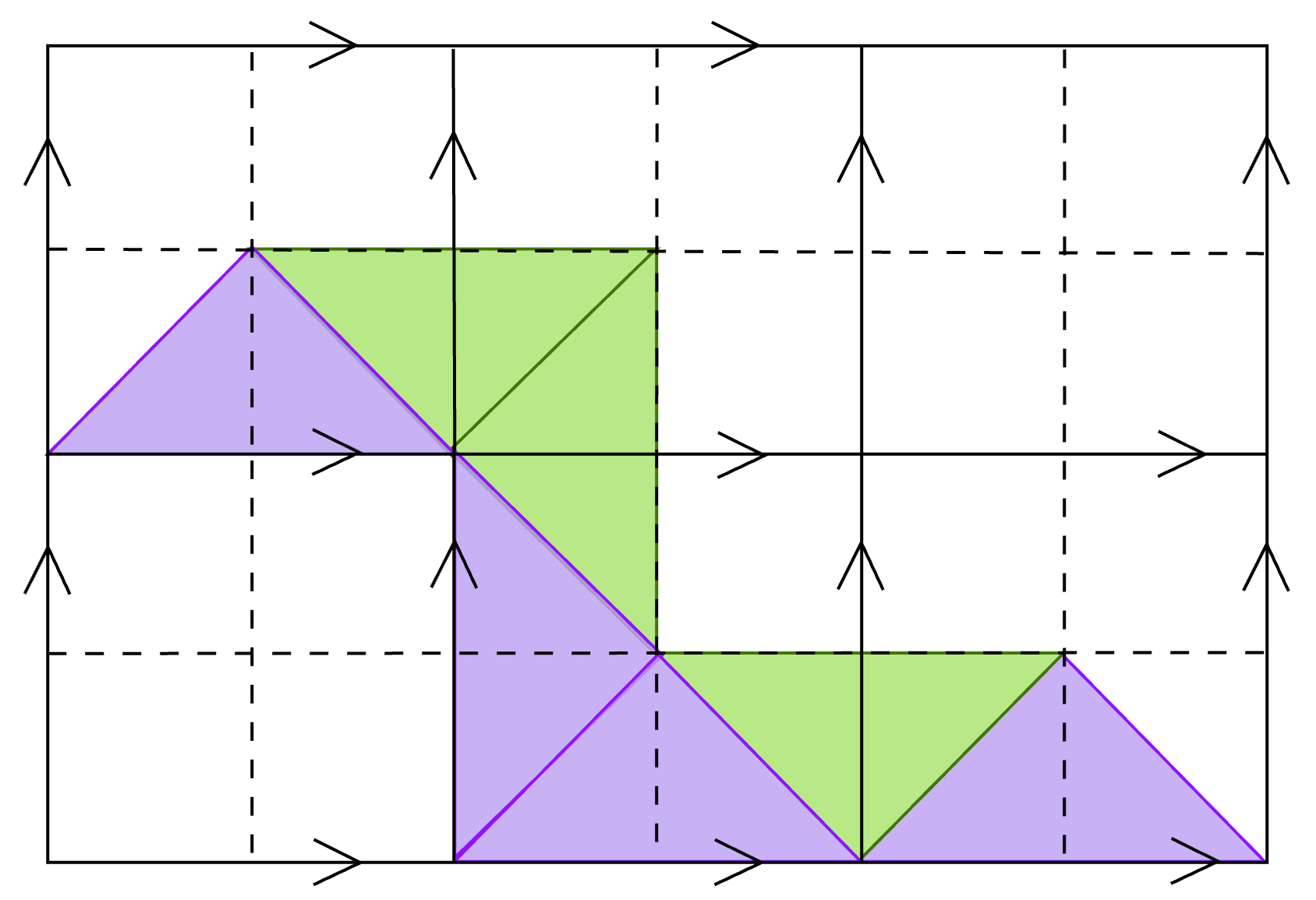}
\caption{An example of a ribbon. Solid lines indicate the direct lattice, whereas the dashed lines indicate the dual lattice. The purple and green triangles form building blocks of the ribbon---purple \textit{direct} triangles have edges on the direct lattice, and green \textit{dual} triangles have edges on the dual lattice.}
\label{fig:ribbon-operator-example}
\end{figure}

We will build up the ribbon operators by first defining the triangle operators, the smallest building blocks of the ribbon. We divide these triangles into two types: ``direct'' triangles with the longest edge along the direct lattice, and ``dual'' triangles with the longest edge along the dual lattice (see Fig. \ref{fig:triangles}). Each direct and dual triangle operator acts on the direct or dual lattice edge respectively--- we define their action in Fig. \ref{fig:triangle_operator}. The action is dependent on two properties:

\begin{enumerate}
    \item \textbf{Alignment:} This is defined as the orientation of the lattice edge relative to the direction of the triangle operator. The lattice edge orientations are set by convention. Dual lattice edges inherit their orientation from the direct lattice edge they intersect; they are assigned the orientation such that they cross the direct lattice edge from right to left if the lattice is rotated so the direct lattice orientation points upwards. The direction of the triangle operator need to follow the direction of the ribbon; the direction of the ribbon is a convention that designate the start as the flavor degree of freedom and the end as the color degree of freedom.
    \item \textbf{Local Orientation:} The ``local'' orientation of a triangle can be clockwise or counterclockwise. We can determine the local orientation by considering the plaquette at its start as a hinge; if swinging the leg of the triangle at the starting site clockwise sweeps out the inside of the triangle, it's a clockwise triangle, and vice versa for a counterclockwise triangle. Our convention of local orientation for direct and dual triangles ensures that the local orientation is the same for all triangles along a ribbon. Most previous literature on the quantum double do not refer to this aspects of the triangle operator (and only pointed out recently in \cite{yan_ribbon_2022}), but this is necessary to ensure consistency of the commutation relation between ribbon operators and the vertex and plaquette projectors, as will be discussed in Section \ref{ribbon-commutation}.
\end{enumerate}

\begin{figure}[h!]
\centering
\includegraphics[width=0.3\textwidth]{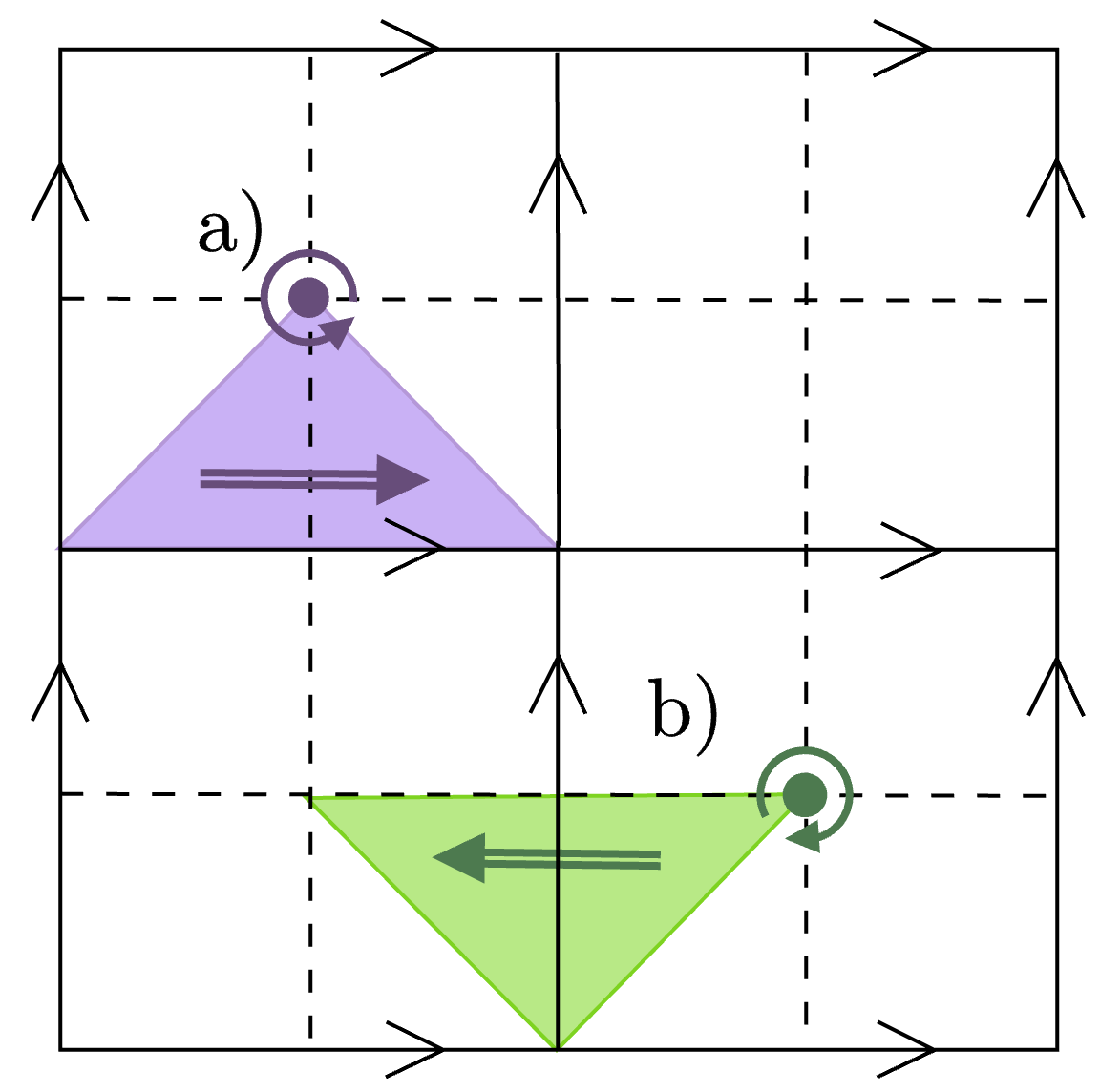}
\caption{Dual (green) and direct (purple) triangles illustrated on a square lattice. Each triangle is oriented with an arrow specifying its direction from one short edge (the starting site) to another. (the ending site). a) Direct triangles have their longest edge on the direct lattice (solid lines); the example in the figure is an ``aligned'' and ``counterclockwise'' direct triangle. b) Dual triangles have their longest edge on the dual lattice (dotted lines); the example in the figure is an ``aligned'' and ``clockwise'' dual triangle. For both direct and dual triangles, their vertices for determining local orientation are at the center of a plaquette.}
\label{fig:triangles}
\end{figure}

\begin{figure}[h!]
\centering
\includegraphics[width=0.9\textwidth]{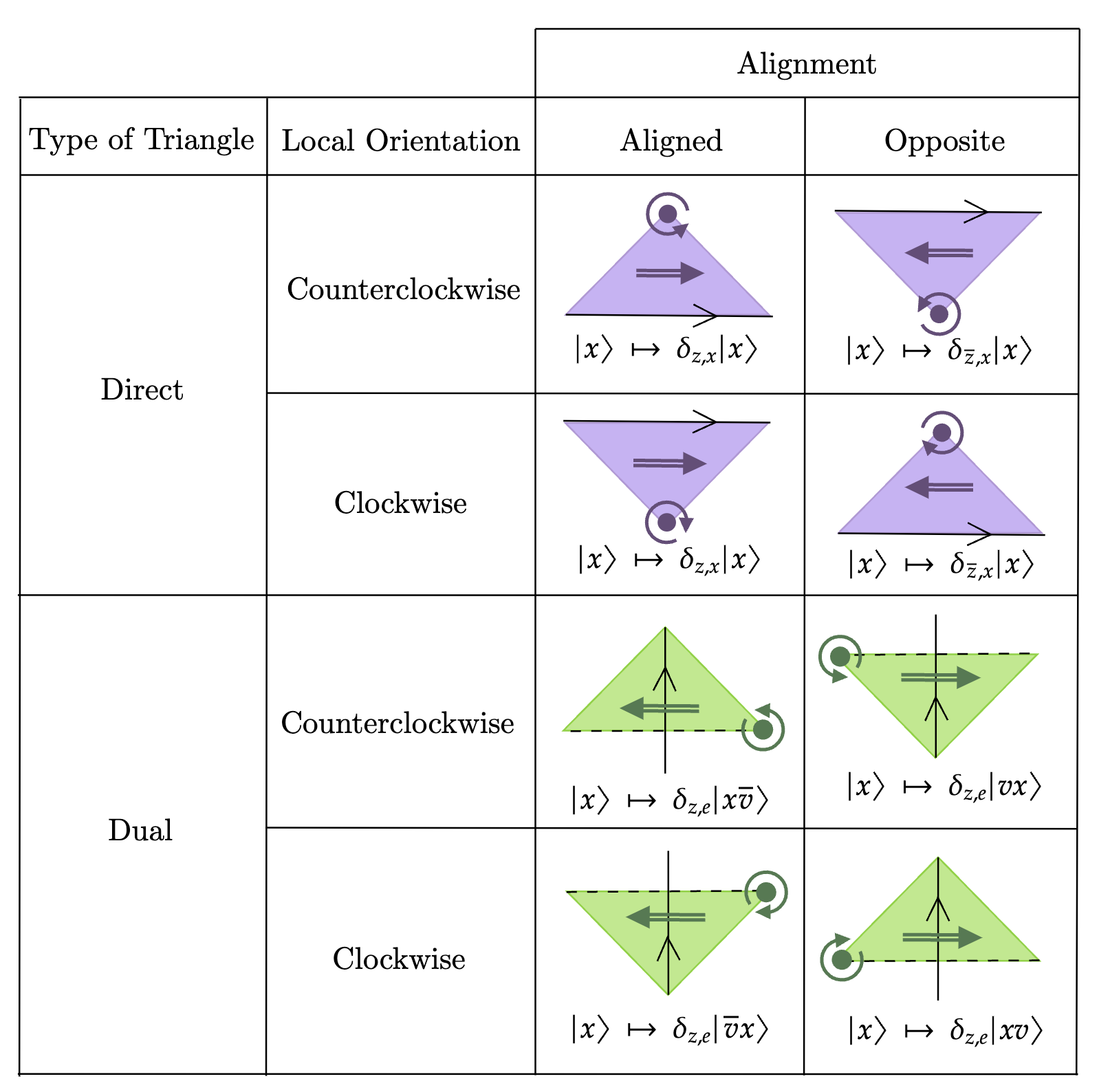}
\caption{Definition of triangle operators, based on the type (direct or dual), alignment (aligned or opposite), and local orientation (clockwise or counterclockwise) of the triangle operator. Direct and dual edge is denoted by solid and dotted line respectively. The orientation of the edge is indicated by the single arrow; the direction of the ribbon is indicated by the double arrow. The local orientation of the triangle operator is determined by first fixing the hinge (dark circle) at the plaquette center touching the triangle operator, and then turn along the ribbon direction.}
\label{fig:triangle_operator}
\end{figure}

Now that we have defined the building blocks, we can define a longer ribbon operator $F$. Each $F$ operator has two associated group elements, $z, v \in G$. These are the same $z$ and $v$ we have seen before; as we will see, $F^{(z, v)}$ will create the state $|z, w=\overline{z} v z\rangle$ where $z$ starts at the beginning vertex of the ribbon:
\begin{equation}
    F^{(z, v)} \ket{\text{vac}} = F^{(z, zw\overline{z})}\ket{\text{vac}} = \ket{z, w}
\end{equation}

\begin{defi}[Recursive definition of ribbon operator]

We build up $F^{(z, v)}$ out of our triangle operators recursively.

\begin{equation}\label{eq: recursive ribbon}
    F^{(z, v)}(\rho) = \sum_{k\in G} F^{(k, v)}(\rho_1) F^{(\overline{k}z, \overline{k}vk)}(\rho_2)
\end{equation}

where $\rho$ is a ribbon composed of multiple triangles, such that $\rho = \rho_1 \cup \rho_2$, where $\rho_1, \rho_2$ are the segments of the ribbon $\rho$.

\end{defi}

Superficially, it appears that we have chosen a particular decomposition of the overall ribbon. However, we can show that the construction of the ribbon is well-defined regardless the choice of decomposition. Consider a ribbon in three parts, rather than two: $\rho = \rho_1 \cup \rho_2 \cup \rho_3$. 
\begin{equation}
    \begin{aligned}
        F^{(z, v)}(\rho) &= \sum_{k\in G} F^{(k, v)}(\rho_1) F^{(\overline{k}z, \overline{k}vk)}(\rho_2 \cup \rho_3) \\
        &= \sum_{k, l\in G} F^{(k, v)}(\rho_1) F^{(l, \overline{k}vk)}(\rho_2) F^{(\overline{(kl)}z, \overline{(kl)}v(kl))}(\rho_3)\\
        &= \sum_{k, kl = m \in G} F^{(k, v)}(\rho_1) F^{(\overline{k}m, \overline{k}vk)}(\rho_2) F^{(\overline{m}z, \overline{m}vm)}(\rho_3)\\
        &= \sum_{m \in G} F^{(m, v)}(\rho_1 \cup \rho_2) F^{(\overline{m}z, \overline{m}vm)}(\rho_3)
    \end{aligned}
\end{equation}
We have shown that the decomposition of $\rho$ in terms of $\rho_1$ and $\rho_2 \cup \rho_3$ is equal to that in terms of $\rho_1 \cup \rho_2$ and $\rho_3$.

An intuition for the recursive formula for constructing ribbon operator is that we sum over all possible intermediate z string that is consistent with the end point, such that there is no excitation created along the ribbon except at its two end points. Another perspective on this is that the z string provides a consistency relation between the local and global flux, and when extending the ribbon, the values of the z string need to be updated in order to move the excitation coherently to only occur at the end points of a ribbon operator. It is reminiscent of the path integral formalism where all path consistent with the initial conditions are summed over; here, all z string consistent with the end point fluxes are summed over. 

To make concrete the recursive definition of ribbon operator, it is instructive to consider an example as shown in Fig. \ref{fig:ribbon-example}. We begin by decomposing the full ribbon into the first triangle and the rest of the ribbon, and continue until we have decomposed the whole thing into the constituent triangles:

\begin{figure}[h!]
\centering
\includegraphics[width=0.8\textwidth]{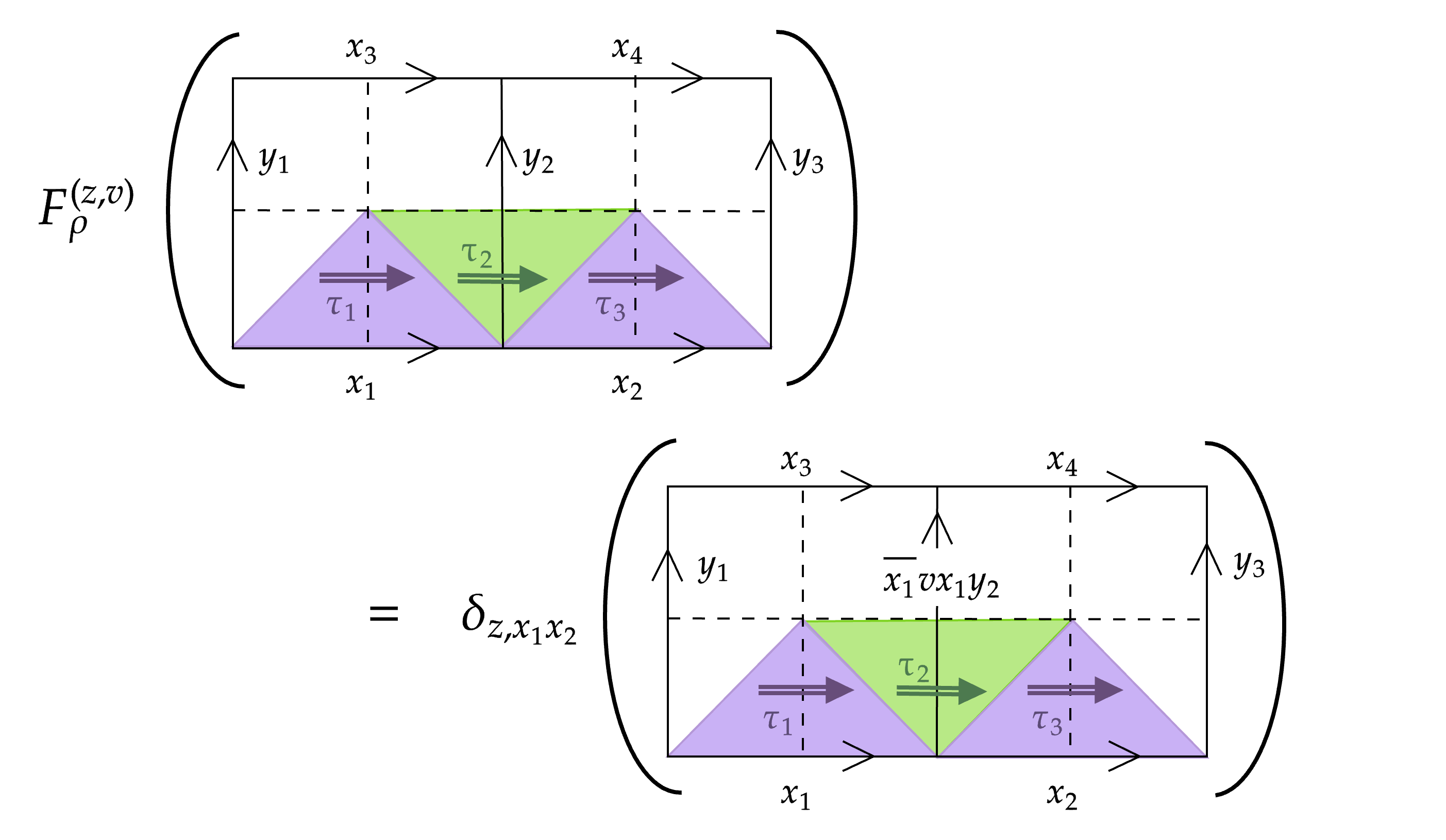}
\caption{The action of a ribbon operator on a given state (labeled by a group element on each edge). $\tau_i$ label triangle operators. The ribbon operator acts by projecting the product of all direct edges along the ribbon to $z$ and multiplying each dual edge by some group element dependent on $v$ and all the previous direct edges.}
\label{fig:ribbon-example}
\end{figure}

Using our triangle operator definitions (see Fig. \ref{fig:triangle_operator}) we find:
\begin{equation}
    \begin{aligned}
        F^{(z, v)}(\tau_1 \cup \tau_2 \cup \tau_3) &= \sum_{k_1, k_2 \in G} F^{(k_1, v)}(\tau_1) F^{(k_2, \overline{k}_1vk_1)}(\tau_2)F^{(\overline{(k_1k_2)}z, \overline{(k_1k_2)}v(k_1k_2))}(\tau_3)\\
        &= \sum_{k_1, k_2 \in G} \delta_{k_1, x_1} \delta_{k_2, e} L^{\overline{k}_1 v k_1}_{y_2} \delta_{\overline{(k_1k_2)}z, x_2}\\
        &= \delta_{z, x_1x_2} L^{\overline{x}_1 v x_1}_{y_2}
    \end{aligned}
\label{eq:ribbon-example}
\end{equation}
where $L^v_x: |x\rangle \rightarrow |vx\rangle$ is an action by left multiplication, where $|x\rangle$ is a state on the edge $x$ of the direct lattice. Since we can always express the state as a linear superposition of states in the regular representation, we abuse our notation here by labeling the group element as $x_i$ on the edge at $x_i$. 

\subsection{Ribbons in the Microscopic and Anyon Bases}\label{ribbon-def}

Does the ribbon operator in Eq. \ref{eq:ribbon-example} create the excitations we expect? First, treating the start of the ribbon as the location of our excitation and the end of the ribbon as the origin, we see that the $z$ in the ribbon operator is exactly the $z$ from the microscopic basis discussed in Sec. \ref{excited-states}. Additionally, suppose the initial state is a ground state (such that all plaquettes have trivial flux: $e=x_1y_2\bar x_3 \bar y_1$), after the action of the ribbon operator, the local flux at the starting plaquette (measured from the vertex where the ribbon starts) is $x_1(\bar x_1 v x_1 y_2) \bar x_3 \bar y_1=v$. Therefore, at the starting plaquette, the ribbon $F^{(z, v)}$ creates the state $\ket{z, w= \bar z v z}$ as expected.

We note that at the end plaquette, the local flux (which is also the topological flux, as this plaquette is located at the origin) is also non-trivial; it has flux $\bar z \bar v z = \bar w$ which is the inverse of the topological flux of the excitation at the starting plaquette, such that the ribbon $F^{(z, v)}$ (acting on the ground state) creates a pair of excitations with total trivial flux (measured at the same base point). The flux $w$ is located at the beginning of the ribbon, while the flux $\bar w$ is located at the end. We set the end of the ribbon to be the origin when we measure any topological flux. 

What if we want to create an excitation in the anyon basis, as will be necessary later to initialize logical states for computation? We simply need to apply the right linear combination of ribbons, just as our anyon basis states were linear combinations of $\ket{z, w}$ states. The basis transformation is as follows:
\begin{equation}\label{eq:anyon-basis-from-microscopic-basis}
    F^{(R, C); \mathbf{u}, \mathbf{u'}} = \frac{|R|}{|\mathbf Z(r)|} \sum_{n \in \mathbf{Z}(r)} \Gamma^R_{j j'}(n) F^{(q_c n \overline{q}_{c'}, c)}
\end{equation}
Here, $R$ is an irreducible representations of $Z(r)$, the centralizer of an representative $r$ in a conjugacy class $C$. We define $\textbf{u} = (c, j)$ and $\textbf{u'} = (c', j')$ where $c, c' \in C$ and $1 \leq j, j' \leq |R|$; $\Gamma^R(n)$ is the representation matrix of the group element $n$ in $R$. The group elements $q_c$ and $q_{c'}$ are defined in the following way with respect to $c, c',r$:
\begin{equation}
    c = q_c r \overline{q}_c, \quad\quad\quad  c' = q_{c'} r \overline{q}_{c'}
\end{equation}
Note that the choice of $q_c,q_{c'}$ is not unique, but the set  $\{q_cn\bar{q}_c' | n \in \mathbf{Z}(r)\}$, which is the set of all possible value of $z$ consistent with the flux configuration, is well-defined. So the sum on the right hand side of Eq. \ref{eq:anyon-basis-from-microscopic-basis} is well-defined.

Applying the ribbon $F^{(R, C); \mathbf{u}, \mathbf{u'}}$ to the vacuum will create an excitation (of a particular anyon type) with charge $R$, flux $C$, at the beginning of the ribbon; at the end of the ribbon, there is another excitation in a internal state such that the anyons at the two ends are in the vacuum fusion channel. Here the local flux (flavor degree of freedom) at the start is $c$, and the local flux at the end is $c'$.

\subsection{Ribbon Commutation Relations}\label{ribbon-commutation}

Here we summarize some key properties of ribbon operators.

\begin{enumerate}
    \item \textbf{Operators acting on the same ribbon:}
For two ribbon operators acting on the same ribbon $t$,
\begin{equation}
        F^{(z_1, v_1)}(t)F^{(z_2, v_2)}(t)=
        \begin{cases}
            \delta_{z_1,z_2}F^{(z_1, v_2v_1)}(t) \quad \quad \quad & \text{for clockwise local orientation}\\
            \delta_{z_1,z_2}F^{(z_1, v_1v_2)}(t) \quad \quad \quad & \text{for counterclockwise local orientation}\\
        \end{cases}
    \end{equation}
See appendix \ref{ap:same-ribbon} for a proof.

\item \textbf{Commutation with projectors in the middle of ribbon:} Ribbon operators commute with all projectors in the middle of the ribbon--- this is key to ensure excitations are only created at the ends of the ribbon. See appendix \ref{ap:middle-ribbon} for an example. 

\item \textbf{Commutation with projectors at end points:}

Ribbon operators don't commute with projectors at the end point, i.e., ribbon operators create excitations at the endpoints).

For plaquette operators, the general commutation relation with ribbon operators are
\begin{align}
    \text{Start: }\quad B_{\mathrm{fl}}^{h}F^{(z, v)}(t)&=
    \begin{cases}
        F^{(z, v)}(t)B_{\mathrm{fl}}^{v h} & \quad \quad \quad  \text{for clockwise}\\
        F^{(z, v)}(t)B_{\mathrm{fl}}^{h v} & \quad \quad \quad  \text{for counterclockwise}
    \end{cases}\\
    \text{End: }\quad B_{\mathrm{col}}^{h}F^{(z,v)}(t)&=
    \begin{cases}
        F^{(z, v)}(t)B_{\mathrm{col}}^{\bar z v zh} & \quad \quad \quad \text{for clockwise}\\
        F^{(z, v)}(t)B_{\mathrm{col}}^{\bar{z} \bar{v} z \bar h} & \quad \quad \quad \text{for counterclockwise}
    \end{cases}
\end{align}

The commutation relations between vertex operators and ribbon operators are:
\begin{align}
    \text{Start: }\quad A_{\mathrm{fl}}^{g}F^{(z,v)}(t)&=F^{(gz, gv\bar g)}(t)A_{\mathrm{fl}}^{g}\\
    \text{End: }\quad A_{\mathrm{col}}^{g}F^{(z,v)}(t)&=F^{(z\bar g, v)}(t)A_{\mathrm{col}}^{g},
\end{align}

regardless of the local orientation of the ribbon. See appendix \ref{ap:end-ribbon} for proof of above commutation relations.

\end{enumerate}

\subsection{Plaquette and Vertex Operators as Closed Ribbons} \label{closed-ribbons}

It turns out that $B^z_p$ and $A^g_s$ are nothing more than special cases of ribbon operators, where the ribbon forms a closed loop.

\begin{figure}[h!]
\centering
\includegraphics[width=0.9\textwidth]{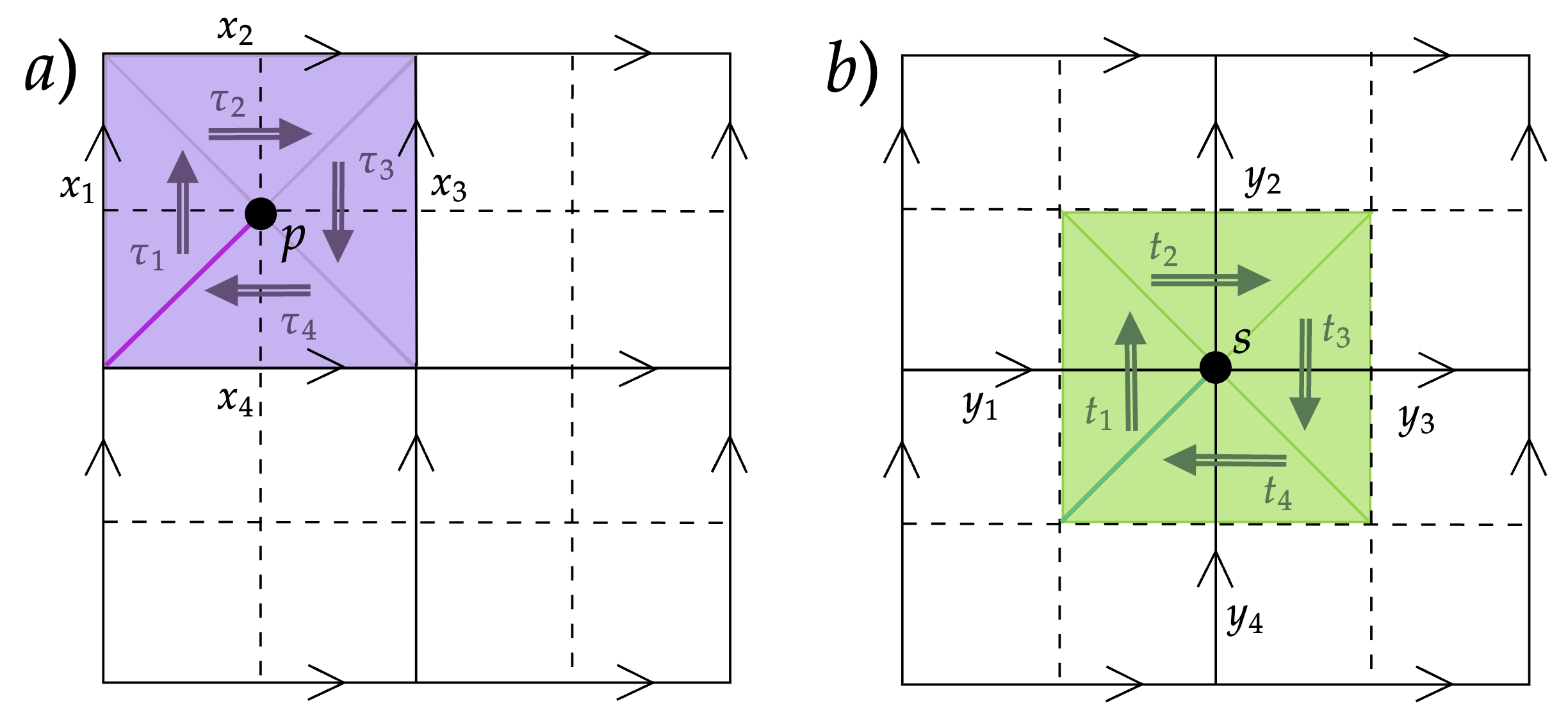}
\caption{a) Closed ribbon made up of direct triangles; the start and end of the closed ribbon is denoted by the bright violet diagonal line. b) Closed ribbon made of dual triangles; the start and end of the closed ribbon is denoted by the bright green diagonal line. }
\label{fig:closed-ribbons}
\end{figure}

Let's consider the simplest closed direct ribbon, consisting of a set of four direct triangles (labeled by $\tau_i$) filling out a plaquette on the square lattice (see Fig. \ref{fig:closed-ribbons}a). First, note that $\tau_1$, $\tau_2$ are aligned, whereas $\tau_3$ and $\tau_4$ are opposite; all 4 direct triangles have clockwise local orientation. Let's write down the corresponding ribbon operator for a pair of group elements $(z, v)$. By using Eq. \ref{eq: recursive ribbon} and the direct triangle definitions based on alignment and local orientation, it is straightforward to show that 
\begin{equation}
        F^{(z, v)}(\tau_1 \cup \tau_2 \cup \tau_3 \cup \tau_4) = \delta_{z, x_1 x_2 \overline{x}_3\overline{x}_4} = \delta_{\overline{z}, x_4x_3\overline{x}_2\overline{x}_1} = B^z_p
\end{equation}

As we have seen before, the ribbon operator projects the product of all the direct edges to the group element $z$. Here, however, these direct edges lie around a plaquette $p$, meaning the action of the ribbon is exactly the action of $B^z_p$, the plaquette operator. Conceptually, this makes sense. Direct triangles transport charges; the closed loop of direct triangles transports a charge around that closed loop. If there is a flux living on the plaquette, the charge will braid non-trivially with the flux and hence detect it. If there is no flux, as we expect in the ground state, then the action of the operator will be trivial, which is exactly the stabilizer condition imposed by the Hamiltonian. Note that the resulting plaquette operator $B^z_p$ is independent of the choice of $v$ for the ribbon, since it is solely composed of direct triangles. 

What about a closed ribbon $F^{g,v}$ consisting of just dual triangles (labeled by $t_i$ in Fig. \ref{fig:closed-ribbons}b)? As you might guess, this will be equivalent to our vertex operators, $A^g$. In terms of the alignments of the dual triangles, $t_1$ and $t_2$ are aligned, while $t_3$ and $t_4$ are opposite-type; all 4 dual triangles have counterclockwise local orientation. By using Eq. \ref{eq: recursive ribbon} and the dual triangle definitions based on alignment and local orientation, we see that the action of the closed ribbon is
\begin{equation}
         F^{(g, v)}(t_1 \cup t_2 \cup t_3 \cup t_4)  = R^{\overline{v}}_{y_1} L^{v}_{y_2} L^{v}_{y_3} R^{\overline v}_{y_4} = A^g_s
\end{equation}
where $L^g_{y_i}$ $(R^{\overline{g}}_{y_i})$ denotes left (right) multiplying the edge $y_i$ by $g$ $(\overline{g})$. As the edges of our dual triangles are nothing but the edges going into or out of a vertex of our lattice, we see that the closed ribbon operator with only dual triangles is exactly the vertex operator $A^g_s$--- incoming edges are right multiplied by $\overline{g}$, while outgoing edges are left-multiplied by $g$. As with the plaquette operator $B$, we can better understand what the vertex operator is doing by considering its ribbon description. Dual triangles transport fluxes between plaquettes; the closed dual ribbon operator creates a flux pair, transports one half around the vertex, and then fuses them. In the vacuum state, there should be no charge at the vertex to braid with the flux, so the action of this operator should be trivial. Note that the resulting vertex operator $A^g_s$ is independent of the choice of $z$ for the ribbon, since it is solely composed of dual triangles. 

\subsection{Particle Exchange and Braiding} \label{particle-braiding-ribbons}

As ribbon operators move anyons around, we expect that they could be used to describe braiding or particle exchange processs in our model as well. In this section we show that the ribbon operators, as they have been constructed, have exactly the necessary commutation relations to match up with anyonic statistics. Most of this discussion is based heavily on Shawn Cui's lecture notes \cite{cui_topological_2018}.

\begin{lem}
    Consider two ribbons $t_1, t_2$ which end at the same site, only sharing one edge. Then the following commutation relations hold:
    \begin{equation}
        F^{(z_2, v_2)}(t_2)F^{(z_1, v_1)}(t_1)=
        \begin{cases}
         F^{(z_1 \overline{z}_2 v_2 z_2, v_1)}(t_1) F^{(z_2, v_2)}(t_2) \quad\quad\quad\text{for clockwise} \\
         F^{(z_1 \overline{z}_2 \overline{v}_2 z_2, v_1)}(t_1) F^{(z_2, v_2)}(t_2)\quad\quad\quad\text{for counterclockwise}
        \end{cases}
    \label{eq:ribbon-exchange}
    \end{equation}
\label{lem:exchange}
\end{lem}

\begin{proof}

We will prove for the case of counterclockwise local orientation (the clockwise case can be shown in the same manner). 

Write each ribbon as the composition of their disjoint parts, $t_1'$ and $t_2'$, and the triangles corresponding to their shared edge, $\tau_1$ and $\tau_2$. First, we show that the lemma holds for the operators acting on $\tau_1$ and $\tau_2$. The left-hand side of equation \ref{eq:ribbon-exchange} acts on their shared edge $|x\rangle$ in the following way:
\begin{equation}
    F^{(z_2, v_2)}(\tau_2)F^{(z_1, v_1)}(\tau_1)|x\rangle = \delta_{e, z_2} \delta_{x, z_1} |x\overline{v}_2\rangle
\end{equation}
where we have assumed that $\tau_1$ is direct and $\tau_2$ is dual (they must be different kinds to share an edge but be distinct). The right-hand side gives:
\begin{equation}
\begin{aligned}
    F^{(z_1 \overline{z}_2 \overline{v}_2 z_2, v_1)}(\tau_1) F^{(z_2, v_2)}(\tau_2)|x\rangle &= \delta_{x\overline{v}_2, z_1 \overline{z}_2 \overline{v}_2 z_2} \delta_{e, z_2} |x\overline{v}_2\rangle\\
    &= \delta_{x, z_1} \delta_{e, z_2} |x\overline{v}_2\rangle
\end{aligned}
\end{equation}
So the right-hand side and the left-hand side do give the same result. We can generalize this to the full ribbons by using the recursive definition for $F^{(v, z)}$:
\begin{equation}
    \begin{aligned}
        F^{(z_2, v_2)}(\tau_2)F^{(z_1, v_1)}(\tau_1) &= \sum_{k_1, k_2} F^{(k_2, v_2)}(t_2')F^{(\overline{k}_2 z_2, \overline{k}_2 v_2 k_2)}(\tau_2) F^{(k_1, v_1)}(t_1')F^{(\overline{k}_1 z_1, \overline{k}_1 v_1 k_1)}(\tau_1)\\
        &= \sum_{k_1, k_2} F^{(k_1, v_1)}(t_1') F^{(k_2, v_2)}(t_2') F^{(\overline{k}_2 z_2), \overline{k}_2 v_2 k_2}(\tau_2) F^{(\overline{k}_1 z_1, \overline{k}_1 v_1 k_1)}(\tau_1)\\
        &=\sum_{k_1, k_2} F^{(k_1, v_1)}(t_1') F^{(k_2, v_2)}(t_2') F^{(\overline{k}_1 z_1 \overline{z}_2\overline{v}_2 z_2, \overline{k}_1 v_1 k_1)}(\tau_1) F^{(\overline{k}_2 z_2, \overline{k}_2 v_2 k_2)}(\tau_2) \\
        &= \sum_{k_1, k_2} F^{(k_1, v_1)}(t_1') F^{(\overline{k}_1 z_1 \overline{z}_2\overline{v}_2 z_2, \overline{k}_1 v_1 k_1)}(\tau_1) F^{(k_2, v_2)}(t_2')  F^{(\overline{k}_2 z_2, \overline{k}_2 v_2 k_2)}(\tau_2) \\
        &= F^{(z_1 \overline{z}_2\overline{v}_2 z_2, v_1)}(t_1) F^{(z_2, v_2)}(t_2)
    \end{aligned}
\end{equation}
where the second and fourth steps followed from the fact $t_1', t_2', \tau_1, \tau_2$ are all disjoint and so their ribbon operators commute.
\end{proof}

\begin{figure}[h!]
\centering
\includegraphics[width=0.75\textwidth]{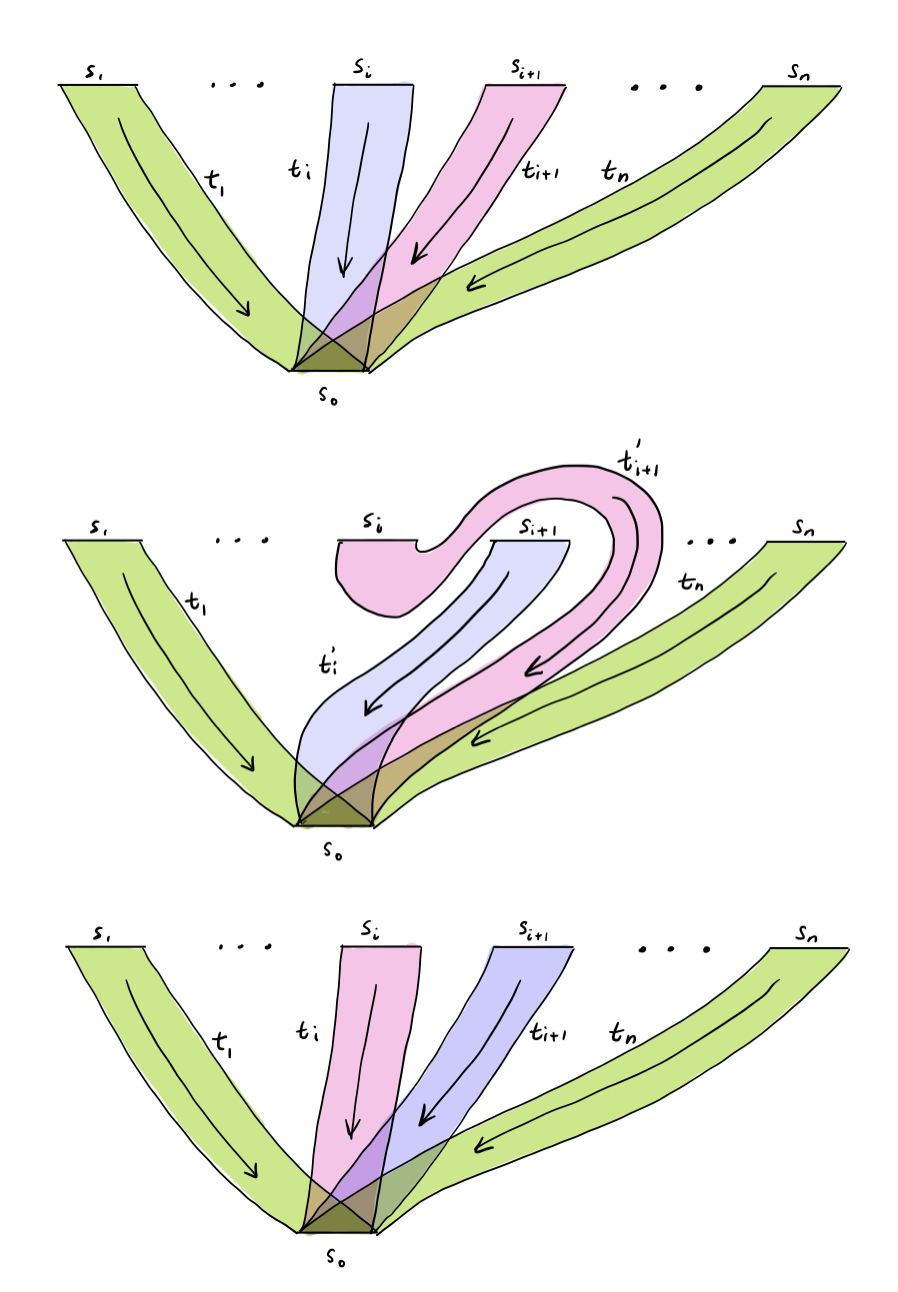}
\caption{A schematic of anyon braiding in terms of ribbon operators. Figure adapted from \cite{cui_topological_2018}.}
\label{fig:ribbon-braiding}
\end{figure}

We can use this intermediate result to demonstrate the ribbon operators reproduce the right braiding properties. Consider a set of particles at sites $s_1, \dots, s_i, s_{i+1}, \dots, s_n$. We can initialize this state from the vacuum by applying ribbon operators from $s_1, \dots, s_i, s_{i+1}, \dots, s_n$ to the origin. See the first diagram in Fig. \ref{fig:ribbon-braiding} for a pictorial representation of this state.
\begin{equation}
    |\Psi\rangle = F^{(z_n, v_n)}(\tau_n)\cdots F^{(z_{i+1}, v_{i+1})}(\tau_{i+1})F^{(z_{i}, v_{i})}(\tau_i) \cdots F^{(z_1, v_1)}(\tau_1)|\Omega\rangle
\end{equation}
We want to swap the positions of the anyons at sites $s_i$ and $s_{i+1}$. We can do this by first deforming $\tau_i$ and $\tau_{i+1}$ to $\tau_{i}'$ and $\tau_{i+1}'$, respectively (see the second diagram in Fig. \ref{fig:ribbon-braiding}). We are allowed to do this because the ribbon operators are topological, so deformations of this kind do not impact the resulting state--- as long as we do not cross any excitations when we deform the ribbons. After doing this, we can keep deforming $\tau_i'$ until it becomes $\tau_{i+1}$. However, we cannot do the same with $\tau_{i+1}'$ because that would require crossing the excitation we have already placed at $s_{i+1}$. Instead, we need to swap the order we apply these operators, which requires the application of the lemma we proved above. Then, finally, we can deform $\tau_{i+1}'$ to $\tau_i$, and we have completed the swap operation.
\begin{equation}
    \begin{aligned}
         |\Psi\rangle &= F^{(z_n, v_n)}(\tau_n)\cdots F^{(z_{i+1}, v_{i+1})}(\tau_{i+1}')F^{(z_{i}, v_{i})}(\tau_i') \cdots F^{(z_1, v_1)}(\tau_1)|\Omega\rangle\\
         &= F^{(z_n, v_n)}(\tau_n)\cdots F^{(z_{i+1}, v_{i+1})}(\tau_{i+1}')F^{(z_{i}, v_{i})}(\tau_{i+1}) \cdots F^{(z_1, v_1)}(\tau_1)|\Omega\rangle\\
         &= F^{(z_n, v_n)}(\tau_n)\cdots F^{(z_{i}\overline{z_{i+1}}\overline{v_{i+1}}z_{i+1}, v_{i})}(\tau_{i+1}) F^{(z_{i+1}, v_{i+1})}(\tau_{i+1}') \cdots F^{(z_1, v_1)}(\tau_1)|\Omega\rangle\\
         &= F^{(z_n, v_n)}(\tau_n)\cdots F^{(z_{i}\overline{z_{i+1}}\overline{v_{i+1}}z_{i+1}, v_{i})}(\tau_{i+1}) F^{(z_{i+1}, v_{i+1})}(\tau_i) \cdots F^{(z_1, v_1)}(\tau_1)|\Omega\rangle
    \end{aligned}
\end{equation}
We can see that the end result of this swapping is to conjugate the topological flux of the $i$th anyon by the topological flux of the $(i+1)$th: we have implemented the counterclockwise exchange operator $R^{-1}$ (Eq. \ref{eq:ccw-exchange-op}) we discussed abstractly in terms of ribbon operators.
\begin{equation}
\begin{aligned}
     w_{i, \mathrm{final}} &= (\overline{z_{i+1}} v_{i+1} z_{i+1} \overline{z_i}) v_i (z_{i}\overline{z_{i+1}}\overline{v_{i+1}}z_{i+1})\\
     &= w_{i+1} \overline{z_i} v_i z_{i} \overline{w_{i+1}}\\
     & = w_{i+1} w_{i, \mathrm{initial}} \overline{w_{i+1}}
\end{aligned}
\end{equation}
Repeating this swap would conjugate both $i$th and $({i+1})$th anyon by their total initial topological flux, again as we expect. So we see the structure of the ribbon operators is consistent with our prior expectations for how anyons should interact. Ribbon operators provide a concrete implementation for us to manipulate anyons directly. 

\subsection{Computation with Ribbons}

The derivation above also helps us see how information can be encoded, protected, and manipulated using anyons. Note that, in the above example of particle exchange, we are manipulating the ends of the ribbons away from the origin --- the part of our excitation that is changeable locally. Yet, the local flux remains unchanged, and only the topological flux is affected. We have performed a logical operation on the topological degrees of freedom, without having to operate on them at the origin! 

Now we can really appreciate what ``topological protection'' means. Consider the scenario where we have applied $n$ ribbons to the vacuum, all ending at the same site $s_0$ (as in the top illustration of Fig. \ref{fig:ribbon-braiding}). Suppose the ribbons are chosen such that the excitations at $s_0$ are overall neutral. Since the color degrees of freedom of the excitations at the free ends of the ribbons are the inverse of the excitations at the origin, this means we have an overall color neutral state. Additionally, to any local operator living at or around $s_0$, it looks like there is no particle there! No local operator will be able to learn any information about the color, since there is no charge or flux to interact with. However, we can still encode information in the Hilbert space living at $s_0$ (which is equivalent to the combined color space of the excitations at $s_1, \cdots, s_n$  by our choice of ribbons), and we can manipulate that information by braiding and fusing the excitations living at the other end of the ribbons. 

\subsection{Generalization: Extended z ribbon operators}

At the start of the section, we provided the recursive definition of ribbon operator that allows one to decompose a long ribbon operator into its triangle operator components; here, we show that there is an alternative ground-up algorithm to construct ribbon operator, and the algorithm is amenable to generalization.

\begin{defi}[Constructive algorithm for ribbon operators]
To construct a ribbon operator $F^{z,v}$,
    \begin{enumerate}
        \item Start with $x=e$. $x$ keep tracks of the product of direct edges along the length of the ribbon.
        \item If current triangle is dual: left multiply by $x^{-1}vx$ or right multiply by $x^{-1}v^{-1}x$ depending on alignment and local orientation of the edge, where $x$ is the product of all previous direct edges.
        \item If current triangle is direct with direct edge $x_i$: update $x \to x \cdot x_i$ if the direct triangle is aligned; otherwise (for opposite alignment), update $x\to x \cdot \overline{x}_i$.
        \item Repeat step 2-3 along the length of the ribbon.
        \item At the end, project the product of all direct edges involved in the ribbon to $z$.
    \end{enumerate}
\end{defi}

Now, we introduce the generalized ribbon operator. We relax the condition that a ribbon operator must be composed of direct and dual triangles that fills a ribbon; in fact, the start and end of a ribbon operator can consist of just arbitrarily long direct edges (highlighted purple line) without necessarily filled in with triangles; crucially, no additional violations are created. As an example, consider the generalized ribbon operator in Fig. \ref{fig:generalized-ribbon}. Based on the constructive algorithm, after the first two direct edges ($y_4$ and $x_1$), the value of $x$ is updated to $x=y_4 x_1$ after step $3$. Then for the first dual triangle with dual edge at $y_1$, its edge value is left multiplied with $x^{-1}vx=\overline{x}_1\overline{y}_4 v y_4 x_1$. Similarly step gives the action for the second dual edge. At the end of the constructive algorithm, we apply a Kronecker delta which projects the multiple of all direct edges (which includes the purple lines prepending and appending the ribbon) $x=y_4 x_1 x_2 x_3 \overline{y}_7$ to the value of $z$. 

\begin{figure}[h!]
\centering
\includegraphics[width=1\textwidth]{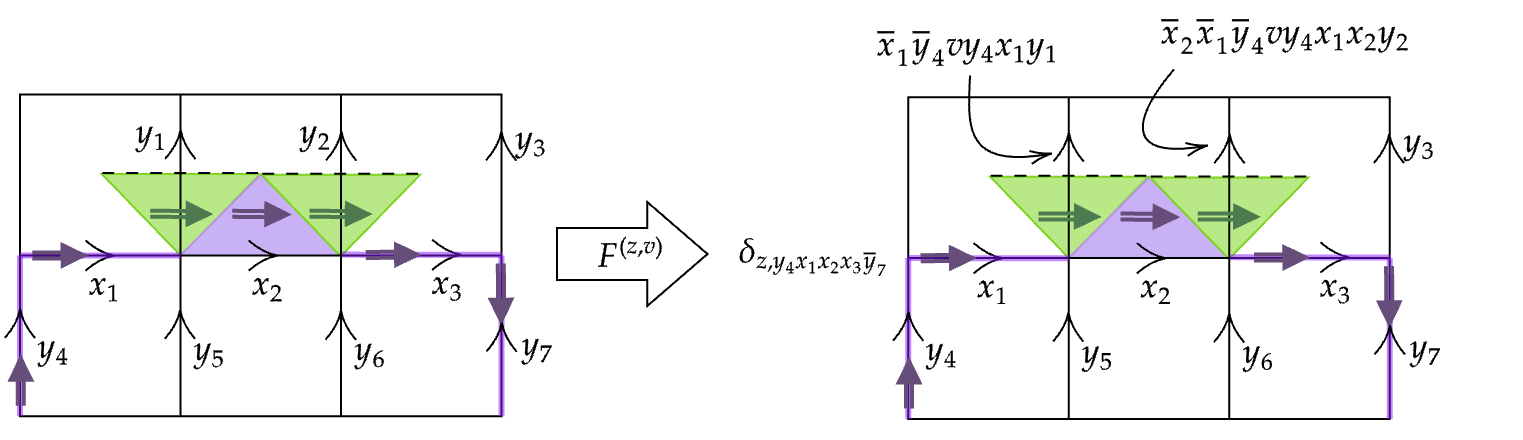}
\caption{An example of a generalized ribbon operator and its action. The generalized ribbon operator allows for the $z$ string to extend beyond the direct edges of the original ribbon, as shown by the thickened purple line. This changes the action of the Kronecker delta, which project the product of all direct edges along the extended $z$ string to the value $z$.}
\label{fig:generalized-ribbon}
\end{figure}

The action of the generalized ribbon operator directly follows from the constructive algorithm. In terms of flux excitation, the generalized ribbon operator create a flux with global flux (i.e. the color degree of freedom) $v$ at the top left plaquette and another flux with global flux $\bar{v}$ at the top right plaquette. There could be charge excitations at the two ends of the purple direct strings; nonetheless, one can check that there is no additional flux and charge excitation created by the generalized ribbon operator. The net effect of prepending or appending a ribbon operator with direct triangles is to move any charge excitation originally decorating the end vertex of the origin ribbon to the new end vertex of the generalized ribbon. 

To define the global flux consistently, we actually do not need to require that all ribbon begin at the same site as shown in Fig. \ref{fig:ribbon-braiding}; instead, the necessary condition is that the extended $z$ string begins at the same vertex, i.e. all global flux are measured with respect to the same origin. This ensures that the global flux of different ribbon pairs can be consistently compared with each other, such that the logical state encoding discussed in Sec. \ref{state_initialization} can be initialized consistently for the lattice-level implementation of universal quantum computation.

\subsection{Lattice-level state initialization}

\subsubsection{Computational basis on the lattice}

We will now discuss how to initialize the computational basis states discussed in Sec. \ref{state_initialization} using the ribbon operator in the $C_{2}$ flux anyons basis (the anyon basis ribbon is discussed in Appendix \ref{ap:basis-change}). The anyon basis ribbon allow one to specify the internal state at both ends of the ribbon. There are nine possible $C_{2}$ anyon basis ribbons, which we will refer to as the flux basis ribbons. 
\begin{equation}
\begin{aligned}
F^{[C_{2}];\sigma,\sigma},F^{[C_{2}];\sigma,\mu\sigma},F^{[C_{2}];\sigma,\bar{\mu}\sigma}\\
F^{[C_{2}];\mu\sigma,\sigma},F^{[C_{2}];\mu\sigma,\mu\sigma},F^{[C_{2}]; \mu\sigma,\bar{\mu}\sigma}\\
F^{[C_{2}];\bar{\mu}\sigma,\sigma},F^{[C_{2}];\bar{\mu}\sigma,\mu\sigma},F^{[C_{2}]; \bar{\mu}\sigma,\bar{\mu}\sigma}
\end{aligned}
\end{equation}

At first sight, it might appear that a given flux basis ribbon does not create a vacuum pair, since the flux at both ends can be different; this seem to be in tension with the global neutrality condition. This tension is resolved by realizing that the two flux labels for a given anyon ribbon is that of a local flux (around just the plaquette of each end of the ribbon), and by a judicious choice of $z$ string the global flux is make sure to be neutral.
In addition, with just one flux basis ribbon like $F^{[C_{2}];\sigma,\sigma}$, it seems that it might have a net charge as toggling both the start and end internal state would not keep the state invariant; again this is in tension with the global neutrality condition. This paradox is resolved by realizing that there are actually $[2]$ charges decorated at the two ends of the $z$ string, as shown in the figure below:

\begin{figure}[h!]
\centering
\includegraphics[width=0.5\textwidth]{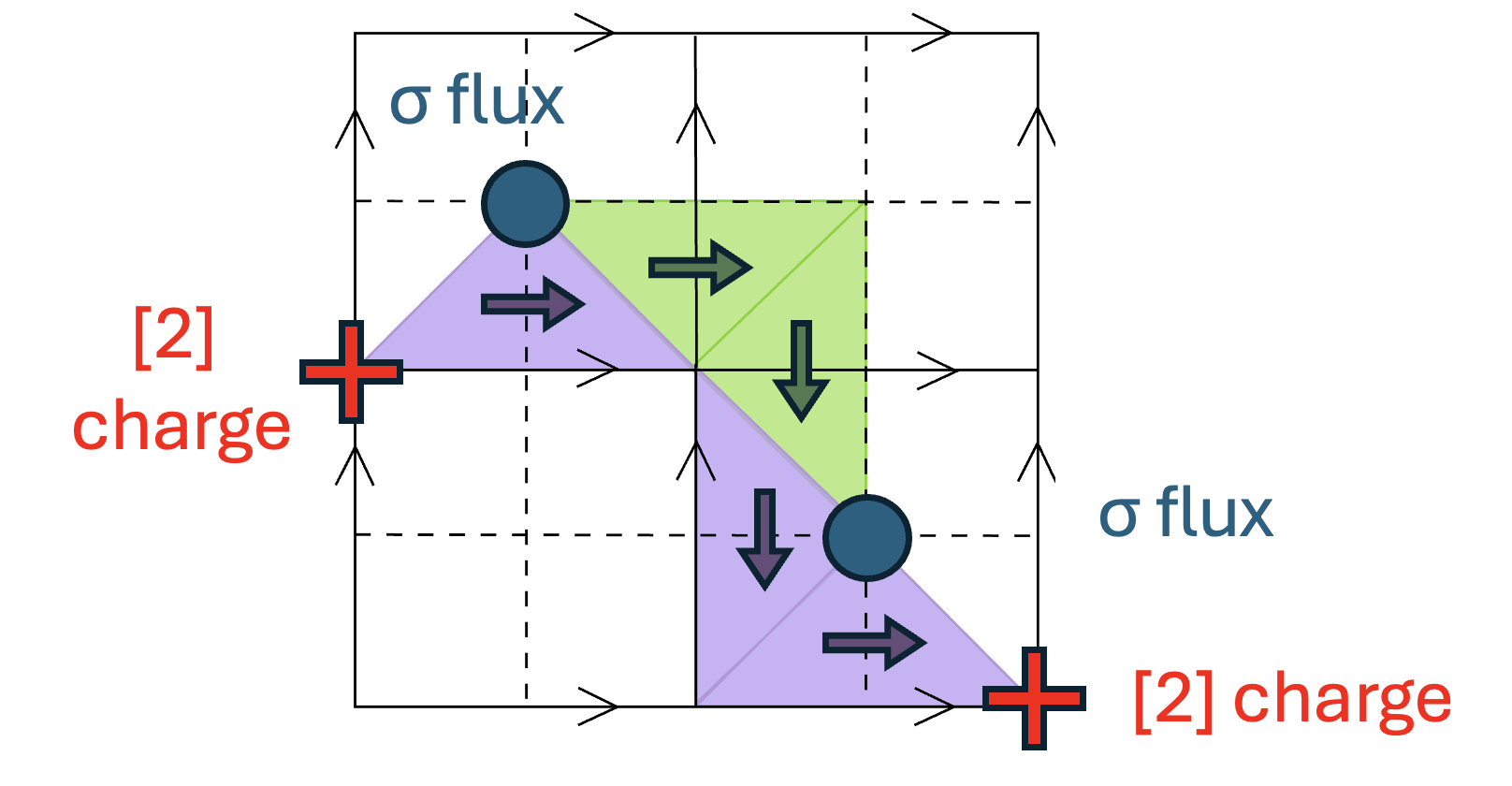}
\caption{Excitations created by the ribbon operator $F^{[C_2];\sigma,\sigma}$. In additional to the two $\sigma$ flux violations at the two plaquettes at the end of the ribbon, there are $[2]$ charge violations at the two ending vertices.}
\label{fig:excitation-of-ribbon}
\end{figure}

The remnant particle of fusing the two $\sigma$ fluxes is in the vacuum fusion channel with the two $[2]$ charge decorating at the ending vertices of the ribbon; hence, as a whole, the ribbon is a neutral pair. By similar argument, all 9 possible $C_2$ anyon basis ribbons create neutral pair; since the computational basis states are symmetric superposition of neutral pairs, they too can be created from vacuum in accordance with neutrality condition.

Let
\begin{equation}
\begin{aligned}
F^{[C_{2}];\sigma} & \equiv F^{[C_{2}];\sigma,\sigma}+F^{[C_{2}];\sigma,\mu\sigma}+F^{[C_{2}];\sigma,\bar{\mu}\sigma}\\
F^{[C_{2}];\mu\sigma} & \equiv F^{[C_{2}];\mu\sigma,\sigma}+F^{[C_{2}];\mu\sigma,\mu\sigma}+F^{[C_{2}];\mu\sigma,\bar{\mu}\sigma}\\
F^{[C_{2}];\bar\mu\sigma} & \equiv F^{[C_{2}];\bar\mu\sigma,\sigma}+F^{[C_{2}];\bar\mu\sigma,\mu\sigma}+F^{[C_{2}];\bar\mu\sigma,\bar{\mu}\sigma}
\end{aligned}
\end{equation}
above linear superposition of ribbon operators have the special property that there is no charge at the end vertex, since it is a symmetric superposition over all possible end internal states. We will use this state to encode the logical $\ket{0}$ state. 
So the concrete implementation of the $\mathcal{Z}$ basis encoding is:
\begin{equation}
\begin{aligned}
\ket{0} &= F^{[C_{2}];\sigma}(t_1)F^{[C_{2}];\sigma}(t_2) \ket{GS} \\
\ket{1} &= F^{[C_{2}];\mu\sigma}(t_1)F^{[C_{2}];\mu\sigma}(t_2) \ket{GS} \\
\ket{2} &= F^{[C_{2}];\bar{\mu}\sigma}(t_1)F^{[C_{2}];\bar{\mu}\sigma}(t_2) \ket{GS} 
\end{aligned}
\end{equation}
where $t_1$ and $t_2$ are two ribbons with one end overlapped at the origin, and $\ket{GS}$ denotes a ground state of the $S_3$ quantum double Hamiltonian. This provides the lattice-level prescription for creating the logical dual basis states defined in Eq. \ref{eq:computational-basis-def}.

\subsubsection{Dual basis on the lattice}

The dual basis states can be initialized by appropriate superposition of the flux basis ribbon operators:
\begin{equation}
\begin{aligned}
\ket{\tilde0} &= \frac{1}{\sqrt3}(F^{[C_{2}];\sigma}(t_1)F^{[C_{2}];\sigma}(t_2) + F^{[C_{2}];\mu\sigma}(t_1)F^{[C_{2}];\mu\sigma}(t_2)+F^{[C_{2}];\bar\mu\sigma}(t_1)F^{[C_{2}];\bar\mu\sigma}(t_2) \ket{GS} \\
\ket{\tilde1} &= \frac{1}{\sqrt3}(F^{[C_{2}];\sigma}(t_1)F^{[C_{2}];\sigma}(t_2) +\omega F^{[C_{2}];\mu\sigma}(t_1)F^{[C_{2}];\mu\sigma}(t_2)+ \omega^* F^{[C_{2}];\bar\mu\sigma}(t_1)F^{[C_{2}];\bar\mu\sigma}(t_2) \ket{GS} \\
\ket{\tilde2} &= \frac{1}{\sqrt3}(F^{[C_{2}];\sigma}(t_1)F^{[C_{2}];\sigma}(t_2) + \omega^* F^{[C_{2}];\mu\sigma}(t_1)F^{[C_{2}];\mu\sigma}(t_2)+ \omega F^{[C_{2}];\bar\mu\sigma}(t_1)F^{[C_{2}];\bar\mu\sigma}(t_2) \ket{GS} 
\end{aligned}
\end{equation}
where $t_1$ and $t_2$ are two ribbons with one end overlapped at the origin, and $\ket{GS}$ denotes a ground state of the $S_3$ quantum double Hamiltonian. This provides the lattice-level prescription for creating the logical dual basis states defined in Eq. \ref{eq:dual-basis-def}.

\section{Conclusion}

    We have reviewed the basics of $S_3$ topological order, the logical encoding, and the universal set of qubit gates that can be explicitly implemented using braiding and fusion of $S_3$ anyons. We hope these notes, by consolidating several previous works, will help clarify and concretize how quantum double models work and how computation is done with non-Abelian anyons. As explained in the introduction, there are compelling reasons to focus on $S_3$ anyons in particular---they are powerful enough to host universal computation, while being simple enough to realize in current quantum devices. We have shown how to manipulate anyons with ribbon operators to implement the three operations that constitute the universal gate set: pull-through gate, qutrit $\mathcal{Z}$ basis measurement, and qutrit $\mathcal{X}$ basis measurement. We demonstrate the universality of the above three operations by implementing the \textit{qubit} gate set of Clifford gates and a non-Clifford gate (CCZ gate); we are not aware of such demonstration in the literature. Another new result that we present is an extended ribbon operator that is useful for state initialization on the lattice. Recently, it has been shown by us and by other authors how these ribbon operators can be made finite depth for solvable groups using adaptive circuits \cite{bravyi_adaptive_2022, Ren24,lyons_protocols_2024}, allowing for more scalable computation. There has also been enormous progress on the experimental side. Various non-Abelian phases have been produced on near-term quantum devices \cite{iqbal_non-abelian_2024, Xu_2024, minev2024realizingstringnetcondensationfibonacci}; non-Abelian defects on top of an Abelian topological order have also been created \cite{andersen_non-abelian_2023, Xu22,iqbal2024qutrittoriccodeparafermions}. By combining theory proposals for how to prepare $S_3$ topological order in near-term quantum devices \cite{verresen_efficiently_2022} with the recipe we have outlined in these notes, we demonstrate that a scalable implementation of the universal gate set is within reach.

    We note that there are currently two ways to realize topological order in experiment: (1) engineer or stabilize a Hamiltonian whose ground state is topologically-ordered, as in quantum Hall systems, or (2) prepare the ground state directly without a Hamiltonian in a digital quantum device. Our notes are written with the second scenario in mind, as quantum double models are currently more readily accessible in these setups. In such scenarios, instead of passively cooling the system into its ground state space, one would actively `cool' by means of error correction, whereby only locally measures the system and applies feedback to lower the effective energy \cite{dennis_topological_2002,terhal_quantum_2015}.
    We leave the stabilization protocols of these states to future work---we note that error correction of non-Abelian states is an active and rapidly-developing field \cite{Wootton_2014,brell_thermalization_2014,wootton_active_2016,knapp_nature_2016,zhu_universal_2020,zhu_instantaneous_2020,schotte_2022,schotte_fault-tolerant_2022,sala_decoherence_2024,sala_stability_2024,hsin_non-abelian_2024}.

The question remains: given near-term, noisy devices, how can we make such a gate set robust? While the non-Abelian properties of our logical anyons enable universal quantum computation, they also make error correction difficult; unlike an Abelian code, we cannot pair up anyons and be guaranteed they will fuse to the vacuum. Some work has been done exploring error-correction protocols for non-Abelian topological orders \cite{Wootton_2014, Hutter_2015, Dauphinais_2017, schotte_2022}, and more recent work has shown how to pair up non-Abelian anyons in solvable quantum doubles in constant time using adaptive quantum circuits \cite{bravyi_adaptive_2022,lyons_protocols_2024,Ren24}, but we still understand relatively little about the theory of non-Abelian error correction. Recently, it has been shown that, counter-intuitively, some non-Abelian orders may be significantly more robust under certain types of local noise than their Abelian counterparts \cite{sala2024decoherencewavefunctiondeformationd4, sala2024stabilityloopmodelsdecohering,Chirame24}, which is an encouraging sign for the existence of error-correction protocols for such codes. Another open question regards the fault-tolerance of such protocols in the presence of measurement errors. We leave these questions for future study.

\subsection{Further reading}

These notes are meant to be a relatively self-contained overview of universal topological quantum computation using the $S_3$ quantum double as a platform; here we list some further resources that discuss in more depth various aspects of the field we did not focus on.

\begin{itemize}
    \item Useful resources on topological quantum computing and topological order in general:
        \begin{itemize}
            \item John Preskill's lecture notes on topological quantum computation \cite{preskill_lecture_2004}, which can be found on the \href{http://theory.caltech.edu/~preskill/ph229/}{course website} for Physics 219 at Caltech. These very clear and concise notes cover the fundamentals of topological quantum computation. 
            
            \item Steve Simon's \emph{Topological Quantum} \cite{simon_topological_2023}, which is a great resource for a broad set of topics related to topological quantum phenomena, and also includes a chapter on quantum double models.
            
            \item Alexei Kitaev's original paper, \emph{Fault-tolerant quantum computation by anyons}, on topological quantum computing \cite{kitaev_fault-tolerant_2003}. This paper lays out all the essential ideas of topologically-protected computing with anyons. In particular, Kitaev makes very clear the set-up needed to ensure the internal states of anyons are protected from local noise. 

            \item \emph{Monoidal Categories and Topological Field Theory} by Vladimir Turaev and Alexis Virelizier\cite{turaevbook}. This is an exposition of topological field theories and the mathematics underpinning them; it is written by and for mathematicians, but is very clear and contains useful exercises.
        \end{itemize}
    \item More detailed treatments of the original ribbon operators:
        \begin{itemize}
            \item Bombin and Martin-Delgado, {\emph{Family of non-Abelian Kitaev models on a lattice: Topological condensation and confinement}} \cite{bombin_family_2008}. See section II and appendix B for a detailed treatment of both quantum doubles and ribbon operators.
            
            \item Shawn Cui's lecture notes from his Topological Quantum Computation course at Stanford \cite{cui_topological_2018}, which can be found \href{https://www.math.purdue.edu/~cui177/Lecture_Combined.pdf}{here}.
            These notes take a more mathematical perspective, and give a great overview of quantum doubles, ribbon operators, and how to do quantum computation in general with non-Abelian anyons. Additionally, there are some sections on Unitary Modular Tensor Categories (UMTCs), which are the mathematical objects that describe general anyon theories. 

            \item Chen, Cui, and Yan, \emph{Ribbon operators in the generalized Kitaev quantum double model based on Hopf algebras} \cite{yan_ribbon_2022}. This paper introduced the local orientation of a ribbon operator, distinguishing clockwise and counterclockwise ribbon operators. They treat ribbon operators from a very mathematical point of view, and so this paper is a good resource for any reader wanting to learn about the quantum double model in the general setting of Hopf algebras.
        \end{itemize}
    
    \item Other implementations of universal gate sets with $S_3$ anyons:
        \begin{itemize}
            \item Mochon, \emph{Anyon computers with smaller groups} \cite{mochon_anyon_2004}. This paper outlines a universal \emph{qutrit} gate set instead of a qubit one; see also Mochon's paper on a generic gate set for nonsolvable finite groups \cite{mochon_anyons_2003}.
            
            \item Cui, Hong, and Wang, \emph{Universal quantum computation with weakly integral anyons} \cite{cui_universal_2015}. This paper provides a variety of possible encodings for $S_3$, distinct from the encoding discussed here.
        \end{itemize}
\end{itemize}

\section*{Acknowledgments}

\textbf{Note added}: While completing this manuscript, we became aware of an independent work by L. Chen, Y. Ren, R. Fan, and A. Jaffe demonstrating universal gate set in $S_3$ quantum double using a different logical encoding proposed in \cite{chen_s3_universal_2024}; we thank the authors for informing us of their work.

\paragraph{Author contributions}
C.F.B.L. and A.L. contributed equally to this project. 

\paragraph{Funding information}

C.F.B.L. and A.L. acknowledge support from the National Science Foundation Graduate Research Fellowship Program (NSF GRFP). This work is in part supported by the DARPA MeasQuIT program.  N.T. is supported by the Walter Burke Institute for Theoretical Physics at Caltech. A.V. is supported by NSF-DMR 2220703 and A.V. and R.V. are supported by the Simons Collaboration on Ultra-Quantum Matter, which is a grant from the Simons Foundation (618615, A.V.).

\pagebreak
\begin{appendix}
\numberwithin{equation}{section}

\section{Useful S3 Group Theory} \label{group-theory}

\subsection{Group Structure}

$S_3$ is the symmetric group on three elements. Geometrically, we can also view it as the group of symmetry of an equilateral triangle. As a reference for the reader, we describe a few basic facts, computational tools, and conventions below.

As a set,
\begin{equation}
    S_3=\{e, (12), (13), (23), (123), (132)\}, 
\end{equation}
where the elements are cycles describing how elements are interchanged; for example, $(12)$ means switching element 1 and 2, and $(123)$ means sending element 1 to 2, 2 to 3, and 3 to 1. We will sometimes denote $\mu = (123)$ and $\sigma=(23)$, which generates the group. So the group presentation for $S_3$ is:
\begin{equation}
    S_3 = \left<\mu,\sigma|\mu\sigma = \sigma\bar\mu\right>.
\end{equation}
We will use the following convention for the order of operation: when converting from multiple 2-cycles into permutations, the order of operation is from right to left. E.g. for $(12)(23)$, $(23)$ is performed first, then $(12)$.

\subsubsection{S3 is solvable}

\begin{defi}[Solvable group]
    A group $G$ is \textbf{solvable} if $\exists$ normal series 
    $$\{e\}=G_0\triangleleft G_1 \dots G_N=G$$ 
    such that $G_{i+1}/G_i$ is Abelian.
\end{defi}

\begin{prop}
    $S_3$ is solvable.
\end{prop}

\begin{proof}
Consider the following series: 
$$
\qty{e}\triangleleft \mathbb Z_3 \triangleleft S_3.
$$
We check that $G_{i+1}/G_i$ is Abelian throughout the series. 
\begin{align*}
    G_2/G_1 = S_3/\mathbb Z_3 = \mathbb Z_2\\    
    G_1/G_0 =\mathbb Z_3/\qty{e} = \mathbb Z_3
\end{align*}

Therefore, above is a normal series, so $S_3$ is indeed solvable.    
\end{proof}

\subsubsection{S3 is not nilpotent}

\begin{defi}[Nilpotent group]
    A group $G$ is \textbf{nilpotent} if it has a upper central series of finite length; i.e., if there exists a set of normal subgroups $G_0, G_1, \dots G_N$ such that:
$$\{e\} = G_0 \triangleleft G_1 \triangleleft \cdots \triangleleft G_N = G$$
where $G_{i+1}/G_i = Z(G/G_i)$.
\end{defi}

\begin{prop}
$S_3$ is not nilpotent.
\end{prop}

\begin{proof}

We will show that no such upper central series exist.
The normal subgroups of $S_3$ are $\qty{e},\qty{e,(123),(132)}\cong \mathbb Z_3, S_3$. Therefore, the only possible normal series are:
\begin{align*}
    \qty{e}\triangleleft S_3\\
    \qty{e}\triangleleft \mathbb Z_3\triangleleft S_3
\end{align*}
For the first normal series, as $G_1/G_0=S_3\neq Z(S_3/\qty{e})=\qty{e}$, so it is not a central series. For the second series, $G_1/G_0=\mathbb Z_3/\qty{e} = \mathbb Z_3 \neq \qty{e}=Z(S_3/\qty{e}) \implies G_1/G_0 \neq Z(G/G_0)$, so it is also not a central series. Therefore, a upper central series doesn't exist; so $S_3$ is not nilpotent.
\end{proof}

\subsection{Representation theory of $\mathbb{Z}_3$} \label{rep-theory-z3}

Since $\mathbb{Z}_3$ is an Abelian group, by Schur's lemma, all of its irreducible representations are 1D. There are three irreducible representations of $\mathbb{Z}_3$:

\begin{itemize}
    \item \textbf{$[1]$ representation:} This is the trivial representation maps all group elements to $1\in \mathbb C$.
    \item \textbf{$[\omega]$ representation:} This representation is given by the map $a \mapsto \omega^a$. 
    \item \textbf{$[\omega^*]$ representation:} This representation is given by the map $a \mapsto (\omega^{*})^{a}$. 
\end{itemize}

One way to derive the three irreducible representation is to start with the regular representation of $\mathbb{Z}_3$, where the representation matrices are 
\begin{equation*}
        \rho_d(0)=
        \begin{pmatrix}
            1 & 0 & 0\\
            0 & 1 & 0\\
            0 & 0 & 1
        \end{pmatrix}, \quad\quad
        \rho_d(1)=
        \begin{pmatrix}
            0 & 0 & 1\\
            1 & 0 & 0\\
            0 & 1 & 0
        \end{pmatrix}, \quad\quad
        \rho_d(2)=
        \begin{pmatrix}
            0 & 1 & 0\\
            0 & 0 & 1\\
            1 & 0 & 0
        \end{pmatrix}
    \end{equation*}

To find the irreducible representations of an Abelian group, we must simultaneously diagonalize these representation matrices. The eigenvectors of $\rho_d(1)$ are:
\begin{align*}
(1,1,1)^T\\
(1,\omega, \omega^*)^T\\
(1,\omega^*, \omega)^T
\end{align*}
where $\omega=e^{i2\pi/3}$. We see that the three irreducible representations given above just corresponding to the action of $\mathbb{Z}_3$ group elements within each eigensubspace.

\subsection{Representation theory of S3} \label{rep-theory-s3}

$S_3$ has three irreducible representations: the trivial, the alternating (also called the sign representation), and the standard representation. We can get the trivial and standard representation from the defining representation; in other words, the defining representation is a direct sum of the trivial and the standard representation.

\begin{itemize}
    \item \textbf{Trivial representation (irrep):} The trivial representation maps all group elements to $1\in \mathbb C$, so it is a 1D representation.
    
    \item \textbf{Alternating representation (irrep):} The alternating representation maps element $g\in S_3$ to the sign of the permutation $\text{sgn}(g)\in \{-1, 1\}$, so it is also a 1D representation. The sign of a group element is determined by whether it decomposed into even or odd number of 2-cycles; for elements in $S_3$, the 2-cycles are the odd permutations, whereas the 3-cycles are the even permutations.
    
    \item \textbf{Defining representation (not irrep):} To motivate the basis choice in the standard representation (which will be introduced next), we consider the defining representation of $S_3$. We will see later that the defining representation is reducible to two irreducible representation: the trivial and the standard representation. 

    Representation matrices in the defining representation of $S_3$ act by permuting 3 coordinates in $\mathbb R^3$. The representation matrices are
    \begin{align*}
        \rho_d(e)=
        \begin{pmatrix}
            1 & 0 & 0\\
            0 & 1 & 0\\
            0 & 0 & 1
        \end{pmatrix}, \quad\quad
        \rho_d(123)&=
        \begin{pmatrix}
            0 & 0 & 1\\
            1 & 0 & 0\\
            0 & 1 & 0
        \end{pmatrix}, \quad\quad
        \rho_d(132)=
        \begin{pmatrix}
            0 & 1 & 0\\
            0 & 0 & 1\\
            1 & 0 & 0
        \end{pmatrix} \quad\quad \\
        \rho_d(12)=
        \begin{pmatrix}
            0 & 1 & 0\\
            1 & 0 & 0\\
            0 & 0 & 1
        \end{pmatrix}, \quad\quad
        \rho_d(13)&=
        \begin{pmatrix}
            0 & 0 & 1\\
            0 & 1 & 0\\
            1 & 0 & 0
        \end{pmatrix}, \quad\quad
        \rho_d(23)=
        \begin{pmatrix}
            1 & 0 & 0\\
            0 & 0 & 1\\
            0 & 1 & 0
        \end{pmatrix}, \quad\quad
    \end{align*}
    These matrices act on the three standard basis vectors:
    \begin{equation*}
        \ket{1}=(1,0,0)^T, ~\ket{2}=(0,1,0)^T, ~\ket{3}= (0,0,1)^T
    \end{equation*}
    
    To decompose the defining representation into irreducible representations, we can make use of the Abelian subgroup $\mathbb Z_3\subset S_3$. The eigenvectors of a generator of the Abelian subgroup (e.g. the 3-cycle $(123)$) give a nice orthogonal basis for the full vector space of the defining representation (since the representation of an Abelian group has to be 1D). The group elements outside of the Abelian subgroup couple these 1D representations to give a higher-dimensional irreducible representation.
    
    The eigenvectors of $\rho_d(123)$ are:
    \begin{align*}
    (1,1,1)^T\\
    (1,\omega, \omega^*)^T\\
    (1,\omega^*, \omega)^T
    \end{align*}
    where $\omega=e^{i2\pi/3}$. We note that the first eigenvector $(1,1,1)^T$ spans the subspace for the trivial representation, hence leaving the other two vectors to be the basis for an irreducible 2D representation, namely the standard representation.

    \item \textbf{Standard representation (irrep): }The standard representation is a 2D representation. Define the basis vectors (as motivated above, they are eigenvectors of the $\rho_d(123)$ matrix) to be
    \begin{align*}
        \ket{2+}=\frac{1}{\sqrt{3}}(1,\omega^*, \omega)^T\\
        \ket{2-}=\frac{1}{\sqrt{3}}(1,\omega, \omega^*)^T
    \end{align*}
    
    The representation matrices can be calculated by looking at how the defining representation matrices act on $\ket{2+}$ and $\ket{2-}$ in the defining basis. The results are: 
   \begin{align*}
        \rho_s(e)=\begin{pmatrix}
            1 & 0\\
            0 & 1\\
        \end{pmatrix}\quad\quad 
        \rho_s(123)&=\begin{pmatrix}
            \omega & 0\\
            0 & \omega^*\\
        \end{pmatrix}\quad\quad
        \rho_s(132)=\begin{pmatrix}
            \omega^* & 0\\
            0 & \omega\\
        \end{pmatrix}\quad\quad \\
        \rho_s(12)=\begin{pmatrix}
            0 & \omega\\
            \omega^* & 0\\
        \end{pmatrix}\quad\quad
        \rho_s(13)&=\begin{pmatrix}
            0 & \omega^*\\
            \omega & 0\\
        \end{pmatrix}\quad\quad
        \rho_s(23)=\begin{pmatrix}
            0 & 1\\
            1 & 0\\
        \end{pmatrix}\quad\quad
    \end{align*}
    where the indexing order is $\{\ket{2+},\ket{2-}\}$.
    
\end{itemize}

\section{Charge Transfer} \label{ap:charge_transfer}

Using representation theory, we can determine the probability of any pair of anyons to fuse to the vacuum. The probability of charge transfer during a protocol like the one discussed in this section is connected to the character of the representation of the charge (the trace of the matrices assigned to the group elements in that particular representation).

For a given representation, the state corresponding to a totally neutral pair of charges (no flux, no charge, the pair will always fuse into the vacuum) is given by the following superposition:
\begin{equation}
    \ket{+}_R = \frac{1}{\sqrt{\abs{R}}} \sum_{i \in R} \ket{i}_R \otimes \ket{i}_{R^*}
    \label{eq:neutral-state-R-rep}
\end{equation}
where $R^*$ is the conjugate representation to $R$\footnote{In non-rigorous terms, conjugate representations are ones that include the identity in their fusion outcomes:
\begin{equation*}
    R \otimes R^* = \mathbb{I} \oplus \cdots 
\end{equation*}
If an anyon is described by representation $R$, its anti-particle will be described by $R^*$, since we expect at least some of the time they will fuse to the vacuum when brought together. In the case of $S_3$, $R$ is isomorphic to $R^*$ for all anyon types.}. We can see that this state transforms trivially under the action of any flux; whatever happens to $\ket{i}_R$, the opposite will happen to $\ket{i}_{R^*}$ and will cancel out. The key part of this, however, is whatever flux we are braiding with should go around \emph{both} charges in the pair. If we instead only wind the flux around half of the pair, charge can be transferred from one particle to the winding flux. Consider the state of the charges after one of them is braided with flux $a$. The flux will act on the state of the charge it is winding with according to the representation $R$ (or $R^*$, but for this example assume we are winding with the $R$ charge-- often $R$ is isomorphic to $R^*$, as is the case for all the irreducible representations of $S_3$).

When only the $R$ charge is braided with flux $a$, the wavefunction of the charge pair becomes
\begin{equation}
    \ket{+}_R \rightarrow \ket{a}_R = \frac{1}{\sqrt{\abs{R}}} \sum_{i, j \in R} D_{ij}^R(a) \ket{j}_R \otimes \ket{i}_{R^*}
\end{equation}
where $D^R(a)$ is the matrix corresponding to $a$ in representation $R$. We want to find the probability that this new state will fuse to the vacuum, so we take the inner product with equation \ref{eq:neutral-state-R-rep}:
\begin{equation}
    \begin{aligned}
        _R\bra{+} \ket{a}_R &= \frac{1}{\abs{R}} \sum_{i, j, k \in R} D_{ij}^R(a) \bra{k}\ket{j} \bra{k}\ket{i} \\
        &= \frac{1}{\abs{R}} \sum_{i, j, k \in R} D_{ij}^R(a) \delta_{kj} \delta_{ki}\\
        &= \frac{1}{\abs{R}} \sum_{i \in R} D_{ii}^R(a)\\
        &= \frac{\mathrm{Tr}(D^R(a))}{\abs{R}}
    \end{aligned}
\end{equation}

The trace of the matrix $D^R(a)$ is the same for all elements in a given conjugacy class\footnote{This can be proved in one or two lines via the cyclic properties of the trace.}. It is also known as the \emph{character} of the representation $R$ evaluated at the group element $a$, denoted by $\chi^R(a)$. So the probability that the new state will fuse to the vacuum is given by:
\begin{equation}
    \mathrm{Prob}(0) = \abs{\frac{\chi^R(a)}{\abs{R}}}^2
\end{equation}

We can check that $C_3$ fluxes do \emph{not} have character zero in the $[2]$ representation ($\chi^{[2]}(a \in C_3) = -1$). This means that even after braiding, there is always a $\frac{1}{4}$ chance of the $[2]$ singlet fusing back to the vacuum. Our computational basis measurements are not completely projective--- we need to repeat them a few times to reduce our chance of getting a false negative below an acceptable threshold.

\section{Ribbon Operator Identities}
\subsection{Microscopic to Anyon Basis Transformation} \label{ap:basis-change}

The basis transformation from the microscopic ($\ket{z, w}$) to the anyon basis can be written out explicitly in the following form: 
\begin{equation}
    \ket{C, R; u, u'} = \frac{|R|}{|\mathbf Z(r)|} \sum_{n \in \mathbf{Z}(r)} \Gamma^R_{j j'}(n) \ket{q_c n \overline{q}_{c'}, c'}
\end{equation}
where $\textbf{u} = (c, j)$ and $\textbf{u'} = (c', j')$ where $c, c' \in C$ and $1 \leq j, j' \leq |R|$. The pair $u$ determines the flavor of the state, while $u'$ determines the color. The matrix $\Gamma^R(n)$ is the representation of group element $n$ in $R$. The group elements $q_c$ and $q_{c'}$ are defined in the following way with respect to $c, c'$ and some chosen representative $r$ of the conjugacy class $C$:
\begin{equation}
\begin{aligned}
    c &= q_c r \overline{q}_c\\
    c' &= q_{c'} r \overline{q}_{c'}
\end{aligned}
\end{equation}

An analogous transformation must hold for the ribbon operators that create microscopic basis states and the ribbons that create anyons:
\begin{equation}
    F^{(R, C); \mathbf{u}, \mathbf{u'}} = \frac{|R|}{|\mathbf Z(r)|} \sum_{n \in \mathbf{Z}(r)} \Gamma^R_{j j'}(n) F^{(q_c n \overline{q}_{c'}, c)}
\end{equation}
Note that our ribbon operators are defined in terms of $z$ and $v$ rather than $z$ and $w$, hence $F$ depends on $c$ rather than $c'$. 

\subsection{Commutation Relations} \label{ap:commutations}

\subsubsection{Operators acting on the same ribbon} \label{ap:same-ribbon}

\begin{prop}
    For two ribbon operators acting on the same ribbon $t$,
    \begin{equation}
        F^{(z_1, v_1)}(t)F^{(z_2, v_2)}(t)=
        \begin{cases}
            \delta_{z_1,z_2}F^{(z_1, v_2v_1)}(t) \quad \quad \quad & \text{for clockwise local orientation}\\
            \delta_{z_1,z_2}F^{(z_1, v_1v_2)}(t) \quad \quad \quad & \text{for counterclockwise local orientation}\\
        \end{cases}
    \end{equation}
\end{prop}

\begin{proof}
    Based on the bottom-up approach of constructing ribbon, the effect of applying a full ribbon operator on the charge degree of freedom is to project the $z$ string to the particular value of $z_i$ specified in the ribbon operator. Therefore, we have the action $\delta_{z,z_1}\delta_{z,z_2}=\delta_{z_1,z_2}\delta_{z,z_1}$. We can group the Kronecker delta $\delta_{z,z_1}$ into the construction of a new ribbon operator $F^{(z_1,\cdot)}(t)$, where the flux information is yet to be specified. 
    
    For the flux degree of freedom, the key is that the action (determined by alignment and local orientation) is the same across $F^{(z_1,v_1)}(t)$ and $F^{(z_2,v_2)}(t)$ since they are acting on the same ribbon. However, we need to consider the clockwise and counterclockwise local orientation cases separately.

    \textbf{Case 1: clockwise local orientation.}
    The first action by $F^{(z_2,v_2)}(t)$ is to either left multiply by $\bar z' \bar v_2 z'$ for an aligned triangle or to right multiply by $\bar z' v_2 z'$ for an opposite one. Then, the action by $F^{(z_1,v_1)}(t)$ is to left multiply by $\bar z' \bar v_1 z'$ for aligned or right multiply $\bar z' v_1 z'$ for opposite triangles. Therefore, the composed action is either left multiplication by $(\bar z' \bar v_1 z')(\bar z' \bar v_2 z')=\bar z' \overline{v_2v_1} z'$ for aligned triangles or right multiplication by $(\bar z' v_2 z')(\bar z' v_1 z')=\bar z' v_2v_1 z'$ for opposite triangles. Therefore, the composed action is equivalent to the action of $\delta_{z_1,z_2} F^{(z_1,v_2v_1)}(t)$.

    \textbf{Case 2: counterclockwise local orientation.}
    The first action by $F^{(z_2,v_2)}(t)$ is either to left multiply by $\bar z' v_2 z'$ for opposite or right multiply by $\bar z' \bar v_2 z'$ for aligned triangles. Then, the action by $F^{(z_1,v_1)}(t)$ is to left multiply by $\bar z' v_1 z'$ for opposite or right multiply by $\bar z' \bar v_1 z'$ for aligned triangle. Therefore, the composed action is either left multiplication by $(\bar z' v_1 z')(\bar z' v_2 z')=\bar z' v_1v_2 z'$ for opposite or right multiplication by $(\bar z' \bar v_2 z')(\bar z' \bar v_1 z')=\bar z' \overline{v_1 v_2} z'$. Therefore, the composed action is equivalent to the action of $\delta_{z_1,z_2} F^{(z_1,v_1v_2)}(t)$.
\end{proof}

\subsubsection{Commutation with projectors in the middle of ribbon} \label{ap:middle-ribbon}

We want to verify the intuition from the toric code that projectors in the middle of a ribbon should commute with the ribbon operator. It's instructive to look at an generic example of ribbon, as pictured in Fig. \ref{fig:bent-ribbon}.

\begin{figure}[h!]
\centering
\includegraphics[width=0.8\textwidth]{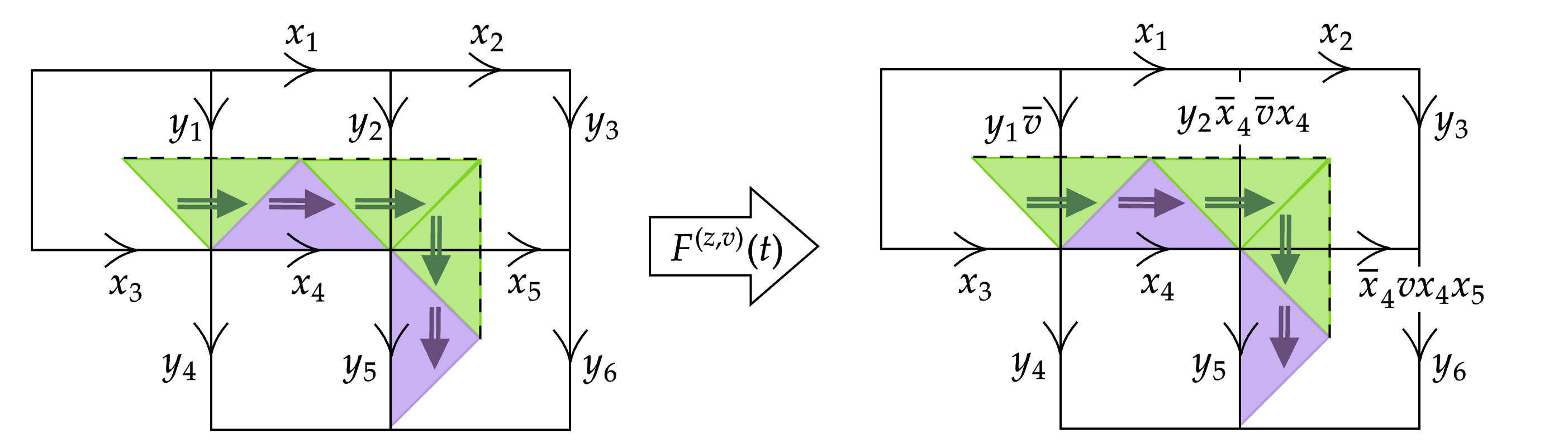}
\caption{A bent ribbon operator with two consecutive dual triangle operators near the end.}
\label{fig:bent-ribbon}
\end{figure}

We can check the commutation with plaquette operators for a few cases: 1) a straight ribbon segment within a plaquette, 2) a bent ribbon segment within a plaquette, 3) with a vertex operator.

\textbf{Case 1: a straight ribbon within a plaquette.} Consider the second plaquette where edges are elements $x_1,y_2,x_4,y_1$ (clockwise from the top). First we show plaquette operators $[B^e,F^{(z, v)}]=0$:
\begin{equation*}
\begin{aligned}
B^e F^{(z, v)}\left(
\begin{gathered}
\begin{tikzpicture}[x=0.75pt,y=0.75pt,yscale=-1,xscale=1]
\draw   (110.2,100.2) -- (160.2,100.2) -- (160.2,150.2) -- (110.2,150.2) -- cycle ;
\draw   (130.72,95.36) .. controls (133.73,97.98) and (136.75,99.56) .. (139.76,100.08) .. controls (136.75,100.6) and (133.73,102.18) .. (130.72,104.8) ;
\draw   (130.72,145.36) .. controls (133.73,147.98) and (136.75,149.56) .. (139.76,150.08) .. controls (136.75,150.6) and (133.73,152.18) .. (130.72,154.8) ;
\draw   (164.76,120.36) .. controls (162.14,123.37) and (160.56,126.39) .. (160.04,129.4) .. controls (159.52,126.39) and (157.94,123.37) .. (155.32,120.36) ;
\draw   (114.76,120.76) .. controls (112.14,123.77) and (110.56,126.79) .. (110.04,129.8) .. controls (109.52,126.79) and (107.94,123.77) .. (105.32,120.76) ;
\draw  [color={rgb, 255:red, 0; green, 0; blue, 0 }  ,draw opacity=1 ][fill={rgb, 255:red, 0; green, 0; blue, 0 }  ,fill opacity=1 ] (108.57,150.2) .. controls (108.57,149.3) and (109.3,148.57) .. (110.2,148.57) .. controls (111.1,148.57) and (111.83,149.3) .. (111.83,150.2) .. controls (111.83,151.1) and (111.1,151.83) .. (110.2,151.83) .. controls (109.3,151.83) and (108.57,151.1) .. (108.57,150.2) -- cycle ;

\draw (128,83.6) node [anchor=north west][inner sep=0.75pt]    {$x_{1}$};
\draw (168.6,120) node [anchor=north west][inner sep=0.75pt]    {$y_{2}$};
\draw (87.4,120) node [anchor=north west][inner sep=0.75pt]    {$y_{1}$};
\draw (128,157) node [anchor=north west][inner sep=0.75pt]    {$x_{4}$};
\end{tikzpicture}
\end{gathered}\right) &= \delta_{z,x_4x_5}B^e\left(
\begin{gathered}
\begin{tikzpicture}[x=0.75pt,y=0.75pt,yscale=-1,xscale=1]
\draw   (110.2,100.2) -- (160.2,100.2) -- (160.2,150.2) -- (110.2,150.2) -- cycle ;
\draw   (130.72,95.36) .. controls (133.73,97.98) and (136.75,99.56) .. (139.76,100.08) .. controls (136.75,100.6) and (133.73,102.18) .. (130.72,104.8) ;
\draw   (130.72,145.36) .. controls (133.73,147.98) and (136.75,149.56) .. (139.76,150.08) .. controls (136.75,150.6) and (133.73,152.18) .. (130.72,154.8) ;
\draw   (164.76,120.36) .. controls (162.14,123.37) and (160.56,126.39) .. (160.04,129.4) .. controls (159.52,126.39) and (157.94,123.37) .. (155.32,120.36) ;
\draw   (114.76,120.76) .. controls (112.14,123.77) and (110.56,126.79) .. (110.04,129.8) .. controls (109.52,126.79) and (107.94,123.77) .. (105.32,120.76) ;
\draw  [color={rgb, 255:red, 0; green, 0; blue, 0 }  ,draw opacity=1 ][fill={rgb, 255:red, 0; green, 0; blue, 0 }  ,fill opacity=1 ] (108.57,150.2) .. controls (108.57,149.3) and (109.3,148.57) .. (110.2,148.57) .. controls (111.1,148.57) and (111.83,149.3) .. (111.83,150.2) .. controls (111.83,151.1) and (111.1,151.83) .. (110.2,151.83) .. controls (109.3,151.83) and (108.57,151.1) .. (108.57,150.2) -- cycle ;

\draw (128,83.6) node [anchor=north west][inner sep=0.75pt]    {$x_{1}$};
\draw (168.6,120) node [anchor=north west][inner sep=0.75pt]    {$y_2 \bar x_4 \bar v x_4$};
\draw (80,120) node [anchor=north west][inner sep=0.75pt]    {$y_{1}\bar v$};
\draw (128,157) node [anchor=north west][inner sep=0.75pt]    {$x_{4}$};
\end{tikzpicture}
\end{gathered}
\right) \\
&= \delta_{z,x_4x_5}\delta_{e, x_4 \bar y_2 \bar x_1 y_1}\left(
\begin{gathered}
\begin{tikzpicture}[x=0.75pt,y=0.75pt,yscale=-1,xscale=1]
\draw   (110.2,100.2) -- (160.2,100.2) -- (160.2,150.2) -- (110.2,150.2) -- cycle ;
\draw   (130.72,95.36) .. controls (133.73,97.98) and (136.75,99.56) .. (139.76,100.08) .. controls (136.75,100.6) and (133.73,102.18) .. (130.72,104.8) ;
\draw   (130.72,145.36) .. controls (133.73,147.98) and (136.75,149.56) .. (139.76,150.08) .. controls (136.75,150.6) and (133.73,152.18) .. (130.72,154.8) ;
\draw   (164.76,120.36) .. controls (162.14,123.37) and (160.56,126.39) .. (160.04,129.4) .. controls (159.52,126.39) and (157.94,123.37) .. (155.32,120.36) ;
\draw   (114.76,120.76) .. controls (112.14,123.77) and (110.56,126.79) .. (110.04,129.8) .. controls (109.52,126.79) and (107.94,123.77) .. (105.32,120.76) ;
\draw  [color={rgb, 255:red, 0; green, 0; blue, 0 }  ,draw opacity=1 ][fill={rgb, 255:red, 0; green, 0; blue, 0 }  ,fill opacity=1 ] (108.57,150.2) .. controls (108.57,149.3) and (109.3,148.57) .. (110.2,148.57) .. controls (111.1,148.57) and (111.83,149.3) .. (111.83,150.2) .. controls (111.83,151.1) and (111.1,151.83) .. (110.2,151.83) .. controls (109.3,151.83) and (108.57,151.1) .. (108.57,150.2) -- cycle ;

\draw (128,83.6) node [anchor=north west][inner sep=0.75pt]    {$x_{1}$};
\draw (168.6,120) node [anchor=north west][inner sep=0.75pt]    {$y_2 \bar x_4 \bar v x_4$};
\draw (80,120) node [anchor=north west][inner sep=0.75pt]    {$y_{1}\bar v$};
\draw (128,157) node [anchor=north west][inner sep=0.75pt]    {$x_{4}$};
\end{tikzpicture}
\end{gathered}
\right)
\end{aligned}
\end{equation*}

In the other ordering,
\begin{equation*}
\begin{aligned}
F^{(z, v)} B^e\left(
\begin{gathered}
\begin{tikzpicture}[x=0.75pt,y=0.75pt,yscale=-1,xscale=1]
\draw   (110.2,100.2) -- (160.2,100.2) -- (160.2,150.2) -- (110.2,150.2) -- cycle ;
\draw   (130.72,95.36) .. controls (133.73,97.98) and (136.75,99.56) .. (139.76,100.08) .. controls (136.75,100.6) and (133.73,102.18) .. (130.72,104.8) ;
\draw   (130.72,145.36) .. controls (133.73,147.98) and (136.75,149.56) .. (139.76,150.08) .. controls (136.75,150.6) and (133.73,152.18) .. (130.72,154.8) ;
\draw   (164.76,120.36) .. controls (162.14,123.37) and (160.56,126.39) .. (160.04,129.4) .. controls (159.52,126.39) and (157.94,123.37) .. (155.32,120.36) ;
\draw   (114.76,120.76) .. controls (112.14,123.77) and (110.56,126.79) .. (110.04,129.8) .. controls (109.52,126.79) and (107.94,123.77) .. (105.32,120.76) ;
\draw  [color={rgb, 255:red, 0; green, 0; blue, 0 }  ,draw opacity=1 ][fill={rgb, 255:red, 0; green, 0; blue, 0 }  ,fill opacity=1 ] (108.57,150.2) .. controls (108.57,149.3) and (109.3,148.57) .. (110.2,148.57) .. controls (111.1,148.57) and (111.83,149.3) .. (111.83,150.2) .. controls (111.83,151.1) and (111.1,151.83) .. (110.2,151.83) .. controls (109.3,151.83) and (108.57,151.1) .. (108.57,150.2) -- cycle ;

\draw (128,83.6) node [anchor=north west][inner sep=0.75pt]    {$x_{1}$};
\draw (168.6,120) node [anchor=north west][inner sep=0.75pt]    {$y_{2}$};
\draw (87.4,120) node [anchor=north west][inner sep=0.75pt]    {$y_{1}$};
\draw (128,157) node [anchor=north west][inner sep=0.75pt]    {$x_{4}$};
\end{tikzpicture}
\end{gathered}\right) &= \delta_{e,x_4\bar y_2 \bar x_1 y_1}F^{(z, v)}\left(
\begin{gathered}
\begin{tikzpicture}[x=0.75pt,y=0.75pt,yscale=-1,xscale=1]
\draw   (110.2,100.2) -- (160.2,100.2) -- (160.2,150.2) -- (110.2,150.2) -- cycle ;
\draw   (130.72,95.36) .. controls (133.73,97.98) and (136.75,99.56) .. (139.76,100.08) .. controls (136.75,100.6) and (133.73,102.18) .. (130.72,104.8) ;
\draw   (130.72,145.36) .. controls (133.73,147.98) and (136.75,149.56) .. (139.76,150.08) .. controls (136.75,150.6) and (133.73,152.18) .. (130.72,154.8) ;
\draw   (164.76,120.36) .. controls (162.14,123.37) and (160.56,126.39) .. (160.04,129.4) .. controls (159.52,126.39) and (157.94,123.37) .. (155.32,120.36) ;
\draw   (114.76,120.76) .. controls (112.14,123.77) and (110.56,126.79) .. (110.04,129.8) .. controls (109.52,126.79) and (107.94,123.77) .. (105.32,120.76) ;
\draw  [color={rgb, 255:red, 0; green, 0; blue, 0 }  ,draw opacity=1 ][fill={rgb, 255:red, 0; green, 0; blue, 0 }  ,fill opacity=1 ] (108.57,150.2) .. controls (108.57,149.3) and (109.3,148.57) .. (110.2,148.57) .. controls (111.1,148.57) and (111.83,149.3) .. (111.83,150.2) .. controls (111.83,151.1) and (111.1,151.83) .. (110.2,151.83) .. controls (109.3,151.83) and (108.57,151.1) .. (108.57,150.2) -- cycle ;

\draw (128,83.6) node [anchor=north west][inner sep=0.75pt]    {$x_{1}$};
\draw (168.6,120) node [anchor=north west][inner sep=0.75pt]    {$y_{2}$};
\draw (87.4,120) node [anchor=north west][inner sep=0.75pt]    {$y_{1}$};
\draw (128,157) node [anchor=north west][inner sep=0.75pt]    {$x_{4}$};
\end{tikzpicture}
\end{gathered}\right)\\
&= \delta_{z,x_4x_5}\delta_{e,x_4\bar y_2\bar x_1 y_1}\left(
\begin{gathered}
\begin{tikzpicture}[x=0.75pt,y=0.75pt,yscale=-1,xscale=1]
\draw   (110.2,100.2) -- (160.2,100.2) -- (160.2,150.2) -- (110.2,150.2) -- cycle ;
\draw   (130.72,95.36) .. controls (133.73,97.98) and (136.75,99.56) .. (139.76,100.08) .. controls (136.75,100.6) and (133.73,102.18) .. (130.72,104.8) ;
\draw   (130.72,145.36) .. controls (133.73,147.98) and (136.75,149.56) .. (139.76,150.08) .. controls (136.75,150.6) and (133.73,152.18) .. (130.72,154.8) ;
\draw   (164.76,120.36) .. controls (162.14,123.37) and (160.56,126.39) .. (160.04,129.4) .. controls (159.52,126.39) and (157.94,123.37) .. (155.32,120.36) ;
\draw   (114.76,120.76) .. controls (112.14,123.77) and (110.56,126.79) .. (110.04,129.8) .. controls (109.52,126.79) and (107.94,123.77) .. (105.32,120.76) ;
\draw  [color={rgb, 255:red, 0; green, 0; blue, 0 }  ,draw opacity=1 ][fill={rgb, 255:red, 0; green, 0; blue, 0 }  ,fill opacity=1 ] (108.57,150.2) .. controls (108.57,149.3) and (109.3,148.57) .. (110.2,148.57) .. controls (111.1,148.57) and (111.83,149.3) .. (111.83,150.2) .. controls (111.83,151.1) and (111.1,151.83) .. (110.2,151.83) .. controls (109.3,151.83) and (108.57,151.1) .. (108.57,150.2) -- cycle ;

\draw (128,83.6) node [anchor=north west][inner sep=0.75pt]    {$x_{1}$};
\draw (168.6,120) node [anchor=north west][inner sep=0.75pt]    {$y_2 \bar x_4 \bar v x_4$};
\draw (80,120) node [anchor=north west][inner sep=0.75pt]    {$y_{1}\bar v$};
\draw (128,157) node [anchor=north west][inner sep=0.75pt]    {$x_{4}$};
\end{tikzpicture}
\end{gathered}
\right)
\end{aligned}
\end{equation*}

So we have shown that $B^h$ commute with $F^{(z,v)}$.

\textbf{Case 2: a bent ribbon within a plaquette.} Consider the third plaquette where edges are elements $x_2,y_3,x_5,y_2$ (clockwise from the top). 

\begin{equation*}
\begin{aligned}
    B^e F^{(z, v)} \left(
    \begin{gathered}
        \begin{tikzpicture}[x=0.75pt,y=0.75pt,yscale=-1,xscale=1]
        
        \draw   (110.2,100.2) -- (160.2,100.2) -- (160.2,150.2) -- (110.2,150.2) -- cycle ;
        \draw   (130.72,95.36) .. controls (133.73,97.98) and (136.75,99.56) .. (139.76,100.08) .. controls (136.75,100.6) and (133.73,102.18) .. (130.72,104.8) ;
        \draw   (130.72,145.36) .. controls (133.73,147.98) and (136.75,149.56) .. (139.76,150.08) .. controls (136.75,150.6) and (133.73,152.18) .. (130.72,154.8) ;
        \draw   (164.76,120.36) .. controls (162.14,123.37) and (160.56,126.39) .. (160.04,129.4) .. controls (159.52,126.39) and (157.94,123.37) .. (155.32,120.36) ;
        \draw   (114.76,120.76) .. controls (112.14,123.77) and (110.56,126.79) .. (110.04,129.8) .. controls (109.52,126.79) and (107.94,123.77) .. (105.32,120.76) ;
        \draw  [color={rgb, 255:red, 0; green, 0; blue, 0 }  ,draw opacity=1 ][fill={rgb, 255:red, 0; green, 0; blue, 0 }  ,fill opacity=1 ] (108.57,150.2) .. controls (108.57,149.3) and (109.3,148.57) .. (110.2,148.57) .. controls (111.1,148.57) and (111.83,149.3) .. (111.83,150.2) .. controls (111.83,151.1) and (111.1,151.83) .. (110.2,151.83) .. controls (109.3,151.83) and (108.57,151.1) .. (108.57,150.2) -- cycle ;
        
        \draw (128,83.6) node [anchor=north west][inner sep=0.75pt]    {$x_{2}$};
        \draw (168.6,120) node [anchor=north west][inner sep=0.75pt]    {$y_{3}$};
        \draw (87.4,120) node [anchor=north west][inner sep=0.75pt]    {$y_{2}$};
        \draw (128,157) node [anchor=north west][inner sep=0.75pt]    {$x_{5}$};
        \end{tikzpicture}
    \end{gathered}
    \right) &= \delta_{x_4x_5}B^e\left(\begin{gathered}
        \begin{tikzpicture}[x=0.75pt,y=0.75pt,yscale=-1,xscale=1]
        \draw   (110.2,100.2) -- (160.2,100.2) -- (160.2,150.2) -- (110.2,150.2) -- cycle ;
        \draw   (130.72,95.36) .. controls (133.73,97.98) and (136.75,99.56) .. (139.76,100.08) .. controls (136.75,100.6) and (133.73,102.18) .. (130.72,104.8) ;
        \draw   (130.72,145.36) .. controls (133.73,147.98) and (136.75,149.56) .. (139.76,150.08) .. controls (136.75,150.6) and (133.73,152.18) .. (130.72,154.8) ;
        \draw   (164.76,120.36) .. controls (162.14,123.37) and (160.56,126.39) .. (160.04,129.4) .. controls (159.52,126.39) and (157.94,123.37) .. (155.32,120.36) ;
        \draw   (114.76,120.76) .. controls (112.14,123.77) and (110.56,126.79) .. (110.04,129.8) .. controls (109.52,126.79) and (107.94,123.77) .. (105.32,120.76) ;
        \draw  [color={rgb, 255:red, 0; green, 0; blue, 0 }  ,draw opacity=1 ][fill={rgb, 255:red, 0; green, 0; blue, 0 }  ,fill opacity=1 ] (108.57,150.2) .. controls (108.57,149.3) and (109.3,148.57) .. (110.2,148.57) .. controls (111.1,148.57) and (111.83,149.3) .. (111.83,150.2) .. controls (111.83,151.1) and (111.1,151.83) .. (110.2,151.83) .. controls (109.3,151.83) and (108.57,151.1) .. (108.57,150.2) -- cycle ;
        
        \draw (128,83.6) node [anchor=north west][inner sep=0.75pt]    {$x_{2}$};
        \draw (168.6,120) node [anchor=north west][inner sep=0.75pt]    {$y_{3}$};
        \draw (55,120) node [anchor=north west][inner sep=0.75pt]    {$y_{2}\bar x_4 \bar v x_4$};
        \draw (111,157) node [anchor=north west][inner sep=0.75pt]    {$\bar x_4 v x_4 x_5$};
        \end{tikzpicture}
    \end{gathered}\right)\\
    &= \delta_{z,x_4x_5}\delta_{e,x_5\bar x_3\bar x_2 y_2} \left(\begin{gathered}
        \begin{tikzpicture}[x=0.75pt,y=0.75pt,yscale=-1,xscale=1]
        \draw   (110.2,100.2) -- (160.2,100.2) -- (160.2,150.2) -- (110.2,150.2) -- cycle ;
        \draw   (130.72,95.36) .. controls (133.73,97.98) and (136.75,99.56) .. (139.76,100.08) .. controls (136.75,100.6) and (133.73,102.18) .. (130.72,104.8) ;
        \draw   (130.72,145.36) .. controls (133.73,147.98) and (136.75,149.56) .. (139.76,150.08) .. controls (136.75,150.6) and (133.73,152.18) .. (130.72,154.8) ;
        \draw   (164.76,120.36) .. controls (162.14,123.37) and (160.56,126.39) .. (160.04,129.4) .. controls (159.52,126.39) and (157.94,123.37) .. (155.32,120.36) ;
        \draw   (114.76,120.76) .. controls (112.14,123.77) and (110.56,126.79) .. (110.04,129.8) .. controls (109.52,126.79) and (107.94,123.77) .. (105.32,120.76) ;
        \draw  [color={rgb, 255:red, 0; green, 0; blue, 0 }  ,draw opacity=1 ][fill={rgb, 255:red, 0; green, 0; blue, 0 }  ,fill opacity=1 ] (108.57,150.2) .. controls (108.57,149.3) and (109.3,148.57) .. (110.2,148.57) .. controls (111.1,148.57) and (111.83,149.3) .. (111.83,150.2) .. controls (111.83,151.1) and (111.1,151.83) .. (110.2,151.83) .. controls (109.3,151.83) and (108.57,151.1) .. (108.57,150.2) -- cycle ;
        
        \draw (128,83.6) node [anchor=north west][inner sep=0.75pt]    {$x_{2}$};
        \draw (168.6,120) node [anchor=north west][inner sep=0.75pt]    {$y_{3}$};
        \draw (55,120) node [anchor=north west][inner sep=0.75pt]    {$y_{2}\bar x_4 \bar v x_4$};
        \draw (111,157) node [anchor=north west][inner sep=0.75pt]    {$\bar x_4 v x_4 x_5$};
        \end{tikzpicture}
    \end{gathered}\right)
\end{aligned}
\end{equation*}

\begin{equation*}
\begin{aligned}
    F^{(z, v)} B^e \left(
    \begin{gathered}
        \begin{tikzpicture}[x=0.75pt,y=0.75pt,yscale=-1,xscale=1]
        
        \draw   (110.2,100.2) -- (160.2,100.2) -- (160.2,150.2) -- (110.2,150.2) -- cycle ;
        \draw   (130.72,95.36) .. controls (133.73,97.98) and (136.75,99.56) .. (139.76,100.08) .. controls (136.75,100.6) and (133.73,102.18) .. (130.72,104.8) ;
        \draw   (130.72,145.36) .. controls (133.73,147.98) and (136.75,149.56) .. (139.76,150.08) .. controls (136.75,150.6) and (133.73,152.18) .. (130.72,154.8) ;
        \draw   (164.76,120.36) .. controls (162.14,123.37) and (160.56,126.39) .. (160.04,129.4) .. controls (159.52,126.39) and (157.94,123.37) .. (155.32,120.36) ;
        \draw   (114.76,120.76) .. controls (112.14,123.77) and (110.56,126.79) .. (110.04,129.8) .. controls (109.52,126.79) and (107.94,123.77) .. (105.32,120.76) ;
        \draw  [color={rgb, 255:red, 0; green, 0; blue, 0 }  ,draw opacity=1 ][fill={rgb, 255:red, 0; green, 0; blue, 0 }  ,fill opacity=1 ] (108.57,150.2) .. controls (108.57,149.3) and (109.3,148.57) .. (110.2,148.57) .. controls (111.1,148.57) and (111.83,149.3) .. (111.83,150.2) .. controls (111.83,151.1) and (111.1,151.83) .. (110.2,151.83) .. controls (109.3,151.83) and (108.57,151.1) .. (108.57,150.2) -- cycle ;
        
        \draw (128,83.6) node [anchor=north west][inner sep=0.75pt]    {$x_{2}$};
        \draw (168.6,120) node [anchor=north west][inner sep=0.75pt]    {$y_{3}$};
        \draw (87.4,120) node [anchor=north west][inner sep=0.75pt]    {$y_{2}$};
        \draw (128,157) node [anchor=north west][inner sep=0.75pt]    {$x_{5}$};
        \end{tikzpicture}
    \end{gathered}
    \right) &= \delta_{e,x_5\bar x_3\bar x_2 y_2} F^{(z, v)}\left(\begin{gathered}
        \begin{tikzpicture}[x=0.75pt,y=0.75pt,yscale=-1,xscale=1]
        \draw   (110.2,100.2) -- (160.2,100.2) -- (160.2,150.2) -- (110.2,150.2) -- cycle ;
        \draw   (130.72,95.36) .. controls (133.73,97.98) and (136.75,99.56) .. (139.76,100.08) .. controls (136.75,100.6) and (133.73,102.18) .. (130.72,104.8) ;
        \draw   (130.72,145.36) .. controls (133.73,147.98) and (136.75,149.56) .. (139.76,150.08) .. controls (136.75,150.6) and (133.73,152.18) .. (130.72,154.8) ;
        \draw   (164.76,120.36) .. controls (162.14,123.37) and (160.56,126.39) .. (160.04,129.4) .. controls (159.52,126.39) and (157.94,123.37) .. (155.32,120.36) ;
        \draw   (114.76,120.76) .. controls (112.14,123.77) and (110.56,126.79) .. (110.04,129.8) .. controls (109.52,126.79) and (107.94,123.77) .. (105.32,120.76) ;
        \draw  [color={rgb, 255:red, 0; green, 0; blue, 0 }  ,draw opacity=1 ][fill={rgb, 255:red, 0; green, 0; blue, 0 }  ,fill opacity=1 ] (108.57,150.2) .. controls (108.57,149.3) and (109.3,148.57) .. (110.2,148.57) .. controls (111.1,148.57) and (111.83,149.3) .. (111.83,150.2) .. controls (111.83,151.1) and (111.1,151.83) .. (110.2,151.83) .. controls (109.3,151.83) and (108.57,151.1) .. (108.57,150.2) -- cycle ;
        
        \draw (128,83.6) node [anchor=north west][inner sep=0.75pt]    {$x_{2}$};
        \draw (168.6,120) node [anchor=north west][inner sep=0.75pt]    {$y_{3}$};
        \draw (87.4,120) node [anchor=north west][inner sep=0.75pt]    {$y_{2}$};
        \draw (128,157) node [anchor=north west][inner sep=0.75pt]    {$x_{5}$};
        \end{tikzpicture}
    \end{gathered}\right)\\
    &=\delta_{z,x_4x_5}\delta_{e,x_5\bar x_3\bar x_2 y_2} \left(\begin{gathered}
        \begin{tikzpicture}[x=0.75pt,y=0.75pt,yscale=-1,xscale=1]
        \draw   (110.2,100.2) -- (160.2,100.2) -- (160.2,150.2) -- (110.2,150.2) -- cycle ;
        \draw   (130.72,95.36) .. controls (133.73,97.98) and (136.75,99.56) .. (139.76,100.08) .. controls (136.75,100.6) and (133.73,102.18) .. (130.72,104.8) ;
        \draw   (130.72,145.36) .. controls (133.73,147.98) and (136.75,149.56) .. (139.76,150.08) .. controls (136.75,150.6) and (133.73,152.18) .. (130.72,154.8) ;
        \draw   (164.76,120.36) .. controls (162.14,123.37) and (160.56,126.39) .. (160.04,129.4) .. controls (159.52,126.39) and (157.94,123.37) .. (155.32,120.36) ;
        \draw   (114.76,120.76) .. controls (112.14,123.77) and (110.56,126.79) .. (110.04,129.8) .. controls (109.52,126.79) and (107.94,123.77) .. (105.32,120.76) ;
        \draw  [color={rgb, 255:red, 0; green, 0; blue, 0 }  ,draw opacity=1 ][fill={rgb, 255:red, 0; green, 0; blue, 0 }  ,fill opacity=1 ] (108.57,150.2) .. controls (108.57,149.3) and (109.3,148.57) .. (110.2,148.57) .. controls (111.1,148.57) and (111.83,149.3) .. (111.83,150.2) .. controls (111.83,151.1) and (111.1,151.83) .. (110.2,151.83) .. controls (109.3,151.83) and (108.57,151.1) .. (108.57,150.2) -- cycle ;
        
        \draw (128,83.6) node [anchor=north west][inner sep=0.75pt]    {$x_{2}$};
        \draw (168.6,120) node [anchor=north west][inner sep=0.75pt]    {$y_{3}$};
        \draw (55,120) node [anchor=north west][inner sep=0.75pt]    {$y_2 \bar x_4 \bar v x_4$};
        \draw (111,157) node [anchor=north west][inner sep=0.75pt]    {$\bar x_4 v x_4 x_5$};
        \end{tikzpicture}
    \end{gathered}\right)
\end{aligned}
\end{equation*}

So we have again shown that $B^e$ commute with $F^{(z,v)}$.

\textbf{Case 3: with a vertex.} Consider the vertex with edges $y_2,x_5,y_5,x_4$ (clockwise from the top). We want to show that the vertex operator $A^g$ commute with $F^{(z,v)}$, as $A^g$ stabilizes ground states.

\begin{equation*}
    A^g F^{(z,v)} \left(
    \begin{gathered}
        \begin{tikzpicture}[x=0.75pt,y=0.75pt,yscale=-1,xscale=1]
        
        \draw   (55.9,88.86) .. controls (58.91,91.48) and (61.93,93.06) .. (64.94,93.58) .. controls (61.93,94.1) and (58.91,95.68) .. (55.9,98.3) ;
        \draw   (91.09,62.36) .. controls (88.47,65.38) and (86.89,68.39) .. (86.37,71.4) .. controls (85.85,68.39) and (84.27,65.38) .. (81.65,62.36) ;
        \draw   (36.45,93.56) -- (136.4,93.56)(86.42,43.67) -- (86.42,143.45) ;
        \draw   (108.22,88.61) .. controls (111.23,91.23) and (114.25,92.81) .. (117.26,93.33) .. controls (114.25,93.85) and (111.23,95.43) .. (108.22,98.05) ;
        \draw   (91.09,122.86) .. controls (88.47,125.88) and (86.89,128.89) .. (86.37,131.9) .. controls (85.85,128.89) and (84.27,125.88) .. (81.65,122.86) ;
        
        \draw (53.03,99.85) node [anchor=north west][inner sep=0.75pt]    {$x_{4}$};
        \draw (63.97,51.69) node [anchor=north west][inner sep=0.75pt]    {$y_{2}$};
        \draw (105.35,99.6) node [anchor=north west][inner sep=0.75pt]    {$x_{5}$};
        \draw (63.69,122.54) node [anchor=north west][inner sep=0.75pt]    {$y_{5}$};
        \end{tikzpicture}
    \end{gathered}\right)=\delta_{z,x_4x_5}A^g\left(
    \begin{gathered}
        \begin{tikzpicture}[x=0.75pt,y=0.75pt,yscale=-1,xscale=1]
        
        \draw   (55.9,88.86) .. controls (58.91,91.48) and (61.93,93.06) .. (64.94,93.58) .. controls (61.93,94.1) and (58.91,95.68) .. (55.9,98.3) ;
        \draw   (91.09,62.36) .. controls (88.47,65.38) and (86.89,68.39) .. (86.37,71.4) .. controls (85.85,68.39) and (84.27,65.38) .. (81.65,62.36) ;
        \draw   (36.45,93.56) -- (136.4,93.56)(86.42,43.67) -- (86.42,143.45) ;
        \draw   (108.22,88.61) .. controls (111.23,91.23) and (114.25,92.81) .. (117.26,93.33) .. controls (114.25,93.85) and (111.23,95.43) .. (108.22,98.05) ;
        \draw   (91.09,122.86) .. controls (88.47,125.88) and (86.89,128.89) .. (86.37,131.9) .. controls (85.85,128.89) and (84.27,125.88) .. (81.65,122.86) ;
        
        \draw (53.03,99.85) node [anchor=north west][inner sep=0.75pt]    {$x_{4}$};
        \draw (63.97,51.69) node [anchor=north west][inner sep=0.75pt]    {$y_{2}\bar x_4 \bar v x_4$};
        \draw (105.35,99.6) node [anchor=north west][inner sep=0.75pt]    {$\bar x_4 v x_4 x_{5}$};
        \draw (63.69,122.54) node [anchor=north west][inner sep=0.75pt]    {$y_{5}$};
        \end{tikzpicture}
    \end{gathered}\right) = \delta_{z,x_4x_5}\left(
    \begin{gathered}
        \begin{tikzpicture}[x=0.75pt,y=0.75pt,yscale=-1,xscale=1]
        
        \draw   (55.9,88.86) .. controls (58.91,91.48) and (61.93,93.06) .. (64.94,93.58) .. controls (61.93,94.1) and (58.91,95.68) .. (55.9,98.3) ;
        \draw   (91.09,62.36) .. controls (88.47,65.38) and (86.89,68.39) .. (86.37,71.4) .. controls (85.85,68.39) and (84.27,65.38) .. (81.65,62.36) ;
        \draw   (36.45,93.56) -- (136.4,93.56)(86.42,43.67) -- (86.42,143.45) ;
        \draw   (108.22,88.61) .. controls (111.23,91.23) and (114.25,92.81) .. (117.26,93.33) .. controls (114.25,93.85) and (111.23,95.43) .. (108.22,98.05) ;
        \draw   (91.09,122.86) .. controls (88.47,125.88) and (86.89,128.89) .. (86.37,131.9) .. controls (85.85,128.89) and (84.27,125.88) .. (81.65,122.86) ;
        
        \draw (53.03,99.85) node [anchor=north west][inner sep=0.75pt]    {$x_{4}\bar g$};
        \draw (63.97,51.69) node [anchor=north west][inner sep=0.75pt]    {$y_{2}\bar x_4 \bar v x_4\bar g$};
        \draw (105.35,99.6) node [anchor=north west][inner sep=0.75pt]    {$g\bar x_4 v x_4 x_{5}$};
        \draw (63.69,122.54) node [anchor=north west][inner sep=0.75pt]    {$gy_{5}$};
        \end{tikzpicture}
    \end{gathered}\right)
\end{equation*}
On the right hand side,
\begin{equation*}
    F^{(z,v)}A^g\left(
    \begin{gathered}
        \begin{tikzpicture}[x=0.75pt,y=0.75pt,yscale=-1,xscale=1]
        
        \draw   (55.9,88.86) .. controls (58.91,91.48) and (61.93,93.06) .. (64.94,93.58) .. controls (61.93,94.1) and (58.91,95.68) .. (55.9,98.3) ;
        \draw   (91.09,62.36) .. controls (88.47,65.38) and (86.89,68.39) .. (86.37,71.4) .. controls (85.85,68.39) and (84.27,65.38) .. (81.65,62.36) ;
        \draw   (36.45,93.56) -- (136.4,93.56)(86.42,43.67) -- (86.42,143.45) ;
        \draw   (108.22,88.61) .. controls (111.23,91.23) and (114.25,92.81) .. (117.26,93.33) .. controls (114.25,93.85) and (111.23,95.43) .. (108.22,98.05) ;
        \draw   (91.09,122.86) .. controls (88.47,125.88) and (86.89,128.89) .. (86.37,131.9) .. controls (85.85,128.89) and (84.27,125.88) .. (81.65,122.86) ;
        
        \draw (53.03,99.85) node [anchor=north west][inner sep=0.75pt]    {$x_{4}$};
        \draw (63.97,51.69) node [anchor=north west][inner sep=0.75pt]    {$y_{2}$};
        \draw (105.35,99.6) node [anchor=north west][inner sep=0.75pt]    {$x_{5}$};
        \draw (63.69,122.54) node [anchor=north west][inner sep=0.75pt]    {$y_{5}$};
        \end{tikzpicture}
    \end{gathered}\right) = F^{(z,v)}\left(
    \begin{gathered}
        \begin{tikzpicture}[x=0.75pt,y=0.75pt,yscale=-1,xscale=1]
        
        \draw   (55.9,88.86) .. controls (58.91,91.48) and (61.93,93.06) .. (64.94,93.58) .. controls (61.93,94.1) and (58.91,95.68) .. (55.9,98.3) ;
        \draw   (91.09,62.36) .. controls (88.47,65.38) and (86.89,68.39) .. (86.37,71.4) .. controls (85.85,68.39) and (84.27,65.38) .. (81.65,62.36) ;
        \draw   (36.45,93.56) -- (136.4,93.56)(86.42,43.67) -- (86.42,143.45) ;
        \draw   (108.22,88.61) .. controls (111.23,91.23) and (114.25,92.81) .. (117.26,93.33) .. controls (114.25,93.85) and (111.23,95.43) .. (108.22,98.05) ;
        \draw   (91.09,122.86) .. controls (88.47,125.88) and (86.89,128.89) .. (86.37,131.9) .. controls (85.85,128.89) and (84.27,125.88) .. (81.65,122.86) ;
        
        \draw (53.03,99.85) node [anchor=north west][inner sep=0.75pt]    {$x_{4}\bar g$};
        \draw (63.97,51.69) node [anchor=north west][inner sep=0.75pt]    {$y_{2}\bar g$};
        \draw (105.35,99.6) node [anchor=north west][inner sep=0.75pt]    {$gx_{5}$};
        \draw (63.69,122.54) node [anchor=north west][inner sep=0.75pt]    {$gy_{5}$};
        \end{tikzpicture}
    \end{gathered}\right) = \delta_{z,x_4x_5}\left(
    \begin{gathered}
        \begin{tikzpicture}[x=0.75pt,y=0.75pt,yscale=-1,xscale=1]
        
        \draw   (55.9,88.86) .. controls (58.91,91.48) and (61.93,93.06) .. (64.94,93.58) .. controls (61.93,94.1) and (58.91,95.68) .. (55.9,98.3) ;
        \draw   (91.09,62.36) .. controls (88.47,65.38) and (86.89,68.39) .. (86.37,71.4) .. controls (85.85,68.39) and (84.27,65.38) .. (81.65,62.36) ;
        \draw   (36.45,93.56) -- (136.4,93.56)(86.42,43.67) -- (86.42,143.45) ;
        \draw   (108.22,88.61) .. controls (111.23,91.23) and (114.25,92.81) .. (117.26,93.33) .. controls (114.25,93.85) and (111.23,95.43) .. (108.22,98.05) ;
        \draw   (91.09,122.86) .. controls (88.47,125.88) and (86.89,128.89) .. (86.37,131.9) .. controls (85.85,128.89) and (84.27,125.88) .. (81.65,122.86) ;
        
        \draw (53.03,99.85) node [anchor=north west][inner sep=0.75pt]    {$x_{4}\bar g$};
        \draw (63.97,51.69) node [anchor=north west][inner sep=0.75pt]    {$y_{2}\bar x_4 \bar v x_4\bar g$};
        \draw (105.35,99.6) node [anchor=north west][inner sep=0.75pt]    {$g\bar x_4 v x_4 x_{5}$};
        \draw (63.69,122.54) node [anchor=north west][inner sep=0.75pt]    {$gy_{5}$};
        \end{tikzpicture}
    \end{gathered}\right)
\end{equation*}

So we have shown that $A^g$ commute with $F^{(z,v)}$.

\subsubsection{Commutation with projectors at end points} \label{ap:end-ribbon}

In this section, we prove the commutation relations stated in the main text. We reproduce them below for convenience. One common argument is that elongating the ribbon doesn't affect the commutation relation of the ribbon with a plaquette operator at the opposite end of the ribbon, so it suffices to consider the first two ribbons to determine the commutation at the start of a ribbon; at the end of the ribbon, it suffices to consider the action of the last dual triangle.\\


\textbf{1. Clockwise local orientation}

Consider the following ribbon with clockwise local orientation as shown in Fig. \ref{fig:clockwise ribbon for plaquette}. Note that the base point from which to define the flux measured by the plaquette operator $B^h$ need to coincide with the starting or ending site of the ribbon. 

\begin{figure}[h!]
    \centering
    \includegraphics[width=0.7\textwidth]{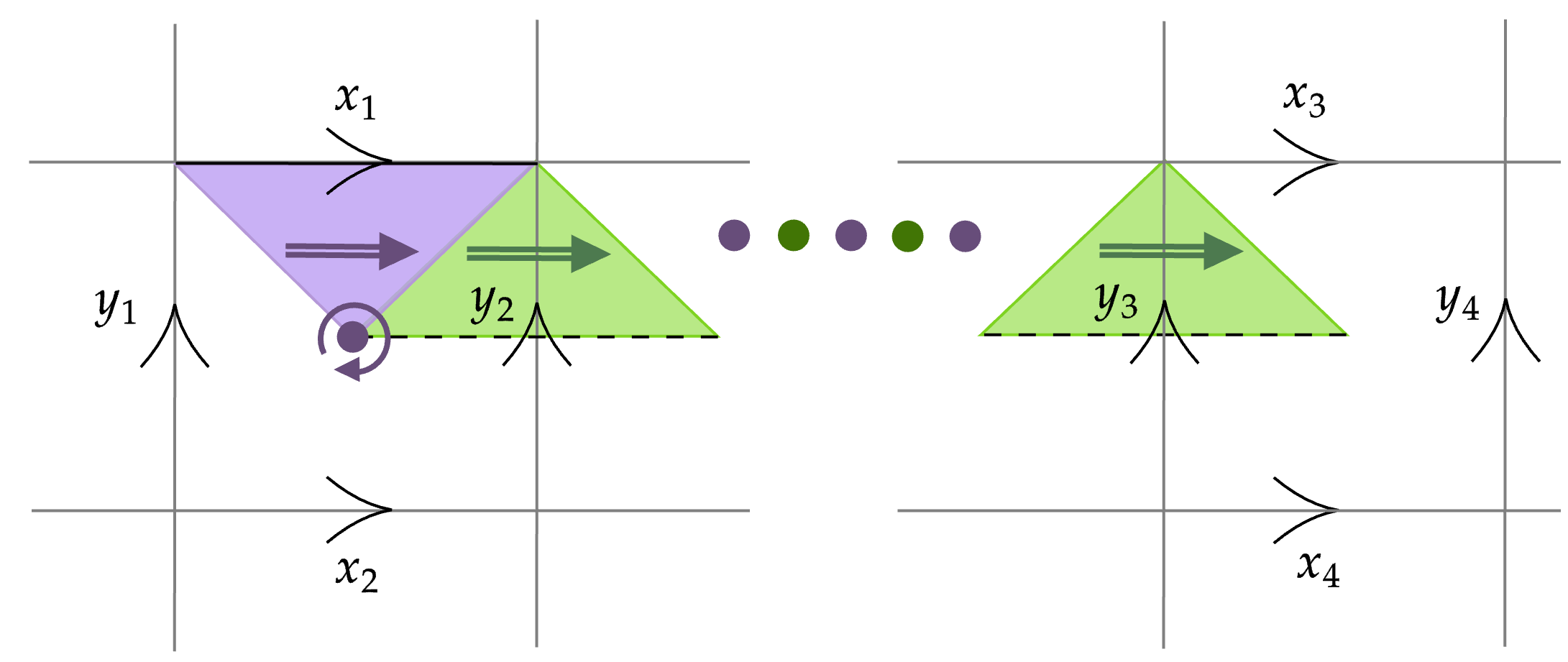}
    \caption{A ribbon operator with clockwise local orientation.}
    \label{fig:clockwise ribbon for plaquette}
\end{figure}

\begin{prop}[Plaquette operator at the start of a clockwise ribbon]
The commutation relation is
    \begin{equation}
        \quad B_{\mathrm{fl}}^{h}F^{(z, v)}(t)=F^{(z, v)}(t)B_{\mathrm{fl}}^{v h}.
    \end{equation}
\end{prop}

\begin{proof}
    WLOG, consider a ribbon (with clockwise local orientation) composed of two triangles, as shown in Fig. \ref{fig:clockwise ribbon for plaquette}. 

    \begin{equation*}
    \begin{aligned}
        B^hF^{(z,v)}(t)
        \left(\begin{gathered}
            \begin{tikzpicture}[x=0.75pt,y=0.75pt,yscale=-1,xscale=1]
            
            \draw   (130.2,120.2) -- (180.2,120.2) -- (180.2,170.2) -- (130.2,170.2) -- cycle ;
            \draw   (150.72,115.36) .. controls (153.73,117.98) and (156.75,119.56) .. (159.76,120.08) .. controls (156.75,120.6) and (153.73,122.18) .. (150.72,124.8) ;
            \draw   (150.72,165.36) .. controls (153.73,167.98) and (156.75,169.56) .. (159.76,170.08) .. controls (156.75,170.6) and (153.73,172.18) .. (150.72,174.8) ;
            \draw   (175.32,149.4) .. controls (177.94,146.39) and (179.52,143.37) .. (180.04,140.36) .. controls (180.56,143.37) and (182.14,146.39) .. (184.76,149.4) ;
            \draw   (125.32,149.8) .. controls (127.94,146.79) and (129.52,143.77) .. (130.04,140.76) .. controls (130.56,143.77) and (132.14,146.79) .. (134.76,149.8) ;
            \draw  [color={rgb, 255:red, 0; green, 0; blue, 0 }  ,draw opacity=1 ][fill={rgb, 255:red, 0; green, 0; blue, 0 }  ,fill opacity=1 ] (128.57,120.2) .. controls (128.57,119.3) and (129.3,118.57) .. (130.2,118.57) .. controls (131.1,118.57) and (131.83,119.3) .. (131.83,120.2) .. controls (131.83,121.1) and (131.1,121.83) .. (130.2,121.83) .. controls (129.3,121.83) and (128.57,121.1) .. (128.57,120.2) -- cycle ;
            
            \draw (148.6,103) node [anchor=north west][inner sep=0.75pt]    {$x_{1}$};
            \draw (188.6,136) node [anchor=north west][inner sep=0.75pt]    {$y_{2}$};
            \draw (107.4,135.4) node [anchor=north west][inner sep=0.75pt]    {$y_{1}$};
            \draw (150.6,175.6) node [anchor=north west][inner sep=0.75pt]    {$x_{2}$};
            \end{tikzpicture}
            \end{gathered}
        \right)
        &=\delta_{z,x_1} B^h
        \left(\begin{gathered}
            \begin{tikzpicture}[x=0.75pt,y=0.75pt,yscale=-1,xscale=1]
            
            \draw   (130.2,120.2) -- (180.2,120.2) -- (180.2,170.2) -- (130.2,170.2) -- cycle ;
            \draw   (150.72,115.36) .. controls (153.73,117.98) and (156.75,119.56) .. (159.76,120.08) .. controls (156.75,120.6) and (153.73,122.18) .. (150.72,124.8) ;
            \draw   (150.72,165.36) .. controls (153.73,167.98) and (156.75,169.56) .. (159.76,170.08) .. controls (156.75,170.6) and (153.73,172.18) .. (150.72,174.8) ;
            \draw   (175.32,149.4) .. controls (177.94,146.39) and (179.52,143.37) .. (180.04,140.36) .. controls (180.56,143.37) and (182.14,146.39) .. (184.76,149.4) ;
            \draw   (125.32,149.8) .. controls (127.94,146.79) and (129.52,143.77) .. (130.04,140.76) .. controls (130.56,143.77) and (132.14,146.79) .. (134.76,149.8) ;
            \draw  [color={rgb, 255:red, 0; green, 0; blue, 0 }  ,draw opacity=1 ][fill={rgb, 255:red, 0; green, 0; blue, 0 }  ,fill opacity=1 ] (128.57,120.2) .. controls (128.57,119.3) and (129.3,118.57) .. (130.2,118.57) .. controls (131.1,118.57) and (131.83,119.3) .. (131.83,120.2) .. controls (131.83,121.1) and (131.1,121.83) .. (130.2,121.83) .. controls (129.3,121.83) and (128.57,121.1) .. (128.57,120.2) -- cycle ;
            
            \draw (148.6,103) node [anchor=north west][inner sep=0.75pt]    {$x_{1}$};
            \draw (188.6,136) node [anchor=north west][inner sep=0.75pt]    {$y_{2}\bar x_1 v x_1$};
            \draw (107.4,135.4) node [anchor=north west][inner sep=0.75pt]    {$y_{1}$};
            \draw (150.6,175.6) node [anchor=north west][inner sep=0.75pt]    {$x_{2}$};
            \end{tikzpicture}
            \end{gathered}
        \right)\\
        &=\delta_{z,x_1}\delta_{\bar h, \bar y_1x_2y_2\bar x_1 v}
        \left(\begin{gathered}
            \begin{tikzpicture}[x=0.75pt,y=0.75pt,yscale=-1,xscale=1]
            
            \draw   (130.2,120.2) -- (180.2,120.2) -- (180.2,170.2) -- (130.2,170.2) -- cycle ;
            \draw   (150.72,115.36) .. controls (153.73,117.98) and (156.75,119.56) .. (159.76,120.08) .. controls (156.75,120.6) and (153.73,122.18) .. (150.72,124.8) ;
            \draw   (150.72,165.36) .. controls (153.73,167.98) and (156.75,169.56) .. (159.76,170.08) .. controls (156.75,170.6) and (153.73,172.18) .. (150.72,174.8) ;
            \draw   (175.32,149.4) .. controls (177.94,146.39) and (179.52,143.37) .. (180.04,140.36) .. controls (180.56,143.37) and (182.14,146.39) .. (184.76,149.4) ;
            \draw   (125.32,149.8) .. controls (127.94,146.79) and (129.52,143.77) .. (130.04,140.76) .. controls (130.56,143.77) and (132.14,146.79) .. (134.76,149.8) ;
            \draw  [color={rgb, 255:red, 0; green, 0; blue, 0 }  ,draw opacity=1 ][fill={rgb, 255:red, 0; green, 0; blue, 0 }  ,fill opacity=1 ] (128.57,120.2) .. controls (128.57,119.3) and (129.3,118.57) .. (130.2,118.57) .. controls (131.1,118.57) and (131.83,119.3) .. (131.83,120.2) .. controls (131.83,121.1) and (131.1,121.83) .. (130.2,121.83) .. controls (129.3,121.83) and (128.57,121.1) .. (128.57,120.2) -- cycle ;
            
            \draw (148.6,103) node [anchor=north west][inner sep=0.75pt]    {$x_{1}$};
            \draw (188.6,136) node [anchor=north west][inner sep=0.75pt]    {$y_{2}\bar x_1 v x_1$};
            \draw (107.4,135.4) node [anchor=north west][inner sep=0.75pt]    {$y_{1}$};
            \draw (150.6,175.6) node [anchor=north west][inner sep=0.75pt]    {$x_{2}$};
            \end{tikzpicture}
            \end{gathered}
        \right)
    \end{aligned}
    \end{equation*}

For the right hand side, we assume a different set of variables ($\zeta$ for charge, $\omega$ for local flux, and $\eta$ for the flux variable of the plaquette operator) and determine their values.
    \begin{equation*}
    \begin{aligned}
        F^{(\zeta,\omega)}(t) B^{\eta}
        \left(\begin{gathered}
            \begin{tikzpicture}[x=0.75pt,y=0.75pt,yscale=-1,xscale=1]
            
            \draw   (130.2,120.2) -- (180.2,120.2) -- (180.2,170.2) -- (130.2,170.2) -- cycle ;
            \draw   (150.72,115.36) .. controls (153.73,117.98) and (156.75,119.56) .. (159.76,120.08) .. controls (156.75,120.6) and (153.73,122.18) .. (150.72,124.8) ;
            \draw   (150.72,165.36) .. controls (153.73,167.98) and (156.75,169.56) .. (159.76,170.08) .. controls (156.75,170.6) and (153.73,172.18) .. (150.72,174.8) ;
            \draw   (175.32,149.4) .. controls (177.94,146.39) and (179.52,143.37) .. (180.04,140.36) .. controls (180.56,143.37) and (182.14,146.39) .. (184.76,149.4) ;
            \draw   (125.32,149.8) .. controls (127.94,146.79) and (129.52,143.77) .. (130.04,140.76) .. controls (130.56,143.77) and (132.14,146.79) .. (134.76,149.8) ;
            \draw  [color={rgb, 255:red, 0; green, 0; blue, 0 }  ,draw opacity=1 ][fill={rgb, 255:red, 0; green, 0; blue, 0 }  ,fill opacity=1 ] (128.57,120.2) .. controls (128.57,119.3) and (129.3,118.57) .. (130.2,118.57) .. controls (131.1,118.57) and (131.83,119.3) .. (131.83,120.2) .. controls (131.83,121.1) and (131.1,121.83) .. (130.2,121.83) .. controls (129.3,121.83) and (128.57,121.1) .. (128.57,120.2) -- cycle ;
            
            \draw (148.6,103) node [anchor=north west][inner sep=0.75pt]    {$x_{1}$};
            \draw (188.6,136) node [anchor=north west][inner sep=0.75pt]    {$y_{2}$};
            \draw (107.4,135.4) node [anchor=north west][inner sep=0.75pt]    {$y_{1}$};
            \draw (150.6,175.6) node [anchor=north west][inner sep=0.75pt]    {$x_{2}$};
            \end{tikzpicture}
            \end{gathered}
        \right)
        &= \delta_{\bar \eta,\bar y_1x_2y_2\bar x_1} F^{(\zeta,\omega)}(t)
        \left(\begin{gathered}
            \begin{tikzpicture}[x=0.75pt,y=0.75pt,yscale=-1,xscale=1]
            
            \draw   (130.2,120.2) -- (180.2,120.2) -- (180.2,170.2) -- (130.2,170.2) -- cycle ;
            \draw   (150.72,115.36) .. controls (153.73,117.98) and (156.75,119.56) .. (159.76,120.08) .. controls (156.75,120.6) and (153.73,122.18) .. (150.72,124.8) ;
            \draw   (150.72,165.36) .. controls (153.73,167.98) and (156.75,169.56) .. (159.76,170.08) .. controls (156.75,170.6) and (153.73,172.18) .. (150.72,174.8) ;
            \draw   (175.32,149.4) .. controls (177.94,146.39) and (179.52,143.37) .. (180.04,140.36) .. controls (180.56,143.37) and (182.14,146.39) .. (184.76,149.4) ;
            \draw   (125.32,149.8) .. controls (127.94,146.79) and (129.52,143.77) .. (130.04,140.76) .. controls (130.56,143.77) and (132.14,146.79) .. (134.76,149.8) ;
            \draw  [color={rgb, 255:red, 0; green, 0; blue, 0 }  ,draw opacity=1 ][fill={rgb, 255:red, 0; green, 0; blue, 0 }  ,fill opacity=1 ] (128.57,120.2) .. controls (128.57,119.3) and (129.3,118.57) .. (130.2,118.57) .. controls (131.1,118.57) and (131.83,119.3) .. (131.83,120.2) .. controls (131.83,121.1) and (131.1,121.83) .. (130.2,121.83) .. controls (129.3,121.83) and (128.57,121.1) .. (128.57,120.2) -- cycle ;
            
            \draw (148.6,103) node [anchor=north west][inner sep=0.75pt]    {$x_{1}$};
            \draw (188.6,136) node [anchor=north west][inner sep=0.75pt]    {$y_{2}$};
            \draw (107.4,135.4) node [anchor=north west][inner sep=0.75pt]    {$y_{1}$};
            \draw (150.6,175.6) node [anchor=north west][inner sep=0.75pt]    {$x_{2}$};
            \end{tikzpicture}
            \end{gathered}
        \right)\\
        &=\delta_{\zeta,x_1}\delta_{\bar \eta, \bar y_1x_2y_2\bar x_1}
        \left(\begin{gathered}
            \begin{tikzpicture}[x=0.75pt,y=0.75pt,yscale=-1,xscale=1]
            
            \draw   (130.2,120.2) -- (180.2,120.2) -- (180.2,170.2) -- (130.2,170.2) -- cycle ;
            \draw   (150.72,115.36) .. controls (153.73,117.98) and (156.75,119.56) .. (159.76,120.08) .. controls (156.75,120.6) and (153.73,122.18) .. (150.72,124.8) ;
            \draw   (150.72,165.36) .. controls (153.73,167.98) and (156.75,169.56) .. (159.76,170.08) .. controls (156.75,170.6) and (153.73,172.18) .. (150.72,174.8) ;
            \draw   (175.32,149.4) .. controls (177.94,146.39) and (179.52,143.37) .. (180.04,140.36) .. controls (180.56,143.37) and (182.14,146.39) .. (184.76,149.4) ;
            \draw   (125.32,149.8) .. controls (127.94,146.79) and (129.52,143.77) .. (130.04,140.76) .. controls (130.56,143.77) and (132.14,146.79) .. (134.76,149.8) ;
            \draw  [color={rgb, 255:red, 0; green, 0; blue, 0 }  ,draw opacity=1 ][fill={rgb, 255:red, 0; green, 0; blue, 0 }  ,fill opacity=1 ] (128.57,120.2) .. controls (128.57,119.3) and (129.3,118.57) .. (130.2,118.57) .. controls (131.1,118.57) and (131.83,119.3) .. (131.83,120.2) .. controls (131.83,121.1) and (131.1,121.83) .. (130.2,121.83) .. controls (129.3,121.83) and (128.57,121.1) .. (128.57,120.2) -- cycle ;
            
            \draw (148.6,103) node [anchor=north west][inner sep=0.75pt]    {$x_{1}$};
            \draw (188.6,136) node [anchor=north west][inner sep=0.75pt]    {$y_{2}\bar x_1 \omega x_1$};
            \draw (107.4,135.4) node [anchor=north west][inner sep=0.75pt]    {$y_{1}$};
            \draw (150.6,175.6) node [anchor=north west][inner sep=0.75pt]    {$x_{2}$};
            \end{tikzpicture}
            \end{gathered}
        \right)
    \end{aligned}
    \end{equation*}
    By comparing like terms, we see that $\zeta=z,\omega=v, \eta=v h$. 
\end{proof}

\begin{prop}[Plaquette operator at the end of a clockwise ribbon]

The commutation relation is 
\begin{equation}
    \quad B_{\mathrm{col}}^{h}F^{(z,v)}(t)=
        F^{(z, v)}(t)B_{\mathrm{col}}^{\bar z v zh}
\end{equation}
    
\end{prop}

\begin{proof}
    Consider a ribbon of arbitrary length which ends with the two triangle operators as shown in the figure below. Say the product of direct edges for the ribbon is $z'$. 

    On the left hand side,
    \begin{equation*}
    \begin{aligned}
        B^h F^{(z,v)} (t) \left(
        \begin{gathered}  
            \begin{tikzpicture}[x=0.75pt,y=0.75pt,yscale=-1,xscale=1]
            
            \draw   (130.2,120.2) -- (180.2,120.2) -- (180.2,170.2) -- (130.2,170.2) -- cycle ;
            \draw   (150.72,115.36) .. controls (153.73,117.98) and (156.75,119.56) .. (159.76,120.08) .. controls (156.75,120.6) and (153.73,122.18) .. (150.72,124.8) ;
            \draw   (150.72,165.36) .. controls (153.73,167.98) and (156.75,169.56) .. (159.76,170.08) .. controls (156.75,170.6) and (153.73,172.18) .. (150.72,174.8) ;
            \draw   (175.32,149.4) .. controls (177.94,146.39) and (179.52,143.37) .. (180.04,140.36) .. controls (180.56,143.37) and (182.14,146.39) .. (184.76,149.4) ;
            \draw   (125.32,149.8) .. controls (127.94,146.79) and (129.52,143.77) .. (130.04,140.76) .. controls (130.56,143.77) and (132.14,146.79) .. (134.76,149.8) ;
            \draw  [color={rgb, 255:red, 0; green, 0; blue, 0 }  ,draw opacity=1 ][fill={rgb, 255:red, 0; green, 0; blue, 0 }  ,fill opacity=1 ] (128.57,120.2) .. controls (128.57,119.3) and (129.3,118.57) .. (130.2,118.57) .. controls (131.1,118.57) and (131.83,119.3) .. (131.83,120.2) .. controls (131.83,121.1) and (131.1,121.83) .. (130.2,121.83) .. controls (129.3,121.83) and (128.57,121.1) .. (128.57,120.2) -- cycle ;
            
            \draw (148.6,97.6) node [anchor=north west][inner sep=0.75pt]    {$x_{3}$};
            \draw (188.6,136) node [anchor=north west][inner sep=0.75pt]    {$y_{4}$};
            \draw (107.4,135.4) node [anchor=north west][inner sep=0.75pt]    {$y_{3}$};
            \draw (150.6,175.6) node [anchor=north west][inner sep=0.75pt]    {$x_{4}$};
            
            \end{tikzpicture}
        \end{gathered}\right)
        &= \delta_{z,z'} B^h \left(
        \begin{gathered}  
            \begin{tikzpicture}[x=0.75pt,y=0.75pt,yscale=-1,xscale=1]
            
            \draw   (130.2,120.2) -- (180.2,120.2) -- (180.2,170.2) -- (130.2,170.2) -- cycle ;
            \draw   (150.72,115.36) .. controls (153.73,117.98) and (156.75,119.56) .. (159.76,120.08) .. controls (156.75,120.6) and (153.73,122.18) .. (150.72,124.8) ;
            \draw   (150.72,165.36) .. controls (153.73,167.98) and (156.75,169.56) .. (159.76,170.08) .. controls (156.75,170.6) and (153.73,172.18) .. (150.72,174.8) ;
            \draw   (175.32,149.4) .. controls (177.94,146.39) and (179.52,143.37) .. (180.04,140.36) .. controls (180.56,143.37) and (182.14,146.39) .. (184.76,149.4) ;
            \draw   (125.32,149.8) .. controls (127.94,146.79) and (129.52,143.77) .. (130.04,140.76) .. controls (130.56,143.77) and (132.14,146.79) .. (134.76,149.8) ;
            \draw  [color={rgb, 255:red, 0; green, 0; blue, 0 }  ,draw opacity=1 ][fill={rgb, 255:red, 0; green, 0; blue, 0 }  ,fill opacity=1 ] (128.57,120.2) .. controls (128.57,119.3) and (129.3,118.57) .. (130.2,118.57) .. controls (131.1,118.57) and (131.83,119.3) .. (131.83,120.2) .. controls (131.83,121.1) and (131.1,121.83) .. (130.2,121.83) .. controls (129.3,121.83) and (128.57,121.1) .. (128.57,120.2) -- cycle ;
            
            \draw (148.6,97.6) node [anchor=north west][inner sep=0.75pt]    {$x_{3}$};
            \draw (188.6,136) node [anchor=north west][inner sep=0.75pt]    {$y_{4}$};
            \draw (73,135.4) node [anchor=north west][inner sep=0.75pt]    {$y_{3}\overline x_1 v x_1$};
            \draw (150.6,175.6) node [anchor=north west][inner sep=0.75pt]    {$x_{4}$};
            
            \end{tikzpicture}
        \end{gathered}\right)\\
        &= \delta_{z,z'}\delta_{\overline h, \overline{z'} \overline v z' \overline y_3 x_4 \overline y_4 \overline x_3} \left(
        \begin{gathered}  
            \begin{tikzpicture}[x=0.75pt,y=0.75pt,yscale=-1,xscale=1]
            
            \draw   (130.2,120.2) -- (180.2,120.2) -- (180.2,170.2) -- (130.2,170.2) -- cycle ;
            \draw   (150.72,115.36) .. controls (153.73,117.98) and (156.75,119.56) .. (159.76,120.08) .. controls (156.75,120.6) and (153.73,122.18) .. (150.72,124.8) ;
            \draw   (150.72,165.36) .. controls (153.73,167.98) and (156.75,169.56) .. (159.76,170.08) .. controls (156.75,170.6) and (153.73,172.18) .. (150.72,174.8) ;
            \draw   (175.32,149.4) .. controls (177.94,146.39) and (179.52,143.37) .. (180.04,140.36) .. controls (180.56,143.37) and (182.14,146.39) .. (184.76,149.4) ;
            \draw   (125.32,149.8) .. controls (127.94,146.79) and (129.52,143.77) .. (130.04,140.76) .. controls (130.56,143.77) and (132.14,146.79) .. (134.76,149.8) ;
            \draw  [color={rgb, 255:red, 0; green, 0; blue, 0 }  ,draw opacity=1 ][fill={rgb, 255:red, 0; green, 0; blue, 0 }  ,fill opacity=1 ] (128.57,120.2) .. controls (128.57,119.3) and (129.3,118.57) .. (130.2,118.57) .. controls (131.1,118.57) and (131.83,119.3) .. (131.83,120.2) .. controls (131.83,121.1) and (131.1,121.83) .. (130.2,121.83) .. controls (129.3,121.83) and (128.57,121.1) .. (128.57,120.2) -- cycle ;
            
            \draw (148.6,97.6) node [anchor=north west][inner sep=0.75pt]    {$x_{3}$};
            \draw (188.6,136) node [anchor=north west][inner sep=0.75pt]    {$y_{4}$};
            \draw (73,135.4) node [anchor=north west][inner sep=0.75pt]    {$y_{3}\overline x_1 v x_1$};
            \draw (150.6,175.6) node [anchor=north west][inner sep=0.75pt]    {$x_{4}$};
            
            \end{tikzpicture}
        \end{gathered}\right)
    \end{aligned}
    \end{equation*}

    On the right hand side, 
    \begin{equation*}
    \begin{aligned}
        F^{(\zeta,\omega)} (t) B^\eta \left(
        \begin{gathered}  
            \begin{tikzpicture}[x=0.75pt,y=0.75pt,yscale=-1,xscale=1]
            
            \draw   (130.2,120.2) -- (180.2,120.2) -- (180.2,170.2) -- (130.2,170.2) -- cycle ;
            \draw   (150.72,115.36) .. controls (153.73,117.98) and (156.75,119.56) .. (159.76,120.08) .. controls (156.75,120.6) and (153.73,122.18) .. (150.72,124.8) ;
            \draw   (150.72,165.36) .. controls (153.73,167.98) and (156.75,169.56) .. (159.76,170.08) .. controls (156.75,170.6) and (153.73,172.18) .. (150.72,174.8) ;
            \draw   (175.32,149.4) .. controls (177.94,146.39) and (179.52,143.37) .. (180.04,140.36) .. controls (180.56,143.37) and (182.14,146.39) .. (184.76,149.4) ;
            \draw   (125.32,149.8) .. controls (127.94,146.79) and (129.52,143.77) .. (130.04,140.76) .. controls (130.56,143.77) and (132.14,146.79) .. (134.76,149.8) ;
            \draw  [color={rgb, 255:red, 0; green, 0; blue, 0 }  ,draw opacity=1 ][fill={rgb, 255:red, 0; green, 0; blue, 0 }  ,fill opacity=1 ] (128.57,120.2) .. controls (128.57,119.3) and (129.3,118.57) .. (130.2,118.57) .. controls (131.1,118.57) and (131.83,119.3) .. (131.83,120.2) .. controls (131.83,121.1) and (131.1,121.83) .. (130.2,121.83) .. controls (129.3,121.83) and (128.57,121.1) .. (128.57,120.2) -- cycle ;
            
            \draw (148.6,97.6) node [anchor=north west][inner sep=0.75pt]    {$x_{3}$};
            \draw (188.6,136) node [anchor=north west][inner sep=0.75pt]    {$y_{4}$};
            \draw (107.4,135.4) node [anchor=north west][inner sep=0.75pt]    {$y_{3}$};
            \draw (150.6,175.6) node [anchor=north west][inner sep=0.75pt]    {$x_{4}$};
            
            \end{tikzpicture}
        \end{gathered}\right)
        &= \delta_{\overline \eta, y_3x_4\overline y_4\overline x_3} F^{(\zeta,\omega)}(t) \left(
        \begin{gathered}  
            \begin{tikzpicture}[x=0.75pt,y=0.75pt,yscale=-1,xscale=1]
            
            \draw   (130.2,120.2) -- (180.2,120.2) -- (180.2,170.2) -- (130.2,170.2) -- cycle ;
            \draw   (150.72,115.36) .. controls (153.73,117.98) and (156.75,119.56) .. (159.76,120.08) .. controls (156.75,120.6) and (153.73,122.18) .. (150.72,124.8) ;
            \draw   (150.72,165.36) .. controls (153.73,167.98) and (156.75,169.56) .. (159.76,170.08) .. controls (156.75,170.6) and (153.73,172.18) .. (150.72,174.8) ;
            \draw   (175.32,149.4) .. controls (177.94,146.39) and (179.52,143.37) .. (180.04,140.36) .. controls (180.56,143.37) and (182.14,146.39) .. (184.76,149.4) ;
            \draw   (125.32,149.8) .. controls (127.94,146.79) and (129.52,143.77) .. (130.04,140.76) .. controls (130.56,143.77) and (132.14,146.79) .. (134.76,149.8) ;
            \draw  [color={rgb, 255:red, 0; green, 0; blue, 0 }  ,draw opacity=1 ][fill={rgb, 255:red, 0; green, 0; blue, 0 }  ,fill opacity=1 ] (128.57,120.2) .. controls (128.57,119.3) and (129.3,118.57) .. (130.2,118.57) .. controls (131.1,118.57) and (131.83,119.3) .. (131.83,120.2) .. controls (131.83,121.1) and (131.1,121.83) .. (130.2,121.83) .. controls (129.3,121.83) and (128.57,121.1) .. (128.57,120.2) -- cycle ;
            
            \draw (148.6,97.6) node [anchor=north west][inner sep=0.75pt]    {$x_{3}$};
            \draw (188.6,136) node [anchor=north west][inner sep=0.75pt]    {$y_{4}$};
            \draw (107.4,135.4) node [anchor=north west][inner sep=0.75pt]    {$y_{3}$};
            \draw (150.6,175.6) node [anchor=north west][inner sep=0.75pt]    {$x_{4}$};
            
            \end{tikzpicture}
        \end{gathered}\right)\\
        &= \delta_{\zeta,z'}\delta_{\overline \eta, y_3x_4\overline y_4\overline x_3} \left(
        \begin{gathered}  
            \begin{tikzpicture}[x=0.75pt,y=0.75pt,yscale=-1,xscale=1]
            
            \draw   (130.2,120.2) -- (180.2,120.2) -- (180.2,170.2) -- (130.2,170.2) -- cycle ;
            \draw   (150.72,115.36) .. controls (153.73,117.98) and (156.75,119.56) .. (159.76,120.08) .. controls (156.75,120.6) and (153.73,122.18) .. (150.72,124.8) ;
            \draw   (150.72,165.36) .. controls (153.73,167.98) and (156.75,169.56) .. (159.76,170.08) .. controls (156.75,170.6) and (153.73,172.18) .. (150.72,174.8) ;
            \draw   (175.32,149.4) .. controls (177.94,146.39) and (179.52,143.37) .. (180.04,140.36) .. controls (180.56,143.37) and (182.14,146.39) .. (184.76,149.4) ;
            \draw   (125.32,149.8) .. controls (127.94,146.79) and (129.52,143.77) .. (130.04,140.76) .. controls (130.56,143.77) and (132.14,146.79) .. (134.76,149.8) ;
            \draw  [color={rgb, 255:red, 0; green, 0; blue, 0 }  ,draw opacity=1 ][fill={rgb, 255:red, 0; green, 0; blue, 0 }  ,fill opacity=1 ] (128.57,120.2) .. controls (128.57,119.3) and (129.3,118.57) .. (130.2,118.57) .. controls (131.1,118.57) and (131.83,119.3) .. (131.83,120.2) .. controls (131.83,121.1) and (131.1,121.83) .. (130.2,121.83) .. controls (129.3,121.83) and (128.57,121.1) .. (128.57,120.2) -- cycle ;
            
            \draw (148.6,97.6) node [anchor=north west][inner sep=0.75pt]    {$x_{3}$};
            \draw (188.6,136) node [anchor=north west][inner sep=0.75pt]    {$y_{4}$};
            \draw (73,135.4) node [anchor=north west][inner sep=0.75pt]    {$y_{3}\overline x_1 \omega x_1$};
            \draw (150.6,175.6) node [anchor=north west][inner sep=0.75pt]    {$x_{4}$};
            
            \end{tikzpicture}
        \end{gathered}\right)
    \end{aligned}
    \end{equation*}

    By comparing like terms, we see that $\zeta=z,\omega=v, \eta=\overline z v z h$. 

\end{proof}

\textbf{2. Counterclockwise local orientation}

Consider a ribbon with counterclockwise local orientation as shown in Fig. \ref{fig:counterclockwise ribbon for plaquette}. 

\begin{figure}[h!]
    \centering
    \includegraphics[width=0.7\textwidth]{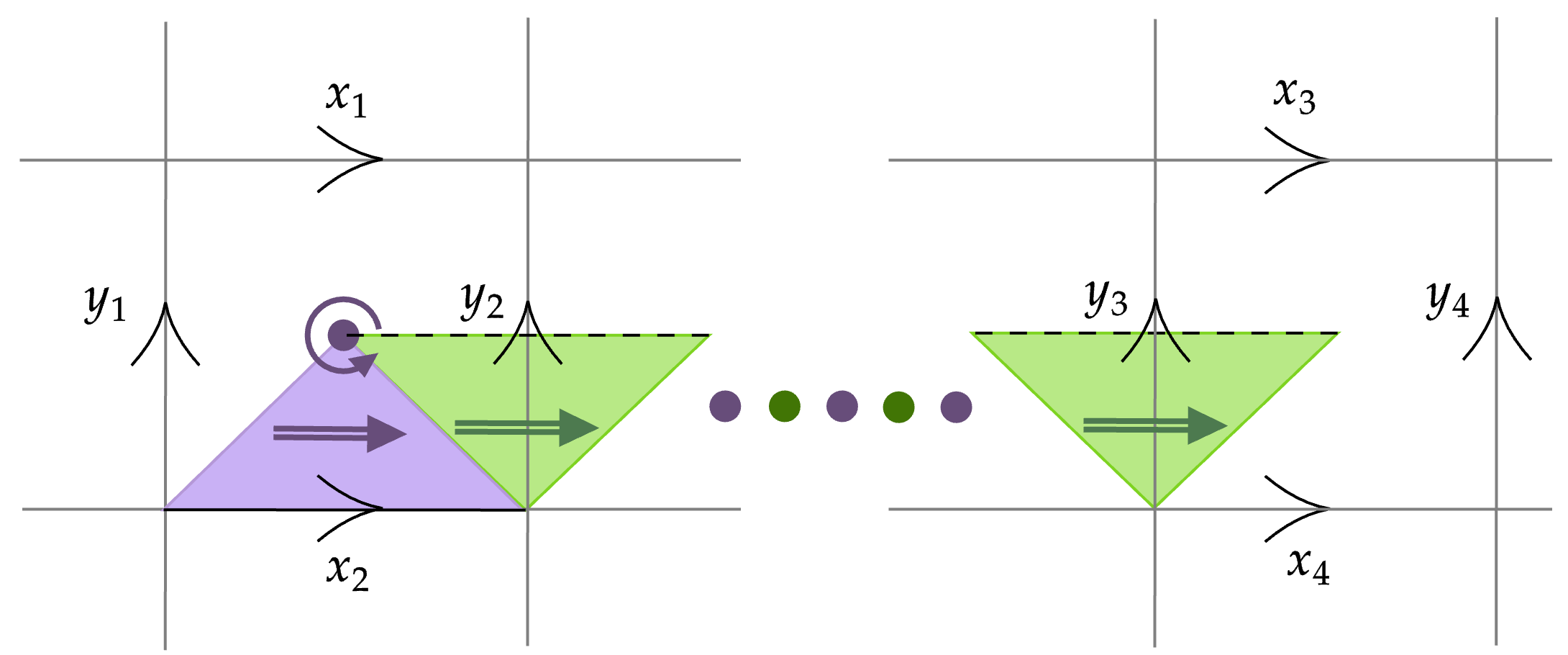}
    \caption{A ribbon operator with counterclockwise local orientation.}
    \label{fig:counterclockwise ribbon for plaquette}
\end{figure}

\begin{prop}[Plaquette operator at the start of a counterclockwise ribbon]

The commutation relation is
    \begin{equation}
        \quad B_{\mathrm{fl}}^{h}F^{(z, v)}(t)=F^{(z, v)}(t)B_{\mathrm{fl}}^{h v}.
    \end{equation}
\end{prop}

\begin{proof}
        
For the left hand side,
    \begin{equation*}
    \begin{aligned}
        B^hF^{(z,v)}(t)
        \left(\begin{gathered}
            \begin{tikzpicture}[x=0.75pt,y=0.75pt,yscale=-1,xscale=1]
            
            \draw   (130.2,120.2) -- (180.2,120.2) -- (180.2,170.2) -- (130.2,170.2) -- cycle ;
            \draw   (150.72,115.36) .. controls (153.73,117.98) and (156.75,119.56) .. (159.76,120.08) .. controls (156.75,120.6) and (153.73,122.18) .. (150.72,124.8) ;
            \draw   (150.72,165.36) .. controls (153.73,167.98) and (156.75,169.56) .. (159.76,170.08) .. controls (156.75,170.6) and (153.73,172.18) .. (150.72,174.8) ;
            \draw   (175.32,149.4) .. controls (177.94,146.39) and (179.52,143.37) .. (180.04,140.36) .. controls (180.56,143.37) and (182.14,146.39) .. (184.76,149.4) ;
            \draw   (125.32,149.8) .. controls (127.94,146.79) and (129.52,143.77) .. (130.04,140.76) .. controls (130.56,143.77) and (132.14,146.79) .. (134.76,149.8) ;
            \draw  [color={rgb, 255:red, 0; green, 0; blue, 0 }  ,draw opacity=1 ][fill={rgb, 255:red, 0; green, 0; blue, 0 }  ,fill opacity=1 ] (128.57,170.2) .. controls (128.57,169.3) and (129.3,168.57) .. (130.2,168.57) .. controls (131.1,168.57) and (131.83,169.3) .. (131.83,170.2) .. controls (131.83,171.1) and (131.1,171.83) .. (130.2,171.83) .. controls (129.3,171.83) and (128.57,171.1) .. (128.57,170.2) -- cycle ;
            
            \draw (148.6,103) node [anchor=north west][inner sep=0.75pt]    {$x_{1}$};
            \draw (188.6,136) node [anchor=north west][inner sep=0.75pt]    {$y_{2}$};
            \draw (107.4,135.4) node [anchor=north west][inner sep=0.75pt]    {$y_{1}$};
            \draw (150.6,175.6) node [anchor=north west][inner sep=0.75pt]    {$x_{2}$};
            \end{tikzpicture}
            \end{gathered}
        \right)
        &=\delta_{z,x_2} B^h
        \left(\begin{gathered}
            \begin{tikzpicture}[x=0.75pt,y=0.75pt,yscale=-1,xscale=1]
            
            \draw   (130.2,120.2) -- (180.2,120.2) -- (180.2,170.2) -- (130.2,170.2) -- cycle ;
            \draw   (150.72,115.36) .. controls (153.73,117.98) and (156.75,119.56) .. (159.76,120.08) .. controls (156.75,120.6) and (153.73,122.18) .. (150.72,124.8) ;
            \draw   (150.72,165.36) .. controls (153.73,167.98) and (156.75,169.56) .. (159.76,170.08) .. controls (156.75,170.6) and (153.73,172.18) .. (150.72,174.8) ;
            \draw   (175.32,149.4) .. controls (177.94,146.39) and (179.52,143.37) .. (180.04,140.36) .. controls (180.56,143.37) and (182.14,146.39) .. (184.76,149.4) ;
            \draw   (125.32,149.8) .. controls (127.94,146.79) and (129.52,143.77) .. (130.04,140.76) .. controls (130.56,143.77) and (132.14,146.79) .. (134.76,149.8) ;
            \draw  [color={rgb, 255:red, 0; green, 0; blue, 0 }  ,draw opacity=1 ][fill={rgb, 255:red, 0; green, 0; blue, 0 }  ,fill opacity=1 ] (128.57,170.2) .. controls (128.57,169.3) and (129.3,168.57) .. (130.2,168.57) .. controls (131.1,168.57) and (131.83,169.3) .. (131.83,170.2) .. controls (131.83,171.1) and (131.1,171.83) .. (130.2,171.83) .. controls (129.3,171.83) and (128.57,171.1) .. (128.57,170.2) -- cycle ;
            
            \draw (148.6,103) node [anchor=north west][inner sep=0.75pt]    {$x_{1}$};
            \draw (188.6,136) node [anchor=north west][inner sep=0.75pt]    {$\bar x_2 v x_2 y_2$};
            \draw (107.4,135.4) node [anchor=north west][inner sep=0.75pt]    {$y_{1}$};
            \draw (150.6,175.6) node [anchor=north west][inner sep=0.75pt]    {$x_{2}$};
            \end{tikzpicture}
            \end{gathered}
        \right)\\
        &=\delta_{z,x_2}\delta_{\bar h, vx_2y_2\bar x_1\bar y_1}
        \left(\begin{gathered}
            \begin{tikzpicture}[x=0.75pt,y=0.75pt,yscale=-1,xscale=1]
            
            \draw   (130.2,120.2) -- (180.2,120.2) -- (180.2,170.2) -- (130.2,170.2) -- cycle ;
            \draw   (150.72,115.36) .. controls (153.73,117.98) and (156.75,119.56) .. (159.76,120.08) .. controls (156.75,120.6) and (153.73,122.18) .. (150.72,124.8) ;
            \draw   (150.72,165.36) .. controls (153.73,167.98) and (156.75,169.56) .. (159.76,170.08) .. controls (156.75,170.6) and (153.73,172.18) .. (150.72,174.8) ;
            \draw   (175.32,149.4) .. controls (177.94,146.39) and (179.52,143.37) .. (180.04,140.36) .. controls (180.56,143.37) and (182.14,146.39) .. (184.76,149.4) ;
            \draw   (125.32,149.8) .. controls (127.94,146.79) and (129.52,143.77) .. (130.04,140.76) .. controls (130.56,143.77) and (132.14,146.79) .. (134.76,149.8) ;
            \draw  [color={rgb, 255:red, 0; green, 0; blue, 0 }  ,draw opacity=1 ][fill={rgb, 255:red, 0; green, 0; blue, 0 }  ,fill opacity=1 ] (128.57,170.2) .. controls (128.57,169.3) and (129.3,168.57) .. (130.2,168.57) .. controls (131.1,168.57) and (131.83,169.3) .. (131.83,170.2) .. controls (131.83,171.1) and (131.1,171.83) .. (130.2,171.83) .. controls (129.3,171.83) and (128.57,171.1) .. (128.57,170.2) -- cycle ;
            
            \draw (148.6,103) node [anchor=north west][inner sep=0.75pt]    {$x_{1}$};
            \draw (188.6,136) node [anchor=north west][inner sep=0.75pt]    {$\bar x_2 v x_2 y_2$};
            \draw (107.4,135.4) node [anchor=north west][inner sep=0.75pt]    {$y_{1}$};
            \draw (150.6,175.6) node [anchor=north west][inner sep=0.75pt]    {$x_{2}$};
            \end{tikzpicture}
            \end{gathered}
        \right)
    \end{aligned}
    \end{equation*}

For the right hand side, we assume a different set of variables ($\zeta$ for charge, $\omega$ for local flux, and $\eta$ for the flux variable of the plaquette operator) and determine their values.
    \begin{equation*}
    \begin{aligned}
        F^{(\zeta,\omega)}(t) B^{\eta}
        \left(\begin{gathered}
            \begin{tikzpicture}[x=0.75pt,y=0.75pt,yscale=-1,xscale=1]
            
            \draw   (130.2,120.2) -- (180.2,120.2) -- (180.2,170.2) -- (130.2,170.2) -- cycle ;
            \draw   (150.72,115.36) .. controls (153.73,117.98) and (156.75,119.56) .. (159.76,120.08) .. controls (156.75,120.6) and (153.73,122.18) .. (150.72,124.8) ;
            \draw   (150.72,165.36) .. controls (153.73,167.98) and (156.75,169.56) .. (159.76,170.08) .. controls (156.75,170.6) and (153.73,172.18) .. (150.72,174.8) ;
            \draw   (175.32,149.4) .. controls (177.94,146.39) and (179.52,143.37) .. (180.04,140.36) .. controls (180.56,143.37) and (182.14,146.39) .. (184.76,149.4) ;
            \draw   (125.32,149.8) .. controls (127.94,146.79) and (129.52,143.77) .. (130.04,140.76) .. controls (130.56,143.77) and (132.14,146.79) .. (134.76,149.8) ;
            \draw  [color={rgb, 255:red, 0; green, 0; blue, 0 }  ,draw opacity=1 ][fill={rgb, 255:red, 0; green, 0; blue, 0 }  ,fill opacity=1 ] (128.57,170.2) .. controls (128.57,169.3) and (129.3,168.57) .. (130.2,168.57) .. controls (131.1,168.57) and (131.83,169.3) .. (131.83,170.2) .. controls (131.83,171.1) and (131.1,171.83) .. (130.2,171.83) .. controls (129.3,171.83) and (128.57,171.1) .. (128.57,170.2) -- cycle ;
            
            \draw (148.6,103) node [anchor=north west][inner sep=0.75pt]    {$x_{1}$};
            \draw (188.6,136) node [anchor=north west][inner sep=0.75pt]    {$y_{2}$};
            \draw (107.4,135.4) node [anchor=north west][inner sep=0.75pt]    {$y_{1}$};
            \draw (150.6,175.6) node [anchor=north west][inner sep=0.75pt]    {$x_{2}$};
            \end{tikzpicture}
            \end{gathered}
        \right)
        &=\delta_{\bar \eta,x_2 y_2 \bar x_1 \bar y_1} F^{(\zeta,\omega)}(t)
        \left(\begin{gathered}
            \begin{tikzpicture}[x=0.75pt,y=0.75pt,yscale=-1,xscale=1]
            
            \draw   (130.2,120.2) -- (180.2,120.2) -- (180.2,170.2) -- (130.2,170.2) -- cycle ;
            \draw   (150.72,115.36) .. controls (153.73,117.98) and (156.75,119.56) .. (159.76,120.08) .. controls (156.75,120.6) and (153.73,122.18) .. (150.72,124.8) ;
            \draw   (150.72,165.36) .. controls (153.73,167.98) and (156.75,169.56) .. (159.76,170.08) .. controls (156.75,170.6) and (153.73,172.18) .. (150.72,174.8) ;
            \draw   (175.32,149.4) .. controls (177.94,146.39) and (179.52,143.37) .. (180.04,140.36) .. controls (180.56,143.37) and (182.14,146.39) .. (184.76,149.4) ;
            \draw   (125.32,149.8) .. controls (127.94,146.79) and (129.52,143.77) .. (130.04,140.76) .. controls (130.56,143.77) and (132.14,146.79) .. (134.76,149.8) ;
            \draw  [color={rgb, 255:red, 0; green, 0; blue, 0 }  ,draw opacity=1 ][fill={rgb, 255:red, 0; green, 0; blue, 0 }  ,fill opacity=1 ] (128.57,170.2) .. controls (128.57,169.3) and (129.3,168.57) .. (130.2,168.57) .. controls (131.1,168.57) and (131.83,169.3) .. (131.83,170.2) .. controls (131.83,171.1) and (131.1,171.83) .. (130.2,171.83) .. controls (129.3,171.83) and (128.57,171.1) .. (128.57,170.2) -- cycle ;
            
            \draw (148.6,103) node [anchor=north west][inner sep=0.75pt]    {$x_{1}$};
            \draw (188.6,136) node [anchor=north west][inner sep=0.75pt]    {$y_{2}$};
            \draw (107.4,135.4) node [anchor=north west][inner sep=0.75pt]    {$y_{1}$};
            \draw (150.6,175.6) node [anchor=north west][inner sep=0.75pt]    {$x_{2}$};
            \end{tikzpicture}
            \end{gathered}
        \right)\\
        &=\delta_{\zeta,x_2}\delta_{\bar \eta, x_2 y_2\bar x_1\bar y_1}
        \left(\begin{gathered}
            \begin{tikzpicture}[x=0.75pt,y=0.75pt,yscale=-1,xscale=1]
            
            \draw   (130.2,120.2) -- (180.2,120.2) -- (180.2,170.2) -- (130.2,170.2) -- cycle ;
            \draw   (150.72,115.36) .. controls (153.73,117.98) and (156.75,119.56) .. (159.76,120.08) .. controls (156.75,120.6) and (153.73,122.18) .. (150.72,124.8) ;
            \draw   (150.72,165.36) .. controls (153.73,167.98) and (156.75,169.56) .. (159.76,170.08) .. controls (156.75,170.6) and (153.73,172.18) .. (150.72,174.8) ;
            \draw   (175.32,149.4) .. controls (177.94,146.39) and (179.52,143.37) .. (180.04,140.36) .. controls (180.56,143.37) and (182.14,146.39) .. (184.76,149.4) ;
            \draw   (125.32,149.8) .. controls (127.94,146.79) and (129.52,143.77) .. (130.04,140.76) .. controls (130.56,143.77) and (132.14,146.79) .. (134.76,149.8) ;
            \draw  [color={rgb, 255:red, 0; green, 0; blue, 0 }  ,draw opacity=1 ][fill={rgb, 255:red, 0; green, 0; blue, 0 }  ,fill opacity=1 ] (128.57,170.2) .. controls (128.57,169.3) and (129.3,168.57) .. (130.2,168.57) .. controls (131.1,168.57) and (131.83,169.3) .. (131.83,170.2) .. controls (131.83,171.1) and (131.1,171.83) .. (130.2,171.83) .. controls (129.3,171.83) and (128.57,171.1) .. (128.57,170.2) -- cycle ;
            
            \draw (148.6,103) node [anchor=north west][inner sep=0.75pt]    {$x_{1}$};
            \draw (188.6,136) node [anchor=north west][inner sep=0.75pt]    {$\bar x_2 \omega x_2 y_2$};
            \draw (107.4,135.4) node [anchor=north west][inner sep=0.75pt]    {$y_{1}$};
            \draw (150.6,175.6) node [anchor=north west][inner sep=0.75pt]    {$x_{2}$};
            \end{tikzpicture}
            \end{gathered}
        \right)
    \end{aligned}
    \end{equation*}

    By comparing like terms, we see that $\zeta=z,\omega=v, \eta=hv$. 
\end{proof}

\begin{prop}[Plaquette operator at the end of a counterclockwise ribbon]

The commutation relation is 

\begin{equation}
    \quad B_{\mathrm{col}}^{h}F^{(z,v)}(t)=
        F^{(z, v)}(t)B_{\mathrm{col}}^{\bar{z} \bar{v} z \bar h}
\end{equation}
    
\end{prop}

\begin{proof}
    Consider a ribbon of arbitrary length which ends with the dual triangle operators as shown in Fig. \ref{fig:counterclockwise ribbon for plaquette}. Let $z'$ be the product of direct edges along the ribbon.

    On the left hand side,
    \begin{equation*}
    \begin{aligned}
        B^h F^{(z,v)} (t) \left(
        \begin{gathered}  
            \begin{tikzpicture}[x=0.75pt,y=0.75pt,yscale=-1,xscale=1]
            
            \draw   (130.2,120.2) -- (180.2,120.2) -- (180.2,170.2) -- (130.2,170.2) -- cycle ;
            \draw   (150.72,115.36) .. controls (153.73,117.98) and (156.75,119.56) .. (159.76,120.08) .. controls (156.75,120.6) and (153.73,122.18) .. (150.72,124.8) ;
            \draw   (150.72,165.36) .. controls (153.73,167.98) and (156.75,169.56) .. (159.76,170.08) .. controls (156.75,170.6) and (153.73,172.18) .. (150.72,174.8) ;
            \draw   (175.32,149.4) .. controls (177.94,146.39) and (179.52,143.37) .. (180.04,140.36) .. controls (180.56,143.37) and (182.14,146.39) .. (184.76,149.4) ;
            \draw   (125.32,149.8) .. controls (127.94,146.79) and (129.52,143.77) .. (130.04,140.76) .. controls (130.56,143.77) and (132.14,146.79) .. (134.76,149.8) ;
            \draw  [color={rgb, 255:red, 0; green, 0; blue, 0 }  ,draw opacity=1 ][fill={rgb, 255:red, 0; green, 0; blue, 0 }  ,fill opacity=1 ] (128.57,170.2) .. controls (128.57,169.3) and (129.3,168.57) .. (130.2,168.57) .. controls (131.1,168.57) and (131.83,169.3) .. (131.83,170.2) .. controls (131.83,171.1) and (131.1,171.83) .. (130.2,171.83) .. controls (129.3,171.83) and (128.57,171.1) .. (128.57,170.2) -- cycle ;
            
            \draw (148.6,97.6) node [anchor=north west][inner sep=0.75pt]    {$x_{3}$};
            \draw (188.6,136) node [anchor=north west][inner sep=0.75pt]    {$y_{3}$};
            \draw (107.4,135.4) node [anchor=north west][inner sep=0.75pt]    {$y_{2}$};
            \draw (150.6,175.6) node [anchor=north west][inner sep=0.75pt]    {$x_{4}$};
            
            \end{tikzpicture}
        \end{gathered}\right)
        &= \delta_{z,z'} B^h \left(
        \begin{gathered}  
            \begin{tikzpicture}[x=0.75pt,y=0.75pt,yscale=-1,xscale=1]
            
            \draw   (130.2,120.2) -- (180.2,120.2) -- (180.2,170.2) -- (130.2,170.2) -- cycle ;
            \draw   (150.72,115.36) .. controls (153.73,117.98) and (156.75,119.56) .. (159.76,120.08) .. controls (156.75,120.6) and (153.73,122.18) .. (150.72,124.8) ;
            \draw   (150.72,165.36) .. controls (153.73,167.98) and (156.75,169.56) .. (159.76,170.08) .. controls (156.75,170.6) and (153.73,172.18) .. (150.72,174.8) ;
            \draw   (175.32,149.4) .. controls (177.94,146.39) and (179.52,143.37) .. (180.04,140.36) .. controls (180.56,143.37) and (182.14,146.39) .. (184.76,149.4) ;
            \draw   (125.32,149.8) .. controls (127.94,146.79) and (129.52,143.77) .. (130.04,140.76) .. controls (130.56,143.77) and (132.14,146.79) .. (134.76,149.8) ;
            \draw  [color={rgb, 255:red, 0; green, 0; blue, 0 }  ,draw opacity=1 ][fill={rgb, 255:red, 0; green, 0; blue, 0 }  ,fill opacity=1 ] (128.57,170.2) .. controls (128.57,169.3) and (129.3,168.57) .. (130.2,168.57) .. controls (131.1,168.57) and (131.83,169.3) .. (131.83,170.2) .. controls (131.83,171.1) and (131.1,171.83) .. (130.2,171.83) .. controls (129.3,171.83) and (128.57,171.1) .. (128.57,170.2) -- cycle ;
            
            \draw (148.6,97.6) node [anchor=north west][inner sep=0.75pt]    {$x_{3}$};
            \draw (188.6,136) node [anchor=north west][inner sep=0.75pt]    {$y_{3}$};
            \draw (73,135.4) node [anchor=north west][inner sep=0.75pt]    {$\overline{z}' v z' y_{2}$};
            \draw (150.6,175.6) node [anchor=north west][inner sep=0.75pt]    {$x_{4}$};
            
            \end{tikzpicture}
        \end{gathered}\right)\\
        &= \delta_{z,z'}\delta_{\overline h, \overline{z'} \overline v z' \overline y_2 x_4 \overline y_3 \overline x_3} \left(
        \begin{gathered}  
            \begin{tikzpicture}[x=0.75pt,y=0.75pt,yscale=-1,xscale=1]
            
            \draw   (130.2,120.2) -- (180.2,120.2) -- (180.2,170.2) -- (130.2,170.2) -- cycle ;
            \draw   (150.72,115.36) .. controls (153.73,117.98) and (156.75,119.56) .. (159.76,120.08) .. controls (156.75,120.6) and (153.73,122.18) .. (150.72,124.8) ;
            \draw   (150.72,165.36) .. controls (153.73,167.98) and (156.75,169.56) .. (159.76,170.08) .. controls (156.75,170.6) and (153.73,172.18) .. (150.72,174.8) ;
            \draw   (175.32,149.4) .. controls (177.94,146.39) and (179.52,143.37) .. (180.04,140.36) .. controls (180.56,143.37) and (182.14,146.39) .. (184.76,149.4) ;
            \draw   (125.32,149.8) .. controls (127.94,146.79) and (129.52,143.77) .. (130.04,140.76) .. controls (130.56,143.77) and (132.14,146.79) .. (134.76,149.8) ;
            \draw  [color={rgb, 255:red, 0; green, 0; blue, 0 }  ,draw opacity=1 ][fill={rgb, 255:red, 0; green, 0; blue, 0 }  ,fill opacity=1 ] (128.57,170.2) .. controls (128.57,169.3) and (129.3,168.57) .. (130.2,168.57) .. controls (131.1,168.57) and (131.83,169.3) .. (131.83,170.2) .. controls (131.83,171.1) and (131.1,171.83) .. (130.2,171.83) .. controls (129.3,171.83) and (128.57,171.1) .. (128.57,170.2) -- cycle ;
            
            \draw (148.6,97.6) node [anchor=north west][inner sep=0.75pt]    {$x_{3}$};
            \draw (188.6,136) node [anchor=north west][inner sep=0.75pt]    {$y_{3}$};
            \draw (73,135.4) node [anchor=north west][inner sep=0.75pt]    {$y_{2}\overline x_1 v x_1$};
            \draw (150.6,175.6) node [anchor=north west][inner sep=0.75pt]    {$x_{4}$};
            
            \end{tikzpicture}
        \end{gathered}\right)
    \end{aligned}
    \end{equation*}

    On the right hand side, 
    \begin{equation*}
    \begin{aligned}
        F^{(\zeta,\omega)} (t) B^\eta \left(
        \begin{gathered}  
            \begin{tikzpicture}[x=0.75pt,y=0.75pt,yscale=-1,xscale=1]
            
            \draw   (130.2,120.2) -- (180.2,120.2) -- (180.2,170.2) -- (130.2,170.2) -- cycle ;
            \draw   (150.72,115.36) .. controls (153.73,117.98) and (156.75,119.56) .. (159.76,120.08) .. controls (156.75,120.6) and (153.73,122.18) .. (150.72,124.8) ;
            \draw   (150.72,165.36) .. controls (153.73,167.98) and (156.75,169.56) .. (159.76,170.08) .. controls (156.75,170.6) and (153.73,172.18) .. (150.72,174.8) ;
            \draw   (175.32,149.4) .. controls (177.94,146.39) and (179.52,143.37) .. (180.04,140.36) .. controls (180.56,143.37) and (182.14,146.39) .. (184.76,149.4) ;
            \draw   (125.32,149.8) .. controls (127.94,146.79) and (129.52,143.77) .. (130.04,140.76) .. controls (130.56,143.77) and (132.14,146.79) .. (134.76,149.8) ;
            \draw  [color={rgb, 255:red, 0; green, 0; blue, 0 }  ,draw opacity=1 ][fill={rgb, 255:red, 0; green, 0; blue, 0 }  ,fill opacity=1 ] (128.57,170.2) .. controls (128.57,169.3) and (129.3,168.57) .. (130.2,168.57) .. controls (131.1,168.57) and (131.83,169.3) .. (131.83,170.2) .. controls (131.83,171.1) and (131.1,171.83) .. (130.2,171.83) .. controls (129.3,171.83) and (128.57,171.1) .. (128.57,170.2) -- cycle ;
            
            \draw (148.6,97.6) node [anchor=north west][inner sep=0.75pt]    {$x_{3}$};
            \draw (188.6,136) node [anchor=north west][inner sep=0.75pt]    {$y_{3}$};
            \draw (107.4,135.4) node [anchor=north west][inner sep=0.75pt]    {$y_{2}$};
            \draw (150.6,175.6) node [anchor=north west][inner sep=0.75pt]    {$x_{4}$};
            
            \end{tikzpicture}
        \end{gathered}\right)
        &= \delta_{\overline \eta, y_2x_4\overline y_3\overline x_3} F^{(\zeta,\omega)}(t) \left(
        \begin{gathered}  
            \begin{tikzpicture}[x=0.75pt,y=0.75pt,yscale=-1,xscale=1]
            
            \draw   (130.2,120.2) -- (180.2,120.2) -- (180.2,170.2) -- (130.2,170.2) -- cycle ;
            \draw   (150.72,115.36) .. controls (153.73,117.98) and (156.75,119.56) .. (159.76,120.08) .. controls (156.75,120.6) and (153.73,122.18) .. (150.72,124.8) ;
            \draw   (150.72,165.36) .. controls (153.73,167.98) and (156.75,169.56) .. (159.76,170.08) .. controls (156.75,170.6) and (153.73,172.18) .. (150.72,174.8) ;
            \draw   (175.32,149.4) .. controls (177.94,146.39) and (179.52,143.37) .. (180.04,140.36) .. controls (180.56,143.37) and (182.14,146.39) .. (184.76,149.4) ;
            \draw   (125.32,149.8) .. controls (127.94,146.79) and (129.52,143.77) .. (130.04,140.76) .. controls (130.56,143.77) and (132.14,146.79) .. (134.76,149.8) ;
            \draw  [color={rgb, 255:red, 0; green, 0; blue, 0 }  ,draw opacity=1 ][fill={rgb, 255:red, 0; green, 0; blue, 0 }  ,fill opacity=1 ] (128.57,170.2) .. controls (128.57,169.3) and (129.3,168.57) .. (130.2,168.57) .. controls (131.1,168.57) and (131.83,169.3) .. (131.83,170.2) .. controls (131.83,171.1) and (131.1,171.83) .. (130.2,171.83) .. controls (129.3,171.83) and (128.57,171.1) .. (128.57,170.2) -- cycle ;
            
            \draw (148.6,97.6) node [anchor=north west][inner sep=0.75pt]    {$x_{3}$};
            \draw (188.6,136) node [anchor=north west][inner sep=0.75pt]    {$y_{3}$};
            \draw (107.4,135.4) node [anchor=north west][inner sep=0.75pt]    {$y_{2}$};
            \draw (150.6,175.6) node [anchor=north west][inner sep=0.75pt]    {$x_{4}$};
            
            \end{tikzpicture}
        \end{gathered}\right)\\
        &= \delta_{\zeta,z'}\delta_{\overline \eta, y_2x_4\overline y_3\overline x_3} \left(
        \begin{gathered}  
            \begin{tikzpicture}[x=0.75pt,y=0.75pt,yscale=-1,xscale=1]
            
            \draw   (130.2,120.2) -- (180.2,120.2) -- (180.2,170.2) -- (130.2,170.2) -- cycle ;
            \draw   (150.72,115.36) .. controls (153.73,117.98) and (156.75,119.56) .. (159.76,120.08) .. controls (156.75,120.6) and (153.73,122.18) .. (150.72,124.8) ;
            \draw   (150.72,165.36) .. controls (153.73,167.98) and (156.75,169.56) .. (159.76,170.08) .. controls (156.75,170.6) and (153.73,172.18) .. (150.72,174.8) ;
            \draw   (175.32,149.4) .. controls (177.94,146.39) and (179.52,143.37) .. (180.04,140.36) .. controls (180.56,143.37) and (182.14,146.39) .. (184.76,149.4) ;
            \draw   (125.32,149.8) .. controls (127.94,146.79) and (129.52,143.77) .. (130.04,140.76) .. controls (130.56,143.77) and (132.14,146.79) .. (134.76,149.8) ;
            \draw  [color={rgb, 255:red, 0; green, 0; blue, 0 }  ,draw opacity=1 ][fill={rgb, 255:red, 0; green, 0; blue, 0 }  ,fill opacity=1 ] (128.57,170.2) .. controls (128.57,169.3) and (129.3,168.57) .. (130.2,168.57) .. controls (131.1,168.57) and (131.83,169.3) .. (131.83,170.2) .. controls (131.83,171.1) and (131.1,171.83) .. (130.2,171.83) .. controls (129.3,171.83) and (128.57,171.1) .. (128.57,170.2) -- cycle ;
            
            \draw (148.6,97.6) node [anchor=north west][inner sep=0.75pt]    {$x_{3}$};
            \draw (188.6,136) node [anchor=north west][inner sep=0.75pt]    {$y_{3}$};
            \draw (73,135.4) node [anchor=north west][inner sep=0.75pt]    {$y_{2}\overline x_1 \omega x_1$};
            \draw (150.6,175.6) node [anchor=north west][inner sep=0.75pt]    {$x_{4}$};
            
            \end{tikzpicture}
        \end{gathered}\right)
    \end{aligned}
    \end{equation*}

    By comparing like terms, we see that $\zeta=z,\omega=v, \eta=\overline z v z h$. 

\end{proof}

For vertex operators, the commutation relations with the ribbon operators are
\begin{align}
    \text{Start: }\quad A_{\mathrm{fl}}^{g}F^{(z,v)}(t)&=F^{(gz, gv\bar g)}(t)A_{\mathrm{fl}}^{g}\\
    \text{End: }\quad A_{\mathrm{col}}^{g}F^{(z,v)}(t)&=F^{(z\bar g, v)}(t)A_{\mathrm{col}}^{g}.
\end{align}

\textbf{Vetex operator with clockwise local orientation}

\begin{figure}[h!]
    \centering
    \includegraphics[width=0.7\textwidth]{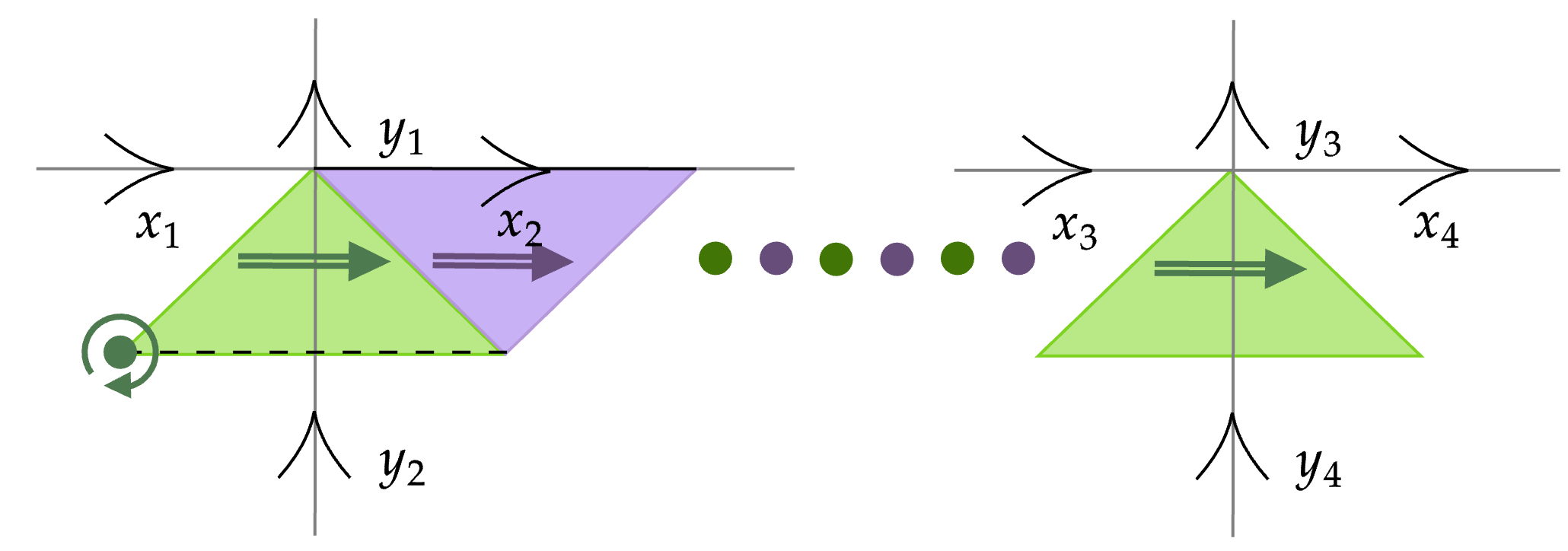}
    \caption{Ribbon operator with clockwise local orientation.}
    \label{fig:clockwise ribbon for vertex}
\end{figure}

\begin{prop}[Vertex operator at the start of a clockwise ribbon]
The commutation relation is
    \begin{equation}
        A_{\mathrm{fl}}^{g}F^{(z,v)}(t)=F^{(gz, gv\bar g)}(t)A_{\mathrm{fl}}^{g}.
    \end{equation}
\end{prop}

\begin{proof}
    WLOG, consider a ribbon (with clockwise local orientation) composed of two triangles, as shown in Fig. \ref{fig:clockwise ribbon for vertex}. 

For the left hand side,
\begin{equation*}
    A^g F^{(z,v)}(t) 
    \left(\begin{gathered}
        \begin{tikzpicture}[x=0.75pt,y=0.75pt,yscale=-1,xscale=1]
        
        \draw   (120.92,165.16) .. controls (123.93,167.78) and (126.95,169.36) .. (129.96,169.88) .. controls (126.95,170.4) and (123.93,171.98) .. (120.92,174.6) ;
        \draw   (100,170) -- (200,170)(150,120) -- (150,220) ;
        \draw   (173.72,165.16) .. controls (176.73,167.78) and (179.75,169.36) .. (182.76,169.88) .. controls (179.75,170.4) and (176.73,171.98) .. (173.72,174.6) ;
        \draw   (145.32,147.2) .. controls (147.94,144.19) and (149.52,141.17) .. (150.04,138.16) .. controls (150.56,141.17) and (152.14,144.19) .. (154.76,147.2) ;
        \draw   (145.32,198.91) .. controls (147.95,195.9) and (149.52,192.88) .. (150.04,189.87) .. controls (150.57,192.88) and (152.14,195.9) .. (154.76,198.91) ;
        
        \draw (118.8,153.4) node [anchor=north west][inner sep=0.75pt]    {$x_{1}$};
        \draw (171.6,153.4) node [anchor=north west][inner sep=0.75pt]    {$x_{2}$};
        \draw (159.47,186.07) node [anchor=north west][inner sep=0.75pt]    {$y_{2}$};
        \draw (158.13,126.27) node [anchor=north west][inner sep=0.75pt]    {$y_{1}$};
        
        \end{tikzpicture}
    \end{gathered}\right)
    = \delta_{z,x_2} A^g
    \left(\begin{gathered}
        \begin{tikzpicture}[x=0.75pt,y=0.75pt,yscale=-1,xscale=1]
        
        \draw   (120.92,165.16) .. controls (123.93,167.78) and (126.95,169.36) .. (129.96,169.88) .. controls (126.95,170.4) and (123.93,171.98) .. (120.92,174.6) ;
        \draw   (100,170) -- (200,170)(150,120) -- (150,220) ;
        \draw   (173.72,165.16) .. controls (176.73,167.78) and (179.75,169.36) .. (182.76,169.88) .. controls (179.75,170.4) and (176.73,171.98) .. (173.72,174.6) ;
        \draw   (145.32,147.2) .. controls (147.94,144.19) and (149.52,141.17) .. (150.04,138.16) .. controls (150.56,141.17) and (152.14,144.19) .. (154.76,147.2) ;
        \draw   (145.32,198.91) .. controls (147.95,195.9) and (149.52,192.88) .. (150.04,189.87) .. controls (150.57,192.88) and (152.14,195.9) .. (154.76,198.91) ;
        
        \draw (118.8,153.4) node [anchor=north west][inner sep=0.75pt]    {$x_{1}$};
        \draw (171.6,153.4) node [anchor=north west][inner sep=0.75pt]    {$x_{2}$};
        \draw (159.47,186.07) node [anchor=north west][inner sep=0.75pt]    {$y_{2}v$};
        \draw (158.13,126.27) node [anchor=north west][inner sep=0.75pt]    {$y_{1}$};
        
        \end{tikzpicture}
    \end{gathered}\right)
    =\delta_{z,x_2}
    \left(\begin{gathered}
        \begin{tikzpicture}[x=0.75pt,y=0.75pt,yscale=-1,xscale=1]
        
        \draw   (120.92,165.16) .. controls (123.93,167.78) and (126.95,169.36) .. (129.96,169.88) .. controls (126.95,170.4) and (123.93,171.98) .. (120.92,174.6) ;
        \draw   (100,170) -- (200,170)(150,120) -- (150,220) ;
        \draw   (173.72,165.16) .. controls (176.73,167.78) and (179.75,169.36) .. (182.76,169.88) .. controls (179.75,170.4) and (176.73,171.98) .. (173.72,174.6) ;
        \draw   (145.32,147.2) .. controls (147.94,144.19) and (149.52,141.17) .. (150.04,138.16) .. controls (150.56,141.17) and (152.14,144.19) .. (154.76,147.2) ;
        \draw   (145.32,198.91) .. controls (147.95,195.9) and (149.52,192.88) .. (150.04,189.87) .. controls (150.57,192.88) and (152.14,195.9) .. (154.76,198.91) ;
        
        \draw (118.8,153.4) node [anchor=north west][inner sep=0.75pt]    {$x_{1}\overline{g}$};
        \draw (171.6,153.4) node [anchor=north west][inner sep=0.75pt]    {$gx_{2}$};
        \draw (159.47,186.07) node [anchor=north west][inner sep=0.75pt]    {$y_{2}v\overline{g}$};
        \draw (158.13,126.27) node [anchor=north west][inner sep=0.75pt]    {$gy_{1}$};
        
        \end{tikzpicture}
    \end{gathered}\right)
\end{equation*}

For the right hand side, we assume a different set of variables ($\zeta$ for charge, $\omega$ for local flux, and $\gamma$ for the flux variable of the plaquette operator) and determine their values.
\begin{equation*}
    F^{(\zeta,\omega)}(t) A^\gamma
    \left(\begin{gathered}
        \begin{tikzpicture}[x=0.75pt,y=0.75pt,yscale=-1,xscale=1]
        
        \draw   (120.92,165.16) .. controls (123.93,167.78) and (126.95,169.36) .. (129.96,169.88) .. controls (126.95,170.4) and (123.93,171.98) .. (120.92,174.6) ;
        \draw   (100,170) -- (200,170)(150,120) -- (150,220) ;
        \draw   (173.72,165.16) .. controls (176.73,167.78) and (179.75,169.36) .. (182.76,169.88) .. controls (179.75,170.4) and (176.73,171.98) .. (173.72,174.6) ;
        \draw   (145.32,147.2) .. controls (147.94,144.19) and (149.52,141.17) .. (150.04,138.16) .. controls (150.56,141.17) and (152.14,144.19) .. (154.76,147.2) ;
        \draw   (145.32,198.91) .. controls (147.95,195.9) and (149.52,192.88) .. (150.04,189.87) .. controls (150.57,192.88) and (152.14,195.9) .. (154.76,198.91) ;
        
        \draw (118.8,153.4) node [anchor=north west][inner sep=0.75pt]    {$x_{1}$};
        \draw (171.6,153.4) node [anchor=north west][inner sep=0.75pt]    {$x_{2}$};
        \draw (159.47,186.07) node [anchor=north west][inner sep=0.75pt]    {$y_{2}$};
        \draw (158.13,126.27) node [anchor=north west][inner sep=0.75pt]    {$y_{1}$};
        
        \end{tikzpicture}
    \end{gathered}\right)
    = \delta_{\zeta,x_2} F^{(\zeta,\omega)}(t)
    \left(\begin{gathered}
        \begin{tikzpicture}[x=0.75pt,y=0.75pt,yscale=-1,xscale=1]
        
        \draw   (120.92,165.16) .. controls (123.93,167.78) and (126.95,169.36) .. (129.96,169.88) .. controls (126.95,170.4) and (123.93,171.98) .. (120.92,174.6) ;
        \draw   (100,170) -- (200,170)(150,120) -- (150,220) ;
        \draw   (173.72,165.16) .. controls (176.73,167.78) and (179.75,169.36) .. (182.76,169.88) .. controls (179.75,170.4) and (176.73,171.98) .. (173.72,174.6) ;
        \draw   (145.32,147.2) .. controls (147.94,144.19) and (149.52,141.17) .. (150.04,138.16) .. controls (150.56,141.17) and (152.14,144.19) .. (154.76,147.2) ;
        \draw   (145.32,198.91) .. controls (147.95,195.9) and (149.52,192.88) .. (150.04,189.87) .. controls (150.57,192.88) and (152.14,195.9) .. (154.76,198.91) ;
        
        \draw (118.8,153.4) node [anchor=north west][inner sep=0.75pt]    {$x_{1}\overline{\gamma}$};
        \draw (171.6,153.4) node [anchor=north west][inner sep=0.75pt]    {$\gamma x_{2}$};
        \draw (159.47,186.07) node [anchor=north west][inner sep=0.75pt]    {$y_{2}\overline{\gamma}$};
        \draw (158.13,126.27) node [anchor=north west][inner sep=0.75pt]    {$\gamma y_{1}$};
        
        \end{tikzpicture}
    \end{gathered}\right)
    =\delta_{\zeta,\gamma x_2}
    \left(\begin{gathered}
        \begin{tikzpicture}[x=0.75pt,y=0.75pt,yscale=-1,xscale=1]
        
        \draw   (120.92,165.16) .. controls (123.93,167.78) and (126.95,169.36) .. (129.96,169.88) .. controls (126.95,170.4) and (123.93,171.98) .. (120.92,174.6) ;
        \draw   (100,170) -- (200,170)(150,120) -- (150,220) ;
        \draw   (173.72,165.16) .. controls (176.73,167.78) and (179.75,169.36) .. (182.76,169.88) .. controls (179.75,170.4) and (176.73,171.98) .. (173.72,174.6) ;
        \draw   (145.32,147.2) .. controls (147.94,144.19) and (149.52,141.17) .. (150.04,138.16) .. controls (150.56,141.17) and (152.14,144.19) .. (154.76,147.2) ;
        \draw   (145.32,198.91) .. controls (147.95,195.9) and (149.52,192.88) .. (150.04,189.87) .. controls (150.57,192.88) and (152.14,195.9) .. (154.76,198.91) ;
        
        \draw (118.8,153.4) node [anchor=north west][inner sep=0.75pt]    {$x_{1}\overline{\gamma}$};
        \draw (171.6,153.4) node [anchor=north west][inner sep=0.75pt]    {$\gamma x_{2}$};
        \draw (159.47,186.07) node [anchor=north west][inner sep=0.75pt]    {$y_{2}\overline{\gamma} \omega$};
        \draw (158.13,126.27) node [anchor=north west][inner sep=0.75pt]    {$\gamma y_{1}$};
        
        \end{tikzpicture}
    \end{gathered}\right)
\end{equation*}

By comparing like terms, we see that $\gamma = g,\zeta=gz,\omega=gv\overline{g}$. 

\end{proof}

\begin{prop}[Vertex operator at the end of a clockwise ribbon]
The commutation relation is
    \begin{equation}
        A_{\mathrm{col}}^{g}F^{(z,v)}(t)=F^{(z\bar g,v)}(t)A_{\mathrm{col}}^{g}.
    \end{equation}
\end{prop}

\begin{proof}

Let $z'$ be the product of direct edges along the ribbon.

For the left hand side,
\begin{equation*}
    A^g F^{(z,v)}(t) 
    \left(\begin{gathered}
        \begin{tikzpicture}[x=0.75pt,y=0.75pt,yscale=-1,xscale=1]
        
        \draw   (120.92,165.16) .. controls (123.93,167.78) and (126.95,169.36) .. (129.96,169.88) .. controls (126.95,170.4) and (123.93,171.98) .. (120.92,174.6) ;
        \draw   (100,170) -- (200,170)(150,120) -- (150,220) ;
        \draw   (173.72,165.16) .. controls (176.73,167.78) and (179.75,169.36) .. (182.76,169.88) .. controls (179.75,170.4) and (176.73,171.98) .. (173.72,174.6) ;
        \draw   (145.32,147.2) .. controls (147.94,144.19) and (149.52,141.17) .. (150.04,138.16) .. controls (150.56,141.17) and (152.14,144.19) .. (154.76,147.2) ;
        \draw   (145.32,198.91) .. controls (147.95,195.9) and (149.52,192.88) .. (150.04,189.87) .. controls (150.57,192.88) and (152.14,195.9) .. (154.76,198.91) ;
        
        \draw (118.8,153.4) node [anchor=north west][inner sep=0.75pt]    {$x_{3}$};
        \draw (171.6,153.4) node [anchor=north west][inner sep=0.75pt]    {$x_{4}$};
        \draw (159.47,186.07) node [anchor=north west][inner sep=0.75pt]    {$y_{4}$};
        \draw (158.13,126.27) node [anchor=north west][inner sep=0.75pt]    {$y_{3}$};
        
        \end{tikzpicture}
    \end{gathered}\right)
    = \delta_{z,z'} A^g
    \left(\begin{gathered}
        \begin{tikzpicture}[x=0.75pt,y=0.75pt,yscale=-1,xscale=1]
        
        \draw   (120.92,165.16) .. controls (123.93,167.78) and (126.95,169.36) .. (129.96,169.88) .. controls (126.95,170.4) and (123.93,171.98) .. (120.92,174.6) ;
        \draw   (100,170) -- (200,170)(150,120) -- (150,220) ;
        \draw   (173.72,165.16) .. controls (176.73,167.78) and (179.75,169.36) .. (182.76,169.88) .. controls (179.75,170.4) and (176.73,171.98) .. (173.72,174.6) ;
        \draw   (145.32,147.2) .. controls (147.94,144.19) and (149.52,141.17) .. (150.04,138.16) .. controls (150.56,141.17) and (152.14,144.19) .. (154.76,147.2) ;
        \draw   (145.32,198.91) .. controls (147.95,195.9) and (149.52,192.88) .. (150.04,189.87) .. controls (150.57,192.88) and (152.14,195.9) .. (154.76,198.91) ;
        
        \draw (118.8,153.4) node [anchor=north west][inner sep=0.75pt]    {$x_{3}$};
        \draw (171.6,153.4) node [anchor=north west][inner sep=0.75pt]    {$x_{4}$};
        \draw (159.47,186.07) node [anchor=north west][inner sep=0.75pt]    {$y_{4}\bar z'v z$};
        \draw (158.13,126.27) node [anchor=north west][inner sep=0.75pt]    {$y_{3}$};
        
        \end{tikzpicture}
    \end{gathered}\right)
    =\delta_{z,z'}
    \left(\begin{gathered}
        \begin{tikzpicture}[x=0.75pt,y=0.75pt,yscale=-1,xscale=1]
        
        \draw   (120.92,165.16) .. controls (123.93,167.78) and (126.95,169.36) .. (129.96,169.88) .. controls (126.95,170.4) and (123.93,171.98) .. (120.92,174.6) ;
        \draw   (100,170) -- (200,170)(150,120) -- (150,220) ;
        \draw   (173.72,165.16) .. controls (176.73,167.78) and (179.75,169.36) .. (182.76,169.88) .. controls (179.75,170.4) and (176.73,171.98) .. (173.72,174.6) ;
        \draw   (145.32,147.2) .. controls (147.94,144.19) and (149.52,141.17) .. (150.04,138.16) .. controls (150.56,141.17) and (152.14,144.19) .. (154.76,147.2) ;
        \draw   (145.32,198.91) .. controls (147.95,195.9) and (149.52,192.88) .. (150.04,189.87) .. controls (150.57,192.88) and (152.14,195.9) .. (154.76,198.91) ;
        
        \draw (118.8,153.4) node [anchor=north west][inner sep=0.75pt]    {$x_{3}\overline{g}$};
        \draw (171.6,153.4) node [anchor=north west][inner sep=0.75pt]    {$gx_{4}$};
        \draw (159.47,186.07) node [anchor=north west][inner sep=0.75pt]    {$y_{4} \bar z' v z' \overline{g}$};
        \draw (158.13,126.27) node [anchor=north west][inner sep=0.75pt]    {$gy_{3}$};
        
        \end{tikzpicture}
    \end{gathered}\right)
\end{equation*}

For the right hand side, we assume a different set of variables ($\zeta$ for charge, $\omega$ for local flux, and $\gamma$ for the flux variable of the plaquette operator) and determine their values.
\begin{equation*}
    F^{(\zeta,\omega)}(t) A^\gamma
    \left(\begin{gathered}
        \begin{tikzpicture}[x=0.75pt,y=0.75pt,yscale=-1,xscale=1]
        
        \draw   (120.92,165.16) .. controls (123.93,167.78) and (126.95,169.36) .. (129.96,169.88) .. controls (126.95,170.4) and (123.93,171.98) .. (120.92,174.6) ;
        \draw   (100,170) -- (200,170)(150,120) -- (150,220) ;
        \draw   (173.72,165.16) .. controls (176.73,167.78) and (179.75,169.36) .. (182.76,169.88) .. controls (179.75,170.4) and (176.73,171.98) .. (173.72,174.6) ;
        \draw   (145.32,147.2) .. controls (147.94,144.19) and (149.52,141.17) .. (150.04,138.16) .. controls (150.56,141.17) and (152.14,144.19) .. (154.76,147.2) ;
        \draw   (145.32,198.91) .. controls (147.95,195.9) and (149.52,192.88) .. (150.04,189.87) .. controls (150.57,192.88) and (152.14,195.9) .. (154.76,198.91) ;
        
        \draw (118.8,153.4) node [anchor=north west][inner sep=0.75pt]    {$x_{3}$};
        \draw (171.6,153.4) node [anchor=north west][inner sep=0.75pt]    {$x_{4}$};
        \draw (159.47,186.07) node [anchor=north west][inner sep=0.75pt]    {$y_{4}$};
        \draw (158.13,126.27) node [anchor=north west][inner sep=0.75pt]    {$y_{3}$};
        
        \end{tikzpicture}
    \end{gathered}\right)
    = \delta_{\zeta,z'} F^{(\zeta,\omega)}(t)
    \left(\begin{gathered}
        \begin{tikzpicture}[x=0.75pt,y=0.75pt,yscale=-1,xscale=1]
        
        \draw   (120.92,165.16) .. controls (123.93,167.78) and (126.95,169.36) .. (129.96,169.88) .. controls (126.95,170.4) and (123.93,171.98) .. (120.92,174.6) ;
        \draw   (100,170) -- (200,170)(150,120) -- (150,220) ;
        \draw   (173.72,165.16) .. controls (176.73,167.78) and (179.75,169.36) .. (182.76,169.88) .. controls (179.75,170.4) and (176.73,171.98) .. (173.72,174.6) ;
        \draw   (145.32,147.2) .. controls (147.94,144.19) and (149.52,141.17) .. (150.04,138.16) .. controls (150.56,141.17) and (152.14,144.19) .. (154.76,147.2) ;
        \draw   (145.32,198.91) .. controls (147.95,195.9) and (149.52,192.88) .. (150.04,189.87) .. controls (150.57,192.88) and (152.14,195.9) .. (154.76,198.91) ;
        
        \draw (118.8,153.4) node [anchor=north west][inner sep=0.75pt]    {$x_{3}\overline{\gamma}$};
        \draw (171.6,153.4) node [anchor=north west][inner sep=0.75pt]    {$\gamma x_{4}$};
        \draw (159.47,186.07) node [anchor=north west][inner sep=0.75pt]    {$y_{4}\overline{\gamma}$};
        \draw (158.13,126.27) node [anchor=north west][inner sep=0.75pt]    {$\gamma y_{3}$};
        
        \end{tikzpicture}
    \end{gathered}\right)
    =\delta_{\zeta,\gamma z'}
    \left(\begin{gathered}
        \begin{tikzpicture}[x=0.75pt,y=0.75pt,yscale=-1,xscale=1]
        
        \draw   (120.92,165.16) .. controls (123.93,167.78) and (126.95,169.36) .. (129.96,169.88) .. controls (126.95,170.4) and (123.93,171.98) .. (120.92,174.6) ;
        \draw   (100,170) -- (200,170)(150,120) -- (150,220) ;
        \draw   (173.72,165.16) .. controls (176.73,167.78) and (179.75,169.36) .. (182.76,169.88) .. controls (179.75,170.4) and (176.73,171.98) .. (173.72,174.6) ;
        \draw   (145.32,147.2) .. controls (147.94,144.19) and (149.52,141.17) .. (150.04,138.16) .. controls (150.56,141.17) and (152.14,144.19) .. (154.76,147.2) ;
        \draw   (145.32,198.91) .. controls (147.95,195.9) and (149.52,192.88) .. (150.04,189.87) .. controls (150.57,192.88) and (152.14,195.9) .. (154.76,198.91) ;
        
        \draw (118.8,153.4) node [anchor=north west][inner sep=0.75pt]    {$x_{3}\overline{\gamma}$};
        \draw (171.6,153.4) node [anchor=north west][inner sep=0.75pt]    {$\gamma x_{4}$};
        \draw (159.47,186.07) node [anchor=north west][inner sep=0.75pt]    {$y_{4} \bar z' \omega z'  \bar\gamma$};
        \draw (158.13,126.27) node [anchor=north west][inner sep=0.75pt]    {$\gamma y_{3}$};
        
        \end{tikzpicture}
    \end{gathered}\right)
\end{equation*}

By comparing like terms, we see that $\gamma = g,\zeta=z\var g,\omega=v$. 

\end{proof}

\textbf{Vertex operator with counterclockwise local orientation}

\begin{figure}[h!]
    \centering
    \includegraphics[width=0.7\textwidth]{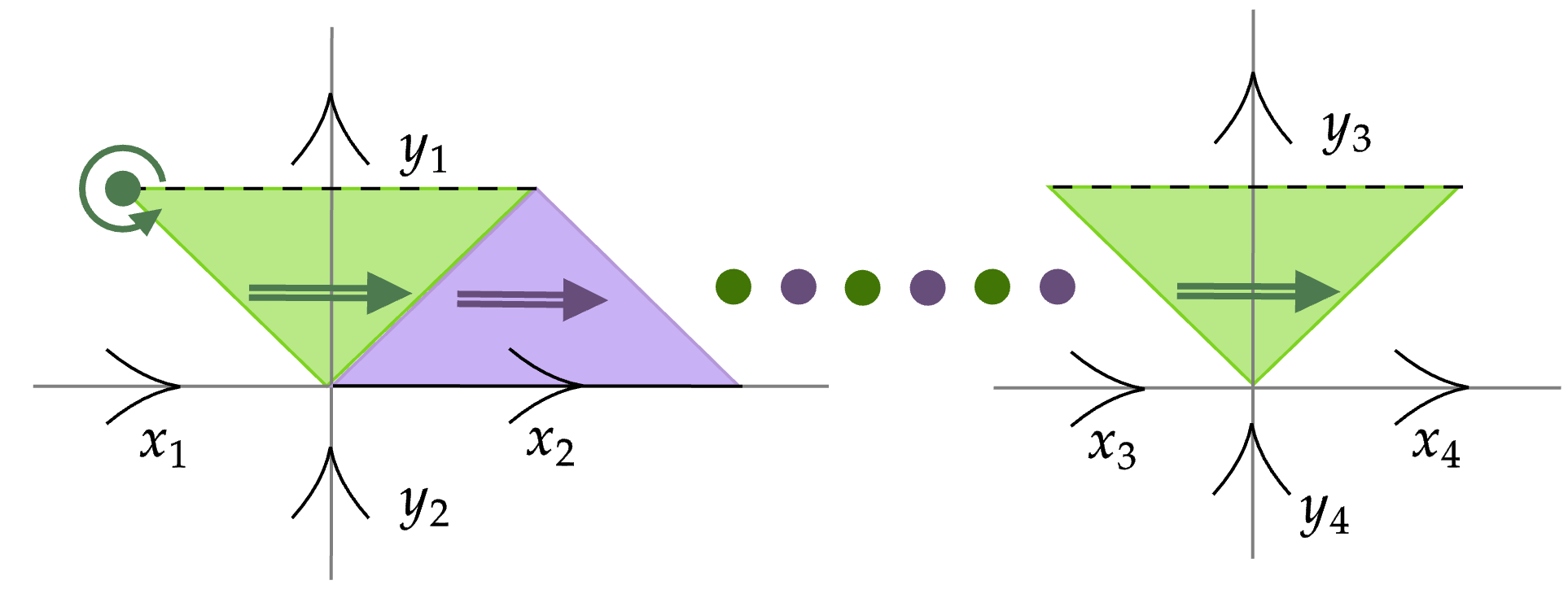}
    \caption{Ribbon operator with counterclockwise local orientation.}
    \label{fig:counterclockwise ribbon for vertex}
\end{figure}

\begin{prop}[Vertex operator at the start of a counterclockwise ribbon]
The commutation relation is
    \begin{equation}
        A_{\mathrm{fl}}^{g}F^{(z,v)}(t)=F^{(gz, gv\bar g)}(t)A_{\mathrm{fl}}^{g}.
    \end{equation}
\end{prop}

\begin{proof}
    WLOG, consider a following ribbon (with counterclockwise local orientation) composed of two triangles, as shown in Fig. \ref{fig:counterclockwise ribbon for vertex}.

For the left hand side,
\begin{equation*}
    A^g F^{(z,v)}(t) 
    \left(\begin{gathered}
        \begin{tikzpicture}[x=0.75pt,y=0.75pt,yscale=-1,xscale=1]
        
        \draw   (120.92,165.16) .. controls (123.93,167.78) and (126.95,169.36) .. (129.96,169.88) .. controls (126.95,170.4) and (123.93,171.98) .. (120.92,174.6) ;
        \draw   (100,170) -- (200,170)(150,120) -- (150,220) ;
        \draw   (173.72,165.16) .. controls (176.73,167.78) and (179.75,169.36) .. (182.76,169.88) .. controls (179.75,170.4) and (176.73,171.98) .. (173.72,174.6) ;
        \draw   (145.32,147.2) .. controls (147.94,144.19) and (149.52,141.17) .. (150.04,138.16) .. controls (150.56,141.17) and (152.14,144.19) .. (154.76,147.2) ;
        \draw   (145.32,198.91) .. controls (147.95,195.9) and (149.52,192.88) .. (150.04,189.87) .. controls (150.57,192.88) and (152.14,195.9) .. (154.76,198.91) ;
        
        \draw (118.8,153.4) node [anchor=north west][inner sep=0.75pt]    {$x_{1}$};
        \draw (171.6,153.4) node [anchor=north west][inner sep=0.75pt]    {$x_{2}$};
        \draw (159.47,186.07) node [anchor=north west][inner sep=0.75pt]    {$y_{2}$};
        \draw (158.13,126.27) node [anchor=north west][inner sep=0.75pt]    {$y_{1}$};
        
        \end{tikzpicture}
    \end{gathered}\right)
    = \delta_{z,x_2} A^g
    \left(\begin{gathered}
        \begin{tikzpicture}[x=0.75pt,y=0.75pt,yscale=-1,xscale=1]
        
        \draw   (120.92,165.16) .. controls (123.93,167.78) and (126.95,169.36) .. (129.96,169.88) .. controls (126.95,170.4) and (123.93,171.98) .. (120.92,174.6) ;
        \draw   (100,170) -- (200,170)(150,120) -- (150,220) ;
        \draw   (173.72,165.16) .. controls (176.73,167.78) and (179.75,169.36) .. (182.76,169.88) .. controls (179.75,170.4) and (176.73,171.98) .. (173.72,174.6) ;
        \draw   (145.32,147.2) .. controls (147.94,144.19) and (149.52,141.17) .. (150.04,138.16) .. controls (150.56,141.17) and (152.14,144.19) .. (154.76,147.2) ;
        \draw   (145.32,198.91) .. controls (147.95,195.9) and (149.52,192.88) .. (150.04,189.87) .. controls (150.57,192.88) and (152.14,195.9) .. (154.76,198.91) ;
        
        \draw (118.8,153.4) node [anchor=north west][inner sep=0.75pt]    {$x_{1}$};
        \draw (171.6,153.4) node [anchor=north west][inner sep=0.75pt]    {$x_{2}$};
        \draw (159.47,186.07) node [anchor=north west][inner sep=0.75pt]    {$y_{2}$};
        \draw (158.13,126.27) node [anchor=north west][inner sep=0.75pt]    {$vy_{1}$};
        
        \end{tikzpicture}
    \end{gathered}\right)
    =\delta_{z,x_2}
    \left(\begin{gathered}
        \begin{tikzpicture}[x=0.75pt,y=0.75pt,yscale=-1,xscale=1]
        
        \draw   (120.92,165.16) .. controls (123.93,167.78) and (126.95,169.36) .. (129.96,169.88) .. controls (126.95,170.4) and (123.93,171.98) .. (120.92,174.6) ;
        \draw   (100,170) -- (200,170)(150,120) -- (150,220) ;
        \draw   (173.72,165.16) .. controls (176.73,167.78) and (179.75,169.36) .. (182.76,169.88) .. controls (179.75,170.4) and (176.73,171.98) .. (173.72,174.6) ;
        \draw   (145.32,147.2) .. controls (147.94,144.19) and (149.52,141.17) .. (150.04,138.16) .. controls (150.56,141.17) and (152.14,144.19) .. (154.76,147.2) ;
        \draw   (145.32,198.91) .. controls (147.95,195.9) and (149.52,192.88) .. (150.04,189.87) .. controls (150.57,192.88) and (152.14,195.9) .. (154.76,198.91) ;
        
        \draw (118.8,153.4) node [anchor=north west][inner sep=0.75pt]    {$x_{1}\overline{g}$};
        \draw (171.6,153.4) node [anchor=north west][inner sep=0.75pt]    {$gx_{2}$};
        \draw (159.47,186.07) node [anchor=north west][inner sep=0.75pt]    {$y_{2}\overline{g}$};
        \draw (158.13,126.27) node [anchor=north west][inner sep=0.75pt]    {$gvy_{1}$};
        
        \end{tikzpicture}
    \end{gathered}\right)
\end{equation*}

For the right hand side, we assume a different set of variables ($\zeta$ for charge, $\omega$ for local flux, and $\gamma$ for the flux variable of the plaquette operator) and determine their values.
\begin{equation*}
    F^{(\zeta,\omega)}(t) A^\gamma
    \left(\begin{gathered}
        \begin{tikzpicture}[x=0.75pt,y=0.75pt,yscale=-1,xscale=1]
        
        \draw   (120.92,165.16) .. controls (123.93,167.78) and (126.95,169.36) .. (129.96,169.88) .. controls (126.95,170.4) and (123.93,171.98) .. (120.92,174.6) ;
        \draw   (100,170) -- (200,170)(150,120) -- (150,220) ;
        \draw   (173.72,165.16) .. controls (176.73,167.78) and (179.75,169.36) .. (182.76,169.88) .. controls (179.75,170.4) and (176.73,171.98) .. (173.72,174.6) ;
        \draw   (145.32,147.2) .. controls (147.94,144.19) and (149.52,141.17) .. (150.04,138.16) .. controls (150.56,141.17) and (152.14,144.19) .. (154.76,147.2) ;
        \draw   (145.32,198.91) .. controls (147.95,195.9) and (149.52,192.88) .. (150.04,189.87) .. controls (150.57,192.88) and (152.14,195.9) .. (154.76,198.91) ;
        
        \draw (118.8,153.4) node [anchor=north west][inner sep=0.75pt]    {$x_{1}$};
        \draw (171.6,153.4) node [anchor=north west][inner sep=0.75pt]    {$x_{2}$};
        \draw (159.47,186.07) node [anchor=north west][inner sep=0.75pt]    {$y_{2}$};
        \draw (158.13,126.27) node [anchor=north west][inner sep=0.75pt]    {$y_{1}$};
        
        \end{tikzpicture}
    \end{gathered}\right)
    = \delta_{\zeta,x_2} F^{(\zeta,\omega)}(t)
    \left(\begin{gathered}
        \begin{tikzpicture}[x=0.75pt,y=0.75pt,yscale=-1,xscale=1]
        
        \draw   (120.92,165.16) .. controls (123.93,167.78) and (126.95,169.36) .. (129.96,169.88) .. controls (126.95,170.4) and (123.93,171.98) .. (120.92,174.6) ;
        \draw   (100,170) -- (200,170)(150,120) -- (150,220) ;
        \draw   (173.72,165.16) .. controls (176.73,167.78) and (179.75,169.36) .. (182.76,169.88) .. controls (179.75,170.4) and (176.73,171.98) .. (173.72,174.6) ;
        \draw   (145.32,147.2) .. controls (147.94,144.19) and (149.52,141.17) .. (150.04,138.16) .. controls (150.56,141.17) and (152.14,144.19) .. (154.76,147.2) ;
        \draw   (145.32,198.91) .. controls (147.95,195.9) and (149.52,192.88) .. (150.04,189.87) .. controls (150.57,192.88) and (152.14,195.9) .. (154.76,198.91) ;
        
        \draw (118.8,153.4) node [anchor=north west][inner sep=0.75pt]    {$x_{1}\overline{\gamma}$};
        \draw (171.6,153.4) node [anchor=north west][inner sep=0.75pt]    {$\gamma x_{2}$};
        \draw (159.47,186.07) node [anchor=north west][inner sep=0.75pt]    {$y_{2}\overline{\gamma}$};
        \draw (158.13,126.27) node [anchor=north west][inner sep=0.75pt]    {$\gamma y_{1}$};
        
        \end{tikzpicture}
    \end{gathered}\right)
    =\delta_{\zeta,\gamma x_2}
    \left(\begin{gathered}
        \begin{tikzpicture}[x=0.75pt,y=0.75pt,yscale=-1,xscale=1]
        
        \draw   (120.92,165.16) .. controls (123.93,167.78) and (126.95,169.36) .. (129.96,169.88) .. controls (126.95,170.4) and (123.93,171.98) .. (120.92,174.6) ;
        \draw   (100,170) -- (200,170)(150,120) -- (150,220) ;
        \draw   (173.72,165.16) .. controls (176.73,167.78) and (179.75,169.36) .. (182.76,169.88) .. controls (179.75,170.4) and (176.73,171.98) .. (173.72,174.6) ;
        \draw   (145.32,147.2) .. controls (147.94,144.19) and (149.52,141.17) .. (150.04,138.16) .. controls (150.56,141.17) and (152.14,144.19) .. (154.76,147.2) ;
        \draw   (145.32,198.91) .. controls (147.95,195.9) and (149.52,192.88) .. (150.04,189.87) .. controls (150.57,192.88) and (152.14,195.9) .. (154.76,198.91) ;
        
        \draw (118.8,153.4) node [anchor=north west][inner sep=0.75pt]    {$x_{1}\overline{\gamma}$};
        \draw (171.6,153.4) node [anchor=north west][inner sep=0.75pt]    {$\gamma x_{2}$};
        \draw (159.47,186.07) node [anchor=north west][inner sep=0.75pt]    {$y_{2}\overline{\gamma}$};
        \draw (158.13,126.27) node [anchor=north west][inner sep=0.75pt]    {$\omega\gamma y_{1}$};
        
        \end{tikzpicture}
    \end{gathered}\right)
\end{equation*}

By comparing like terms, we see that $\gamma = g,\zeta=gz,\omega=gv\overline{g}$. 

\end{proof}

\begin{prop}[Vertex operator at the end of a counterclockwise ribbon]
The commutation relation is
    \begin{equation}
        A_{\mathrm{col}}^{g}F^{(z,v)}(t)=F^{(z\bar g, v)}(t)A_{\mathrm{col}}^{g}.
    \end{equation}
\end{prop}

\begin{proof}
    WLOG, consider a following ribbon (with counterclockwise local orientation) composed of two triangles, as shown in Fig. \ref{fig:counterclockwise ribbon for vertex}.

For the left hand side,
\begin{equation*}
    A^g F^{(z,v)}(t) 
    \left(\begin{gathered}
        \begin{tikzpicture}[x=0.75pt,y=0.75pt,yscale=-1,xscale=1]
        
        \draw   (120.92,165.16) .. controls (123.93,167.78) and (126.95,169.36) .. (129.96,169.88) .. controls (126.95,170.4) and (123.93,171.98) .. (120.92,174.6) ;
        \draw   (100,170) -- (200,170)(150,120) -- (150,220) ;
        \draw   (173.72,165.16) .. controls (176.73,167.78) and (179.75,169.36) .. (182.76,169.88) .. controls (179.75,170.4) and (176.73,171.98) .. (173.72,174.6) ;
        \draw   (145.32,147.2) .. controls (147.94,144.19) and (149.52,141.17) .. (150.04,138.16) .. controls (150.56,141.17) and (152.14,144.19) .. (154.76,147.2) ;
        \draw   (145.32,198.91) .. controls (147.95,195.9) and (149.52,192.88) .. (150.04,189.87) .. controls (150.57,192.88) and (152.14,195.9) .. (154.76,198.91) ;
        
        \draw (118.8,153.4) node [anchor=north west][inner sep=0.75pt]    {$x_{3}$};
        \draw (171.6,153.4) node [anchor=north west][inner sep=0.75pt]    {$x_{4}$};
        \draw (159.47,186.07) node [anchor=north west][inner sep=0.75pt]    {$y_{4}$};
        \draw (158.13,126.27) node [anchor=north west][inner sep=0.75pt]    {$y_{3}$};
        
        \end{tikzpicture}
    \end{gathered}\right)
    = \delta_{z,z'} A^g
    \left(\begin{gathered}
        \begin{tikzpicture}[x=0.75pt,y=0.75pt,yscale=-1,xscale=1]
        
        \draw   (120.92,165.16) .. controls (123.93,167.78) and (126.95,169.36) .. (129.96,169.88) .. controls (126.95,170.4) and (123.93,171.98) .. (120.92,174.6) ;
        \draw   (100,170) -- (200,170)(150,120) -- (150,220) ;
        \draw   (173.72,165.16) .. controls (176.73,167.78) and (179.75,169.36) .. (182.76,169.88) .. controls (179.75,170.4) and (176.73,171.98) .. (173.72,174.6) ;
        \draw   (145.32,147.2) .. controls (147.94,144.19) and (149.52,141.17) .. (150.04,138.16) .. controls (150.56,141.17) and (152.14,144.19) .. (154.76,147.2) ;
        \draw   (145.32,198.91) .. controls (147.95,195.9) and (149.52,192.88) .. (150.04,189.87) .. controls (150.57,192.88) and (152.14,195.9) .. (154.76,198.91) ;
        
        \draw (118.8,153.4) node [anchor=north west][inner sep=0.75pt]    {$x_{3}$};
        \draw (171.6,153.4) node [anchor=north west][inner sep=0.75pt]    {$x_{4}$};
        \draw (159.47,186.07) node [anchor=north west][inner sep=0.75pt]    {$y_{4}$};
        \draw (158.13,126.27) node [anchor=north west][inner sep=0.75pt]    {$\bar z' v z' y_{3}$};
        
        \end{tikzpicture}
    \end{gathered}\right)
    =\delta_{z,z'}
    \left(\begin{gathered}
        \begin{tikzpicture}[x=0.75pt,y=0.75pt,yscale=-1,xscale=1]
        
        \draw   (120.92,165.16) .. controls (123.93,167.78) and (126.95,169.36) .. (129.96,169.88) .. controls (126.95,170.4) and (123.93,171.98) .. (120.92,174.6) ;
        \draw   (100,170) -- (200,170)(150,120) -- (150,220) ;
        \draw   (173.72,165.16) .. controls (176.73,167.78) and (179.75,169.36) .. (182.76,169.88) .. controls (179.75,170.4) and (176.73,171.98) .. (173.72,174.6) ;
        \draw   (145.32,147.2) .. controls (147.94,144.19) and (149.52,141.17) .. (150.04,138.16) .. controls (150.56,141.17) and (152.14,144.19) .. (154.76,147.2) ;
        \draw   (145.32,198.91) .. controls (147.95,195.9) and (149.52,192.88) .. (150.04,189.87) .. controls (150.57,192.88) and (152.14,195.9) .. (154.76,198.91) ;
        
        \draw (118.8,153.4) node [anchor=north west][inner sep=0.75pt]    {$x_{3}\overline{g}$};
        \draw (171.6,153.4) node [anchor=north west][inner sep=0.75pt]    {$gx_{4}$};
        \draw (159.47,186.07) node [anchor=north west][inner sep=0.75pt]    {$y_{4}\bar g$};
        \draw (158.13,126.27) node [anchor=north west][inner sep=0.75pt]    {$g\bar z' v z' y_{3}$};
        
        \end{tikzpicture}
    \end{gathered}\right)
\end{equation*}

For the right hand side, we assume a different set of variables ($\zeta$ for charge, $\omega$ for local flux, and $\gamma$ for the flux variable of the plaquette operator) and determine their values.
\begin{equation*}
    F^{(\zeta,\omega)}(t) A^\gamma
    \left(\begin{gathered}
        \begin{tikzpicture}[x=0.75pt,y=0.75pt,yscale=-1,xscale=1]
        
        \draw   (120.92,165.16) .. controls (123.93,167.78) and (126.95,169.36) .. (129.96,169.88) .. controls (126.95,170.4) and (123.93,171.98) .. (120.92,174.6) ;
        \draw   (100,170) -- (200,170)(150,120) -- (150,220) ;
        \draw   (173.72,165.16) .. controls (176.73,167.78) and (179.75,169.36) .. (182.76,169.88) .. controls (179.75,170.4) and (176.73,171.98) .. (173.72,174.6) ;
        \draw   (145.32,147.2) .. controls (147.94,144.19) and (149.52,141.17) .. (150.04,138.16) .. controls (150.56,141.17) and (152.14,144.19) .. (154.76,147.2) ;
        \draw   (145.32,198.91) .. controls (147.95,195.9) and (149.52,192.88) .. (150.04,189.87) .. controls (150.57,192.88) and (152.14,195.9) .. (154.76,198.91) ;
        
        \draw (118.8,153.4) node [anchor=north west][inner sep=0.75pt]    {$x_{3}$};
        \draw (171.6,153.4) node [anchor=north west][inner sep=0.75pt]    {$x_{4}$};
        \draw (159.47,186.07) node [anchor=north west][inner sep=0.75pt]    {$y_{4}$};
        \draw (158.13,126.27) node [anchor=north west][inner sep=0.75pt]    {$y_{3}$};
        
        \end{tikzpicture}
    \end{gathered}\right)
    = \delta_{\zeta,z'} F^{(\zeta,\omega)}(t)
    \left(\begin{gathered}
        \begin{tikzpicture}[x=0.75pt,y=0.75pt,yscale=-1,xscale=1]
        
        \draw   (120.92,165.16) .. controls (123.93,167.78) and (126.95,169.36) .. (129.96,169.88) .. controls (126.95,170.4) and (123.93,171.98) .. (120.92,174.6) ;
        \draw   (100,170) -- (200,170)(150,120) -- (150,220) ;
        \draw   (173.72,165.16) .. controls (176.73,167.78) and (179.75,169.36) .. (182.76,169.88) .. controls (179.75,170.4) and (176.73,171.98) .. (173.72,174.6) ;
        \draw   (145.32,147.2) .. controls (147.94,144.19) and (149.52,141.17) .. (150.04,138.16) .. controls (150.56,141.17) and (152.14,144.19) .. (154.76,147.2) ;
        \draw   (145.32,198.91) .. controls (147.95,195.9) and (149.52,192.88) .. (150.04,189.87) .. controls (150.57,192.88) and (152.14,195.9) .. (154.76,198.91) ;
        
        \draw (118.8,153.4) node [anchor=north west][inner sep=0.75pt]    {$x_{3}\overline{\gamma}$};
        \draw (171.6,153.4) node [anchor=north west][inner sep=0.75pt]    {$\gamma x_{4}$};
        \draw (159.47,186.07) node [anchor=north west][inner sep=0.75pt]    {$y_{4}\overline{\gamma}$};
        \draw (158.13,126.27) node [anchor=north west][inner sep=0.75pt]    {$\gamma y_{3}$};
        
        \end{tikzpicture}
    \end{gathered}\right)
    =\delta_{\zeta,\gamma z'}
    \left(\begin{gathered}
        \begin{tikzpicture}[x=0.75pt,y=0.75pt,yscale=-1,xscale=1]
        
        \draw   (120.92,165.16) .. controls (123.93,167.78) and (126.95,169.36) .. (129.96,169.88) .. controls (126.95,170.4) and (123.93,171.98) .. (120.92,174.6) ;
        \draw   (100,170) -- (200,170)(150,120) -- (150,220) ;
        \draw   (173.72,165.16) .. controls (176.73,167.78) and (179.75,169.36) .. (182.76,169.88) .. controls (179.75,170.4) and (176.73,171.98) .. (173.72,174.6) ;
        \draw   (145.32,147.2) .. controls (147.94,144.19) and (149.52,141.17) .. (150.04,138.16) .. controls (150.56,141.17) and (152.14,144.19) .. (154.76,147.2) ;
        \draw   (145.32,198.91) .. controls (147.95,195.9) and (149.52,192.88) .. (150.04,189.87) .. controls (150.57,192.88) and (152.14,195.9) .. (154.76,198.91) ;
        
        \draw (118.8,153.4) node [anchor=north west][inner sep=0.75pt]    {$x_{3}\overline{\gamma}$};
        \draw (171.6,153.4) node [anchor=north west][inner sep=0.75pt]    {$\gamma x_{4}$};
        \draw (159.47,186.07) node [anchor=north west][inner sep=0.75pt]    {$y_{4}\overline{\gamma}$};
        \draw (158.13,126.27) node [anchor=north west][inner sep=0.75pt]    {$g \bar z' \omega z' y_{3}$};
        
        \end{tikzpicture}
    \end{gathered}\right)
\end{equation*}

By comparing like terms, we see that $\gamma = g,\zeta=z\bar g,\omega=v$. 

\end{proof}

\section{Gate implementation details}

\subsection{Error Correction for the Sign-flip gate} \label{ap:repeat}

\begin{center}
    \begin{tikzpicture}
    \node[scale=1.25]{ 
        \begin{quantikz}
            \lstick{$\ket{\psi}$} & \ctrl{1} & \qw & \rstick{$\sigma_{(i+2)}^z \ket{\psi}$} \qw \\
            \lstick{$\ket{\xi}$} & \gate{U_+} & \meter{$i=0,1,2$} \arrow[r] & \rstick{$\ket{i}$}
        \end{quantikz}
    };
    \end{tikzpicture}
\end{center}

The exact gate implemented by our sign-flip circuit (pictured above) is dependent on the outcome of the measurement. We don't have to resort to post-selection, however--- we can take advantage of the fact the sign-flip gates are closed under multiplication. We will construct a ``repeat until success procedure'', where (as the name suggests) we repeat the gate until we have a sequence of measurement outcomes that correct back to the gate we wanted.

What are the conditions for success? First, let's consider what happens when we apply two sign-flip gates in sequence. Suppose we implement $\sigma^z_{(1)}$ first and $\sigma^z_{(2)}$ second. We will have flipped both $c_1$ and $c_2$, which means the output state is equivalent up to a global phase to one with just $c_0$ flipped:
\begin{equation}
    \sigma^z_{(2)}\sigma^z_{(1)} = \sigma^z_{(1)}\sigma^z_{(2)} = - \sigma^z_{(0)}
\end{equation}
Additionally, we know that repeating the same sign-flip gate twice in a row will be the identity. More generally, we have:
\begin{equation}
    \sigma^z_{(i)}\sigma^z_{(j)} = \delta_{ij} - \abs{\epsilon_{ijk}}\sigma^z_{(k)}
\end{equation}
Another useful fact is that all the different sign-flip gates commute with one another, so if two of the same measurement result show up anywhere in our measurement record, they ``cancel'' to give us the identity. Now we can see how the correction process will go; we want to repeat our procedure for implementing the sign-flip gates until we have an odd number of the result we want, and an even number of the other two. This ensures that we have cancelled out any wrong sign flips and end up with the correct output. So we can classify a given measurement record as corrected or uncorrected based on the parity of the number of zeroes, ones, and twos. Note that any successful ``corrected'' string will be of odd length, since it will have even amounts of two kinds of letters and an odd amount of the remaining type (where letter here means 0, 1, or 2). In Fig. \ref{fig:sign-tree}, we show the tree of gates implemented up to three repetitions; highlighted are the measurement records that give $\sigma_{(0)}$ overall.

\begin{figure}[h!]
\centering
\includegraphics[width=0.85\textwidth]{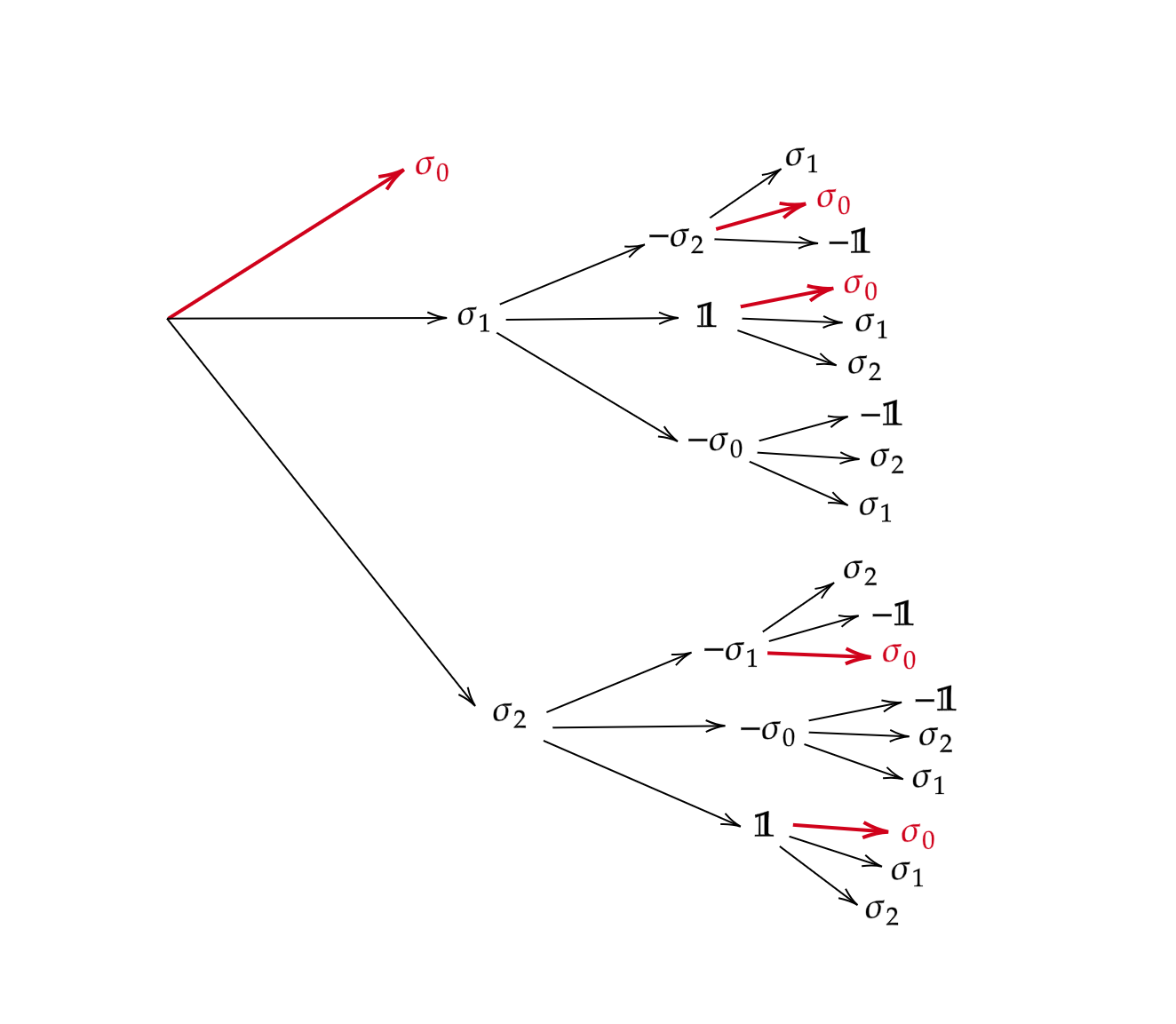}
\caption{Tree of possible outcomes for $N\leq3$ repetitions of the sign-flip gate. Going from top to bottom, each trio of branches indicates a measurement result of $0, 1, \text{or } 2$ for that round. Highlighted in red are branches that lead to an overall $\sigma_{(0)}$ gate; we can see that $\frac{2}{9}$ of the $N=3$ branches lead to $\sigma_{(0)}$.}
\label{fig:sign-tree}
\end{figure}

We can calculate exactly how many branches of the tree in Fig. \ref{fig:sign-tree} terminate in the gate we want; this will allow us to calculate how many times we need to repeat to have a good chance of success. Suppose we have a certain uncorrected sequence of measurement outcomes after an odd number $L$ repetitions, and we are trying to implement the $\sigma^z_{(0)}$ gate. After the next two measurement outcomes, what are the chances we have turned this uncorrected string into a corrected one? Consider the two kinds of uncorrected strings:

\begin{enumerate}
    \item The string has an odd number of all three outcomes: if the sequence of two measurements appended is $\{1,2\}$ or $\{2,1\}$ we have corrected the string. All others will leave it uncorrected.
    \item The string has an even number of zeros and ones (twos), and an odd number of twos (ones): if the appended sequence is $\{0,1(2)\}$ or $\{1(2),0\}$, then we have corrected the string. All others leave it uncorrected.
\end{enumerate}

There are $9$ possible sequences of two measurements, so in either case we have a $2/9$ probability of turning the uncorrected string into a corrected one after two repetitions, or alternately, a $7/9$ probability of not correcting the string. This will be the case no matter the length $L$ of the initial string; the only exception is the first time we implement the gate we have a $1/3$ chance of success. So the chance of success $p$ after $N$ steps is
\begin{equation}
    p = 1 - \prod_i^N P(i)
\end{equation}
where $P(i)$ is the probability of not correcting the string after step $i$, and takes the following form:
\begin{equation}
    P(i) =
    \begin{cases}
        \frac{2}{3} & \text{if } i = 1 \\
        1 & \text{if } i \text{ even} \\
        \frac{7}{9} & \text{if } i \text{ odd and } i > 1
    \end{cases}
\end{equation}
Plugging in and compressing the product, we find our chance of success scales exponentially with the number of steps:
\begin{equation}
    p = 1 - \frac{2}{3}\left(\frac{7}{9}\right)^{\frac{N-1}{2}}
\end{equation}
where the exponent being $(N-1)/2$ makes sure we are only counting the odd steps after the first one. This formula tells us that for a $99\%$ chance of success, we need to repeat at least 35 times, and for a $99.9\%$ chance of success, we should repeat at least 53 times. A more modest number of gates, say 10 repetitions, gives us a $78.5\%$ success rate.

\subsection{Showing CCZ is a non-Clifford gate} \label{ap:lambda3}

Why is $CCZ$ not Clifford? Clifford gates map Pauli strings to Pauli strings, i.e. conjugating a Pauli operator by a Clifford gate should always return a Pauli. We can see with a simple example this is not true of $CCZ$. Consider the Pauli string $X \otimes \mathbb{I}\otimes \mathbb{I}$ conjugated by $CCZ$. The total gate $CCZ(X \otimes \mathbb{I}\otimes \mathbb{I})CCZ$ gives a $(-1)$ prefactor in two cases:
\begin{enumerate}
    \item $x=y=z=1$
        \begin{equation}
        \begin{aligned}
            CCZ \cdot X_1 \cdot CCZ \ket{111} &= -CCZ \cdot X_1\ket{111}\\
            &= -CCZ\ket{011}\\
            &= -\ket{011} \\
            &= -X_1\ket{111}
        \end{aligned}
        \end{equation}
    \item $x=0$ and $y=z=1$
        \begin{equation}
        \begin{aligned}
            CCZ \cdot X_1 \cdot CCZ \ket{011} &= CCZ \cdot X_1\ket{011}\\
            &= CCZ\ket{111}\\
            &= -\ket{111} \\
            &= -X_1\ket{011}
        \end{aligned}
        \end{equation}
\end{enumerate}
In all other cases, the action is just $X_1$. Because this negative sign only comes when both $y=z=1$, we see that the effect of conjugating $X_1 = X \otimes I \otimes I$ by $CCZ$ is to introduce a $CZ$ gate acting on the 2nd and 3rd qubit:
\begin{equation}
    CCZ(X \otimes \mathbb{I}\otimes \mathbb{I})CCZ = X \otimes CZ
\end{equation}
Since $CZ$ is not a Pauli gate, $CCZ$ cannot be a Clifford gate.

\subsection{Preparing and measuring the Y eigenstate}\label{subsec:prepare Y eigenstate}
\label{ap:y-preparation}

In order to implement the $S$ gate, we need to prepare $\ket{+_Y}$, the eigenstate of Pauli $Y$ with +1 eigenvalue. Alternately, we can implement $S^\dagger$ using $\ket{-_Y}$. The protocol we will outline here cannot absolutely distinguish between $\ket{\pm_Y}$, but we can tell if two states are the same eigenstate or not. This allows us to set a convention for $S$ and $S^\dagger$.

The preparation of the state $\ket{{\pm}_Y}=\frac{1+i}{2}\ket0 \mp \frac{1-i}{2}\ket1$ is adapted from \cite{mochon_anyons_2003}. We first create a Bell pair: this can be done by creating a $\ket{+}$ and entangling it with a $\ket{0}$ state via the $U_+$ gate. Once we have this Bell pair, we discard one half to obtain a maximally mixed state which can be expressed as  
\begin{equation}
    \frac{1}{2} I = \frac{1}{2}\ketbra{+_Y} +  \frac{1}{2}\ketbra{-_Y}
\end{equation}
where we use the $Y$ eigenbasis as our resolution to the identity. We have a classical mixture of $Y$ eigenstates; our goal will be to prepare a larger classically-correlated mixture of $Y$ eigenstates. Namely, the state
\begin{align}
    \frac{1}{2}\ket{+_Y}^{\otimes N}\bra{+_Y}^{\otimes N} +  \frac{1}{2}\ket{-_Y}^{\otimes N}\bra{-_Y}^{\otimes N}
\end{align}
so that we have a pool of ``$\ket{+_Y}$'' and ``$\ket{-_Y}$'' to use as resources. This starting mixed state will serve as the reference state for creating this classically correlated mixture, using the following circuit:
\begin{center}
    \begin{tikzpicture}
    \node[scale=1.25]{ 
        \begin{quantikz}
            \lstick{$\ket{+}$} & \ctrl{1} & \ctrl{1} & \rstick{$\ket{\pm_Y}$}\qw\\
            \lstick{$\ket{\pm_Y}$} & \gate{X} & \gate{Z} & \rstick{$\ket{\pm_Y}$}\qw
        \end{quantikz}
    };
    \end{tikzpicture}
\end{center}
This circuit copies the mixed state onto the input $\ket{+}$, creating two \emph{classically correlated} (not entangled!) mixed states. If we could measure in the $Y$ eigenbasis, measuring these two states would always yield the same outcome. We now have the means to create a pile of ``$\ket{+_Y}$'' states, if we set our convention such that the initial reference mixed state is $\ket{+_Y}$. 

To create a $\ket{-_Y}$ state, we begin with a fresh Bell pair, uncorrelated with our $\ket{+_Y}$ pile. Using only one half of this new Bell pair (which, as before, we can resolve into a mixture of $Y$ eigenstates), we compare it to a $\ket{+_Y}$ using the following circuit:
\begin{center}
    \begin{tikzpicture}
    \node[scale=1.25]{ 
        \begin{quantikz}
            \lstick{$\ket{\pm_Y}$} & \ctrl{1} & \ctrl{1} & \meter{$\ket{\pm}$} \\
            \lstick{$\ket{+_Y}$} & \gate{Z} & \gate {X} & \rstick{$\ket{+_Y}$}\qw
        \end{quantikz}
    };
    \end{tikzpicture}
\end{center}
We measure one half of the new Bell pair in the $X$ basis: if the result is $\ket-$, we have projected the other half of the fresh Bell pair into $\ket{-_Y}$. If the result is $\ket+$, we have created another $\ket{+_Y}$. We can proceed this way to create as many $\ket{-_Y}$ as we need. 

\bibliography{BiBTeX_File.bib}

\begin{thebibliography}{100}
\providecommand{\url}[1]{\texttt{#1}}
\providecommand{\urlprefix}{URL }
\expandafter\ifx\csname urlstyle\endcsname\relax
  \providecommand{\doi}[1]{doi:\discretionary{}{}{}#1}\else
  \providecommand{\doi}{doi:\discretionary{}{}{}\begingroup \urlstyle{rm}\Url}\fi
\providecommand{\eprint}[2][]{\url{#2}}

\bibitem{kitaev_fault-tolerant_2003}
A.~Y. Kitaev,
\newblock \emph{Fault-tolerant quantum computation by anyons},
\newblock Annals of Physics \textbf{303}(1), 2,
\newblock \doi{10.1016/S0003-4916(02)00018-0},
\newblock \eprint{quant-ph/9707021}.

\bibitem{freedman_topological_2002}
M.~H. Freedman, A.~Kitaev, M.~J. Larsen and Z.~Wang,
\newblock \emph{Topological quantum computation},
\newblock Bull. Amer. Math. Soc. \textbf{40}(1), 31,
\newblock \doi{10.1090/S0273-0979-02-00964-3}.

\bibitem{nayak_non-abelian_2008}
C.~Nayak, S.~H. Simon, A.~Stern, M.~Freedman and S.~Das~Sarma,
\newblock \emph{Non-{{Abelian}} anyons and topological quantum computation},
\newblock Rev. Mod. Phys. \textbf{80}(3), 1083,
\newblock \doi{10.1103/RevModPhys.80.1083}.

\bibitem{bravyi_quantum_1998}
S.~B. Bravyi and A.~Y. Kitaev,
\newblock \emph{Quantum codes on a lattice with boundary},
\newblock \doi{10.48550/arXiv.quant-ph/9811052},
\newblock \urlprefix\url{http://arxiv.org/abs/quant-ph/9811052},
\newblock \eprint{quant-ph/9811052}.

\bibitem{freedman_projective_2001}
M.~H. Freedman and D.~A. Meyer,
\newblock \emph{Projective {{Plane}} and {{Planar Quantum Codes}}} \textbf{1}(3), 325,
\newblock \doi{10.1007/s102080010013}.

\bibitem{fowler_surface_2012}
A.~G. Fowler, M.~Mariantoni, J.~M. Martinis and A.~N. Cleland,
\newblock \emph{Surface codes: {{Towards}} practical large-scale quantum computation} \textbf{86}(3), 032324,
\newblock \doi{10.1103/PhysRevA.86.032324}.

\bibitem{dennis_topological_2002}
E.~Dennis, A.~Kitaev, A.~Landahl and J.~Preskill,
\newblock \emph{Topological quantum memory} \textbf{43}(9), 4452,
\newblock \doi{10.1063/1.1499754}.

\bibitem{leinaas_theory_1977}
J.~M. Leinaas and J.~Myrheim,
\newblock \emph{On the theory of identical particles},
\newblock Nuovo Cim B \textbf{37}(1), 1,
\newblock \doi{10.1007/BF02727953}.

\bibitem{wilczek_quantum_1982}
F.~Wilczek,
\newblock \emph{Quantum {{Mechanics}} of {{Fractional-Spin Particles}}},
\newblock Phys. Rev. Lett. \textbf{49}(14), 957,
\newblock \doi{10.1103/PhysRevLett.49.957}.

\bibitem{Halperin84}
B.~I. Halperin,
\newblock \emph{Statistics of quasiparticles and the hierarchy of fractional quantized hall states},
\newblock Phys. Rev. Lett. \textbf{52}, 1583 (1984),
\newblock \doi{10.1103/PhysRevLett.52.1583}.

\bibitem{wu_general_1984}
Y.-S. Wu,
\newblock \emph{General theory for quantum statistics in two dimensions},
\newblock Phys. Rev. Lett. \textbf{52}, 2103 (1984),
\newblock \doi{10.1103/PhysRevLett.52.2103}.

\bibitem{Arovas84}
D.~Arovas, J.~R. Schrieffer and F.~Wilczek,
\newblock \emph{Fractional statistics and the quantum hall effect},
\newblock Phys. Rev. Lett. \textbf{53}, 722 (1984),
\newblock \doi{10.1103/PhysRevLett.53.722}.

\bibitem{wen_ground-state_1990}
X.~G. Wen and Q.~Niu,
\newblock \emph{Ground-state degeneracy of the fractional quantum {{Hall}} states in the presence of a random potential and on high-genus {{Riemann}} surfaces} \textbf{41}(13), 9377,
\newblock \doi{10.1103/PhysRevB.41.9377}.

\bibitem{einarsson_fractional_1990}
T.~Einarsson,
\newblock \emph{Fractional statistics on a torus} \textbf{64}(17), 1995,
\newblock \doi{10.1103/PhysRevLett.64.1995}.

\bibitem{bravyi_universal_2005}
S.~Bravyi and A.~Kitaev,
\newblock \emph{Universal quantum computation with ideal {{Clifford}} gates and noisy ancillas},
\newblock Phys. Rev. A \textbf{71}(2), 022316,
\newblock \doi{10.1103/PhysRevA.71.022316}.

\bibitem{horsman_surface_2012}
D.~Horsman, A.~G. Fowler, S.~Devitt and R.~V. Meter,
\newblock \emph{Surface code quantum computing by lattice surgery},
\newblock New J. Phys. \textbf{14}(12), 123011,
\newblock \doi{10.1088/1367-2630/14/12/123011}.

\bibitem{chamberland_universal_2022}
C.~Chamberland and E.~T. Campbell,
\newblock \emph{Universal {{Quantum Computing}} with {{Twist-Free}} and {{Temporally Encoded Lattice Surgery}}},
\newblock PRX Quantum \textbf{3}(1), 010331,
\newblock \doi{10.1103/PRXQuantum.3.010331}.

\bibitem{litinski_lattice_2018}
D.~Litinski and F.~v. Oppen,
\newblock \emph{Lattice {{Surgery}} with a {{Twist}}: Simplifying {{Clifford }}{{Gates}} of {{Surface Codes}}},
\newblock Quantum \textbf{2}, 62 (2018),
\newblock \doi{10.22331/q-2018-05-04-62}.

\bibitem{bombin_topological_2010}
H.~Bombin,
\newblock \emph{Topological {{Order}} with a {{Twist}}: {{Ising Anyons}} from an {{Abelian Model}}},
\newblock Phys. Rev. Lett. \textbf{105}(3), 030403,
\newblock \doi{10.1103/PhysRevLett.105.030403}.

\bibitem{barkeshli_twist_2013}
M.~Barkeshli, C.-M. Jian and X.-L. Qi,
\newblock \emph{Twist defects and projective non-{{Abelian}} braiding statistics},
\newblock Phys. Rev. B \textbf{87}(4), 045130,
\newblock \doi{10.1103/PhysRevB.87.045130}.

\bibitem{barkeshli2015physical}
M.~Barkeshli and J.~D. Sau,
\newblock \emph{Physical architecture for a universal topological quantum computer based on a network of majorana nanowires},
\newblock arXiv preprint arXiv:1509.07135  (2015).

\bibitem{benhemou_non-abelian_2022}
A.~Benhemou, J.~K. Pachos and D.~E. Browne,
\newblock \emph{Non-{{Abelian}} statistics with mixed-boundary punctures on the toric code},
\newblock Phys. Rev. A \textbf{105}(4), 042417,
\newblock \doi{10.1103/PhysRevA.105.042417},
\newblock \eprint{2103.08381}.

\bibitem{yoder_surface_2017}
T.~J. Yoder and I.~H. Kim,
\newblock \emph{The surface code with a twist},
\newblock {Quantum} \textbf{1}, 2 (2017),
\newblock \doi{10.22331/q-2017-04-25-2}.

\bibitem{kesselring_boundaries_2018}
M.~S. Kesselring, F.~Pastawski, J.~Eisert and B.~J. Brown,
\newblock \emph{The boundaries and twist defects of the color code and their applications to topological quantum computation},
\newblock Quantum \textbf{2}, 101,
\newblock \doi{10.22331/q-2018-10-19-101}.

\bibitem{burton_genons_2024}
S.~Burton, E.~Durso-Sabina and N.~C. Brown,
\newblock \emph{Genons, {{Double Covers}} and {{Fault-tolerant Clifford Gates}}},
\newblock \doi{10.48550/arXiv.2406.09951},
\newblock \eprint{2406.09951}.

\bibitem{wen1989}
X.~G. Wen,
\newblock \emph{Vacuum degeneracy of chiral spin states in compactified space},
\newblock Phys. Rev. B \textbf{40}, 7387 (1989),
\newblock \doi{10.1103/PhysRevB.40.7387}.

\bibitem{wen_topological_1990}
X.~G. Wen,
\newblock \emph{Topological orders in rigid states},
\newblock Int. J. Mod. Phys. B \textbf{04}(02), 239,
\newblock \doi{10.1142/S0217979290000139}.

\bibitem{wen_book_2004}
X.-G. Wen,
\newblock \emph{Quantum Field Theory of Many-Body Systems: From the Origin of Sound to an Origin of Light and Electrons},
\newblock Oxford Graduate Texts. Oxford University Press,
\newblock ISBN 978-0-19-171301-9,
\newblock \doi{10.1093/acprof:oso/9780199227259.001.0001}.

\bibitem{sachdev_book_2023}
S.~Sachdev,
\newblock \emph{Quantum {{Phases}} of {{Matter}}},
\newblock Cambridge University Press,
\newblock ISBN 978-1-009-21269-4,
\newblock \doi{10.1017/9781009212717}.

\bibitem{simon_topological_2023}
S.~H. Simon,
\newblock \emph{Topological {{Quantum}}},
\newblock Oxford University Press,
\newblock ISBN 978-0-19-888672-3.

\bibitem{Wegner71}
F.~J. Wegner,
\newblock \emph{Duality in generalized ising models and phase transitions without local order parameters},
\newblock Journal of Mathematical Physics \textbf{12}(10), 2259 (1971),
\newblock \doi{10.1063/1.1665530},
\newblock \eprint{https://pubs.aip.org/aip/jmp/article-pdf/12/10/2259/19106483/2259\_1\_online.pdf}.

\bibitem{Kogut79}
J.~B. Kogut,
\newblock \emph{An introduction to lattice gauge theory and spin systems},
\newblock Rev. Mod. Phys. \textbf{51}, 659 (1979),
\newblock \doi{10.1103/RevModPhys.51.659}.

\bibitem{Fradkin79}
E.~Fradkin and S.~H. Shenker,
\newblock \emph{Phase diagrams of lattice gauge theories with higgs fields},
\newblock Phys. Rev. D \textbf{19}, 3682 (1979),
\newblock \doi{10.1103/PhysRevD.19.3682}.

\bibitem{AB}
Y.~Aharonov and D.~Bohm,
\newblock \emph{Significance of electromagnetic potentials in the quantum theory},
\newblock Phys. Rev. \textbf{115}, 485 (1959),
\newblock \doi{10.1103/PhysRev.115.485}.

\bibitem{propitius_discrete_1996}
M.~d.~W. Propitius and F.~A. Bais,
\newblock \emph{Discrete gauge theories},
\newblock \eprint{hep-th/9511201}.

\bibitem{cui_universal_2015}
S.~X. Cui, S.-M. Hong and Z.~Wang,
\newblock \emph{Universal quantum computation with weakly integral anyons},
\newblock Quantum Inf Process \textbf{14}(8), 2687,
\newblock \doi{10.1007/s11128-015-1016-y},
\newblock \eprint{1401.7096}.

\bibitem{mochon_anyons_2003}
C.~Mochon,
\newblock \emph{Anyons from nonsolvable finite groups are sufficient for universal quantum computation},
\newblock Phys. Rev. A \textbf{67}(2), 022315,
\newblock \doi{10.1103/PhysRevA.67.022315}.

\bibitem{mochon_anyon_2004}
C.~Mochon,
\newblock \emph{Anyon computers with smaller groups},
\newblock Phys. Rev. A \textbf{69}(3), 032306,
\newblock \doi{10.1103/PhysRevA.69.032306}.

\bibitem{Goldin81}
G.~A. Goldin, R.~Menikoff and D.~H. Sharp,
\newblock \emph{Representations of a local current algebra in nonsimply connected space and the aharonov–bohm effect},
\newblock Journal of Mathematical Physics \textbf{22}(8), 1664 (1981),
\newblock \doi{10.1063/1.525110},
\newblock \eprint{https://pubs.aip.org/aip/jmp/article-pdf/22/8/1664/19078813/1664\_1\_online.pdf}.

\bibitem{frohlich_statistics_1988}
J.~Fröhlich,
\newblock \emph{Statistics of {{Fields}}, the {{Yang-Baxter Equation}}, and the {{Theory}} of {{Knots}} and {{Links}}},
\newblock In G.~'t~Hooft, A.~Jaffe, G.~Mack, P.~K. Mitter and R.~Stora, eds., \emph{Nonperturbative {{Quantum Field Theory}}}, pp. 71--100. Springer US,
\newblock ISBN 978-1-4613-0729-7,
\newblock \doi{10.1007/978-1-4613-0729-7_4}.

\bibitem{Moore1988}
G.~W. Moore and N.~Seiberg,
\newblock \emph{{Classical and Quantum Conformal Field Theory}},
\newblock Commun. Math. Phys. \textbf{123}, 177 (1989),
\newblock \doi{10.1007/BF01238857}.

\bibitem{Fredenhagen1988}
K.~Fredenhagen, K.-H. Rehren and B.~Schroer,
\newblock \emph{{Superselection Sectors with Braid Group Statistics and Exchange Algebras. 1. General Theory}},
\newblock Commun. Math. Phys. \textbf{125}, 201 (1989),
\newblock \doi{10.1007/BF01217906}.

\bibitem{bais_flux_1980}
F.~A. Bais,
\newblock \emph{Flux metamorphosis},
\newblock Nuclear Physics B \textbf{170}(1), 32,
\newblock \doi{10.1016/0550-3213(80)90474-5}.

\bibitem{Moore1991}
G.~W. Moore and N.~Read,
\newblock \emph{{Nonabelions in the fractional quantum Hall effect}},
\newblock Nucl. Phys. B \textbf{360}, 362 (1991),
\newblock \doi{10.1016/0550-3213(91)90407-O}.

\bibitem{alexander_bais_quantum_1992}
F.~Bais, P.~van Driel and M.~de~Wild~Propitius,
\newblock \emph{Quantum symmetries in discrete gauge theories},
\newblock Physics Letters B \textbf{280}(1), 63,
\newblock \doi{10.1016/0370-2693(92)90773-W}.

\bibitem{Gabbiani1992}
F.~Gabbiani and J.~Frohlich,
\newblock \emph{{Operator algebras and conformal field theory}},
\newblock Commun. Math. Phys. \textbf{155}, 569 (1993),
\newblock \doi{10.1007/BF02096729}.

\bibitem{alexander_bais_anyons_1993}
F.~Bais, P.~van Driel and M.~de~Wild~Propitius,
\newblock \emph{Anyons in discrete gauge theories with {{Chern-Simons}} terms},
\newblock Nuclear Physics B \textbf{393}(3), 547,
\newblock \doi{10.1016/0550-3213(93)90073-X}.

\bibitem{ogburn_preskill_topological_1999}
R.~Walter~Ogburn and J.~Preskill,
\newblock \emph{Topological {{Quantum Computation}}},
\newblock In C.~P. Williams, ed., \emph{Quantum {{Computing}} and {{Quantum Communications}}}, pp. 341--356. Springer,
\newblock ISBN 978-3-540-49208-5,
\newblock \doi{10.1007/3-540-49208-9_31}.

\bibitem{freedman_modular_2002}
M.~H. Freedman, M.~Larsen and Z.~Wang,
\newblock \emph{A {{Modular Functor Which}} is {{Universal}} for {{Quantum Computation}}},
\newblock Commun. Math. Phys. \textbf{227}(3), 605,
\newblock \doi{10.1007/s002200200645}.

\bibitem{freedman_two-eigenvalue_2002}
M.~H. Freedman, M.~J. Larsen and Z.~Wang,
\newblock \emph{The {{Two-Eigenvalue Problem}} and {{Density}} of {{Jones Representation}} of {{Braid Groups}}},
\newblock Commun. Math. Phys. \textbf{228}(1), 177,
\newblock \doi{10.1007/s002200200636}.

\bibitem{trebst_short_2008}
S.~Trebst, M.~Troyer, Z.~Wang and A.~W.~W. Ludwig,
\newblock \emph{A {{Short Introduction}} to {{Fibonacci Anyon Models}}},
\newblock Progress of Theoretical Physics Supplement \textbf{176}, 384,
\newblock \doi{10.1143/PTPS.176.384}.

\bibitem{lu_measurement_2022}
T.-C. Lu, L.~A. Lessa, I.~H. Kim and T.~H. Hsieh,
\newblock \emph{Measurement as a {{Shortcut}} to {{Long-Range Entangled Quantum Matter}}} \textbf{3}(4), 040337,
\newblock \doi{10.1103/PRXQuantum.3.040337}.

\bibitem{tantivasadakarn_hierarchy_2023}
N.~Tantivasadakarn, A.~Vishwanath and R.~Verresen,
\newblock \emph{Hierarchy of {{Topological Order From Finite-Depth Unitaries}}, {{Measurement}}, and {{Feedforward}}},
\newblock PRX Quantum \textbf{4}(2), 020339,
\newblock \doi{10.1103/PRXQuantum.4.020339}.

\bibitem{verresen_efficiently_2022}
R.~Verresen, N.~Tantivasadakarn and A.~Vishwanath,
\newblock \emph{Efficiently preparing {{Schr}}{\"o}dinger's cat, fractons and non-{{Abelian}} topological order in quantum devices},
\newblock \eprint{2112.03061}.

\bibitem{tanti_long-range_2024}
N.~Tantivasadakarn, R.~Thorngren, A.~Vishwanath and R.~Verresen,
\newblock \emph{Long-range entanglement from measuring symmetry-protected topological phases} \textbf{14}(2), 021040,
\newblock \doi{10.1103/PhysRevX.14.021040},
\newblock \eprint{2112.01519}.

\bibitem{bravyi_adaptive_2022}
S.~Bravyi, I.~Kim, A.~Kliesch and R.~Koenig,
\newblock \emph{Adaptive constant-depth circuits for manipulating non-abelian anyons},
\newblock \doi{10.48550/arXiv.2205.01933} (2022), \eprint{2205.01933}.

\bibitem{kitaev_anyons_nodate}
A.~Kitaev,
\newblock \emph{Anyons based on a finite group},
\newblock \urlprefix\url{http://theory.caltech.edu/~preskill/ph219/prob7_07-kitaev.pdf}.

\bibitem{raussendorf2001}
R.~Raussendorf and H.~J. Briegel,
\newblock \emph{A one-way quantum computer},
\newblock Phys. Rev. Lett. \textbf{86}, 5188 (2001),
\newblock \doi{10.1103/PhysRevLett.86.5188}.

\bibitem{Raussendorf_2005}
R.~Raussendorf, S.~Bravyi and J.~Harrington,
\newblock \emph{Long-range quantum entanglement in noisy cluster states},
\newblock Physical Review A \textbf{71}(6) (2005),
\newblock \doi{10.1103/physreva.71.062313}.

\bibitem{ashkenazi_duality_2022}
S.~Ashkenazi,
\newblock \emph{Duality as a feasible physical transformation for quantum simulation} \textbf{105}(2),
\newblock \doi{10.1103/PhysRevA.105.022431}.

\bibitem{piroli_approximating_2024}
L.~Piroli, G.~Styliaris and J.~I. Cirac,
\newblock \emph{Approximating {{Many-Body Quantum States}} with {{Quantum Circuits}} and {{Measurements}}} \textbf{133}(23), 230401,
\newblock \doi{10.1103/PhysRevLett.133.230401}.

\bibitem{piroli2021}
L.~Piroli, G.~Styliaris and J.~I. Cirac,
\newblock \emph{Quantum circuits assisted by local operations and classical communication: Transformations and phases of matter},
\newblock Phys. Rev. Lett. \textbf{127}, 220503 (2021),
\newblock \doi{10.1103/PhysRevLett.127.220503}.

\bibitem{zhu2023}
G.-Y. Zhu, N.~Tantivasadakarn, A.~Vishwanath, S.~Trebst and R.~Verresen,
\newblock \emph{Nishimori's cat: Stable long-range entanglement from finite-depth unitaries and weak measurements},
\newblock Phys. Rev. Lett. \textbf{131}, 200201 (2023),
\newblock \doi{10.1103/PhysRevLett.131.200201}.

\bibitem{lee2022decodingmeasurementpreparedquantumphases}
J.~Y. Lee, W.~Ji, Z.~Bi and M.~P.~A. Fisher,
\newblock \emph{Decoding measurement-prepared quantum phases and transitions: from ising model to gauge theory, and beyond} (2022), \eprint{2208.11699}.

\bibitem{li2023measuringtopologicalfieldtheories}
Y.~Li, M.~Litvinov and T.-C. Wei,
\newblock \emph{Measuring topological field theories: Lattice models and field-theoretic description} (2023), \eprint{2310.17740}.

\bibitem{gunn2023phasesmatrixproductstates}
D.~Gunn, G.~Styliaris, T.~Kraft and B.~Kraus,
\newblock \emph{Phases of matrix product states with symmetric quantum circuits and symmetric measurements with feedforward} (2023), \eprint{2312.13838}.

\bibitem{li2023}
Y.~Li, H.~Sukeno, A.~P. Mana, H.~P. Nautrup and T.-C. Wei,
\newblock \emph{Symmetry-enriched topological order from partially gauging symmetry-protected topologically ordered states assisted by measurements},
\newblock Phys. Rev. B \textbf{108}, 115144 (2023),
\newblock \doi{10.1103/PhysRevB.108.115144}.

\bibitem{liu_StringNet_2022}
Y.-J. Liu, K.~Shtengel, A.~Smith and F.~Pollmann,
\newblock \emph{Methods for {{Simulating String-Net States}} and {{Anyons}} on a {{Digital Quantum Computer}}} \textbf{3}(4), 040315,
\newblock \doi{10.1103/PRXQuantum.3.040315}.

\bibitem{Ren24}
Y.~Ren, N.~Tantivasadakarn and D.~J. Williamson,
\newblock \emph{Efficient preparation of solvable anyons with adaptive quantum circuits},
\newblock arXiv preprint arXiv:2411.04985  (2024).

\bibitem{preskill_lecture_2004}
J.~Preskill,
\newblock \emph{Lecture {{Notes}} for {{Physics}} 219: {{Quantum Computation}}},
\newblock \urlprefix\url{http://theory.caltech.edu/~preskill/ph219/topological.pdf}.

\bibitem{cui_topological_2018}
S.~X. Cui,
\newblock \emph{Topological {{Quantum Computation}}},
\newblock \urlprefix\url{https://www.math.purdue.edu/~cui177/Lecture_Combined.pdf}.

\bibitem{overbosch_inequivalent_2001}
B.~J. Overbosch and F.~A. Bais,
\newblock \emph{Inequivalent classes of interference experiments with non-{{Abelian}} anyons},
\newblock Phys. Rev. A \textbf{64}(6), 062107,
\newblock \doi{10.1103/PhysRevA.64.062107}.

\bibitem{bonderson_interferometry_2008}
P.~Bonderson, K.~Shtengel and J.~K. Slingerland,
\newblock \emph{Interferometry of non-{{Abelian}} anyons},
\newblock Annals of Physics \textbf{323}(11), 2709,
\newblock \doi{10.1016/j.aop.2008.01.012}.

\bibitem{bombin_family_2008}
H.~Bombin and M.~A. Martin-Delgado,
\newblock \emph{Family of non-{{Abelian Kitaev}} models on a lattice: {{Topological}} condensation and confinement},
\newblock Phys. Rev. B \textbf{78}(11), 115421,
\newblock \doi{10.1103/PhysRevB.78.115421}.

\bibitem{iqbal_non-abelian_2024}
M.~Iqbal, N.~Tantivasadakarn, R.~Verresen, S.~L. Campbell, J.~M. Dreiling, C.~Figgatt, J.~P. Gaebler, J.~Johansen, M.~Mills, S.~A. Moses, J.~M. Pino, A.~Ransford \emph{et~al.},
\newblock \emph{Non-{{Abelian}} topological order and anyons on a trapped-ion processor} \textbf{626}(7999), 505,
\newblock \doi{10.1038/s41586-023-06934-4}.

\bibitem{Ritz24}
A.~Ritz-Zwilling, J.-N. Fuchs, S.~H. Simon and J.~Vidal,
\newblock \emph{Topological and nontopological degeneracies in generalized string-net models},
\newblock Phys. Rev. B \textbf{109}, 045130 (2024),
\newblock \doi{10.1103/PhysRevB.109.045130}.

\bibitem{gottesman_fault-tolerant_1999}
D.~Gottesman,
\newblock \emph{Fault-{{Tolerant Quantum Computation}} with {{Higher-Dimensional Systems}}},
\newblock In C.~P. Williams, ed., \emph{Quantum {{Computing}} and {{Quantum Communications}}}, pp. 302--313. Springer,
\newblock ISBN 978-3-540-49208-5,
\newblock \doi{10.1007/3-540-49208-9_27}.

\bibitem{yan_ribbon_2022}
B.~Yan, P.~Chen and S.~X. Cui,
\newblock \emph{Ribbon operators in the generalized {{Kitaev}} quantum double model based on {{Hopf}} algebras},
\newblock J. Phys. A: Math. Theor. \textbf{55}(18), 185201,
\newblock \doi{10.1088/1751-8121/ac552c}.

\bibitem{lyons_protocols_2024}
A.~Lyons, C.~F.~B. Lo, N.~Tantivasadakarn, A.~Vishwanath and R.~Verresen,
\newblock \emph{Protocols for {{Creating Anyons}} and {{Defects}} via {{Gauging}}} (2024), \eprint{http://arxiv.org/abs/2411.04181}.

\bibitem{Xu_2024}
S.~Xu, Z.-Z. Sun, K.~Wang, H.~Li, Z.~Zhu, H.~Dong, J.~Deng, X.~Zhang, J.~Chen, Y.~Wu \emph{et~al.},
\newblock \emph{Non-abelian braiding of fibonacci anyons with a superconducting processor},
\newblock Nature Physics  (2024),
\newblock \doi{10.1038/s41567-024-02529-6}.

\bibitem{minev2024realizingstringnetcondensationfibonacci}
Z.~K. Minev, K.~Najafi, S.~Majumder, J.~Wang, A.~Stern, E.-A. Kim, C.-M. Jian and G.~Zhu,
\newblock \emph{Realizing string-net condensation: Fibonacci anyon braiding for universal gates and sampling chromatic polynomials} (2024), \eprint{https://arxiv.org/abs/2406.12820}.

\bibitem{andersen_non-abelian_2023}
T.~I. Andersen, Y.~D. Lensky, K.~Kechedzhi, I.~K. Drozdov, A.~Bengtsson, S.~Hong, A.~Morvan, X.~Mi, A.~Opremcak, R.~Acharya, R.~Allen, M.~Ansmann \emph{et~al.},
\newblock \emph{Non-{{Abelian}} braiding of graph vertices in a superconducting processor},
\newblock Nature \textbf{618}(7964), 264,
\newblock \doi{10.1038/s41586-023-05954-4}.

\bibitem{Xu22}
S.~Xu \emph{et~al.},
\newblock \emph{{Digital Simulation of Projective Non-Abelian Anyons with 68 Superconducting Qubits}},
\newblock Chin. Phys. Lett. \textbf{40}(6), 060301 (2023),
\newblock \doi{10.1088/0256-307X/40/6/060301},
\newblock \eprint{2211.09802}.

\bibitem{iqbal2024qutrittoriccodeparafermions}
M.~Iqbal, A.~Lyons, C.~F.~B. Lo, N.~Tantivasadakarn, J.~Dreiling, C.~Foltz, T.~M. Gatterman, D.~Gresh, N.~Hewitt, C.~A. Holliman, J.~Johansen, B.~Neyenhuis \emph{et~al.},
\newblock \emph{Qutrit toric code and parafermions in trapped ions} (2024), \eprint{2411.04185}.

\bibitem{terhal_quantum_2015}
B.~M. Terhal,
\newblock \emph{Quantum error correction for quantum memories} \textbf{87}(2), 307,
\newblock \doi{10.1103/RevModPhys.87.307}.

\bibitem{Wootton_2014}
J.~R. Wootton, J.~Burri, S.~Iblisdir and D.~Loss,
\newblock \emph{Error correction for non-abelian topological quantum computation},
\newblock Phys. Rev. X \textbf{4}, 011051 (2014),
\newblock \doi{10.1103/PhysRevX.4.011051}.

\bibitem{brell_thermalization_2014}
C.~G. Brell, S.~Burton, G.~Dauphinais, S.~T. Flammia and D.~Poulin,
\newblock \emph{Thermalization, {{Error Correction}}, and {{Memory Lifetime}} for {{Ising Anyon Systems}}} \textbf{4}(3), 031058,
\newblock \doi{10.1103/PhysRevX.4.031058}.

\bibitem{wootton_active_2016}
J.~R. Wootton and A.~Hutter,
\newblock \emph{Active error correction for {{Abelian}} and non-{{Abelian}} anyons} \textbf{93}(2), 022318,
\newblock \doi{10.1103/PhysRevA.93.022318}.

\bibitem{knapp_nature_2016}
C.~Knapp, M.~Zaletel, D.~E. Liu, M.~Cheng, P.~Bonderson and C.~Nayak,
\newblock \emph{The {{Nature}} and {{Correction}} of {{Diabatic Errors}} in {{Anyon Braiding}}} \textbf{6}(4), 041003,
\newblock \doi{10.1103/PhysRevX.6.041003}.

\bibitem{zhu_universal_2020}
G.~Zhu, A.~Lavasani and M.~Barkeshli,
\newblock \emph{Universal {{Logical Gates}} on {{Topologically Encoded Qubits}} via {{Constant-Depth Unitary Circuits}}} \textbf{125}(5), 050502,
\newblock \doi{10.1103/PhysRevLett.125.050502}.

\bibitem{zhu_instantaneous_2020}
G.~Zhu, A.~Lavasani and M.~Barkeshli,
\newblock \emph{Instantaneous braids and {{Dehn}} twists in topologically ordered states} \textbf{102}(7), 075105,
\newblock \doi{10.1103/PhysRevB.102.075105}.

\bibitem{schotte_2022}
A.~Schotte, G.~Zhu, L.~Burgelman and F.~Verstraete,
\newblock \emph{Quantum error correction thresholds for the universal fibonacci turaev-viro code},
\newblock Phys. Rev. X \textbf{12}, 021012 (2022),
\newblock \doi{10.1103/PhysRevX.12.021012}.

\bibitem{schotte_fault-tolerant_2022}
A.~Schotte, L.~Burgelman and G.~Zhu,
\newblock \emph{Fault-tolerant error correction for a universal non-{{Abelian}} topological quantum computer at finite temperature},
\newblock \doi{10.48550/arXiv.2301.00054},
\newblock \eprint{2301.00054}.

\bibitem{sala_decoherence_2024}
P.~Sala, J.~Alicea and R.~Verresen,
\newblock \emph{Decoherence and wavefunction deformation of \${{D}}\_4\$ non-{{Abelian}} topological order},
\newblock \doi{10.48550/arXiv.2409.12948},
\newblock \eprint{2409.12948}.

\bibitem{sala_stability_2024}
P.~Sala and R.~Verresen,
\newblock \emph{Stability and {{Loop Models}} from {{Decohering Non-Abelian Topological Order}}},
\newblock \doi{10.48550/arXiv.2409.12230},
\newblock \eprint{2409.12230}.

\bibitem{hsin_non-abelian_2024}
P.-S. Hsin, R.~Kobayashi and G.~Zhu,
\newblock \emph{Non-{{Abelian Self-Correcting Quantum Memory}}},
\newblock \doi{10.48550/arXiv.2405.11719},
\newblock \eprint{2405.11719}.

\bibitem{Hutter_2015}
A.~Hutter, D.~Loss and J.~R. Wootton,
\newblock \emph{Improved hdrg decoders for qudit and non-abelian quantum error correction},
\newblock New Journal of Physics \textbf{17}(3), 035017 (2015),
\newblock \doi{10.1088/1367-2630/17/3/035017}.

\bibitem{Dauphinais_2017}
G.~Dauphinais and D.~Poulin,
\newblock \emph{Fault-tolerant quantum error correction for non-abelian anyons},
\newblock Communications in Mathematical Physics \textbf{355}(2), 519–560 (2017),
\newblock \doi{10.1007/s00220-017-2923-9}.

\bibitem{sala2024decoherencewavefunctiondeformationd4}
P.~Sala, J.~Alicea and R.~Verresen,
\newblock \emph{Decoherence and wavefunction deformation of $d_4$ non-abelian topological order} (2024), \eprint{2409.12948}.

\bibitem{sala2024stabilityloopmodelsdecohering}
P.~Sala and R.~Verresen,
\newblock \emph{Stability and loop models from decohering non-abelian topological order} (2024), \eprint{2409.12230}.

\bibitem{Chirame24}
S.~Chirame, A.~Prem, S.~Gopalakrishnan and F.~J. Burnell,
\newblock \emph{{Stabilizing Non-Abelian Topological Order against Heralded Noise via Local Lindbladian Dynamics}}  (2024),
\newblock \eprint{2410.21402}.

\bibitem{turaevbook}
V.~Turaev and A.~Virelizier,
\newblock \emph{Monoidal {{Categories and}} {{Topological}} {{Field}} {{Theory}}},
\newblock Birkhäuser Cham,
\newblock ISBN 978-3-319-49833-1 (2017).

\bibitem{chen_s3_universal_2024}
L.~Chen, Y.~Ren, R.~Fan and A.~Jaffe,
\newblock \emph{{A Universal Circuit Set Using the $S_3$ Quantum Double}}  (2024),
\newblock \eprint{2411.09697}.

\end{thebibliography}

\end{appendix}

\end{document}